\documentclass{article}
\usepackage{graphicx} 
\usepackage{geometry}
\usepackage[utf8]{inputenc}
\usepackage[english]{babel}
\usepackage{authblk}
\usepackage{multirow}
\usepackage{adjustbox}
\usepackage{ragged2e}
\usepackage{float}
\usepackage{enumitem}
\usepackage{titlesec}
\usepackage{stackrel}
\usepackage{physics}
\usepackage{dsfont} 
\usepackage{booktabs}
\usepackage{letltxmacro}
\usepackage{relsize}
\usepackage[font=footnotesize]{caption}
\usepackage{algorithm,algpseudocode}
\usepackage{blindtext}
\usepackage{mathtools}
\mathtoolsset{showonlyrefs}
\usepackage{etoolbox}
\usepackage{booktabs}
\usepackage{lmodern}
\usepackage{tikz}
\usepackage{scalerel}
\usepackage{hyperref}
\hypersetup{
    colorlinks=true,
    linkcolor=blue,
    citecolor=magenta,      
    urlcolor=cyan,
    pdfpagemode=FullScreen,
    }
\usepackage{natbib}
\usepackage{amsmath, amsfonts, amssymb, amsthm}

\newtheorem{theorem}{Theorem}
\newtheorem{proposition}{Proposition}
\newtheorem{lemma}{Lemma}
\newtheorem{definition}{Definition}

\newtheorem{assumption}{Assumption}

\newtheorem{remark}{Remark}
\newtheorem{cor}{Corollary}
\usepackage{wrapfig}
\usepackage{subcaption}
\usepackage{subfloat}
\graphicspath{ {images/} }

\usepackage{arydshln}
\usepackage{tikz}
\usetikzlibrary{patterns,decorations.pathreplacing}
\definecolor{controlblue}{RGB}{194,214,236}
\definecolor{treatedred}{RGB}{234,132,121}
\definecolor{questionred}{RGB}{234,51,35}

\DeclareRobustCommand{\Rb}{\mathbb{R}}
\DeclareRobustCommand{\bD}{\boldsymbol{D}}
\DeclareRobustCommand{\bZ}{\boldsymbol{Z}}
\DeclareRobustCommand{\bA}{\boldsymbol{A}}
\DeclareRobustCommand{\bH}{\boldsymbol{H}}

\DeclareRobustCommand{\bU}{\boldsymbol{U}}
\DeclareRobustCommand{\bV}{\boldsymbol{V}}
\DeclareRobustCommand{\bS}{\boldsymbol{S}}
\DeclareRobustCommand{\bM}{\boldsymbol{M}}
\DeclareRobustCommand{\bQ}{\boldsymbol{Q}}
\DeclareRobustCommand{\bP}{\boldsymbol{P}}
\DeclareRobustCommand{\bX}{\boldsymbol{X}}

\DeclareRobustCommand{\bB}{\boldsymbol{B}}
\DeclareRobustCommand{\bW}{\boldsymbol{W}}

\newcommand{\bbW}{\bar{\bW}}
\newcommand{\bbY}{\bar{\bY}}
\newcommand{\bbZ}{\bar{\bZ}}
\newcommand{\bbz}{\bar{\bz}}

\newcommand{\bI}{\boldsymbol{I}}

\newcommand{\bY}{\boldsymbol{Y}}

\DeclareRobustCommand{\bhA}{\widehat{\boldsymbol{A}}}

\DeclareRobustCommand{\hbP}{\widehat{\boldsymbol{P}}}

\DeclareRobustCommand{\Ac}{\mathcal{A}}

\DeclareRobustCommand{\Hc}{\mathcal{H}}

\DeclareRobustCommand{\Ic}{\mathcal{I}}
\DeclareRobustCommand{\Ec}{\mathcal{E}}

\DeclareRobustCommand{\Oc}{\mathcal{O}}
\newcommand{\Gc}{\mathcal{G}}

\newcommand{\Fc}{\mathcal{F}}
\newcommand{\Vc}{\mathcal{V}}
\newcommand{\hVc}{\widehat{\Vc}}
\DeclareRobustCommand{\Sc}{\mathcal{S}}
\DeclareRobustCommand{\Ac}{\mathcal{A}}

\DeclareRobustCommand{\Rb}{\mathbb R}
\DeclareRobustCommand{\Nb}{\mathbb N}
\DeclareRobustCommand{\Pb}{\mathbb P}
\DeclareRobustCommand{\Zb}{\mathbb Z}

\DeclareMathOperator{\Var}{Var}

\DeclareMathOperator*{\argmin}{\arg\!\min}

\newcommand{\hV}{\widehat{V}}
\newcommand{\bz}{\boldsymbol{z}} 
\newcommand{\bDelta}{\boldsymbol{\Delta}} 
 
\newcommand{\bXi}{\boldsymbol{\Xi}}
\newcommand{\bomega}{\boldsymbol{\omega}}

\newcommand{\bphi}{\boldsymbol{\phi}}

\newcommand{\hsigma}{\widehat{\sigma}}

\newcommand{\bq}{\boldsymbol{q}}
\newcommand{\hbq}{\widehat{\bq}} 

\newcommand{\distas}[1]{\mathbin{\overset{#1}{\kern\z@\sim}}}%
\newsavebox{\mybox}\newsavebox{\mysim}
\newcommand{\distras}[1]{%
  \savebox{\mybox}{\hbox{\kern3pt$\scriptstyle#1$\kern3pt}}%
  \savebox{\mysim}{\hbox{$\sim$}}%
  \mathbin{\overset{#1}{\kern\z@\resizebox{\wd\mybox}{\ht\mysim}{$\sim$}}}%
}

\newcommand{\htheta}{\widehat{\theta}}

\newcommand{\boldeta}{\boldsymbol{\eta}}

\newcommand{\Ex}{\mathbb{E}}

\newcommand\independent{\protect\mathpalette{\protect\independenT}{\perp}}
\def\independenT#1#2{\mathrel{\rlap{$#1#2$}\mkern2mu{#1#2}}}

\newcommand{\htau}{\widehat{\tau}}

\newcommand{\hs}{\widehat{s}}
\newcommand{\hS}{\widehat{S}}
\newcommand{\bv}{\boldsymbol{v}}

\newcommand{\bx}{\boldsymbol{x}}
\newcommand{\bu}{\boldsymbol{u}}

\newcommand{\bbeta}{\boldsymbol{\beta}} 
\newcommand{\balpha}{\boldsymbol{\alpha}} 
\newcommand{\hbbeta}{\widehat{\bbeta}}
\newcommand{\bzero}{\boldsymbol{0}}
\newcommand{\hbU}{\widehat{\bU}}
\newcommand{\hbV}{\widehat{\bV}} 
\newcommand{\hbalpha}{\widehat{\balpha}}

\newcommand{\Nc}{\mathcal{N}}

\newcommand{\by}{\boldsymbol{y}}

\newcommand{\bg}{\boldsymbol{g}}
 
\newcommand{\bgamma}{\boldsymbol{\gamma}} 
\newcommand{\hbgamma}{\widehat{\bgamma}} 
\newcommand{\hbA}{\widehat{\bA}} 
\newcommand{\hbS}{\widehat{\bS}} 
\newcommand{\bzeta}{\boldsymbol{\zeta}} 

\newcommand{\txtop}{\text{op}}
\newcommand{\eop}{\emph{op}}

\newcommand{\SI}{\texttt{SI}}

\newcommand{\mssa}{\texttt{mSSA}}

\newcommand{\SC}{\texttt{SC}}

\newcommand{\HSVT}{\texttt{HSVT}}

\newcommand{\TWSF}{\texttt{TWSF}}

\newcommand{\PCR}{\texttt{PCR}}

\newcommand{\pre}{\texttt{pre}}
\newcommand{\post}{\texttt{post}}
\newcommand{\lag}{\texttt{lag}}
\newcommand{\tnext}{\texttt{next}}
\newcommand{\epre}{\emph{\pre}}

\newcommand{\elag}{\emph{\lag}}
\newcommand{\etnext}{\emph{\tnext}}
\newcommand{\proj}{\texttt{proj}}

\newcommand{\col}{\text{col}}
\newcommand{\row}{\text{row}}
\newcommand{\ecol}{\emph{\col}}
\newcommand{\erow}{\emph{\row}} 

\newcommand{\noise}{\texttt{noise}}

\newcommand{\tcross}{\texttt{cross}}

\newcommand{\bbX}{\bar{\bX}}

\newcommand{\bby}{\bar{\by}}
\newcommand{\ba}{\boldsymbol{a}}

\newcommand{\bxi}{\boldsymbol{\xi}}

\newcommand{\bdelta}{\boldsymbol{\delta}}
\newcommand{\bPi}{\boldsymbol{\Pi}}

\newcommand{\hbx}{\widehat{\bx}}
\newcommand{\be}{\boldsymbol{e}}

\newcommand{\hUpsilon}{\widehat{\Upsilon}}

\newcommand{\direct}{\texttt{dir}}
\newcommand{\edirect}{\emph{\direct}}
\newcommand{\rec}{\texttt{rec}}
\newcommand{\erec}{\emph{\rec}}

\newcommand{\tQ}{\tilde{Q}}

\newcommand{\bJ}{\boldsymbol{J}}
\newcommand{\hbJ}{\widehat{\bJ}}

\newcommand{\hbY}{\widehat{\bY}}
\newcommand{\hbZ}{\widehat{\bZ}}

\newcommand{\tbq}{\tilde{\bq}}
\newcommand{\tq}{\tilde{q}}

\newcommand{\bolde}{\boldsymbol{e}}
\newcommand{\bvarrho}{\boldsymbol{\varrho}}

\newcommand{\bmu}{\boldsymbol{\mu}}
\newcommand{\bw}{\boldsymbol{w}}
\newcommand{\full}{\texttt{full}}

\title{\scshape From Control to Treatment: \\ Causal Forecasting Beyond the Observed Panel}
\author{Dennis Shen}

\affil{\small Department of Data Sciences \& Operations, USC Marshall \\ \texttt{dennis.shen@marshall.usc.edu}}

\date{}

\begin{document}

\maketitle

\vspace{-15pt}
\begin{abstract}
%
Conventional panel data methods recover retrospective counterfactuals within the observed horizon. Many policy decisions, however, are prospective: how would an untreated unit evolve beyond the observed panel under an intervention it has not previously experienced? We develop a causal forecasting framework that combines the counterfactual logic of synthetic controls with multivariate time-series extrapolation. Under a latent factor model, we impose low-rank structure on the treated-state time factors. This yields the Two-Way Synthetic Forecasting estimator, which learns cross-unit weights from pre-treatment outcomes and temporal dynamics from treated donors' post-treatment trajectories. We establish pointwise identification, finite-sample error bounds, consistency, asymptotic normality, and inference for prespecified linear summaries over fixed forecast horizons. Simulations support the theory, and an application to the opening of NFL stadiums during the 2020 season illustrates how the method can inform a prospective policy change using only information available at the decision date.

\end{abstract}

\medskip
\noindent
{\smaller
\textbf{Keywords:}
synthetic controls; synthetic interventions; multivariate singular spectrum analysis; pointwise inference
}

\section{Introduction} \label{sec:intro} 
In September 2020, amid the COVID-19 pandemic, the National Football League (NFL) prepared for a season unlike any other in its history. Each team faced a highly consequential and localized policy choice: should spectators be admitted to the stadium or should games be played behind closed doors?
In coordination with public-health authorities, teams weighed local guidelines, disease prevalence, risk tolerance, and operational constraints. 
Policies varied sharply across cities: the Baltimore Ravens, for example, admitted spectators, whereas the Buffalo Bills kept their stadium closed. 
Teams that opened their doors adopted mitigation measures such masking requirements, social distancing, and capacity limits \citep{mac21a}. 
This uneven policy landscape created natural policy experiments and raised causal questions about the public-health consequences of opening stadiums. 


For a city that opened its stadium, the natural counterfactual is immediate: {\em what would have happened had it remained closed?}
For example, how would COVID-19 case counts in Baltimore have evolved had the Ravens not admitted spectators?
Synthetic controls (\SC) offers an elegant answer \citep{abadie1, abadie2}. 
It combines closed-stadium cities into a weighted ``synthetic Baltimore'' chosen to resemble Baltimore before the season. 
After Baltimore opens, the same weights trace an estimated path for Baltimore under continued closure. 
Comparing that path with Baltimore’s observed trajectory yields an estimate of the effect of opening.
%

A city that remained closed poses the mirror-image question: {\em what would have happened had it opened?}
For example, how would COVID-19 case counts in Buffalo have evolved had the Bills admitted spectators?
Synthetic interventions (\SI) addresses this question by extending the logic of \SC~to multiple intervention states \citep{synth_iv}. 
As with \SC, \SI~constructs a ``synthetic Buffalo,'' but it does so using cities that admitted spectators. 
The relevant cross-city relationship is learned during the period in which all cities were closed and then transported to the period in which the donor cities were open.

At first glance, \SC~and \SI~appear to complete the causal picture: 
\SC~reconstructs what treated units would have experienced under control, whereas \SI~reconstructs what control units would have experienced under treatment. 
Yet this symmetry also reveals a common boundary: the donor outcomes must be observed throughout the period for which the counterfactual is sought. 
%
%
In this precise sense, both methods are retrospective. 
They can reconstruct an alternative past, but they cannot forecast an alternative future. 


This distinction becomes consequential when policy decisions must be made in real time. 
Consider a closed-stadium city midway through the season. 
The relevant question is no longer what would have happened after a decision already made, but what will happen after a decision yet to be made: {\em how will case counts evolve if the city opens for the next game?} 
%
The target trajectory is now doubly out of sample. 
It lies beyond the observed panel and under a policy the city has never experienced. 
It is at once a forecast and a counterfactual.

Conventional time-series methods offer a natural source of extrapolative structure, but they forecast the continuation of the regime already in place: Baltimore under continued opening or Buffalo under continued closure. 
On their own, they do not predict the consequences of switching a target unit into a treatment state absent from its own history. 
Thus, \SC~and \SI~move across policy states but remain within the observed horizon; standard time-series forecasting moves beyond that horizon but remains within an observed policy state. 
This gap motivates the central question of this article: 
{\em Can we forecast beyond the observed panel to predict how a control unit will evolve under a treatment it has not previously experienced?} 


\subsection{Contributions} 
We develop a framework that combines the retrospective counterfactual logic of \SC~and \SI~with the prospective structure of time-series forecasting. 
Our approach builds on their latent factor foundations, under which potential outcomes depend on latent unit and time factors. 
As in \SI, unit factors are shared across intervention states, while time factors may vary with the intervention. 
We augment this model with low-dimensional time-series structure on the treated temporal factors. 
The resulting model expands the estimands supported by conventional \SC~and \SI~designs from retrospective counterfactual reconstruction within the observed panel to prospective causal forecasting beyond it.
We provide an identification result for these causal forecasts. 

To operationalize this strategy, we introduce the Two-Way Synthetic Forecasting (\TWSF) estimator, which combines cross-sectional and temporal structure. As in \SC~and \SI, \TWSF~learns cross-unit relationships from the pre-treatment period. As in multivariate time-series forecasting, it learns dynamic relationships from the post-treatment trajectories of treated units. 
By synthesizing these two components, \TWSF~forecasts the treated potential outcome of a control unit beyond the observed panel. Under suitable conditions, we establish finite-sample pointwise error bounds, consistency, asymptotic normality, and feasible inference for prespecified linear summaries over fixed horizons.

Finally, we support our statistical claims through simulation studies and revisit the NFL stadium-opening study of \cite{nfl_pnas}. In the empirical application, we ask what would have happened had closed-stadium cities switched to opening at a future decision point using only information available up to that time. The findings are broadly consistent with the original analysis: opening stadiums does not appear to produce a systematic increase in local case counts, although the timing of the decision may matter substantively.

\subsection{Related Work} 
This article connects causal panel models with multivariate time-series forecasting. 
Synthetic control (\SC) estimates a treated unit’s path under control from other units, while the synthetic interventions (\SI) framework extends this argument to multiple intervention states \citep{abadie1, abadie2, synth_iv}.  
These methods are often formalized through low-rank factor or matrix completion models \citep{bai2019matrix, fernandez2020low, CAHAN2023113}.  
Another line relates unit- and time-side regressions \citep{athey2021matrix, sameroot} and combines their adjustments through augmented or doubly robust procedures \citep{asc, sdid}. 
Neural forecasting, Bayesian autoregressive factors, and multitask Gaussian processes enrich the temporal model but retain a within-panel target 
\citep{synbeats, Pang_Liu_Xu_2022, eli23, chen23}. 


Multivariate singular spectrum analysis (\mssa) extrapolates beyond the panel by exploiting low-rank trajectory matrices \citep{mssa}. 
%
%
Although \mssa~does not address counterfactual regime changes, its treatment of temporal dynamics informs our approach. 
The closest work is that of \cite{focus}, who also derive pointwise guarantees for causal forecasting in panel data under a low-rank factor model. 
Their setting is complementary. 
Observationally, their design requires target units to have previously experienced treatment, whereas we study target units that remain untreated throughout the observed panel. 
From a modeling standpoint, they impose a stationary autoregressive structure on the temporal factors while we impose a low-rank temporal structure. 

\subsection{Notation} 
For a positive integer $a$, let $[a] = \{1, \dots, a\}$. For a vector $\bv \in \Rb^a$, let $\| \bv \|_p$ denote its $\ell_p$-norm. Define the inner product between $\bu, \bv \in \Rb^a$ as $\langle \bu, \bv \rangle = \bu^\top \bv = \sum_{\ell=1}^a u_\ell v_\ell$. For a matrix $\bX \in \Rb^{a \times b}$, denote its operator and Frobenius norms as $\| \bX \|_\txtop$ and $\| \bX \|_F$. Let $\| \bx \|_{\psi_2}$ denote the subgaussian norm of a random vector $\bx$. Denote by $\bI_a$ the $a \times a$ identity matrix, $\bzero_{a \times b}$ the $a \times b$ zero matrix, and $\be_a$ the $a$th standard basis vector. Let $\dagger$ denote the Moore-Penrose pseudoinverse. 
 Convergence in probability and distribution are $\xrightarrow{p}$ and $\rightsquigarrow$.

\section{Causal Forecasting Panel Framework} \label{sec:framework}
Consider a panel of $N$ units observed across $T \coloneqq T_0 + T_1$ time periods. 
Let $Y_{it}(0)$ and $Y_{it}(1)$ denote the potential outcomes of unit $i \in [N]$ at time $t \in [T]$ under control and treatment \citep{neyman, rubin1976}. 
With $D_{it} \in \{0,1\}$ as the treatment status, the observed outcome is 
\begin{align}
	Y_{it} \coloneqq D_{it} \cdot Y_{it}(1) + (1-D_{it}) \cdot Y_{it}(0). 
	\label{eq:sutva}
\end{align} 
The observation law of \eqref{eq:sutva} encodes the stable unit treatment value assumption, which implicitly rules out spillovers and network effects \citep{imbens_rubin_2015}. 

We study the observation pattern depicted in Figure~\ref{fig:data}. 
During the first $T_0$ periods, which we call the pre-treatment period, all $N$ units are observed under control.
Over the next $T_1$ periods, which we call the post-treatment period, unit $N$ remains under control while the other units receive treatment. 
Let $\Ic_0 \coloneqq \{N\}$, $\Ic_1 \coloneqq [N-1]$, and $N_1 \coloneqq | \Ic_1 |$. 
Thus, $D_{it} = 0$ for every $i \in [N]$ and $t \le T_0$. 
For $T_0 < t \le T$, $D_{it} = 1$ for $i \in \Ic_1$ and $D_{it} = 0$ for $i = N$. 
Let $\bD \in \{0,1\}^{N \times T}$ and $\bY \in \Rb^{N \times T}$ collect the treatment assignments and observed outcomes. 


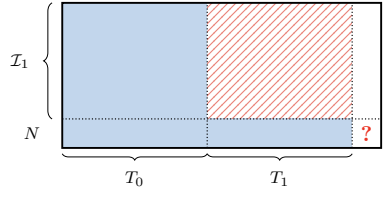
\begin{figure}[t]
\centering

\scalebox{0.80}{%
\begin{tikzpicture}[
    x=1.2cm,
    y=1.2cm,
    every node/.style={font=\footnotesize}
]

    \fill[controlblue]
        (0,0) rectangle (2,2);

    \fill[controlblue]
        (2,0) rectangle (4,0.4);

    \path[
        pattern=north east lines,
        pattern color=treatedred
    ]
        (2,0.4) rectangle (4,2);

    \fill[white]
        (4,0) rectangle (4.4,2);

    \draw[line width=0.9pt]
        (0,0) rectangle (4.4,2);

    \draw[densely dotted,line width=0.55pt]
        (2,0) -- (2,2);

    \draw[densely dotted,line width=0.55pt]
        (4,0) -- (4,2);

    \draw[densely dotted,line width=0.55pt]
        (0,0.4) -- (4.4,0.4);

    \draw[
        decorate,
        decoration={
            brace,
            amplitude=4pt
        }
    ]
        (-0.14,0.4)
        --
        node[midway,left=7pt] {$\mathcal{I}_1$}
        (-0.14,2);

    \node[left=7pt]
        at (0,0.2)
        {$N$};

    \node[
        text=questionred,
        font=\normalsize\bfseries
    ]
        at (4.2,0.2)
        {?};

    \draw[
        decorate,
        decoration={
            brace,
            mirror,
            amplitude=3pt
        }
    ]
        (0,-0.08)
        --
        node[midway,below=5pt] {$T_0$}
        (2,-0.08);

    \draw[
        decorate,
        decoration={
            brace,
            mirror,
            amplitude=3pt
        }
    ]
        (2,-0.08)
        --
        node[midway,below=5pt] {$T_1$}
        (4,-0.08);
\end{tikzpicture}%
}
\caption{
Observation pattern for $Y$.
Blue denotes observed control outcomes,
red stripes denote observed treated outcomes,
white denotes unobserved future outcomes,
and the red ``?'' denotes $\theta$.
}
\label{fig:data}

\end{figure}


\subsection{Tensor Factor Model} 
We adopt the factor model underlying the \SI~framework. 

\begin{assumption} 
\label{assump:lfm}
For $i \in [N]$, $t \in \Zb$, and $d \in \{0,1\}$, let 
$$Y_{it}(d) \coloneqq \sum_{a=1}^r U_{i a} V_{t a}(d) + \varepsilon_{it}(d),$$
where $\bu_i \coloneqq \big[U_{ia}: a \in [r] \big] \in \Rb^r$ is the latent factor for unit $i$, $\bv_t(d) \coloneqq \big[V_{ta}(d): a \in [r] \big] \in \Rb^r$ is the latent factor at time $t$ under state $d$, and $\varepsilon_{it}(d) \in \Rb$ is idiosyncratic noise. 
\end{assumption}
The central restriction is that the unit factor $\bu_i$ is invariant across both time and intervention states. 
This invariance enables cross-unit relationships learned under control to be transported to treatment. 
By contrast, the time factor $\bv_t(d)$ may vary flexibly across intervention states, allowing treatment to alter the temporal evolution of potential outcomes.  
%

\subsubsection{Causal Forecasting Estimand} \label{sec:estimand}
We define the one-step-ahead causal forecasting estimand as 
\begin{align}
	\theta \coloneqq \Ex \left[Y_{N, T+1}(1) \mid \bu_N, \bv_{T+1}(1) \right]. \label{eq:estimand}
\end{align}
In words, $\theta$ is the conditional mean of the target unit’s treated potential outcome one period beyond the observed panel. In the NFL example, if Baltimore first admitted spectators on November 1, then $\theta$ would represent its expected case count on November 2 under the open-stadium policy. Section~\ref{sec:forecast.horizon} extends this object to fixed multi-step summaries.

\subsubsection{Unobserved Confounding}
Treatment assignment may depend on local characteristics, recent trends, and anticipated outcomes. We allow such confounding to operate through the latent factors. Define 
\begin{align}
	\Ec \coloneqq \sigma \left( \left\{ \bu_i: i \in [N] \right\}, \left\{ \bv_t(d): t \in [T+1], d \in \{0,1\} \right\}, \bD \right)
\end{align} 
as the sigma-field generated by the latent factors and intervention assignments. 

\begin{assumption} 
\label{assump:mean_ind}
For every $(i, t, d)$, let $\Ex[\varepsilon_{it}(d) \mid \Ec] = 0$.  
\end{assumption}
Conditional on the latent structure, the idiosyncratic noise is mean zero. Together with Assumption~\ref{assump:lfm}, Assumption~\ref{assump:mean_ind} implies that potential outcomes are mean independent of treatment assignment once the latent factors are held fixed. Thus, the latent factors play the role of unobserved confounders, analogous to observed covariates under a classical selection-on-observables assumption.
This latent ignorability condition parallels assumptions used in causal panel models \citep{athey2021matrix, asc, synth_iv} and related contexts \citep{kallus2018causal}. As with any restriction on unobserved confounding, Assumption~\ref{assump:mean_ind} is not directly testable and should be assessed using domain knowledge about the data-generating process. 

\subsubsection{Latent Unit Structure} \label{sec:unit.structure} 
%

\begin{assumption} 
\label{assump:units}
Let $\bu_N \in \emph{span}\{\bu_j: j \in \Ic_1\}$. 
\end{assumption}
Assumption~\ref{assump:units} requires the latent factor of the control unit to lie within the linear span of the latent factors of the treated donor units. 
The low-rank structure in Assumption~\ref{assump:lfm} makes this condition plausible but does not guarantee it. 
In practical terms, Assumption~\ref{assump:units} requires a sufficiently large and rich donor pool. 
The condition is thus analogous to common support, where treated and control units must overlap in their covariate distributions. 
Although the latent factors are unobserved, Assumption~\ref{assump:units} can be assessed indirectly through pre-treatment fit diagnostics, as is standard in the \SC~literature \citep{abadie_survey}. 


%
\begin{proposition} \label{prop:beta} 
Let Assumptions~\ref{assump:lfm}-\ref{assump:units} hold. Then, there exists a $\bbeta \in \Rb^{N_1}$ such that
\begin{enumerate} [label=(\alph*)]
	\item $\theta = \sum_{j \in \Ic_1} \beta_j \cdot \Ex[Y_{j, T+1} (1) \mid \Ec]$,
	
	\item $\Ex[Y_{Nt} \mid \Ec ] = \sum_{j \in \Ic_1} \beta_j \cdot \Ex[ Y_{jt} \mid \Ec]$ for all $t \le T_0$. 
\end{enumerate}
\end{proposition} 

Proposition~\ref{prop:beta}, which parallels \cite[Proposition 1]{synth_iv}, shows that the same unit weights can be used on both sides of the treatment boundary. Part (b) provides an observable period in which those weights can be learned: during the pre-treatment period, the target and donors are all observed under control. 
Part (a) then transports the learned cross-unit relation to the treated state. 
It does not, however, identify $\theta$ by itself because the donors' treated outcomes at $T+1$ lie beyond the observed panel. 

\subsubsection{Latent Temporal Structure} \label{sec:time.structure} 
For each latent treated-state component 
$a \in [r]$, define the time series $f_{a}(t) \coloneqq V_{ta} (1)$ for $t \in \mathbb{Z}$. 
For positive integers $m, n$ and shift $s \in \Zb$, define the corresponding Hankel matrix by 
$\bH(f_a; m, n, s) \coloneqq \big[f_a(s+i+j-2): i \in [m], j \in [n] \big] \in \Rb^{m \times n}$. 
%
%

\begin{assumption} 
\label{assump:hankel}
For every $a \in [r]$, positive integers $m, n \in \mathbb{N}$, and shift $s \in \mathbb{Z}$, let $\emph{rank}(\bH(f_a; m, n, s)) \le Q$.
\end{assumption}
Assumption~\ref{assump:hankel} imposes that each latent treated temporal factor admits a low-rank Hankel representation. 
This condition accommodates finite sums of products of harmonics, polynomials, and exponentials, and thus covers numerous time series with trend or periodicity \citep{mssa}. 
%
The class of time series with low-rank Hankel structure is also closed under component-wise addition and multiplication. 
The uniform rank restriction implies that each $f_a$ satisfies a linear recurrence of order at most $Q$. Combining the $r$ components yields a common recurrence of order at most $rQ$ for all treated donors.



\begin{proposition} \label{prop:alpha}
Let Assumptions~\ref{assump:lfm}, \ref{assump:mean_ind}, and \ref{assump:hankel} hold. Fix a lag length $L \in \Zb_+$ satisfying $rQ \le L \le T_1-1$. Then, there exists an $\balpha \in \Rb^{L}$ such that 
\begin{enumerate} [label=(\alph*)]
	\item $\Ex[Y_{j, T+1}(1) \mid \Ec] = \sum_{\ell =1}^{L} \alpha_\ell \cdot \Ex[Y_{j, T -L  + \ell} \mid \Ec]$ for all $j \in \Ic_1$, 
	
	\item $\Ex[Y_{jt} \mid \Ec] = \sum_{\ell = 1}^{L} \alpha_\ell \cdot \Ex[Y_{j, t-L-1+\ell} \mid \Ec]$ for all $j \in \Ic_1$ and $t \in \{T_0+L+1, \dots, T\}$. 
\end{enumerate}
\end{proposition}
Proposition~\ref{prop:alpha} is similar in spirit to \cite[Proposition 4.1]{mssa}. 
Part (a) expresses each donor unit's treated outcome at time $T+1$ as a common linear function of its recent treated outcomes. Part (b) shows that the same relation is visible throughout the observed post-treatment period. It therefore supplies an observable temporal regression with which the forecasting rule may be learned.

\subsection{A Two-Way Structural Representation} \label{sec:estimation.result}
%

\begin{theorem} \label{thm:estimation}
Let Assumptions~\ref{assump:lfm}-\ref{assump:hankel} hold. 
For any $(\balpha, \bbeta)$ satisfying Propositions~\ref{prop:beta}-\ref{prop:alpha}, 
$$
	\theta = \sum_{\ell = 1}^{L} \sum_{j \in \Ic_1} \alpha_\ell \cdot \beta_j \cdot \Ex[Y_{j, T - L + \ell} \mid \Ec]. 
$$
%
%
\end{theorem} 

\begin{proof}
The result follows directly from Propositions~\ref{prop:beta}(a) and \ref{prop:alpha}(a). 
\end{proof}

Theorem~\ref{thm:estimation} demonstrates that $\theta$ can be represented using only donor outcomes observed up to time $T$ by combining two distinct forms of transport. The unit weights $\bbeta$ transport information from treated donors to the target unit, while the time weights $\balpha$ transport information from the observed panel into the future. The observable relations in Propositions~\ref{prop:beta}(b) and \ref{prop:alpha}(b) motivate the estimation of these two coefficient vectors. 

\section{The Two-Way Synthetic Forecasting Estimator} \label{sec:alg}
%
Motivated by Theorem~\ref{thm:estimation}, we introduce the Two-Way Synthetic Forecasting (\TWSF) estimator.

\subsection{Setup} \label{sec:alg.setup}

\noindent \textbf{Unit-side data.}
Collect the pre-treatment outcomes of the target unit $N$ in $$\by_{N, \pre} \coloneqq \left[Y_{Nt}: t \le T_0 \right] \in \Rb^{T_0},$$ 
and arrange the corresponding donor outcomes as $$\bY_{\Ic_1, \pre} \coloneqq \left[Y_{jt}: j \in \Ic_1, t \le T_0 \right] \in \Rb^{N_1 \times T_0}.$$ 

\medskip \noindent \textbf{Time-side data.}
Let $L \in \mathbb{Z}_+$ be the lag length, and define $B \coloneqq T_1 / (L+1) \in \mathbb{N}$ and $M \coloneqq (B-1)N_1$. 
For donor $j \in \Ic_1$, organize the first $B-1$ nonoverlapping post-treatment blocks into the Page matrix $(\bP_j)_{ab} \coloneqq Y_{j, T_0+(b-1)(L+1)+a}$ for $a \in [L+1]$ and $b \in [B-1]$.
%
Within each column, the first $L$ entries form a lag vector, and the final entry is its one-step-ahead response. Stacking the donor-specific Page matrices horizontally gives 
\begin{align}
	\bP \coloneqq
	\left[ \bP_{1} ~ \cdots ~ \bP_{N_1} \right]
	\eqqcolon 
	\begin{pmatrix}
		\bZ_\lag
		\\
		\hdashline
		\bz_{\tnext}^\top
	\end{pmatrix}
	\in \Rb^{(L+1) \times M}. \label{eq:Zlag.znext} 
\end{align}
Here, $\bZ_\lag \in \Rb^{L \times M}$ contains the first $L$ rows of the stacked Page matrix, while $\bz_{\tnext} \in \Rb^M$ contains the final row. 
The Page matrix excludes the final $L+1$ post-treatment observations as they are reserved for forecasting the target outcome at time $T+1$.
The most recent $L$ observations form the forecast state
$$\bW \coloneqq \big[ Y_{j, T-L+\ell}: j \in \Ic_1, \ell \in [L] \big] \in \Rb^{N_1 \times L}.$$ 
%

\medskip \noindent \textbf{Spectral denoising.} 
A key ingredient of \TWSF~is a spectral de-noising step based on hard singular value thresholding (\HSVT). 
For a matrix $\bA \in \Rb^{a \times b}$, write its singular value decomposition as $\bA = \sum_{\ell=1}^{\min\{a,b\}} s_\ell \bu_\ell \bv_\ell^\top$. 
For any $k \le \min\{a, b\}$, define its rank-$k$ approximation as 
$\bA^{(k)} \coloneqq \HSVT(\bA, k) = \sum_{\ell=1}^k s_\ell \bu_\ell \bv_\ell^\top$. 
%

\subsection{Description} \label{sec:desc}
%
The \TWSF~estimator proceeds in four steps. 
\begin{enumerate} [label=(\alph*)]

\item For $k_y \le \min\{N_1, T_0\}$ and $k_z \le \min\{L, M\}$, define the de-noised design matrices as 
\begin{align}
	\bY_{\Ic_1, \pre}^{(k_y)} &\coloneqq \HSVT\left(\bY_{\Ic_1, \pre}, k_y \right), 
	\quad
	\bZ_{\lag}^{(k_z)} \coloneqq \HSVT\left(\bZ_\lag, k_z \right). \label{eq:hsvt.YZ}
\end{align}

\item Learn the cross-unit weights from the pre-treatment period by solving 
\begin{align}
	\hbbeta &\coloneqq \left(\bY_{\Ic_1, \pre}^{(k_y)} \right)^{\top, \dagger} \by_{N, \pre} \in \argmin_{\bomega \in \Rb^{N_1}} \left\| \by_{N, \pre} - \left( \bY_{\Ic_1, \pre}^{(k_y)} \right)^\top \bomega \right\|_2^2. \label{eq:pcr.beta.hat}
\end{align}

\item Learn the temporal forecasting weights from the treated donor trajectories by solving 
\begin{align}
	\hbalpha &\coloneqq \left(\bZ_{\lag}^{(k_z)} \right)^{\top, \dagger} \bz_{\tnext} \in \argmin_{\bomega \in \Rb^L} \left\| \bz_{\tnext} - \left(\bZ_{\lag}^{(k_z)}\right)^\top \bomega \right\|_2^2. \label{eq:pcr.alpha.hat}
\end{align}

\item Combine the unit and temporal weights to produce the one-step-ahead causal forecast 
\begin{align}
	\htheta &\coloneqq \left\langle \hbalpha, \bW^\top \hbbeta \right \rangle = \left \langle \hbbeta, \bW \hbalpha \right \rangle. \label{eq:htheta}
\end{align} 

\end{enumerate}

\subsubsection{Interpretation} 
Taking Baltimore as the target unit, \eqref{eq:htheta} admits two equivalent interpretations. 
Under the first, the vector $\bW^\top \hbbeta$ represents Baltimore's imputed counterfactual trajectory over the final $L$ observed periods had it already been operating under the open-stadium policy. 
Applying the temporal rule $\hbalpha$ then forecasts Baltimore's treated outcome at $T+1$. 
Under the second, the vector $\bW \hbalpha$ contains a one-step-ahead treated forecast for every donor unit. 
The unit weights $\hbbeta$ then synthesize these donor forecasts into the forecast for Baltimore. 
The two routes coincide algebraically and hence, this duality explains the term {\em two-way}: \TWSF~transports information across both units and time, where its unit stage follows the logic of \SI, while its temporal stage follows the forecasting logic of \mssa. 
Accordingly, \TWSF~can be read either as ``impute, then forecast'' or as ``forecast, then impute.'' 


\subsubsection{Principal Component Regression}
The operations in \eqref{eq:hsvt.YZ}-\eqref{eq:pcr.alpha.hat} are instances of principal component regression (\PCR).
Each observed design matrix is first replaced by a low-rank approximation and then used in a least squares regression. 
Under the latent factor structure of Section~\ref{sec:framework}, the conditional means $\Ex[\bY_{\Ic_1, \pre} \mid \Ec]$ and $\Ex[\bZ_\lag \mid \Ec]$ are low rank, and their observed counterparts are contaminated by idiosyncratic noise. 
Spectral truncation retains the dominant signal directions while limiting noise amplification in weakly identified directions. 

%

The spectral thresholds $k_y$ and $k_z$ may be selected by cross-validation, by retaining enough principal components to explain a prescribed fraction of spectral energy, or by applying a data-dependent singular-value threshold such as those proposed in \cite{donoho14} and \cite{Chatterjee15}. 


\section{Statistical Guarantees} \label{sec:results}
This section establishes statistical guarantees for \TWSF. 
Throughout, for any random object $\bX$, write $\bbX \coloneqq \Ex[\bX \mid \Ec]$. 
All results below are understood conditional on $\Ec$. 

\subsection{Recoverability under \PCR} \label{sec:identification} 
We first address an important subtlety. 
Because the population design matrices may be rank deficient, \PCR~need not identify the coefficient vectors $\balpha$ and $\bbeta$ uniquely. 
Instead, \PCR~identifies only the projection onto the row spaces of the population designs. 
The following condition ensures that these recoverable components nevertheless suffice to identify the causal forecast.

\begin{assumption} 
\label{assump:transport}
	Let $\ecol(\bbW) \subseteq \ecol(\bbY_{\Ic_1, \epre})$ 
	and $\erow(\bbW) \subseteq \erow(\bbZ_\elag^\top)$. 
\end{assumption} 
Assumption~\ref{assump:transport} is an out-of-sample generalization condition. 
It requires the unit- and time-side directions relevant to the forecast to be represented in the designs used to estimate the corresponding weights.
In ordinary statistical learning, the analogous requirement is overlap between the training and test covariate distributions. Here, that requirement is expressed in the linear-algebraic language of a panel subject to a treatment-induced distribution shift. 

\begin{cor} \label{cor:identification}
Let the setup of Theorem~\ref{thm:estimation} hold. 
Define $\balpha^* \coloneqq \bbZ_\elag^{\top, \dagger} \bbZ_\elag^\top  \balpha$ and $\bbeta^* \coloneqq \bbY_{\Ic_1, \epre}^{\top, \dagger} \bbY_{\Ic_1, \epre}^\top  \bbeta$. 
If Assumption~\ref{assump:transport} holds, then
$$
	\theta = \sum_{\ell = 1}^{L} \sum_{j \in \Ic_1} \alpha^*_\ell \cdot \beta^*_j \cdot \Ex[Y_{j, T-L+\ell} \mid \Ec]. 
$$
\end{cor} 
Corollary~\ref{cor:identification} reformulates Theorem~\ref{thm:estimation} in terms of the uniquely defined, recoverable coefficients $\balpha^*$ and $\bbeta^*$. Assumption~\ref{assump:transport} is therefore complementary to previous structural conditions. 
Namely,  Assumption~\ref{assump:units} guarantees that a valid unit representation exists, while Assumption~\ref{assump:transport} ensures that its identifiable component is sufficient for forecasting.  
Assumptions~\ref{assump:hankel} and \ref{assump:transport} play the analogous roles on the temporal side.  

\subsection{Additional Assumptions} \label{sec:assump.consistency} 

\begin{assumption} 
\label{assump:bounded} 
	For every $(i,t,d)$, let $\langle \bu_i, \bv_t(d) \rangle \in [-1,1]$. 
\end{assumption}
The bound is without loss of generality and can be extended to $[a,b]$ for any $a \le b$. 

\begin{assumption} 
\label{assump:subg}
	Conditional on $\Ec$, $\varepsilon_{it}(d)$ are independent subgaussian variables satisfying $\emph{Var}(\varepsilon_{it}(d)) \le \sigma^2$ and $\| \varepsilon_{it}(d) \|_{\psi_2} \le C_\varepsilon \sigma$ for some constant $C_\varepsilon > 0$. 
\end{assumption} 
Assumption~\ref{assump:subg} permits the latent temporal factors to be correlated and to obey rich time-series dynamics. The independence requirement applies only to the idiosyncratic shocks. Although restrictive, it provides a transparent starting point for a finite-sample analysis of \TWSF. Comparable independence conditions appear in early analyses of \SC, \SI, and \mssa~\citep{abadie2, synth_iv, mssa}, as well as in related work \citep{focus}. Extending  to more general dependent noise structures remains future work.

\begin{assumption} 
\label{assump:spectra}
	The condition number $\kappa_y$ of $\bbY_{\Ic_1, \epre}$ satisfies $\kappa_y^{-1} \ge c_y$ and $\| \bbY_{\Ic_1, \epre} \|_F^2 \ge c'_y N_1 T_0$ for constants $c_y, c'_y > 0$. 
	Similarly, the condition number $\kappa_z$ of $\bbZ_\elag$ satisfies $\kappa_z^{-1} \ge c_z$ and $\| \bbZ_\elag\|_F^2 \ge c'_z L M$ for constants $c_z, c'_z > 0$. 
\end{assumption} 
Assumption~\ref{assump:spectra} requires the nonzero singular values of the population design matrices to be sufficiently strong and well balanced, so that the low-rank signal can be reliably separated from noise. 
It is analogous to spectral gap, pervasiveness, and beta-min conditions used in factor models, matrix completion, and high-dimensional regression \citep{chamberlainfactor, bai2019matrix, fan2018eigenvector, beta_min}. 
In applications, the plausibility of this condition can be investigated through scree plots, estimated condition numbers, and related spectral diagnostics. 

%

\subsection{Parameter Estimation Error} \label{sec:param.est}
We first control the errors of the two \PCR~subroutines used to recover $(\balpha^*, \bbeta^*)$. 
Define the population ranks $r_y \coloneqq \rank(\bbY_{\Ic_1, \pre})$ and $r_z \coloneqq \rank(\bbZ_{\lag})$.
For $a, m, n \in \Nb$, let  
\begin{align}
	\Lambda(a,m,n) \coloneqq \frac{a}{\min\{\sqrt{m}, \sqrt{n}\}} + \frac{ \sqrt{ a (1 + \log(mn))}}{\sqrt{n}}. 
\end{align}
Define $\Lambda_\alpha \coloneqq \Lambda(r_z, L, M)$ and $\Lambda_\beta \coloneqq \Lambda(r_y, N_1, T_0)$. 
Hereafter, the symbols $\lesssim$ and $\gtrsim$ suppress constants that do not depend on the model dimensions. 

\begin{theorem} \label{prop:parameter.recovery}
Let Assumptions~\ref{assump:lfm}-\ref{assump:spectra} hold. 
Suppose $k_y = r_y$ and 
\begin{align}
	N_1 T_0 \gtrsim \sigma^2 r_y  \left( \sqrt{N_1} + \sqrt{T_0} + \sqrt{\log(N_1T_0)} \right)^2. 
	\label{eq:sep.y}
\end{align} 
Then, with probability at least $1 - \Oc\left((N_1 T_0)^{-10} \right)$, $\| \hbbeta - \bbeta^*  \|_2 \lesssim \sigma \Lambda_{\beta} / \sqrt{N_1}$. 
%
%
Similarly, suppose $k_z = r_z$ and 
\begin{align}
	L M \gtrsim \sigma^2 r_z  \left( \sqrt{L} + \sqrt{M} + \sqrt{\log(LM)} \right)^2, 
	\label{eq:sep.z}
\end{align}
with $L \ge rQ$. 
Then, with probability at least $1 - \Oc\left((LM)^{-10} \right)$, 
$\| \hbalpha - \balpha^*  \|_2 \lesssim \sigma \Lambda_\alpha / \sqrt{L}$. 
%
%
\end{theorem} 
Theorem~\ref{prop:parameter.recovery} provides finite-sample rates for recovering the identifiable components of the two population weight vectors. 
The following regime provides a convenient interpretation. 

\begin{definition} [Balanced regime] \label{def:balanced}
Let $N_1 \asymp T_0 \asymp L \asymp M \asymp n$ and $r_y \asymp r_z \asymp r_n$. 
\end{definition}

With fixed ranks, Theorem~\ref{prop:parameter.recovery} yields an error rate of $\Oc(n^{-1})$, up to logarithmic factors. 
This sharpens the $\Oc(n^{-3/4})$ bound in \cite{pcr_jmlr}. 
Because \SI, \mssa, and related de-noising methods \citep{corrupt} all rely on \PCR~subroutines, the improved bound may have applications beyond the present setting. 

Conditions \eqref{eq:sep.y} and \eqref{eq:sep.z} require the smallest nonzero singular values of the population signal matrices dominate the operator-norm fluctuation of the noise blocks. 
The theorem also assumes oracle knowledge of $r_y$ and $r_z$. 
A formal analysis of rank misspecification is left for future work, although existing results suggest that overestimating the rank is generally less harmful than underestimating it \citep{pcr_jmlr}.  

\subsection{Finite-Sample Forecast Error} \label{sec:forecast.error} 
The parameter bounds imply the following finite-sample guarantee. 

\begin{theorem} 
\label{thm:error}
Let the setup of Theorem~\ref{prop:parameter.recovery} hold. 
Define $\varphi_{\min}^2 \coloneqq 1 + \log(\min\{N_1 T_0, LM \})$. 
Then, with probability at least $1- \Oc((N_1 T_0)^{-10} + (LM)^{-10})$, 
\begin{align}
	\left| \htheta - \theta \right| 
	\lesssim 
	\sigma \left(\Lambda_{\alpha} + \Lambda_{\beta}\right) + \sigma^2 \Lambda_{\alpha} \Lambda_{\beta} + \frac{ \sigma \varphi_{\min}  \left(\sigma \Lambda_{\alpha} + \sqrt{r_z}\right) \left(\sigma \Lambda_{\beta} + \sqrt{r_y} \right)}{\sqrt{N_1L}}. 
	\label{eq:hp.error.bound}
\end{align} 
%
%
\end{theorem} 
Under Definition~\ref{def:balanced}, $\big| \htheta - \theta \big| = \Oc( n^{-1/2} a_n  +  n^{-1} a_n^2)$, where $a_n \coloneqq r_n + \sqrt{r_n \log(n)}$. 
Consequently, when $r_n \gg \log(n)$, \TWSF~is consistent if $r_n = o (\sqrt{n})$; when $r_n \lesssim \log(n)$, it is sufficient that $r_n \log(n) = o(n)$. 

\subsection{Asymptotic Normality} \label{sec:normality}
We next derive the asymptotic distribution of $\htheta$. 
The proof in Appendix~\ref{sec:proofs.inference} uses an orthogonalization expansion that cancels the first-order \PCR~errors in $\hbalpha$ and $\hbbeta$.
This expansion introduces the dual weights 
\begin{align}
	\bq^*_{\beta} &\coloneqq \bbY_{\Ic_1, \pre}^\dagger  \bbW \balpha^*, 
	\quad
	\bq^*_{\alpha} \coloneqq \bbZ_\lag^\dagger  \bbW^\top \bbeta^*.  \label{eq:riesz.q}
\end{align}
Under Assumption~\ref{assump:transport}, these vectors solve the feasible equations $\bbY_{\Ic_1, \pre} \bq^*_{\beta} = \bbW \balpha^*$ and $\bbZ_\lag  \bq^*_{\alpha} = \bbW^\top \bbeta^*$. 
%
%
Thus, $\bq^*_{\beta}$ gives the minimum $\ell_2$-norm combination of pre-treatment columns needed to reproduce the donor forecast vector $\bbW \balpha^*$; likewise, $\bq^*_{\alpha}$ is the minimum $\ell_2$-norm combination of temporal training examples required to reproduce the lag history $\bbW^\top \bbeta^*$ at which the forecast is evaluated. 
We can now state our normality result. 

\begin{theorem} 
\label{thm:inference} 
Let the setup of Theorem~\ref{thm:error} hold with Assumption~\ref{assump:subg} specialized to Gaussian noise with homoskedastic variance $\sigma^2$. 
Then, $\htheta - \theta = G + R$, where $G \sim \mathcal{N}\left(0, \Vc^2 \right)$ and 
$\Vc^2 \coloneqq \sigma^2 \{ \| \balpha^* \|_2^2 \cdot \| \bbeta^* \|_2^2
	+ \| \bq^*_{\beta} \|_2^2 (1 + \| \bbeta^* \|_2^2 )
	+ \| \bq^*_{\alpha} \|_2^2 (1 + \| \balpha^* \|_2^2 ) 
	\}$;
moreover, with probability at least $1 - \Oc(\rho)$, $| R | \lesssim \Psi$, where $\rho \coloneqq (N_1T_0)^{-10} + (LM)^{-10} + (N_1L)^{-10}$ and $\Psi$ is a deterministic remainder envelope defined below in \eqref{eq:remainder.upper.bound}. 
Thus, if $\Psi / \Vc = o(1)$, then, as $N_1, T_0, L, M \rightarrow \infty$, $\big(\htheta - \theta \big) / \Vc \rightsquigarrow \Nc(0,1)$. 
\end{theorem} 

\noindent \textbf{Leading variance.} 
The three terms in $\Vc^2$ correspond to noise in the terminal donor-history block, estimation of the unit-side rule, and estimation of the time-side rule. The dual-weight norms measure extrapolative leverage: weak representation of the forecast direction in either design increases the variance transmitted to the final estimate. Conditional independence of the relevant noise blocks makes these contributions additive. Gaussianity is not conceptually essential but it avoids additional Lyapunov-type conditions and yields a cleaner asymptotic statement. 


\medskip \noindent \textbf{Remainder envelope.} 
Define $\varphi_{\max}^2 \coloneqq 1 +  \log(\max\{ N_1T_0 , LM, N_1L \})$, 
\begin{align}
	\Omega_{\alpha} &\coloneqq\frac{r_z}{\min\{L, M\}} + \frac{\sqrt{r_y r_z}}{\sqrt{M}  \min\{\sqrt{N_1}, \sqrt{L}\}} + \frac{\sqrt{r_z}  \Lambda_{\beta}}{\sqrt{M }},
	\\
	\Omega_{\beta} &\coloneqq \frac{r_y}{\min\{N_1, T_0\}} + \frac{\sqrt{r_y r_z}}{\sqrt{T_0}  \min\{\sqrt{N_1}, \sqrt{L}\}} + \frac{\sqrt{r_y}  \Lambda_{\alpha}}{\sqrt{T_0 }}.
\end{align} 
The remainder envelope is $\Psi / \sigma^2 \coloneqq \Psi_{\alpha \beta} + \Psi_{\alpha} + \Psi_{\beta}$, where 
\begin{align}
	\Psi_{\alpha \beta}
	&\coloneqq 
	\varphi_{\max} \left\{ \frac{ \sqrt{r_y}  \Lambda_{\alpha} + \sqrt{r_z}  \Lambda_{\beta} + \sigma  \Lambda_{\alpha} \Lambda_{\beta}}{\sqrt{N_1 L}} \right\},
	\\ 
	\Psi_\alpha
	&\coloneqq  
	\Omega_{\alpha} \left\{ \frac{\sqrt{r_z  M}  }{\min\{\sqrt{L}, \sqrt{M}\}} + \varphi_{\max} \right\},
	~~
	\Psi_\beta
	\coloneqq
	\Omega_{\beta} \left\{ \frac{\sqrt{r_y  T_0} }{\min\{\sqrt{N_1}, \sqrt{T_0}\}} + \varphi_{\max} \right\}.
	\quad ~ 
	\label{eq:remainder.upper.bound} 
\end{align}
Here, $\Psi_{\alpha \beta}$ collects interactions between the two \PCR~stages, while $\Psi_\alpha$ and $\Psi_\beta$ capture perturbations of the time- and unit-side dual weights.

\medskip \noindent \textbf{Interpretation.} 
Theorem~\ref{thm:inference} establishes asymptotic normality whenever the remainder $\Psi$ is negligible relative to the leading standard deviation $\Vc$.
Under Definition~\ref{def:balanced}, suppose additionally that $\| \bq^*_{\beta} \|_2 + \| \bq^*_{\alpha} \|_2 \gtrsim n^{-1/2}$.
Then, $\Vc \gtrsim n^{-1/2}$, and 
$\Psi / \Vc = \Oc\big( n^{-1/2}  r_n ( \sqrt{r_n} + \sqrt{1+\log(n)} )^2 \big)$. 
When $r_n \gg \log(n)$, the condition $\Psi/ \Vc = o(1)$ holds if $r_n = o(n^{1/4})$. 
When $r_n \lesssim \log(n)$, it is sufficient that $r_n \log(n) = o(\sqrt{n})$.

\section{Pointwise Inference} \label{sec:inference}
Theorem~\ref{thm:inference} identifies the leading Gaussian fluctuation of \TWSF. Feasible inference requires consistent estimates of both the  idiosyncratic variance $\sigma^2$ and the leading variance $\Vc^2$. 
As in Section~\ref{sec:results}, all results below are understood conditional on $\Ec$. 


\subsection{Estimating the Idiosyncratic Variance} \label{sec:var.sigma}
%
%
We consider a variance estimator based on the pooled unit- and time-side \PCR~residuals: 
\begin{align}
	\hsigma^2 &\coloneqq \frac{\| \by_{N, \pre} - \bY_{\Ic_1, \pre}^\top \hbbeta \|_2^2 ~+~ \| \bz_{\tnext} - \bZ_\lag^\top \hbalpha \|_2^2}{(T_0 - k_y) + (M - k_z)}. \label{eq:var.sigma.pool} 
\end{align}
For $a, m, n \in \Nb$ with $a < m$, define  
\begin{align}
	\delta(a, m, n) \coloneqq 
	\frac{a}{m - a}  \left( 1 + \frac{m + \log(mn)}{n} \right)
	+ \frac{\sqrt{\log(mn)}}{\sqrt{m - a}} + \frac{\log(mn)}{m - a}. 
	\label{eq:delta.var.sigma}
\end{align}

\begin{proposition} \label{prop:var.sigma}
Let the setup of Theorem~\ref{thm:inference} hold. 
Then, with probability at least $1 - \Oc((N_1T_0)^{-10} + (LM)^{-10})$, 
\begin{align}
	\left| \frac{\hsigma^2}{\sigma^2}  - 1 \right| 
	&\lesssim 
	\frac{(T_0 - r_y) \cdot \delta(r_y, T_0, N_1) + (M - r_z) \cdot \delta(r_z, M, L)}{(T_0 - r_y) + (M - r_z)}. 
\end{align}
\end{proposition} 
Within Definition~\ref{def:balanced},
$|\hsigma^2/\sigma^2 - 1| = \Oc\big( n^{-1} r_n + n^{-1/2}  \sqrt{1 + \log(n)} \big)$. 
Hence, the variance estimator in \eqref{eq:var.sigma.pool} is consistent whenever $r_n = o(n)$. 

\subsection{Estimating the Leading Variance}
Define the scale-free factor $\Upsilon \coloneqq \Vc^2 / \sigma^2$, and its plug-in estimator by 
%
%
\begin{align}
	\hUpsilon &\coloneqq  \| \hbalpha \|_2^2 \cdot \| \hbbeta \|_2^2
	+ \| \hbq_{\beta} \|_2^2  \left(1 + \| \hbbeta \|_2^2 \right)
	+ \| \hbq_{\alpha} \|_2^2  \left(1 + \| \hbalpha \|_2^2 \right),
	\label{eq:var.est.asymp}
\end{align}
where the empirical analogues of the dual weights in \eqref{eq:riesz.q} are 
$\hbq_{\beta} \coloneqq \big(\bY_{\Ic_1, \pre}^{(k_y)} \big)^\dagger \bW \hbalpha$ 
and $\hbq_{\alpha} \coloneqq \big(\bZ_{\lag}^{(k_z)} \big)^\dagger \bW^\top \hbbeta$. 
Using \eqref{eq:var.sigma.pool}, define the leading variance estimator as $\hVc^2 \coloneqq \hsigma^2 \hUpsilon$. 

\begin{proposition} \label{prop:var.asymp}
Let the setup of Theorem~\ref{thm:inference} hold. 
Then, with probability at least $1 - \Oc( \rho )$, $\big| \hUpsilon - \Upsilon \big| \lesssim \Gamma$, 
where $\Gamma$ is a deterministic variance envelope defined below in \eqref{eq:asymp.var.ub}. 
Therefore, if $\Gamma / \Upsilon = o(1)$, then $\hUpsilon / \Upsilon \xrightarrow{p} 1$. 
If, in addition, $\hsigma^2 / \sigma^2 \xrightarrow{p} 1$, then $\hVc / \Vc \xrightarrow{p} 1$.
\end{proposition}

\noindent \textbf{Variance envelope.} 
For $a, b, m, n \in \Rb_+$, define $\Gamma_1(a,b,m,n) \coloneqq a(a+b)(m+n)^2$ and $\Gamma_2(a,b,m,n) \coloneqq a(a+b)\left\{1 + (m+n)^2 \right\}$. 
Set $\Gamma \coloneqq \Gamma_{\alpha \beta} + \Gamma_{\alpha} + \Gamma_{\beta}$, where 
\begin{align}
	\Gamma_\alpha &\coloneqq 
	\Gamma_1\left(\frac{\sigma \Lambda_\alpha}{\sqrt{L}}, \frac{\sqrt{r_z}}{\sqrt{L}}, \frac{\sqrt{r_z}}{\sqrt{M}}, \sigma \Omega_\alpha\right)
	+ \Gamma_2\left(\sigma \Omega_\alpha, \frac{\sqrt{r_z}}{\sqrt{M}}, \frac{\sqrt{r_z}}{\sqrt{L}}, \frac{\sigma \Lambda_\alpha}{\sqrt{L}} \right),
	\\
	\Gamma_\beta &\coloneqq 
	\Gamma_1\left(\frac{\sigma \Lambda_\beta}{\sqrt{N_1}}, \frac{\sqrt{r_y}}{\sqrt{N_1}}, \frac{\sqrt{r_y}}{\sqrt{T_0}}, \sigma \Omega_\beta \right)
	+ \Gamma_2\left(\sigma \Omega_\beta, \frac{\sqrt{r_y}}{\sqrt{T_0}}, \frac{\sqrt{r_y}}{\sqrt{N_1}}, \frac{\sigma \Lambda_\beta}{\sqrt{N_1}} \right),
	\\ 
	\Gamma_{\alpha \beta} &\coloneqq 
	\Gamma_1\left(\frac{\sigma \Lambda_\alpha}{\sqrt{L}}, \frac{\sqrt{r_z}}{\sqrt{L}}, \frac{\sqrt{r_y}}{\sqrt{N_1}}, \frac{\sigma \Lambda_\beta}{\sqrt{N_1}} \right)
	+ \Gamma_1\left(\frac{\sigma \Lambda_\beta}{\sqrt{N_1}}, \frac{\sqrt{r_y}}{\sqrt{N_1}}, \frac{\sqrt{r_z}}{\sqrt{L}}, \frac{\sigma \Lambda_\alpha}{\sqrt{L}} \right).  
	\label{eq:asymp.var.ub}
\end{align} 
%
The term $\Gamma_{\alpha \beta}$ controls the perturbation of the product of the two primal weights. The terms $\Gamma_{\alpha}$ and $\Gamma_{\beta}$ control estimation of the time- and unit-side dual weights. 

\medskip \noindent \textbf{Interpretation.} 
Under Definition~\ref{def:balanced}, suppose additionally that $\| \bq^*_{\beta} \|_2 + \| \bq^*_{\alpha} \|_2 \gtrsim n^{-1/2}$. 
Then, $\Upsilon \gtrsim n^{-1}$, and Proposition~\ref{prop:var.asymp} implies 
$\Gamma / \Upsilon = \Oc \big( n^{-1/2} r_n^{3/2}  (\sqrt{r_n} + \sqrt{1 + \log(n)} ) \big)$. 
When $r_n \gg \log(n)$, this ratio vanishes if $r_n = o(n^{1/4})$; when $r_n \lesssim \log(n)$, it is sufficient that $r_n^3 \log(n) = o(n)$. 
Thus, $\hVc^2$ is consistent under the same broad scaling regimes that make the remainder $\Psi$ in Theorem~\ref{thm:inference} negligible.

\subsection{A Pointwise Confidence Interval}
We assemble the preceding ingredients to yield the studentized central limit theorem. 

\begin{cor} \label{cor:inference} 
Let the setup of Proposition~\ref{prop:var.asymp} hold. 
Thus, if $\Psi / \Vc = o(1)$, $\Gamma / \Upsilon = o(1)$, and $\hsigma^2 / \sigma^2 \xrightarrow{p} 1$, 
then $\big(\htheta - \theta\big) / \hVc \rightsquigarrow \Nc(0,1)$. 
\end{cor}

\begin{proof}
The result follows from Theorem~\ref{thm:inference}, Proposition~\ref{prop:var.asymp}, and Slutsky's theorem.
\end{proof}

For any significance level $\gamma \in (0,1)$, Corollary~\ref{cor:inference} justifies a pointwise $(1-\gamma)100\%$ confidence interval:  
\begin{align}
	\texttt{CI}_\theta(\gamma) \coloneqq \left[ \htheta \pm \Phi^{-1}(1-\gamma/2) \cdot \hVc \right],
	\label{eq:conf.iv} 
\end{align} 
where $\Phi$ is the standard normal distribution function. 
This provides pointwise inference for a prospective causal forecast, complementing the growing literature on inference in panel-data and \SC~settings \citep{victor18, victor21, li2020, Choi20092024, yan2024entrywiseinferencemissingpanel}.

\section{Forecasting over a Fixed Multi-Step Horizon} \label{sec:forecast.horizon}
We now extend \TWSF~from a one-step causal forecast to a prespecified linear summary of a finite future path. 
Fix an ambient horizon $h \ge 1$ and a nonzero deterministic weight vector $\boldeta = (\eta_1, \dots, \eta_h )^\top \in \Rb^h$. 
Define its support and effective horizon by $\Sc_\eta \coloneqq \{\ell \in [h]: \eta_\ell \neq 0\}$ and $h_{\eta} \coloneqq \max \Sc_\eta$. 
Let $\Fc_{\eta} \coloneqq \sigma(\bu_N, \{\bv_{T+\ell}(1): \ell \in \Sc_\eta \})$. 
The multi-step estimand is 
\begin{align}
	\theta_{\eta} \coloneqq \Ex\left[ \sum_{\ell \in \Sc_\eta} \eta_\ell \cdot Y_{N, T+\ell}(1) \mid \Fc_{\eta} \right]. 
	\label{eq:estimand.horizon}
\end{align}
Several familiar targets are special cases. Taking $\boldeta = \bolde_\ell$ gives the causal forecast at lead $\ell$. 
Taking $\boldeta = h^{-1}\boldsymbol{1}$ gives the expected average outcome over the forecast horizon.  
More generally, cumulative, discounted, and contrast-based summaries follow from other choices of $\boldeta$. 

Define $\Ec_{\eta} \coloneqq \sigma(\Ec \cup \Fc_{\eta})$ as the smallest sigma-field containing the original latent environment and the future treated time factors relevant to $\boldeta$. 
Throughout this section, guarantees are interpreted conditional on $\Ec_{\eta}$. 
%
%
Below, we consider two extensions of one-step \TWSF. 

\renewcommand{\theassumption}{D\arabic{assumption}}
\renewcommand{\thecor}{D\arabic{cor}}
\renewcommand{\thetheorem}{D\arabic{theorem}}
\setcounter{assumption}{0}
\setcounter{cor}{0}
\setcounter{theorem}{0}

\subsection{Direct Estimation} \label{sec:direct}
Direct \TWSF~treats each forecast lead as a separate regression problem. 
The appeal of this strategy is that it avoids compounding one-step forecast errors. Its cost is that each training block must be longer, so fewer non-overlapping Page blocks are available for estimation.

\subsubsection{Direct \TWSF~Estimator} \label{sec:alg.direct} 
Fix a lag length $L$, and define $B_{\eta} \coloneqq T_1 / (L + h_\eta) \in \Nb$ and $M_{\eta} \coloneqq (B_{\eta} - 1)N_1$.
%
Using the first $B_{\eta}-1$ post-treatment blocks, construct the stacked direct Page matrix 
$[\bP^{\direct}_\eta]_{a, (j,b)} \coloneqq Y_{j, T_0+(b-1)(L+h_\eta)+a}$ for $a \in [L+h_\eta]$ and $(j,b) \in \Ic_1 \times [B_{\eta}-1]$. 
We partition $\bP^{\direct}_\eta$ analogous to \eqref{eq:Zlag.znext}: 
\begin{align}
	\bP_{\eta}^{\direct} 
	=  
	\begin{pmatrix}
		\bZ_{\direct\text{-}\lag}
		\\
		\hdashline
		\bz_{\direct\text{-}\tnext, 1}^\top
		\\
		\vdots
		\\
		\bz_{\direct\text{-}\tnext, h_{\eta}}^\top
	\end{pmatrix}
	\in \Rb^{(L+h_\eta) \times M_{\eta}}. 
\end{align}
Here, 
$\bZ_{\direct\text{-}\lag} \in \Rb^{L \times M_{\eta}}$ contains the first $L$ rows, while $\bz_{\direct\text{-}\tnext, \ell} \in \Rb^{M_{\eta}}$ collects the $\ell$-step-ahead responses across all donor units and training blocks. 
The terminal donor histories used at forecast time remain $\bW$. 

\begin{enumerate} [label=(\alph*)]

	\item Define $\bY_{\Ic_1, \pre}^{(k_y)}$ by \eqref{eq:hsvt.YZ}. For $k_z \le \min\{L, M_{\eta}\}$, define 
	$\bZ_{\direct\text{-}\lag}^{(k_z)} \coloneqq \HSVT \left(\bZ_{\direct\text{-}\lag}, k_z\right)$. 
	
	\item Define $\hbbeta$ by \eqref{eq:pcr.beta.hat}. For each $\ell \in \Sc_\eta$, define 
	$\hbalpha^\direct_\ell \coloneqq (\bZ_{\direct\text{-}\lag}^{(k_z)})^{\top, \dagger} \cdot \bz_{\direct\text{-}\tnext, \ell}$
	%
	and aggregate these coefficients as $\hbalpha^\direct_{\eta} \coloneqq \sum_{\ell \in \Sc_\eta} \eta_\ell \hbalpha^\direct_\ell$.
	
	\item Define the direct point estimate as $\htheta^\direct_{\eta} \coloneqq \left\langle \hbalpha^\direct_{\eta}, \bW^\top \hbbeta \right\rangle$. 
	
	\item Define the leading variance estimate as $(\hVc_{\eta}^\direct)^2 \coloneqq \hsigma^2 \hUpsilon^\direct$, where
	\begin{align}
		\hUpsilon^\direct &\coloneqq \| \hbalpha^\direct_{\eta} \|_2^2 \| \hbbeta \|_2^2
		+ \| \hbq^\direct_{\beta \eta} \|_2^2 \left(1 + \| \hbbeta \|_2^2 \right)
		+ \| \hbq^\direct_{\alpha} \|_2^2 \left(\| \boldeta \|_2^2 + \| \hbalpha^\direct_{\eta} \|_2^2 \right),
		\label{eq:direct.var.est}
		\\
		\hbq^\direct_{\beta \eta} &\coloneqq \left(\bY_{\Ic_1, \pre}^{(k_y)} \right)^\dagger \bW \hbalpha^\direct_{\eta},
		\quad
		\hbq^\direct_{\alpha} \coloneqq \left(\bZ_{\direct\text{-}\lag}^{(k_z)} \right)^\dagger \bW^\top  \hbbeta.
		\label{eq:riesz.direct}
	\end{align}
	%

\end{enumerate}
%
The vector $\bW^\top \hbbeta$ remains the imputed treated history of unit $N$ over the final $L$ observed periods. 
Each $\hbalpha^\direct_\ell$ maps this common history directly to lead $\ell$, while $\hbalpha^\direct_{\eta}$ maps it directly to the selected linear summary. 
Thus, direct \TWSF~estimates one temporal rule for each active lead in $\Sc_\eta$ using a common set of unit weights.

\subsubsection{Formal Results for Direct \TWSF} \label{sec:results.direct}

\begin{theorem} \label{thm:direct.normality}
Let the setup of Theorem~\ref{thm:inference} hold  with $(\bZ_{\edirect\text{-}\elag}, \bz_{\edirect\text{-}\etnext, \ell})$ for $\ell \in \Sc_\eta$ replacing $(\bZ_{\elag}, \bz_\etnext)$.
Let $r^\edirect_z \coloneqq \rank(\Ex[\bZ_{\edirect\text{-}\elag} \mid \Ec_{\eta}])$, and define $(\rho_{\eta}, \Psi_{\eta}, \Gamma_{\eta})$ from 
$(\rho, \Psi, \Gamma)$ by replacing $(M, r_z)$ with $(M_{\eta}, r^\edirect_z)$. 
Then, $\htheta^\edirect_{\eta} - \theta_{\eta} = G^\edirect_{\eta} + R^\edirect_{\eta}$, 
with $G^\edirect_{\eta} \sim \mathcal{N}\big(0, (\Vc^\edirect_{\eta})^2\big)$; 
moreover, with probability at least $1 - \Oc(\rho_{\eta})$, 
$\big| R^\edirect_{\eta} \big| \lesssim \| \boldeta \|_1 \cdot \Psi_{\eta}$
and
$\big| \hUpsilon^\edirect - \Upsilon^\edirect \big| \lesssim \| \boldeta \|_1^2 \Gamma_{\eta}$, where $\Upsilon^\edirect \coloneqq (\Vc^\edirect_{\eta})^2 / \sigma^2$. 
Thus, if $\| \boldeta \|_1 \Psi_{\eta} / \Vc^\edirect_{\eta} = o(1)$, $\| \boldeta \|_1^2 \Gamma_{\eta} / \Upsilon^\edirect = o(1)$, and $\hsigma^2 / \sigma^2 \xrightarrow{p} 1$, then 
$\big(\htheta^\edirect_{\eta} - \theta_{\eta}\big) / \hVc^\edirect_{\eta} \rightsquigarrow \Nc(0,1)$. 
\end{theorem}
The population standard deviation $\Vc^\direct_{\eta}$ is defined in Section~\ref{sec:proof.normality.direct} (Supplementary Material). 
Each active lead gives rise to a temporal regression with the same structure as one-step regression. The principal difference is the reduced temporal sample size $M_{\eta}$. 
Because $\boldeta$ is fixed, the balanced-regime conditions that justify one-step inference continue to suffice.

%
\renewcommand{\theassumption}{R\arabic{assumption}}
\renewcommand{\thecor}{R\arabic{cor}}
\renewcommand{\thetheorem}{R\arabic{theorem}}
\setcounter{assumption}{0}
\setcounter{cor}{0}
\setcounter{theorem}{0}
\subsection{Recursive Estimation} \label{sec:recursive}
Recursive \TWSF~estimates the one-step temporal rule once and iterates it through the forecast horizon. It preserves the temporal sample size of the one-step estimator, but estimation error may propagate through the recursion.


\subsubsection{Recursive \TWSF~Estimator} \label{sec:alg.recursive}
For any vector $\bx \in \Rb^L$, define the companion map 
\begin{align}
	\bPi(\bx) &\coloneqq \begin{pmatrix}
	\bzero_{(L-1) \times 1} & \bI_{L-1}
	\\
	\bx^\top
	\end{pmatrix} \in \Rb^{L \times L}. 
	\label{eq:pi.matrix}
\end{align}
For $\ell \ge 1$, define $\bg_\ell(\bx) \coloneqq \{ \bPi(\bx)^\ell \}^\top \be_L$, which maps the initial $L$-lag state directly to its forecast at lead $\ell$. Let $\bJ_\ell(\bx)$ denote its Jacobian, which admits the representation
$$
	\bJ_\ell(\bx) = \sum_{a=0}^{\ell - 1} \big\{\be_L^\top  \bPi(\bx)^a  \be_L \big\} \cdot \big\{\bPi(\bx)^{\ell - 1 - a} \big\}^\top. 
$$
Define $\bJ_{\eta}(\bx) \coloneqq \sum_{\ell \in \Sc_\eta} \eta_\ell \bJ_\ell(\bx)$. 

\begin{enumerate} [label=(\alph*)]

	\item Define $\bY_{\Ic_1, \pre}^{(k_y)}$ and $\bZ_{\lag}^{(k_z)}$ by \eqref{eq:hsvt.YZ}. 
	
	\item Define $\hbbeta$ and $\hbalpha$ by \eqref{eq:pcr.beta.hat} and \eqref{eq:pcr.alpha.hat}. For each $\ell \in \Sc_\eta$, define $\hbalpha^\rec_\ell \coloneqq \bg_\ell(\hbalpha)$, and aggregate these coefficients as $\hbalpha^\rec_{\eta} \coloneqq \sum_{\ell \in \Sc_\eta} \eta_\ell \hbalpha^\rec_\ell$. 
	
	\item Define the recursive point estimate as $\htheta_{\eta}^\rec \coloneqq \left\langle \hbalpha^\rec_{\eta}, \bW^\top \hbbeta \right\rangle$. 
	
	\item Define the leading variance estimate as $(\hVc_{\eta}^\rec)^2 \coloneqq \hsigma^2 \hUpsilon^\rec$, where
	\begin{align}
		\hUpsilon^\rec &\coloneqq 
		\| \hbalpha^\rec_{\eta} \|_2^2  \cdot \| \hbbeta \|_2^2 
		+ \| \hbq^\rec_{\beta \eta} \|_2^2 \left(1 + \| \hbbeta \|_2^2 \right)
		+ \| \hbq^\rec_{\alpha \eta} \|_2^2 \left(1 + \| \hbalpha \|_2^2 \right),
		\label{eq:rec.var.est} 
		\\
		\hbq^\rec_{\beta \eta} &\coloneqq \left(\bY_{\Ic_1, \pre}^{(k_y)} \right)^\dagger \bW \hbalpha^\rec_{\eta},
		\quad
		\hbq^\rec_{\alpha \eta} \coloneqq \left(\bZ_{\lag}^{(k_z)} \right)^\dagger \bJ_{\eta}(\hbalpha)^\top \bW^\top \hbbeta. 
	\end{align} 
	%

\end{enumerate}
The vector $\bW^\top \hbbeta$ is again the imputed treated $L$-lag state of the target unit at time $T$. 
The map $\hbalpha^\rec_{\eta}$ advances this state recursively and returns the selected linear summary of the resulting forecasts. Its Jacobian measures how a first-order perturbation in the estimated one-step rule propagates into that summary. 

\subsubsection{Formal Results for Recursive \TWSF} \label{sec:results.recursive}
The recursive method requires a strengthened recoverability condition. 

\begin{assumption} [Recursive recovery] \label{assump:recursive} 
Let $\ecol(\Ex[\bW \mid \Ec_{\eta}]) \subseteq \ecol(\Ex[\bY_{\Ic_1, \epre} \mid \Ec_{\eta}])$. 
Further, for any $\balpha$ satisfying Proposition~\ref{prop:alpha} and each $\ell \in \{0, \dots, h_{\eta}-1\}$, let $\erow(\Ex[\bW \mid \Ec_{\eta}] \cdot \{\bPi(\balpha)^\ell \}^\top) \subseteq \erow(\Ex[\bZ_{\elag}^\top \mid \Ec_{\eta}])$.
\end{assumption}
When $h_{\eta}=1$, Assumption~\ref{assump:recursive} reduces to Assumption~\ref{assump:transport} used in the one-step analysis. For $h_{\eta} > 1$, it requires not only the terminal lag block but also its recursively shifted versions to lie in the temporal training span. The condition therefore becomes more demanding as the horizon grows and is most plausible for short or moderate horizons. %

To quantify the recursive error amplification, define $\mathfrak{C}_\eta \coloneqq C_{h_{\eta}} \big(1 + \sigma \Lambda_{\alpha}/ \sqrt{L} \big)^{h_{\eta}-1}$, where $C_{h_{\eta}} \ge 1$ depends only on $h_{\eta}$ and the constants within Assumption~\ref{assump:spectra}. 

\begin{theorem} \label{thm:recursive.normality}
Let the setup of Theorem~\ref{thm:inference} hold 
with Assumption~\ref{assump:recursive} replacing Assumption~\ref{assump:transport}. 
Then, $\htheta^\erec_{\eta} - \theta_{\eta} = G^\erec_{\eta} + R^\erec_{\eta}$, 
with $G^\erec_{\eta} \sim \mathcal{N}\big(0, (\Vc^\erec_{\eta})^2\big)$; moreover, with probability at least $1 - \Oc(\rho)$, 
$\big| R^\erec_{\eta} \big| \lesssim \mathfrak{C}_\eta \| \boldeta \|_1 \Psi$ 
and 
$\big| \hUpsilon^\erec - \Upsilon^\erec \big| \lesssim \mathfrak{C}_\eta^2 \| \boldeta \|_1^2 \Gamma$, where $\Upsilon^\erec \coloneqq (\Vc^\erec_{\eta})^2/ \sigma^2$. 
Thus, if $\mathfrak{C}_\eta \| \boldeta \|_1 \Psi / \Vc^\erec_{\eta} = o(1)$, $\mathfrak{C}_\eta^2 \| \boldeta \|_1^2 \Gamma / \Upsilon^\erec = o(1)$, and $\hsigma^2 / \sigma^2 \xrightarrow{p} 1$, then 
$\big(\htheta^\erec_{\eta} - \theta_{\eta}\big) / \hVc^\erec_{\eta} \rightsquigarrow \Nc(0,1)$. 
\end{theorem}
The population standard deviation $\Vc^\rec_{\eta}$ is defined in Section~\ref{sec:proofs.recursive} (Supplementary Material). 
Relative to the one-step result, the primary novelty is the multiplicative factor $\mathfrak{C}_\eta$. 
Since this factor remains controlled for fixed $h_\eta$, the conditions that justify one-step inference continue to justify recursive multi-step inference.
%


\subsection{Takeaways} \label{sec:h.summary}
Theorems~\ref{thm:direct.normality} and \ref{thm:recursive.normality} formalize the trade-off between the two approaches. 
Direct \TWSF~avoids recursive error propagation but its guarantees depend on the reduced temporal sample size $M_{\eta}$. 
Recursive \TWSF~retains the one-step temporal sample size but its error bounds carry the amplification factor $\mathfrak{C}_\eta$. 
At the same time, for $m \in \{\direct, \rec\}$ and $\gamma \in (0,1)$, Theorems~\ref{thm:direct.normality} and \ref{thm:recursive.normality} justify the $(1-\gamma)100\%$ pointwise confidence interval 
\begin{align}
	\texttt{CI}^{m}_{\theta_{\eta}}(\gamma) \coloneqq \left[ \htheta^{m}_{\eta} \pm \Phi^{-1}(1-\gamma/2) \cdot \hVc^{m}_{\eta} \right].
	\label{eq:ci.horizon}
\end{align} 
The interval covers the prespecified scalar summary and does not provide simultaneous coverage of the full forecast path. 

\begin{remark}[Prospective treatment effects]
A prospective treatment effect is obtained by combining the treated-state forecast with a separate forecast of the target unit’s control-state path. 
Appendix~\ref{sec:tau} gives the formal construction and conditions. 
%
%
%
\end{remark}


%

\section{Simulation Studies} \label{sec:simulations} 
We conduct simulations to evaluate the finite-sample accuracy and pointwise confidence interval coverage of direct and recursive \TWSF~under the balanced, fixed-rank regime in Definition~\ref{def:balanced}, which motivates the theoretical comparisons in Sections~\ref{sec:results}-\ref{sec:forecast.horizon}. 

\subsection{Simulation Design} \label{sec:sims.setup}

%
Consider horizons $h \in \{1, 5, 10\}$ and index simulations by $n \in \{25, 50, 75, 100, 125, 150\}$. 
%
%
For each $n$, let $N_1=T_0=L=n$ and $T_1(n) \coloneqq 8(n+10)$. 
At each design point, generate 100 independent latent panels and 10 independent noise realizations conditional on each panel, yielding 1000 replications. 
Standard errors are clustered by latent panel.

%
The target unit factor is a convex combination of eight donor factors. 
The control- and treated-state time factors are generated from a fixed collection of harmonic components, including a low-frequency component that appears only under treatment. 
We scale the signals so that every absolute conditional mean is bounded by $0.8$, then add independent Gaussian noise with standard deviation $\sigma = 0.1$.
The population design ranks are $r_y=4$ and $r_z=8$. 
We supply the estimators with the oracle ranks and lag length so that the experiment isolates the finite-sample behavior described by the theory. 
Appendix~\ref{sec:sims.design} gives the full harmonic specification, nested-panel construction, and implementation details.


\subsection{Simulation Results} \label{sec:sims.results}
For each horizon, the estimand is the average treated outcome over the forecast period, corresponding to $\boldeta(h) = h^{-1} \boldsymbol{1}$. 
Panels~\ref{fig:bias.1}--\ref{fig:rmse.10} report the bias and root mean squared error (RMSE) results; Appendix~\ref{sec:sims.metrics} provides formal definitions of these performance measures. 
At $h=1$, the two estimators coincide exactly. 
The one-step estimator exhibits a modest negative bias over much of the grid, with magnitude below approximately \(0.01\). 
At $h \in \{5,10\}$, both estimators have positive bias when $n=25$; for $n \ge 50$, their biases fluctuate near zero and are small relative to their RMSEs, which decline steadily as the panel grows.

Panels~\ref{fig:cov.1}--\ref{fig:len.10} report pointwise coverage and average length at the 90\% nominal level. 
The one-step interval mildly undercovers across the grid. 
At longer horizons, undercoverage is concentrated in the smaller panels and is more pronounced for direct \TWSF. 
Recursive intervals are generally closer to nominal coverage but are somewhat wider. 
Coverage improves and average interval length decreases as the panel grows, and the differences between the two procedures become modest in the larger panels.

Overall, the results accord with the theory: estimation error and interval length decline with the panel dimensions, while pointwise coverage generally moves toward the nominal level. 
Our findings also illustrate the practical trade-off between the two procedures. 
Direct estimation is more fragile when the panel is small and the horizon is long; 
recursive estimation uses the temporal data more parsimoniously but pays a modest variance cost. 
Because horizon averages may conceal poor performance at the most distant lead, Appendix~\ref{sec:sims.results.terminal} also considers the terminal-lead estimand for $h \in \{5,10\}$ using $\boldeta(h) = \be_h$.  

\begin{figure}[!tp]
	\centering 
	\begin{subfigure}[b]{0.32\textwidth}
		\centering 
		\includegraphics[width=\linewidth]{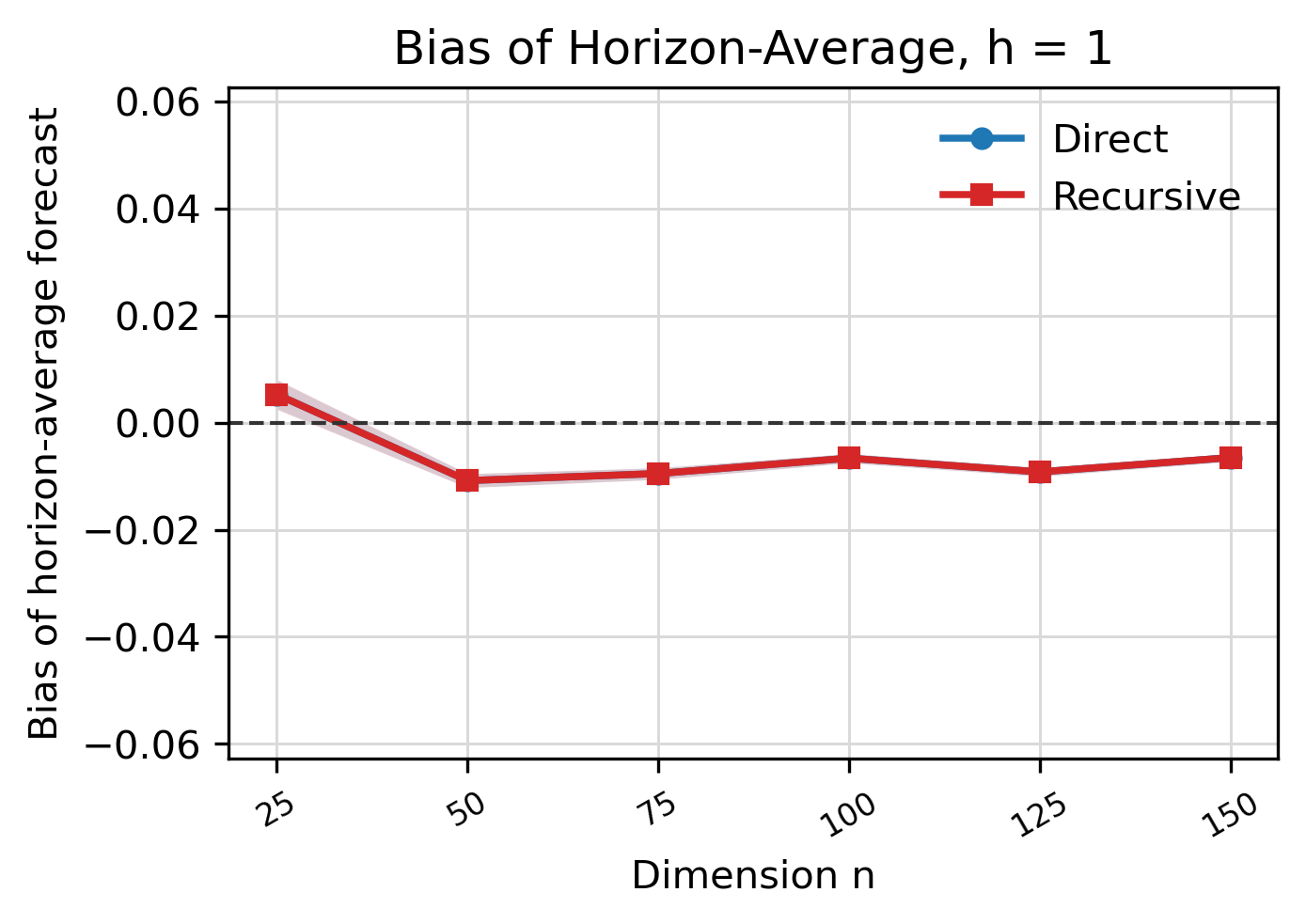}
		\caption{Bias $h=1$}
		\label{fig:bias.1}
	\end{subfigure} 
	\begin{subfigure}[b]{0.32\textwidth}
		\centering 
		\includegraphics[width=\linewidth]{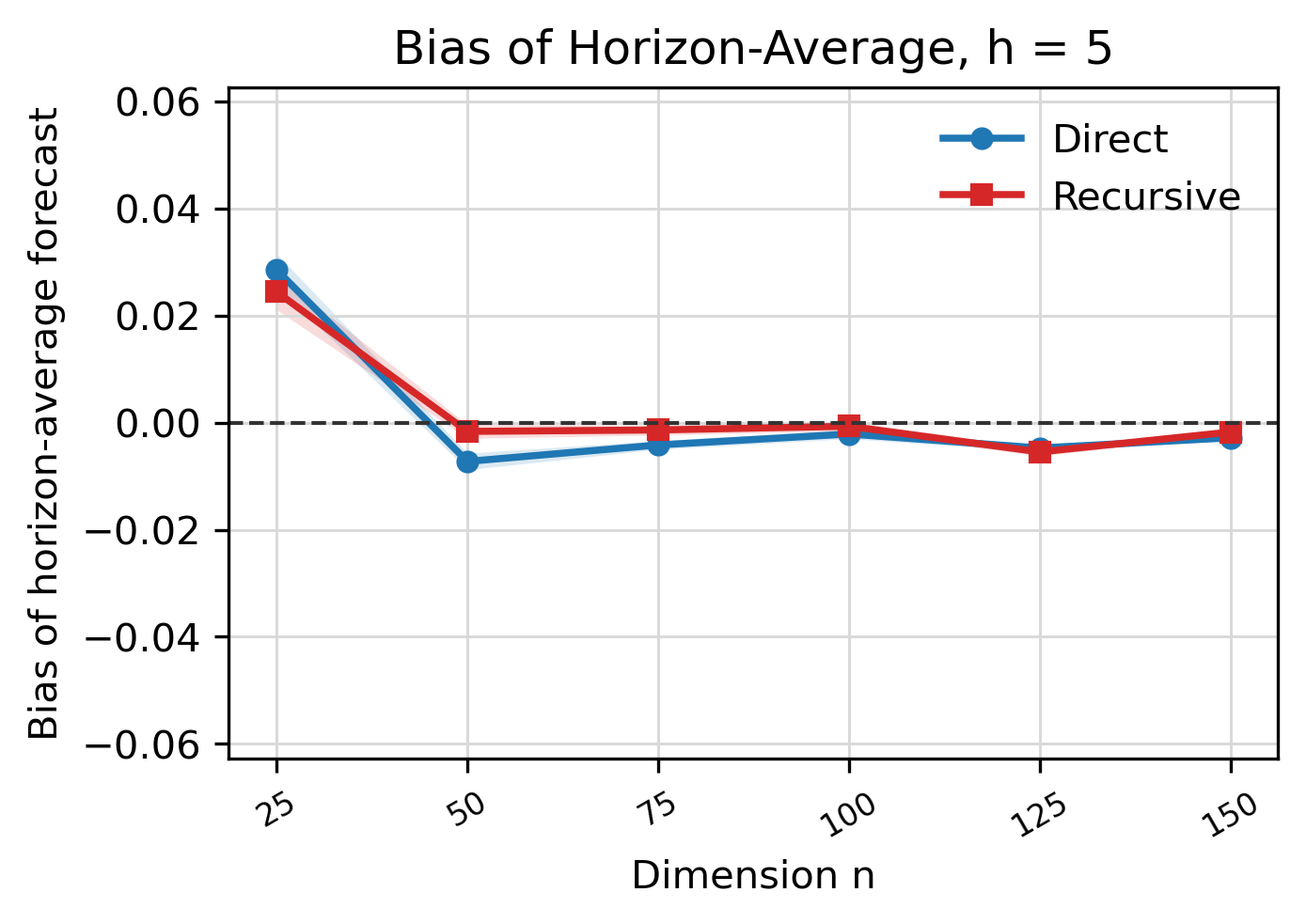}
		\caption{Bias $h=5$} 
		\label{fig:bias.5}
	\end{subfigure} 
	\begin{subfigure}[b]{0.32\textwidth}
		\centering 
		\includegraphics[width=\linewidth]{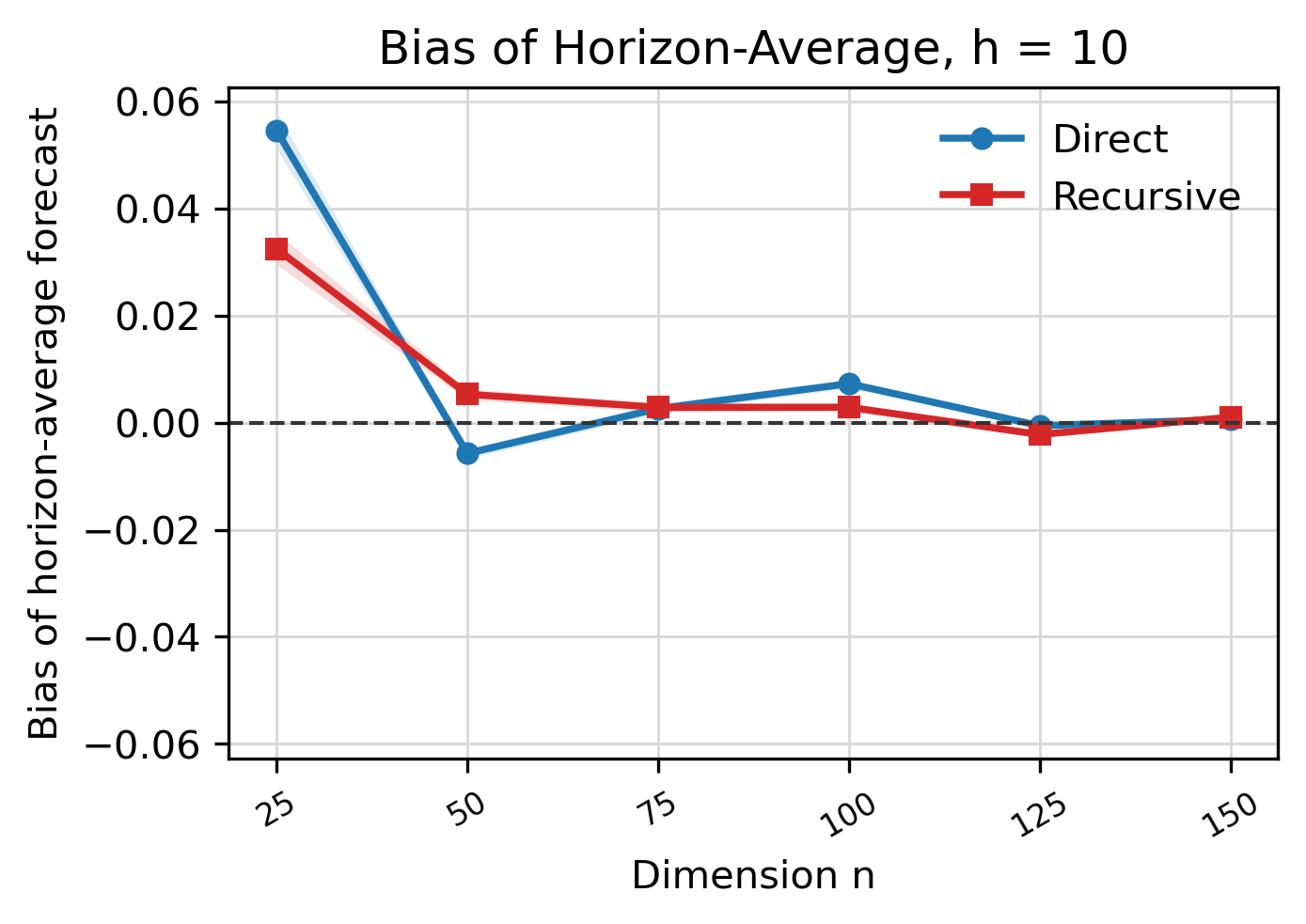}
		\caption{Bias $h=10$} 
		\label{fig:bias.10}
	\end{subfigure}
	\\
	\begin{subfigure}[b]{0.32\textwidth}
		\centering 
		\includegraphics[width=\linewidth]{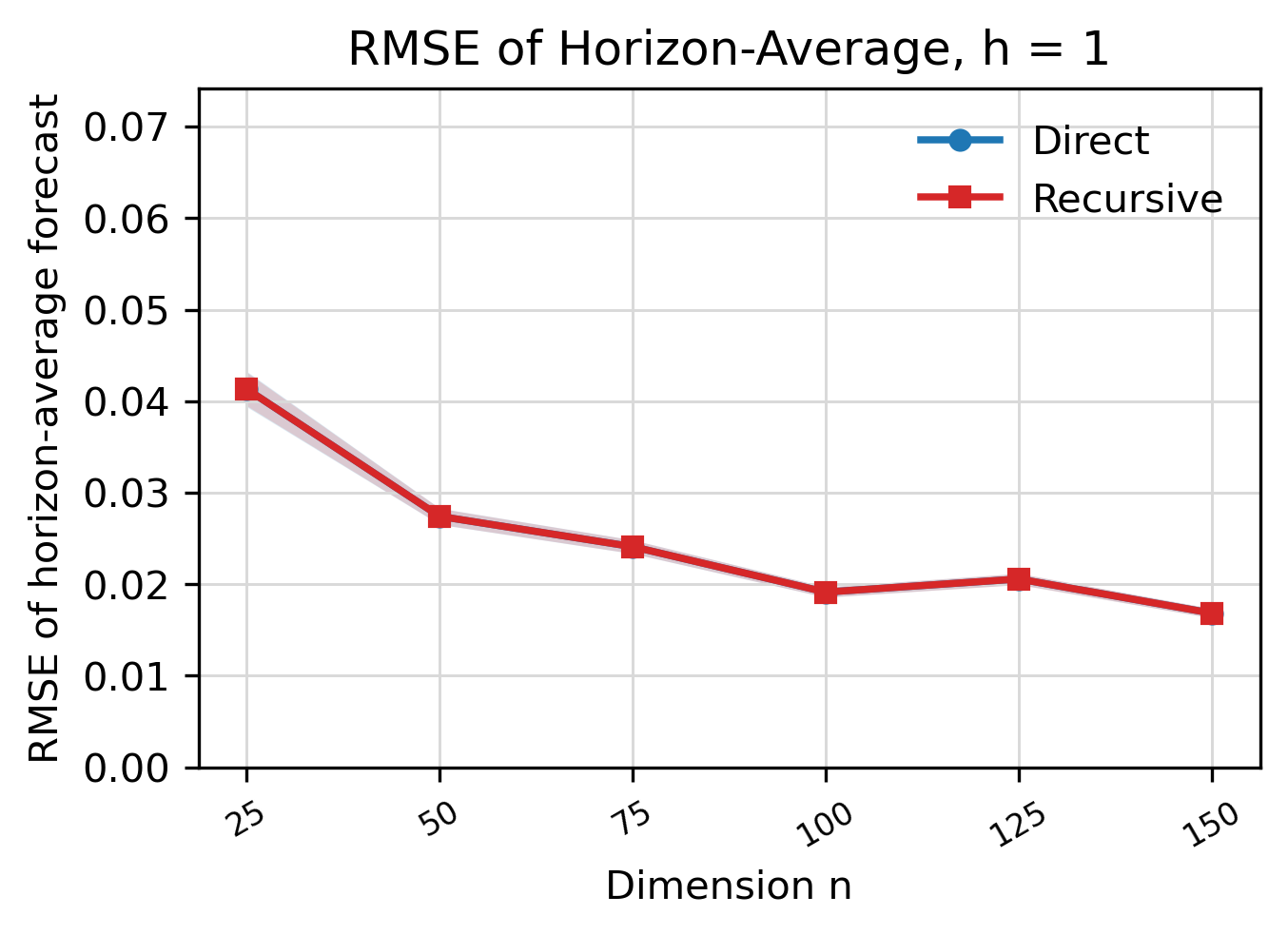}
		\caption{RMSE $h=1$} 
		\label{fig:rmse.1}
	\end{subfigure} 
	\begin{subfigure}[b]{0.32\textwidth}
		\centering 
		\includegraphics[width=\linewidth]{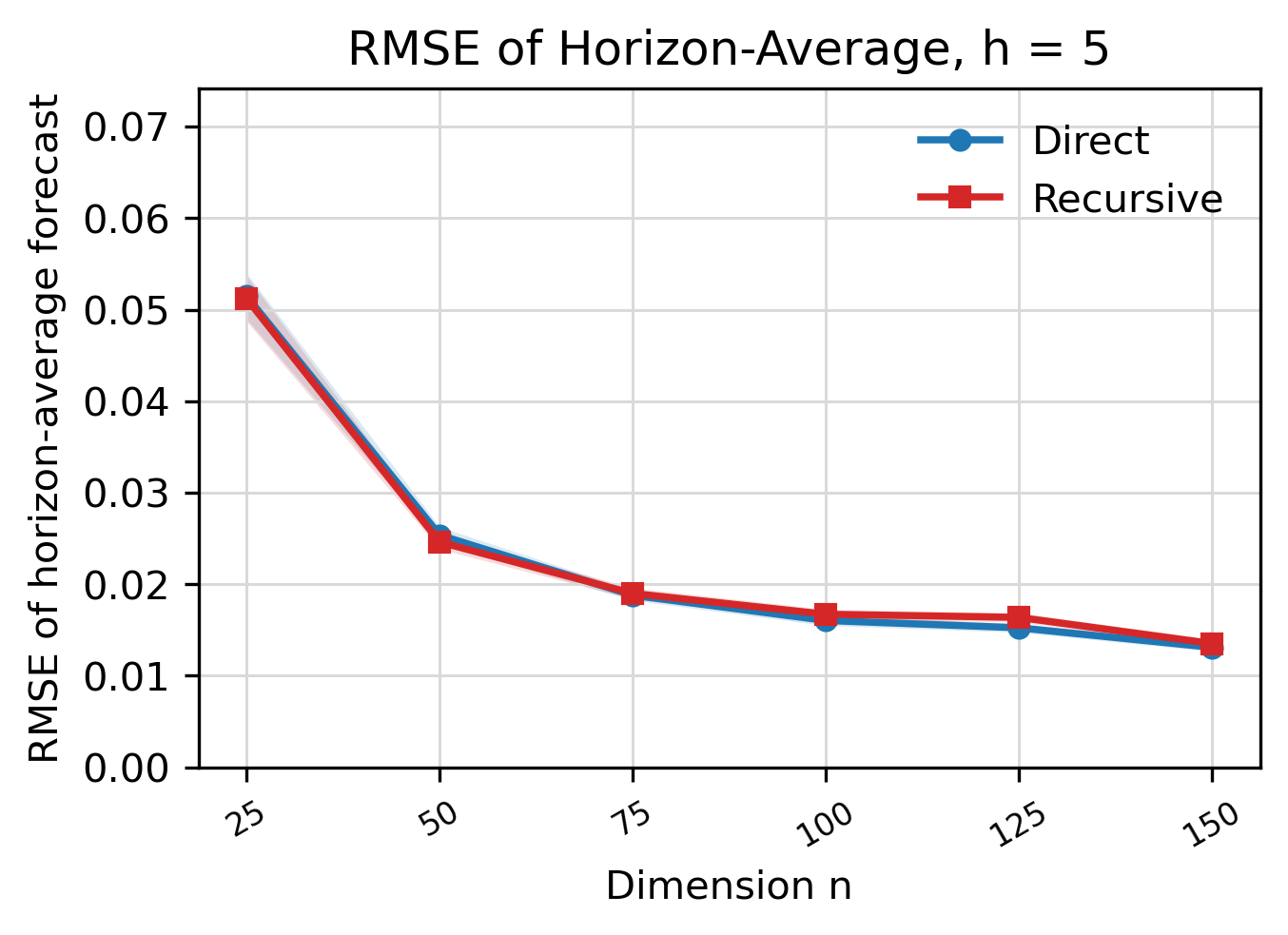}
		\caption{RMSE $h=5$} 
		\label{fig:rmse.5}
	\end{subfigure} 
	\begin{subfigure}[b]{0.32\textwidth}
		\centering 
		\includegraphics[width=\linewidth]{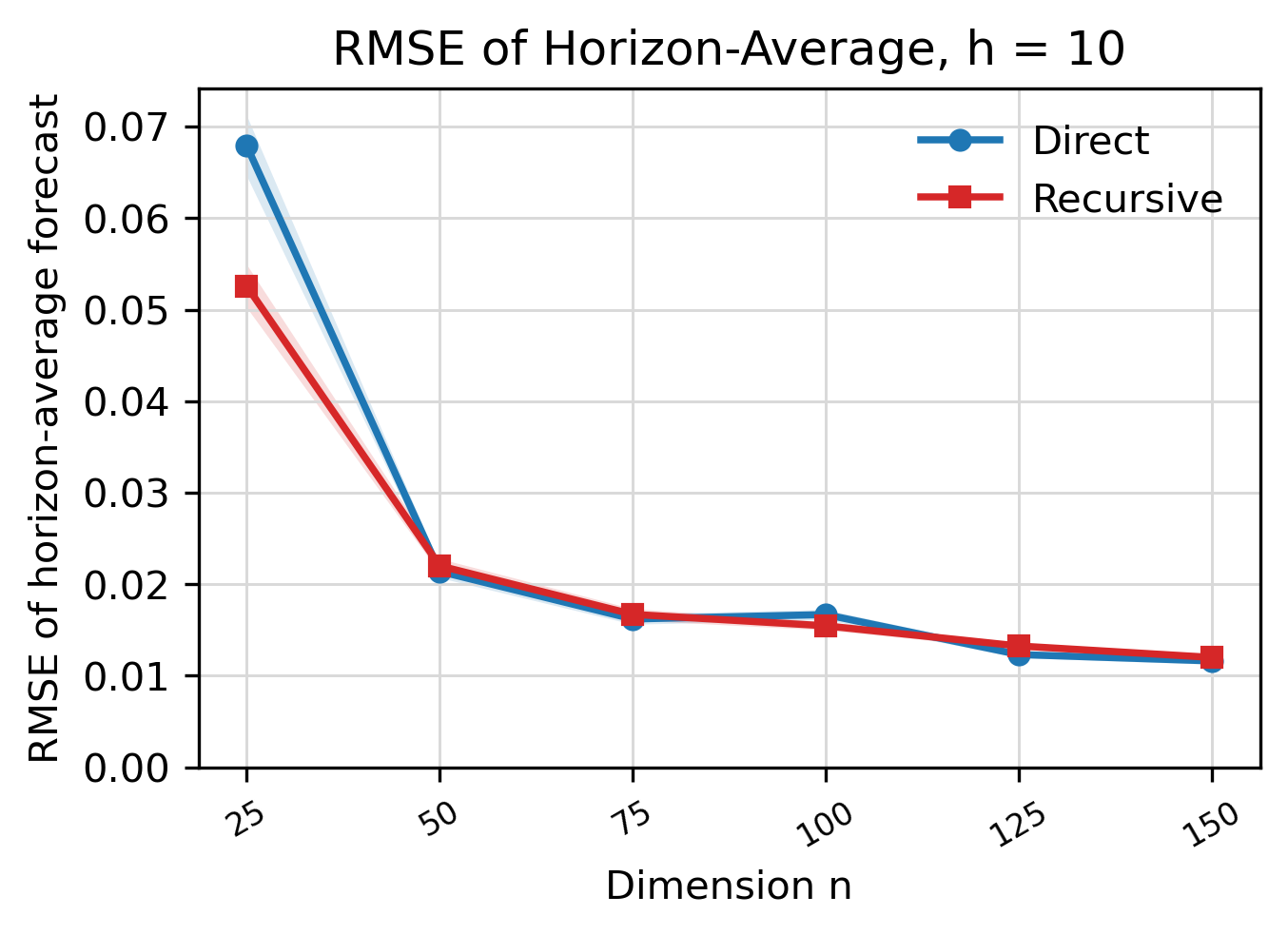}
		\caption{RMSE $h=10$} 
		\label{fig:rmse.10}
	\end{subfigure} 
	\\
	\begin{subfigure}[b]{0.32\textwidth}
		\centering 
		\includegraphics[width=\linewidth]{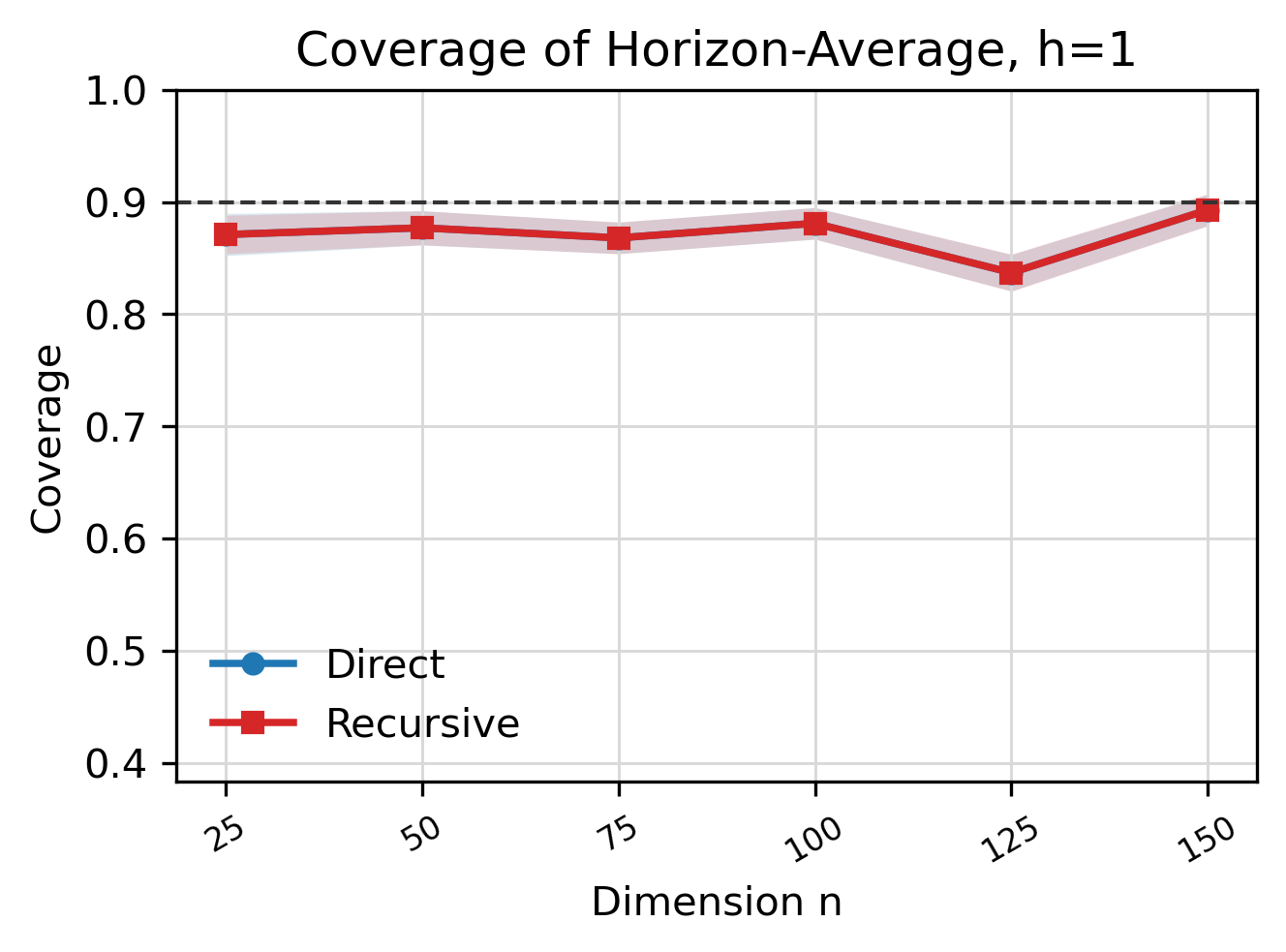}
		\caption{Coverage $h=1$} 
		\label{fig:cov.1}
	\end{subfigure} 
	\begin{subfigure}[b]{0.32\textwidth}
		\centering 
		\includegraphics[width=\linewidth]{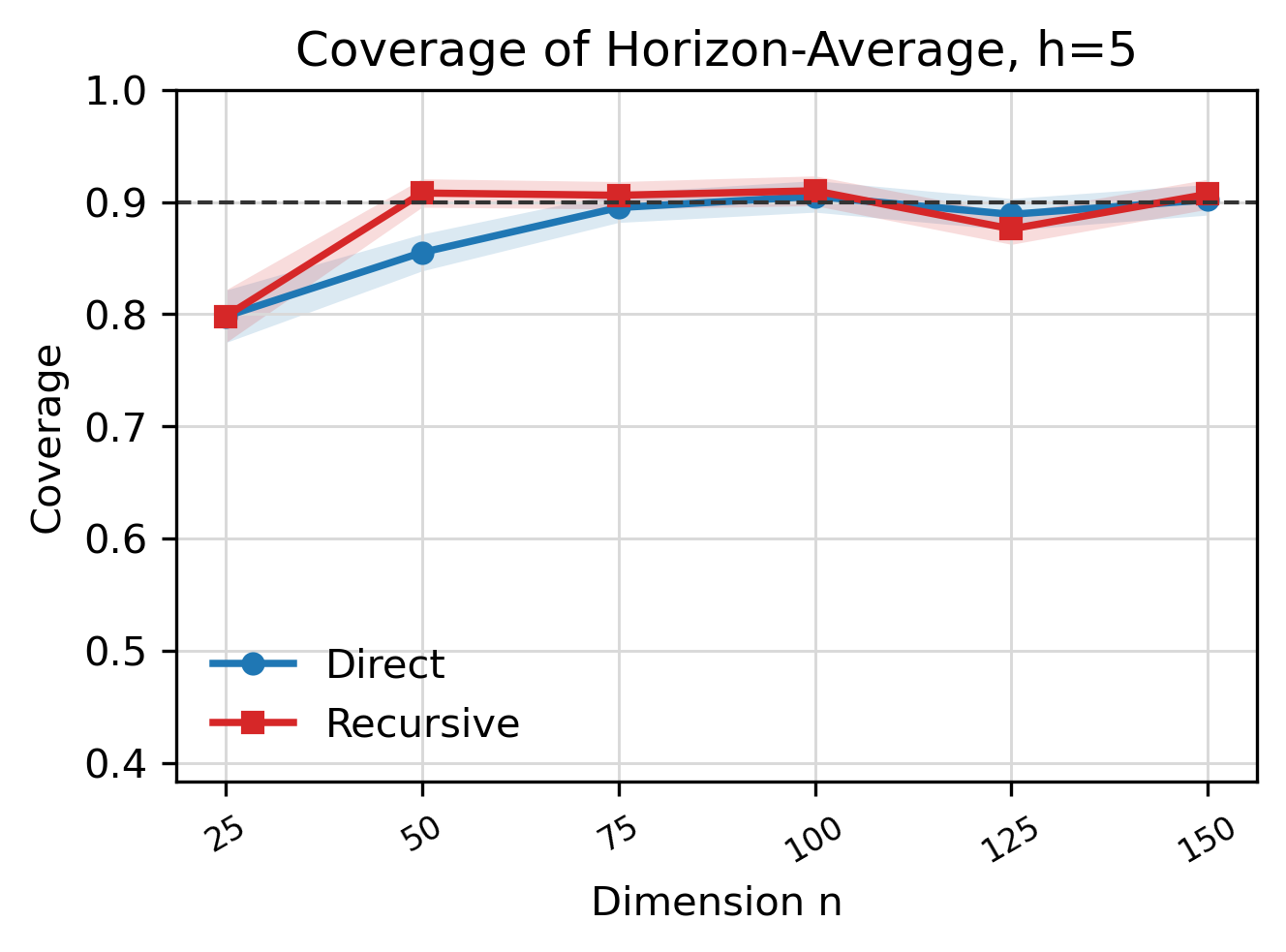}
		\caption{Coverage $h=5$} 
		\label{fig:cov.5}
	\end{subfigure} 
	\begin{subfigure}[b]{0.32\textwidth}
		\centering 
		\includegraphics[width=\linewidth]{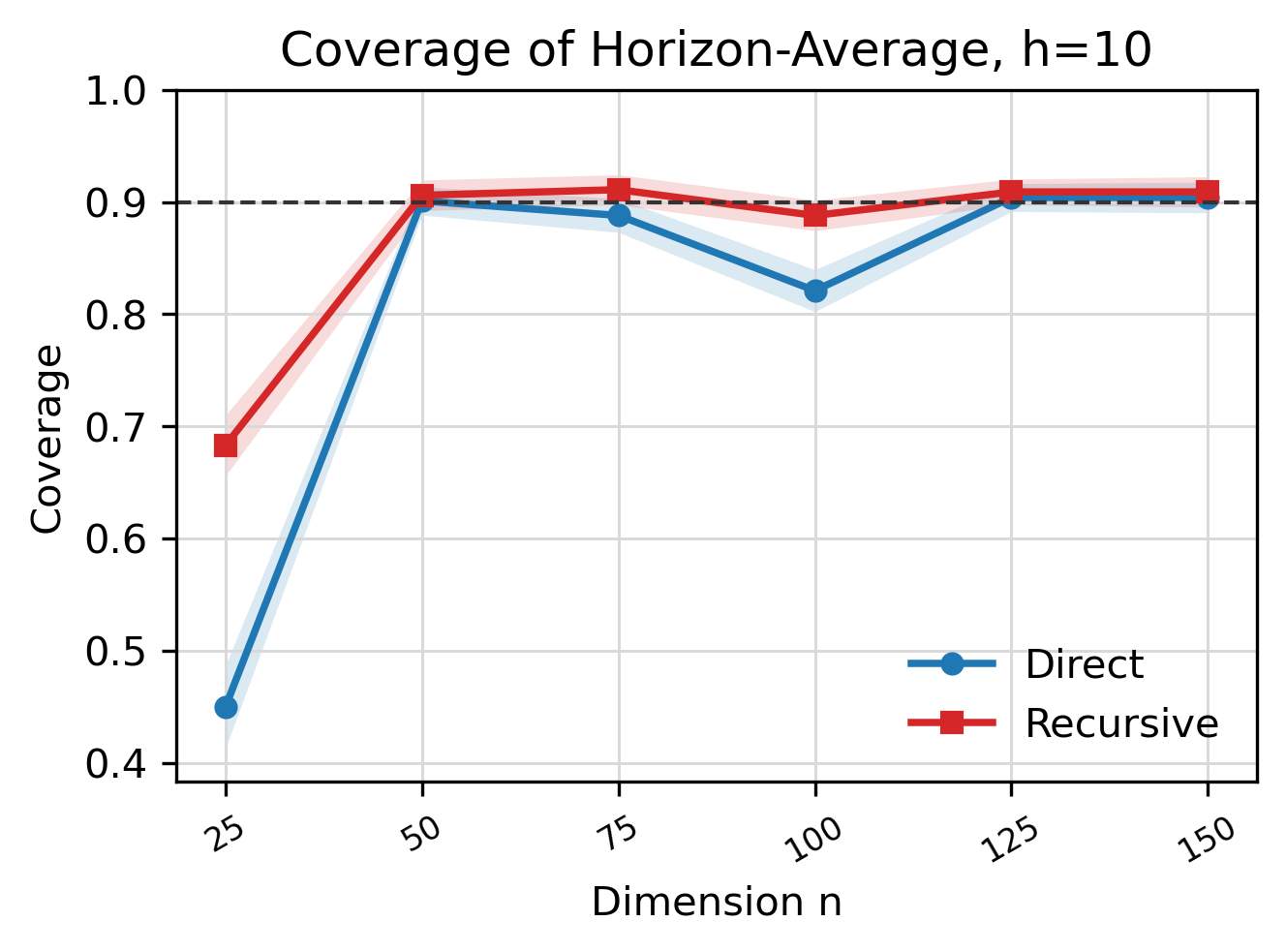}
		\caption{Coverage $h=10$} 
		\label{fig:cov.10}
	\end{subfigure}
	\\
	\begin{subfigure}[b]{0.32\textwidth}
		\centering 
		\includegraphics[width=\linewidth]{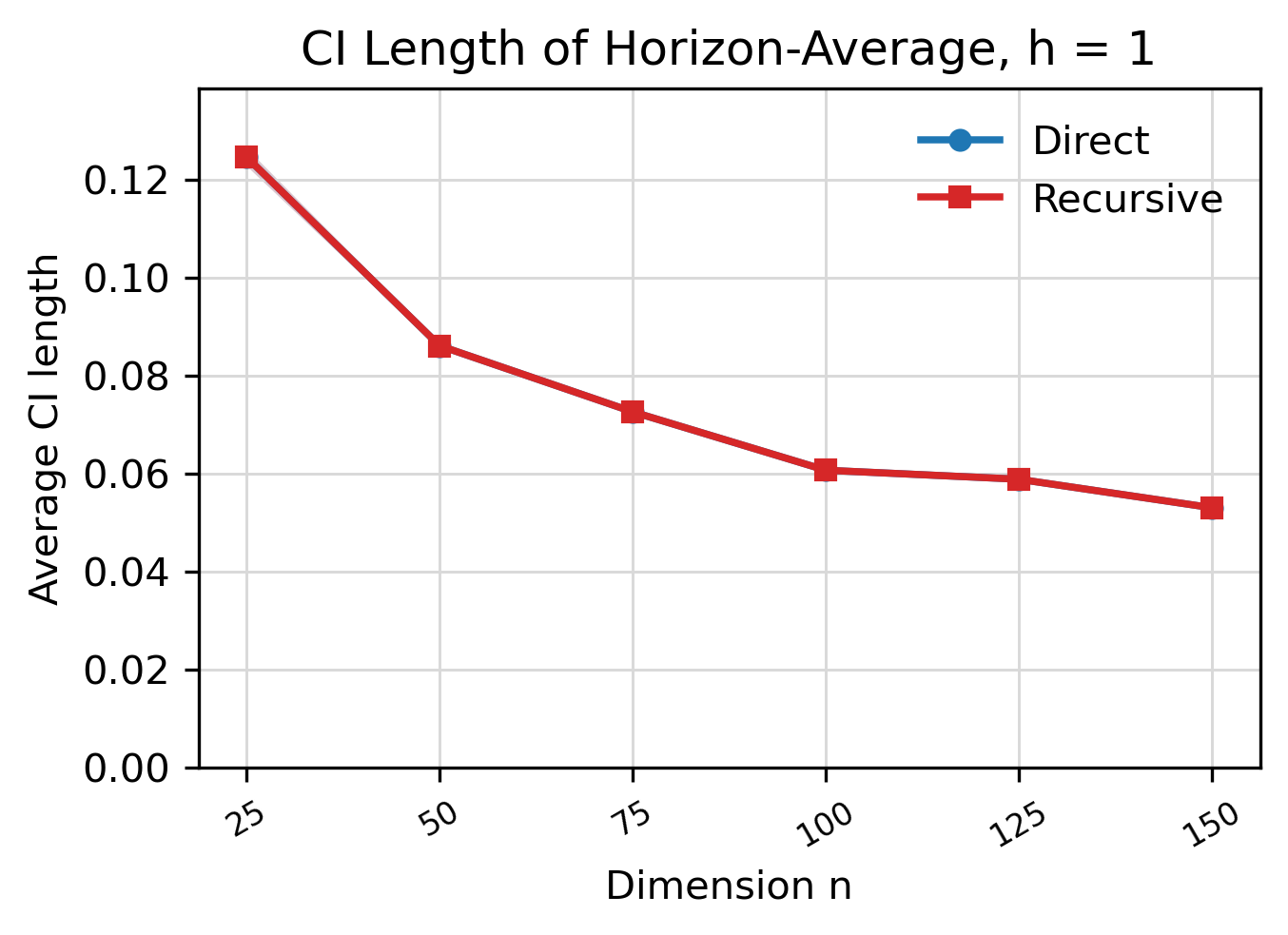}
		\caption{Length $h=1$} 
		\label{fig:len.1}
	\end{subfigure} 
	\begin{subfigure}[b]{0.32\textwidth}
		\centering 
		\includegraphics[width=\linewidth]{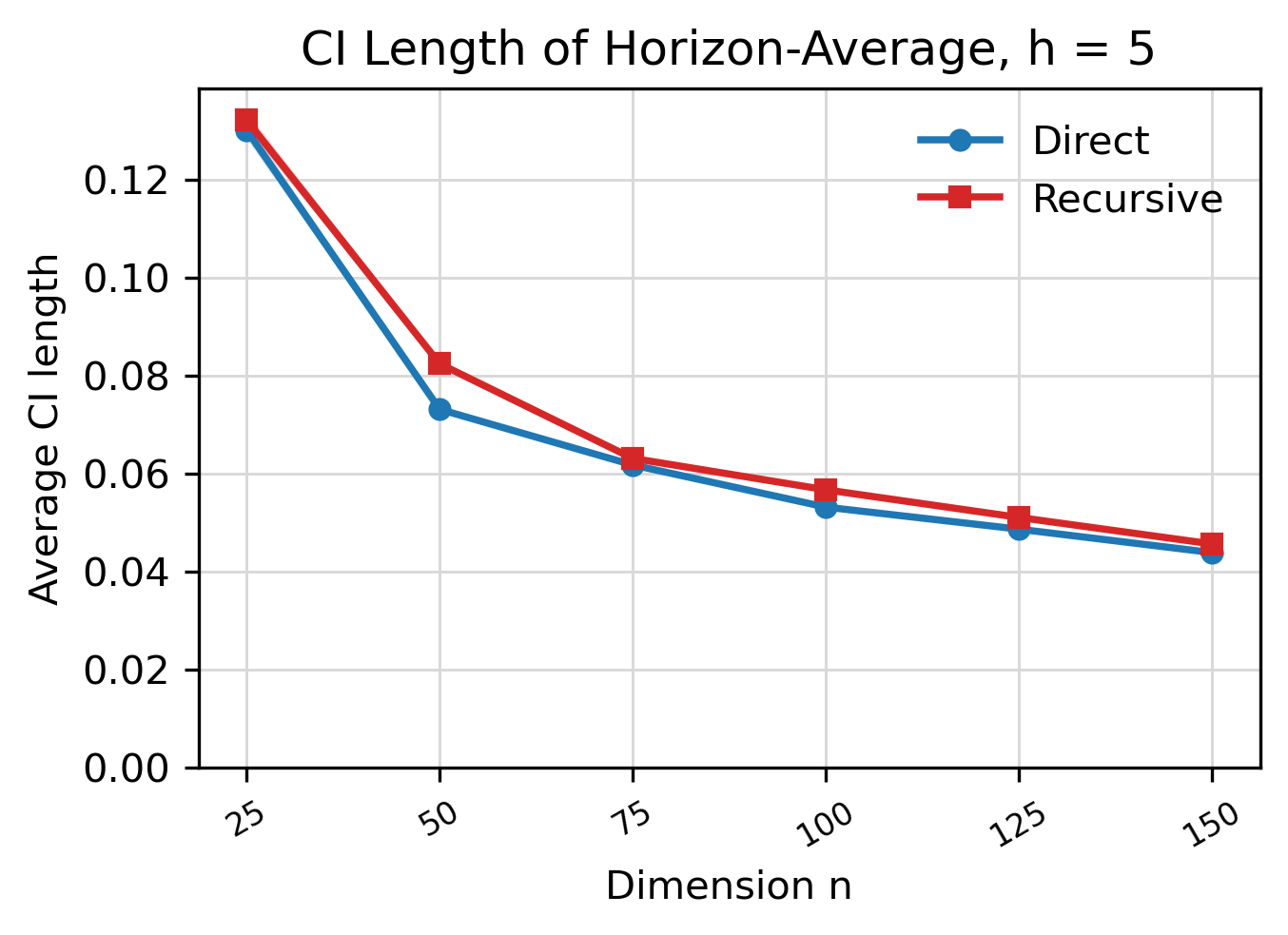}
		\caption{Length $h=5$} 
		\label{fig:len.5}
	\end{subfigure} 
	\begin{subfigure}[b]{0.32\textwidth}
		\centering 
		\includegraphics[width=\linewidth]{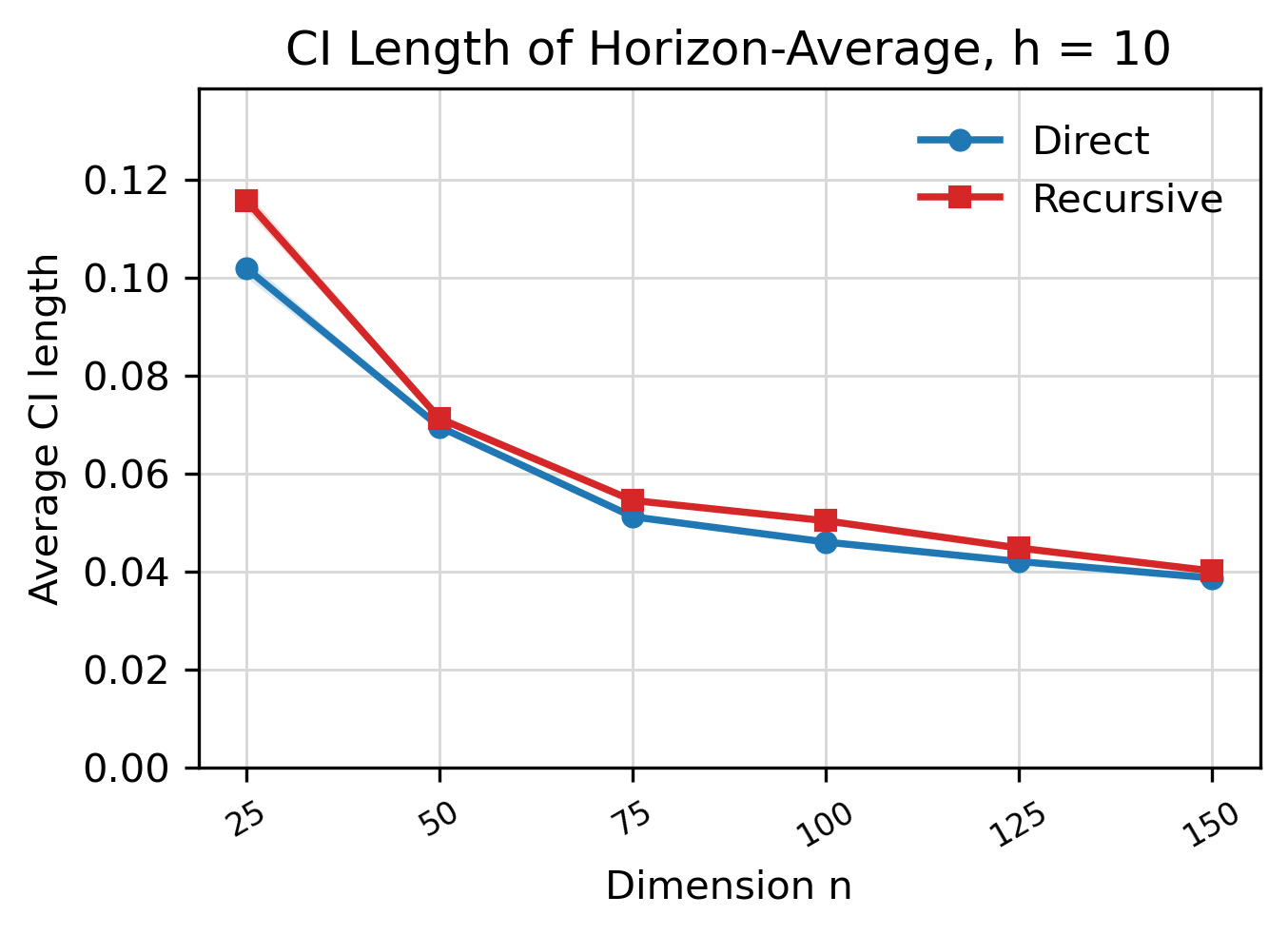}
		\caption{Length $h=10$} 
		\label{fig:len.10}
	\end{subfigure}
	\caption{Finite-sample performance of direct and recursive \TWSF~for horizon-average estimands. Panels~\ref{fig:bias.1}--\ref{fig:bias.10} report bias, \ref{fig:rmse.1}-\ref{fig:rmse.10} report RMSE, \ref{fig:cov.1}-\ref{fig:cov.10} report 90\% pointwise coverage, and \ref{fig:len.1}-\ref{fig:len.10} report average interval length. Blue and red denote direct and recursive \TWSF; shaded bands are clustered Monte Carlo standard errors. Results are based on 1000 replications.}
	\label{fig:sims.avg}  
\end{figure}
%

\section{Case Study: Opening NFL Stadiums During the 2020 Season} \label{sec:case.study}
We revisit the NFL stadium-opening study of \cite{nfl_pnas}, introduced in Section~\ref{sec:intro}. Their analysis addressed a retrospective question: how would cities that admitted spectators have evolved had their stadiums remained closed? We study the prospective counterpart: how would a closed-stadium city have evolved had it begun admitting spectators at a future decision date? 

\subsection{Data and Empirical Design}
%
\subsubsection{Setup}
Following \cite{nfl_pnas}, we use daily county-level COVID-19 case counts from the New York Times \citep{NYT:21} beginning on April 1, 2020. 
The units are the counties containing NFL stadiums, labeled by their associated NFL markets. 
%
%
After consolidating markets with multiple teams, the sample contains $N=30$ cities: 19 admitted spectators at some point during the regular season and 11 remained closed throughout the study period. 
Table~\ref{tab:open} lists the first regular-season date on which each of the 19 opening NFL markets admitted spectators. 
%
%
Notably, we model stadium opening as a binary intervention. 
In reality, outcomes likely depend on attendance levels, mitigation practices, and enforcement intensity. 
Our estimates to follow abstract from this source of heterogeneity, and the results should be interpreted with this limitation in mind.

\begin{table}[t]
\centering
\caption{First stadium-opening dates.}
\label{tab:open}
\footnotesize
    \begin{adjustbox}{max width=\textwidth}
    \begin{tabular}{l c  }
        \toprule
\textbf{First opening date} & \textbf{NFL market} \\
\midrule
    September 10 & Kansas City    	\\
    September 13 & Jacksonville 	\\
    September 17 & Cleveland 		\\
    September 20 & Dallas, Indianapolis, Miami \\
    September 27 & Denver \\
    October 4 & Carolina, Cincinnati, Houston, Tampa Bay\\
    October 11 & Atlanta, Pittsburgh \\
    October 13 & Tennessee \\
    October 18 & Philadelphia \\
    October 25 & Arizona, New Orleans \\
    November 1 & Baltimore \\
    November 8 & Washington \\ 
\bottomrule
\end{tabular}
\end{adjustbox}
\end{table}

\subsubsection{Observation Pattern and Donor Pool} 
Stadium openings were staggered, whereas the theoretical framework in Section~\ref{sec:framework} assumes a common treatment date. 
We therefore embed the empirical setting in an approximate \TWSF~design, following the strategy of \cite{synth_iv}. 
A formal extension to staggered adoption remains an important direction for future work.

The treated donor pool consists of the six earliest cities to admit spectators: 
$\Ic_1 =$ \{Kansas City, Jacksonville, Cleveland, Dallas, Indianapolis, Miami\}, giving a modest $N_1=6$. 
All six donor cities had opened by September 20. 
Denver is excluded because waiting for its September 27 opening would materially shorten the treated period available for estimating temporal dynamics. 
We therefore begin the time-side training window on September 21, the first day after every donor city had admitted spectators.

\subsubsection{Forecast Design} 
For each target city $i$, let $\tau$ denote the opening date. 
We treat the reported count on $\tau$ as pre-effect and forecast the next $h=14$ daily counts. 
The two-week horizon reflects the biological and reporting delays through which any effect of stadium attendance could appear in reported cases \citep{covid1, covid2}. 
In all analyses, the unit-side weights are learned from April 1 through September 10, before any stadium opened. 
The time-side weights are learned from September 21 through $\tau$. 
The forecast period is excluded from both estimation and hyperparameter selection. 
For each forecast origin, Table~\ref{tab:parameters} reports the cross-validated lag length and truncation ranks used by recursive \TWSF. 

\begin{table}[t!]
\centering
\caption{Cross-validated hyperparameters for recursive \TWSF.}
\label{tab:parameters}
\footnotesize
    \begin{adjustbox}{max width=\textwidth}
    \begin{tabular}{l c  c  c  }
        \toprule
\textbf{Opening date} & \textbf{Lag} $L$ 
&\textbf{Rank} $k_y$  & \textbf{Rank} $k_z$  \\
\midrule
    2020-10-04 &            3 &       4 &                   2  \\ 
    2020-10-11 &           8 &       6 &                   1 \\ 
    2020-10-13 &           8 &       2 &                   2 \\ 
    2020-10-18 &            5 &       5 &                   3 \\
    2020-10-25 &            8 &       3 &                   4 \\
    2020-11-01 &            7 &       4&                   4 \\
    2020-11-08 &            8 &      3 &                   7 \\
\bottomrule
\end{tabular}
\end{adjustbox}
\end{table}

Direct \TWSF~is infeasible for the earlier decision dates because the available treated period does not contain enough Page blocks of length $L+14$. We therefore use recursive \TWSF~throughout.
Although \TWSF~is estimated on daily counts, the figures display cumulative trajectories. 
Let $S_{it}$ denote the observed cumulative count through day $t$. 
For $m \in [14]$, define $\boldeta^{(m)}$ as the binary vector whose first $m$ entries are one. 
The cumulative open-stadium forecast at day $\tau + m$ is then given by $\hS_{i, \tau+m} \coloneqq S_{i\tau} + \htheta^\rec_{i, \eta^{(m)}}(\tau)$. 

\subsection{Empirical Results} 

\subsubsection{Validation Study} 
The principal counterfactual outcomes are unobserved, so we first evaluate \TWSF~on cities that eventually admitted spectators. 
For each later-opening city, we treat the date of its first open home game as the end of the observed panel and forecast the following 14 days using only information available by that date. 
This variation offers an informal indication of how temporal sample size affects forecasting performance.

Because the realized post-opening path includes future idiosyncratic shocks, we use prediction intervals rather than confidence intervals for the conditional mean. 
Under the homoskedastic Gaussian version of Assumption~\ref{assump:subg}, the innovation variance accumulated through lead $m \in [14]$ is $\sigma^2 \| \boldeta^{(m)} \|_2^2 = m \sigma^2$. 
The corresponding $(1-\gamma)100\%$ prediction interval is 
$\Big[ \hS_{i, \tau+m} \pm \Phi^{-1}(1-\gamma/2) \sqrt{ \big\{\hVc^\rec_{i m}(\tau)\big\}^2 + m \hsigma^2} \Big]$, 
where $\hVc^\rec_{im}(\tau)$ reflects uncertainty in the expected cumulative path while $m \hsigma^2$ accounts for future daily shocks.

Figure~\ref{fig:validation.same} reports results for all validation cities. 
The central forecasts track most realized trajectories closely relative to their prediction intervals. 
The clearest departures occur in Cincinnati and Pittsburgh, where cumulative counts accelerate near the end of the horizon and exceed the upper prediction limit.
Forecasts are generally more accurate for later-opening cities, which provide longer treated-donor training windows. 
Cincinnati and Pittsburgh are among the earlier and more difficult cases. 
This pattern is consistent with the role of temporal sample size in Theorems~\ref{prop:parameter.recovery} and \ref{thm:recursive.normality}, although the limited number of validation cities precludes a definitive comparison. 
Taken together, the validation results indicate that \TWSF~can often produce treated forecasts close to subsequently observed paths using only information available at the forecast origin.

\begin{figure}[t]
	\centering 
	\begin{subfigure}[b]{0.32\textwidth}
		\centering 
		\includegraphics[width=\linewidth]{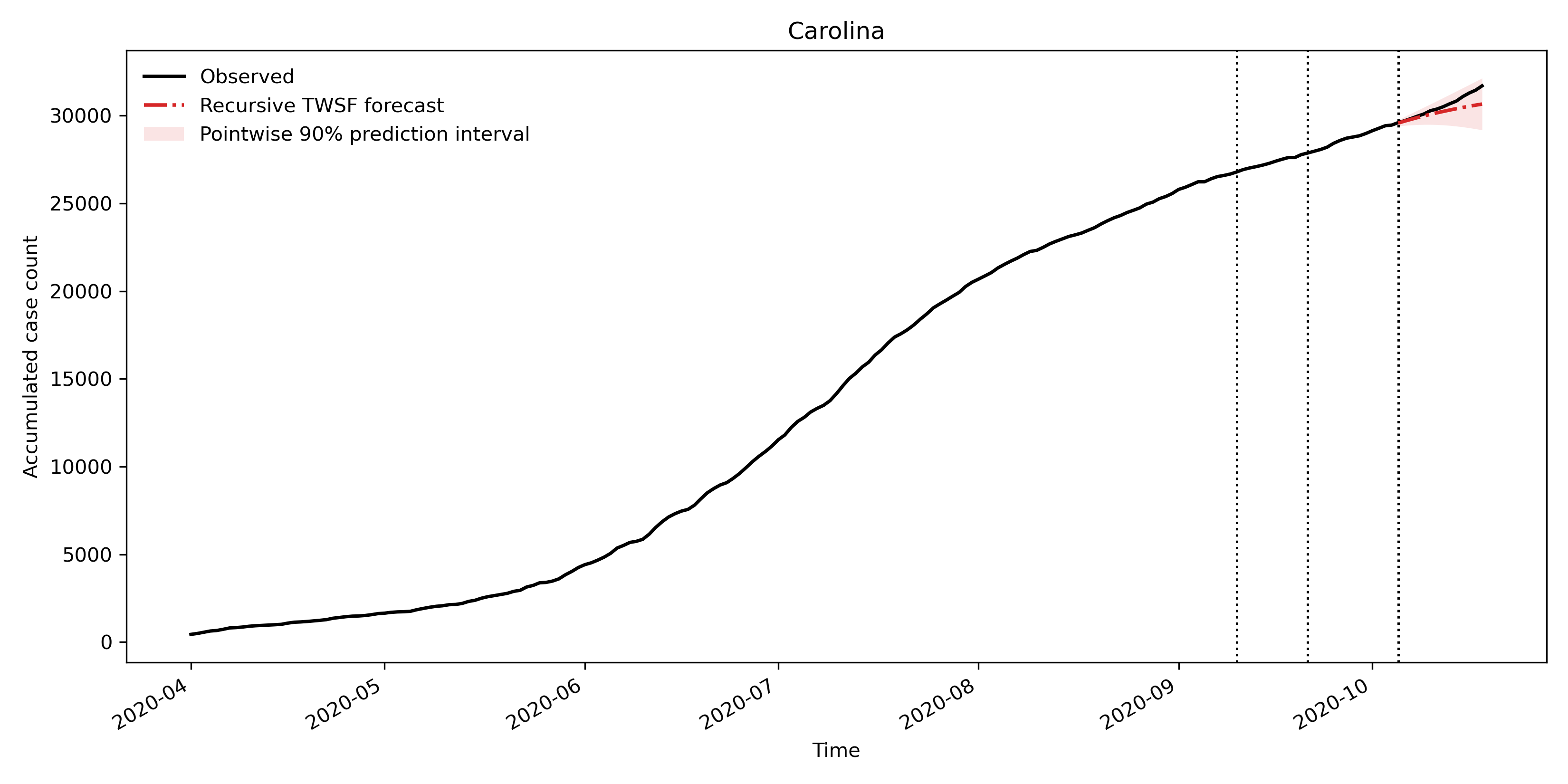}
		\caption{Carolina: 10-04.} 
	\end{subfigure} 
	\begin{subfigure}[b]{0.32\textwidth}
		\centering 
		\includegraphics[width=\linewidth]{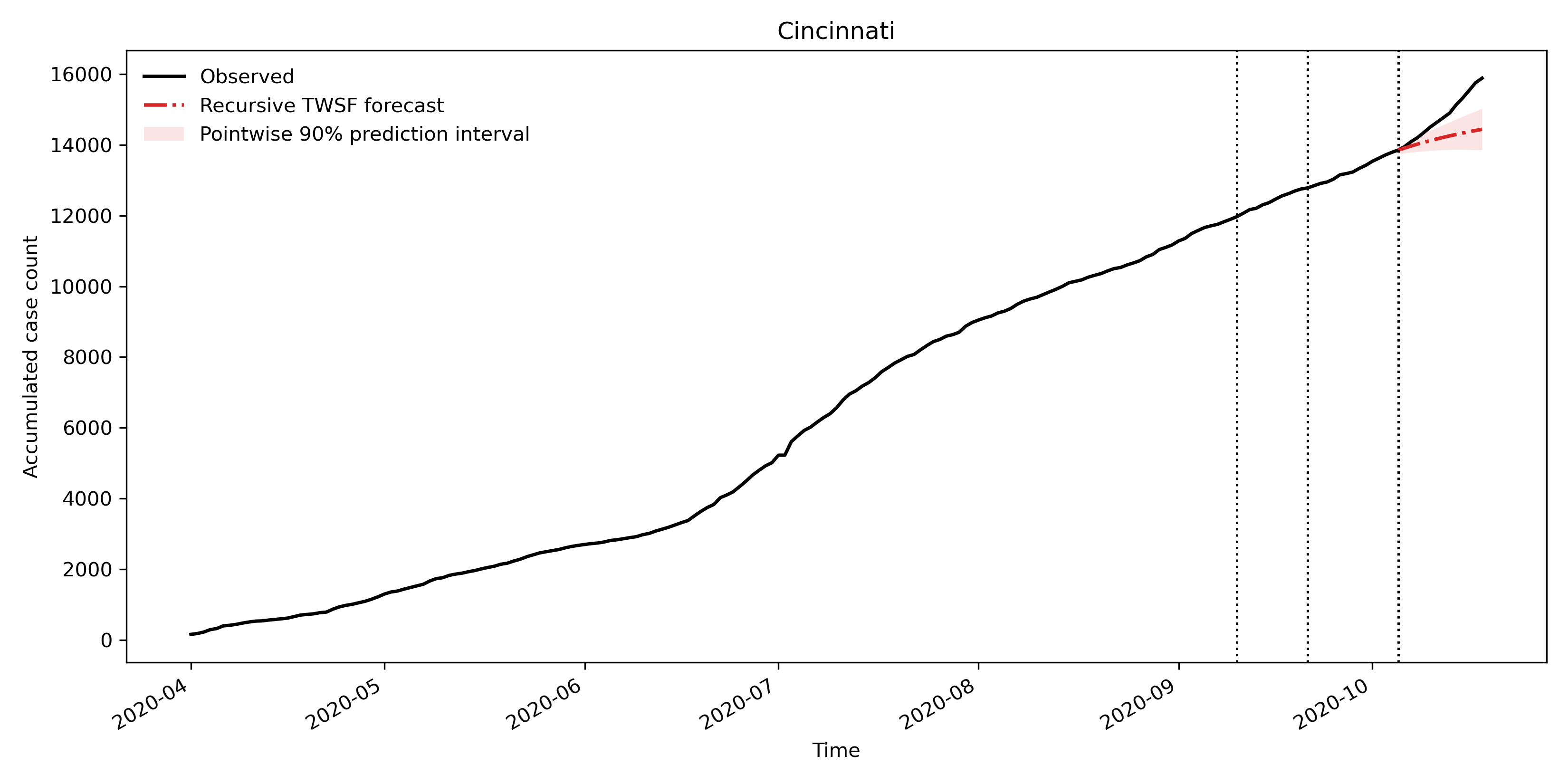}
		\caption{Cincinnati: 10-04.} 
	\end{subfigure} 
	\begin{subfigure}[b]{0.32\textwidth}
		\centering 
		\includegraphics[width=\linewidth]{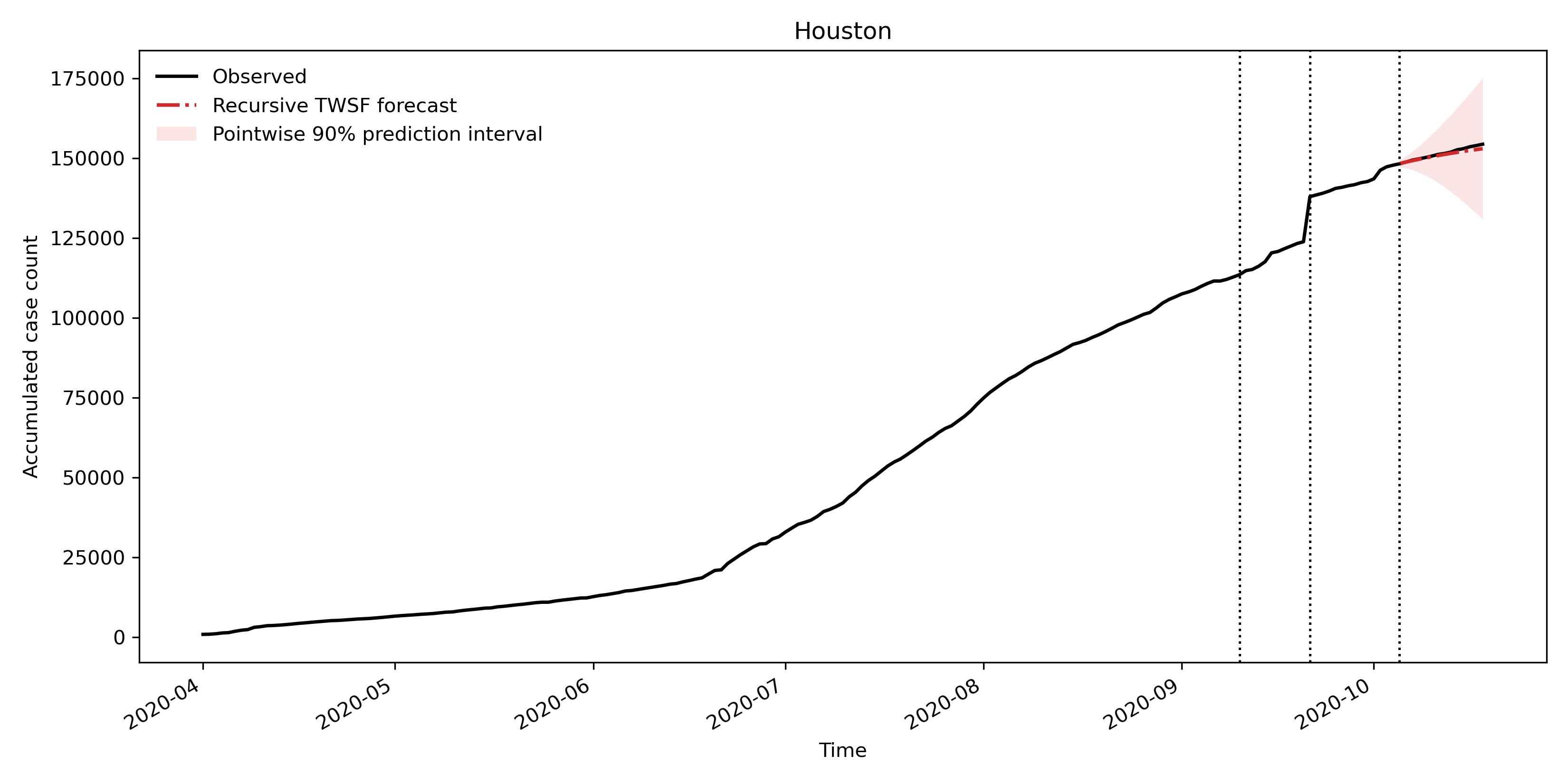}
		\caption{Houston: 10-04.} 
	\end{subfigure} 
	\\
	\begin{subfigure}[b]{0.32\textwidth}
		\centering 
		\includegraphics[width=\linewidth]{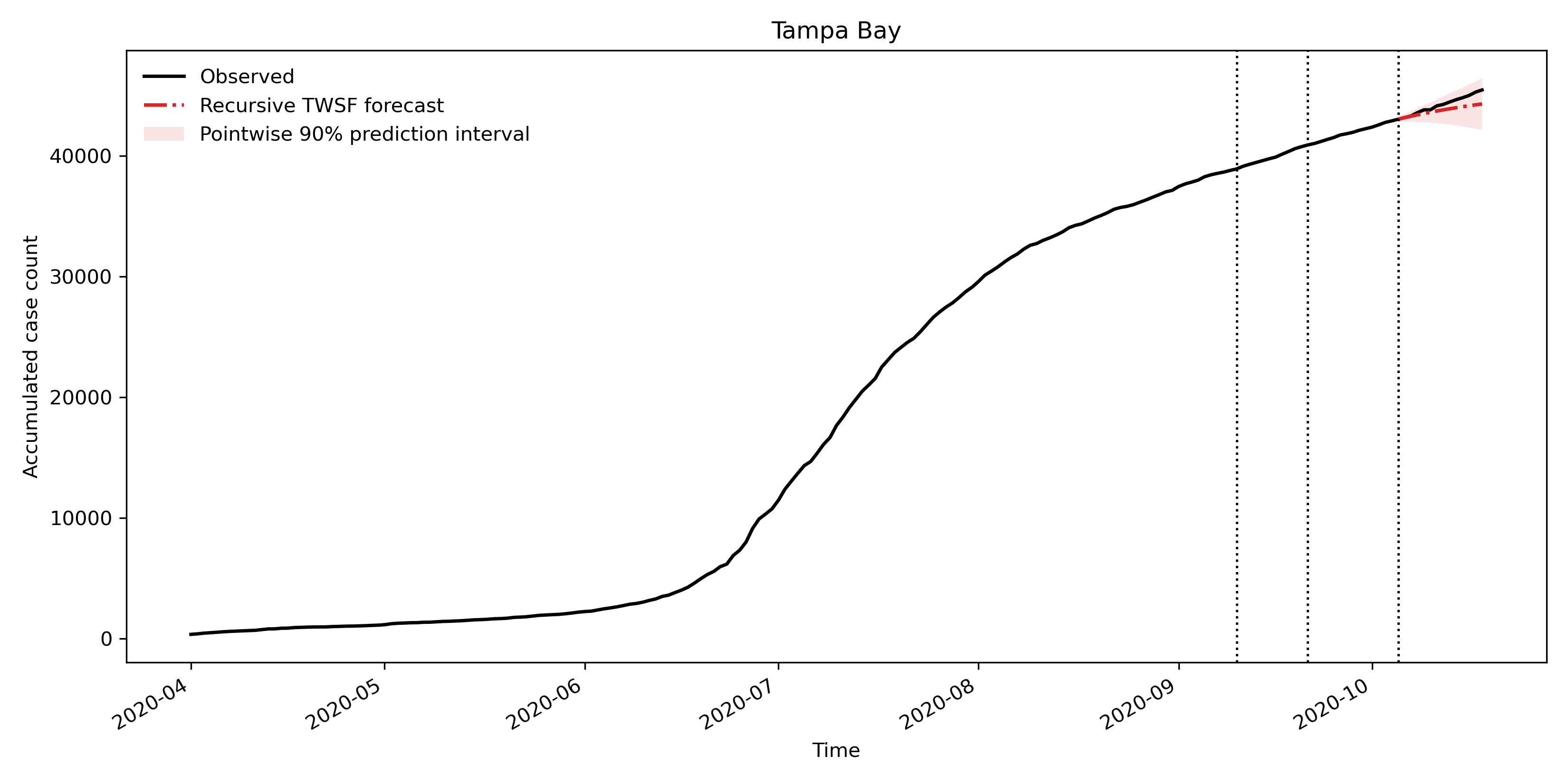}
		\caption{Tampa Bay: 10-04.} 
	\end{subfigure} 
	\begin{subfigure}[b]{0.32\textwidth}
		\centering 
		\includegraphics[width=\linewidth]{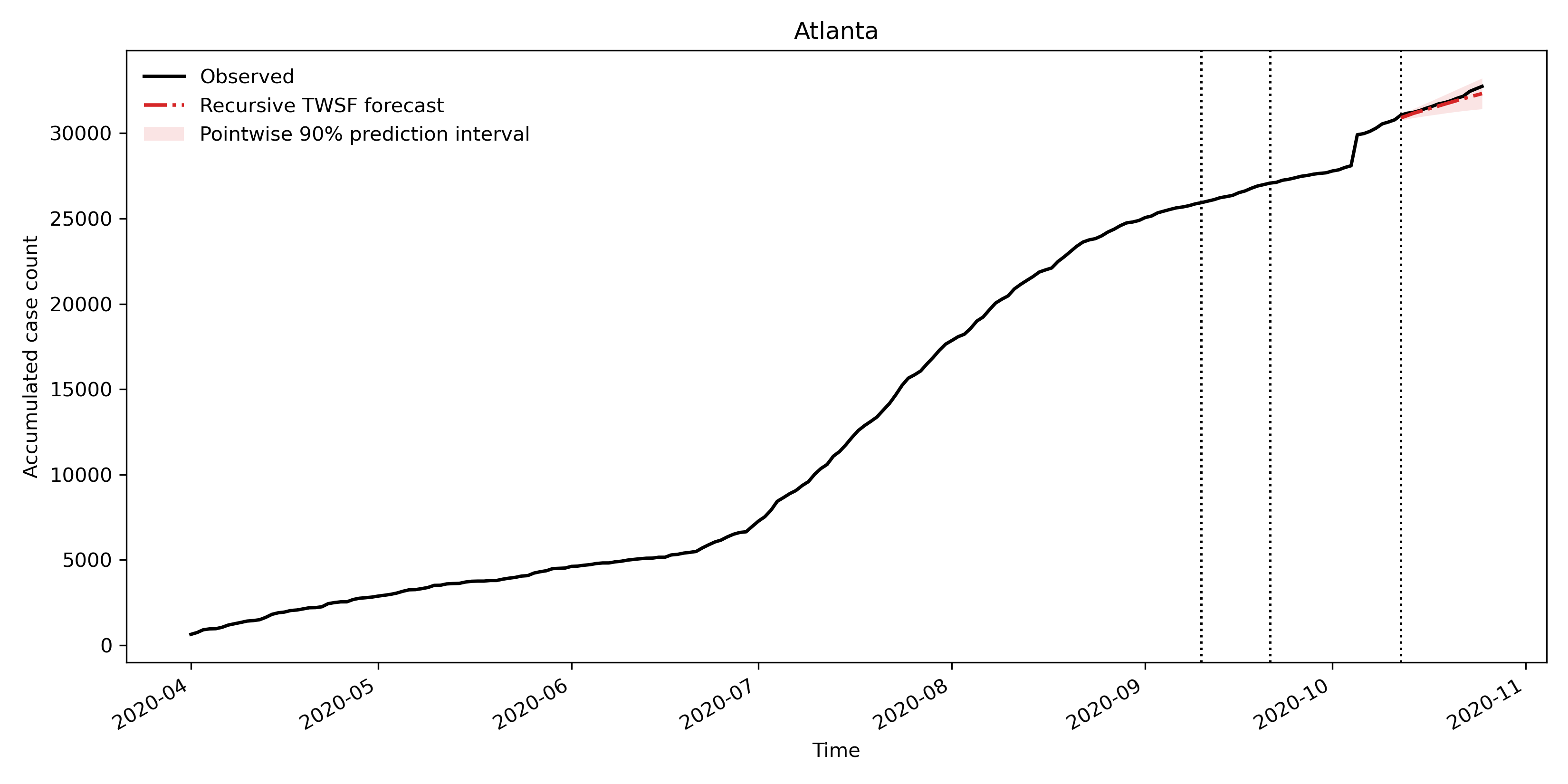}
		\caption{Atlanta: 10-11.} 
	\end{subfigure} 
	\begin{subfigure}[b]{0.32\textwidth}
		\centering 
		\includegraphics[width=\linewidth]{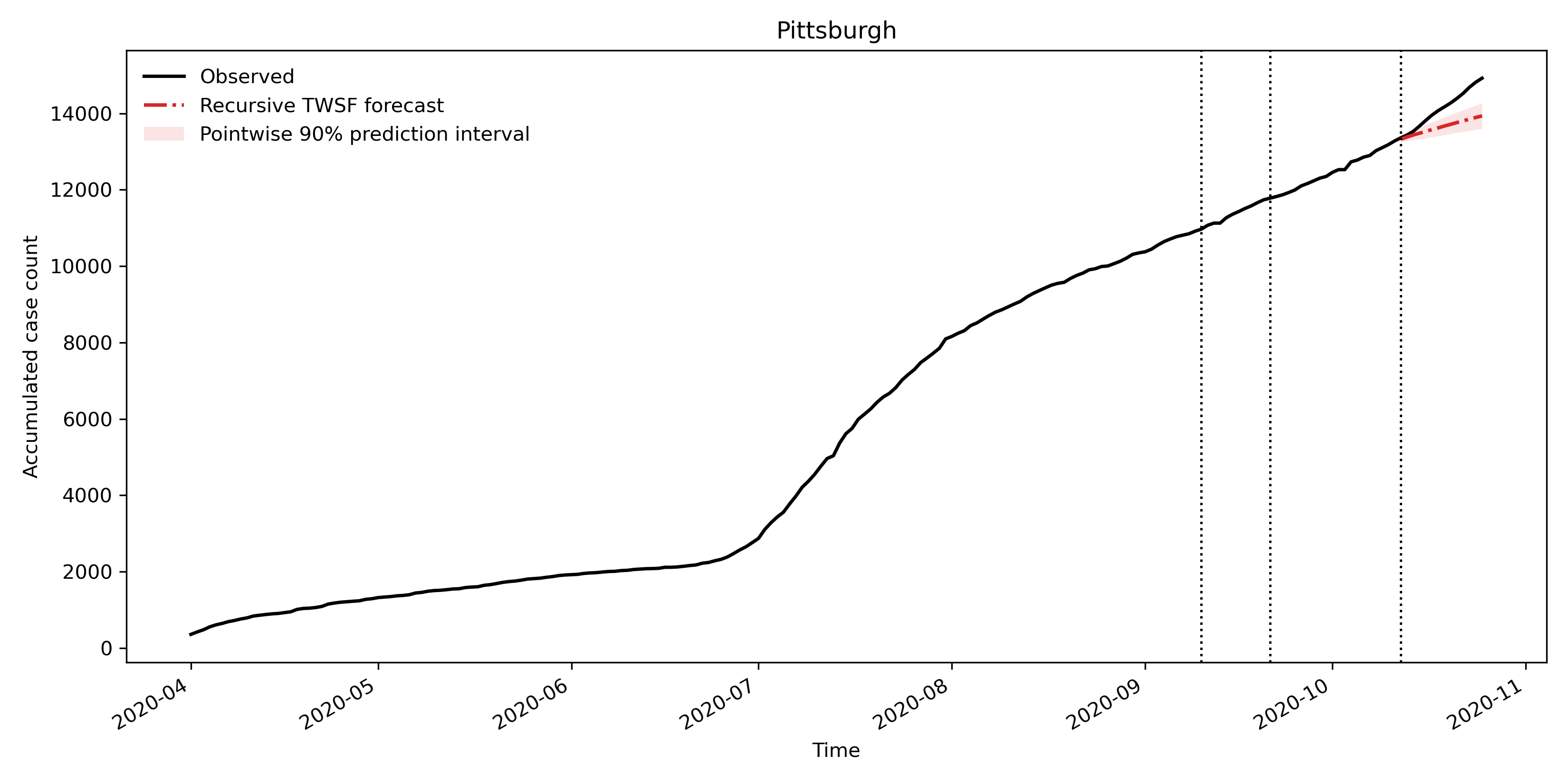}
		\caption{Pittsburgh: 10-11.} 
	\end{subfigure} 
	\\
	\begin{subfigure}[b]{0.32\textwidth}
		\centering 
		\includegraphics[width=\linewidth]{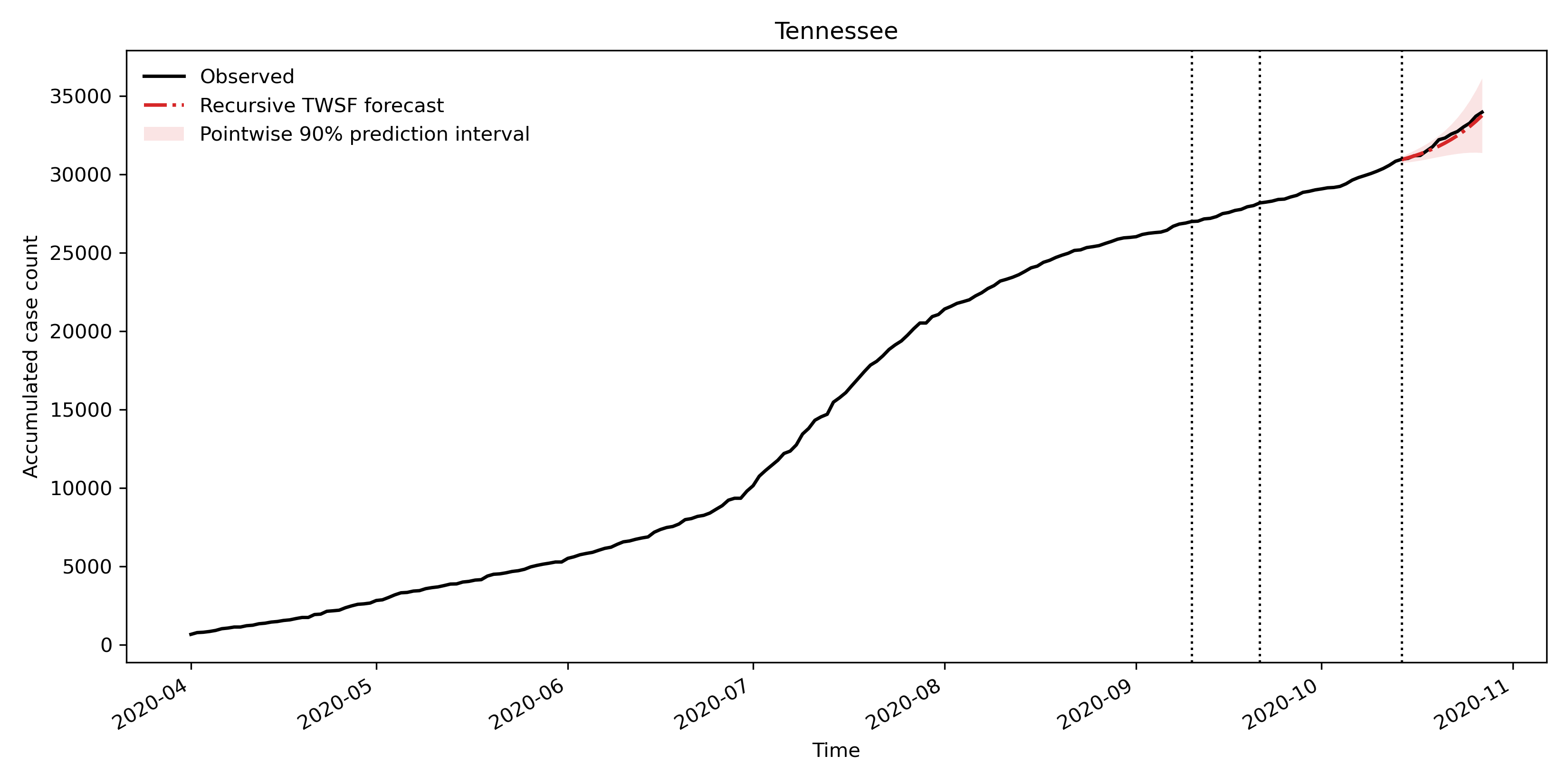}
		\caption{Tennessee: 10-13.} 
	\end{subfigure} 
	\begin{subfigure}[b]{0.32\textwidth}
		\centering 
		\includegraphics[width=\linewidth]{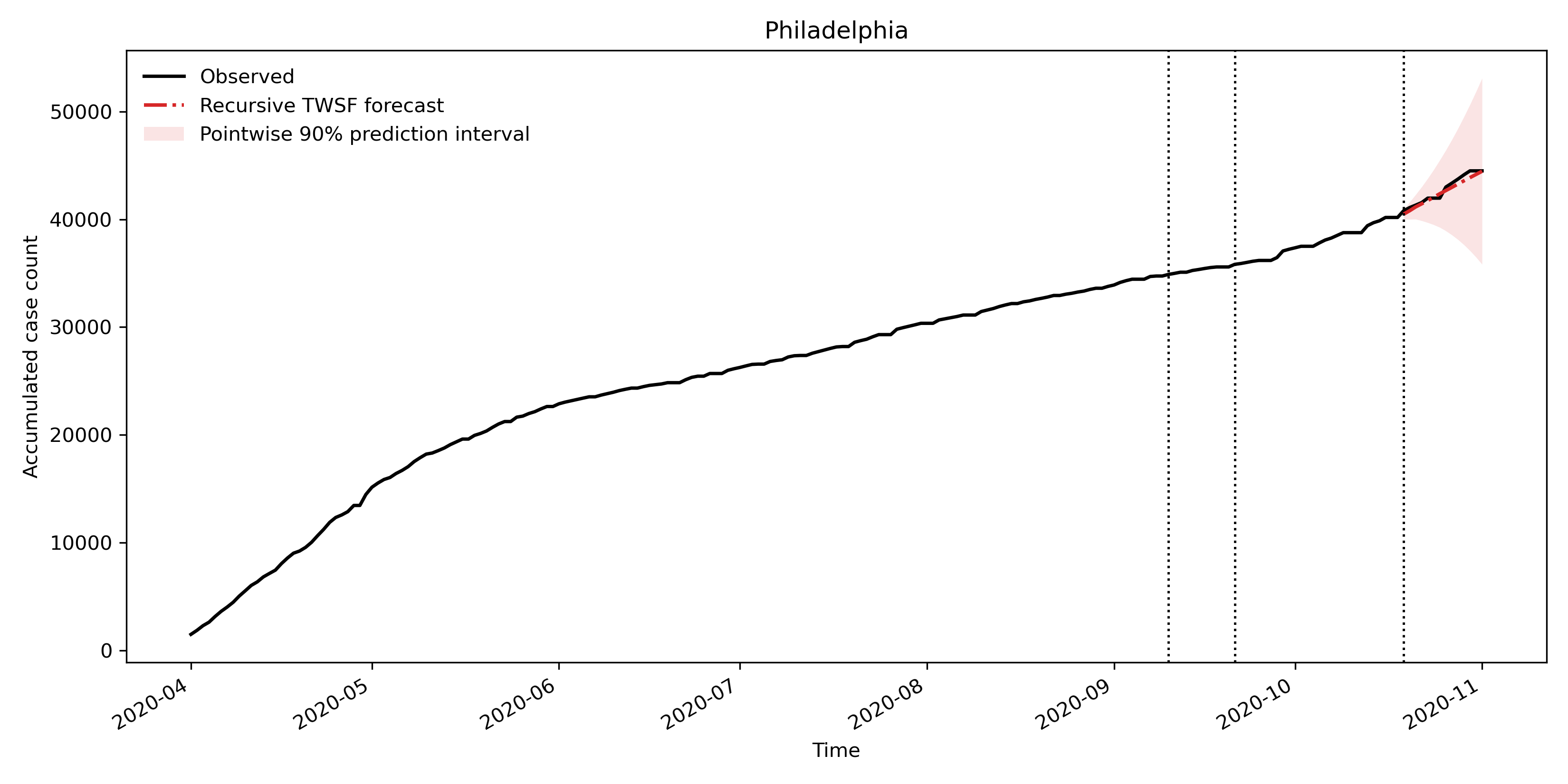}
		\caption{Philadelphia: 10-18.} 
	\end{subfigure} 
	\begin{subfigure}[b]{0.32\textwidth}
		\centering 
		\includegraphics[width=\linewidth]{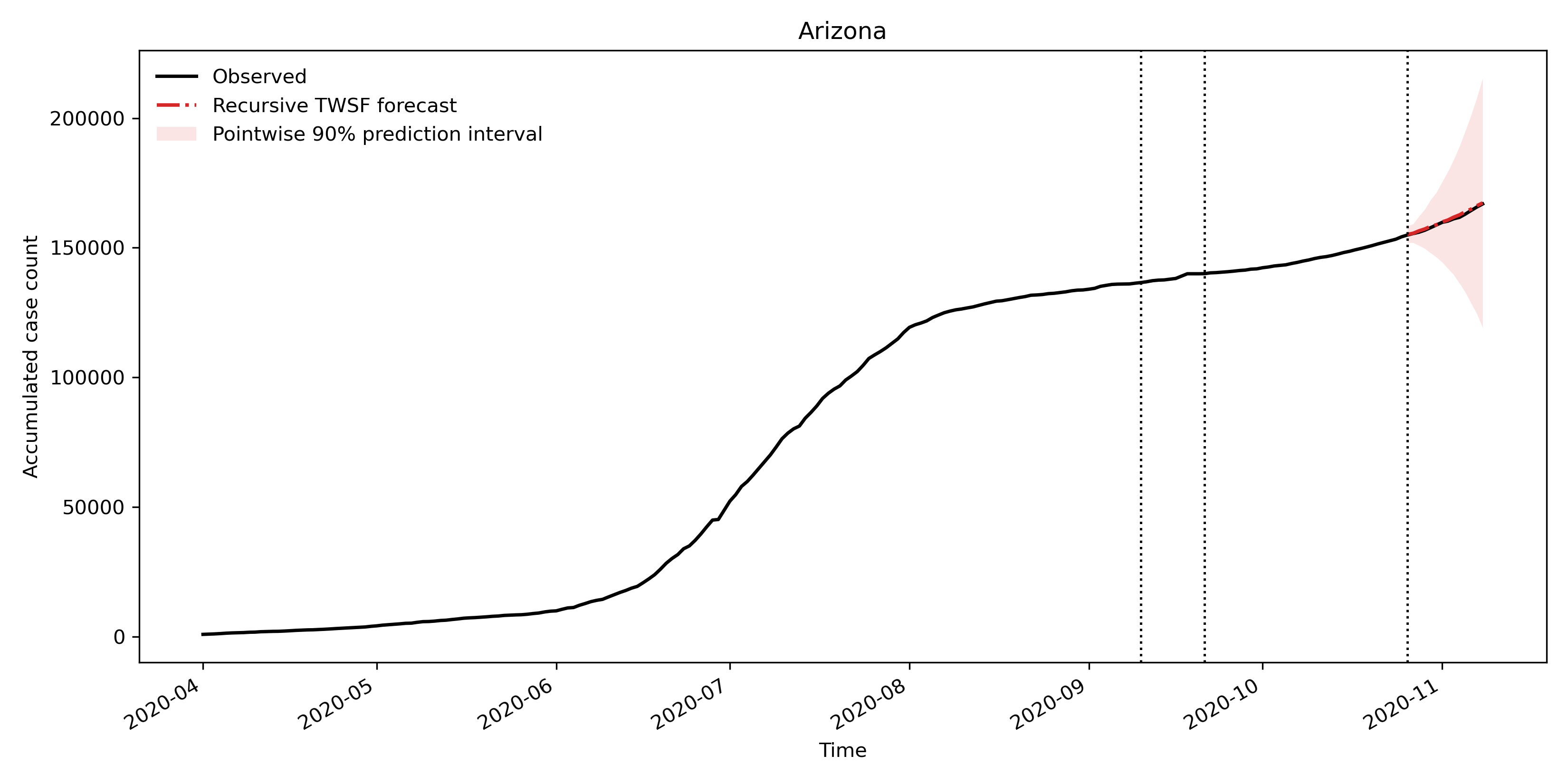}
		\caption{Arizona: 10-25.} 
	\end{subfigure} 
	\\
	\begin{subfigure}[b]{0.32\textwidth}
		\centering 
		\includegraphics[width=\linewidth]{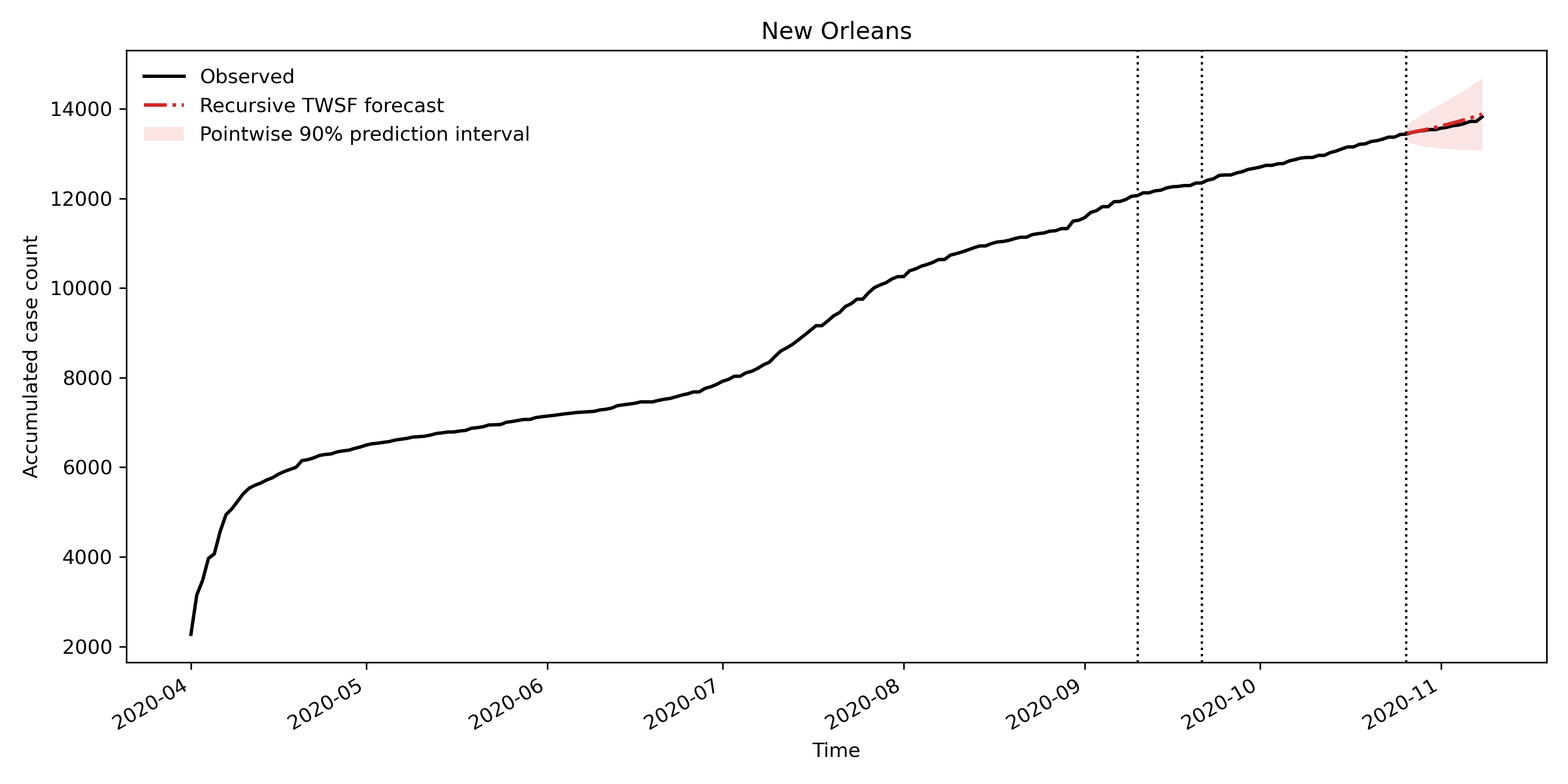}
		\caption{New Orleans: 10-25.} 
	\end{subfigure} 
	\begin{subfigure}[b]{0.32\textwidth}
		\centering 
		\includegraphics[width=\linewidth]{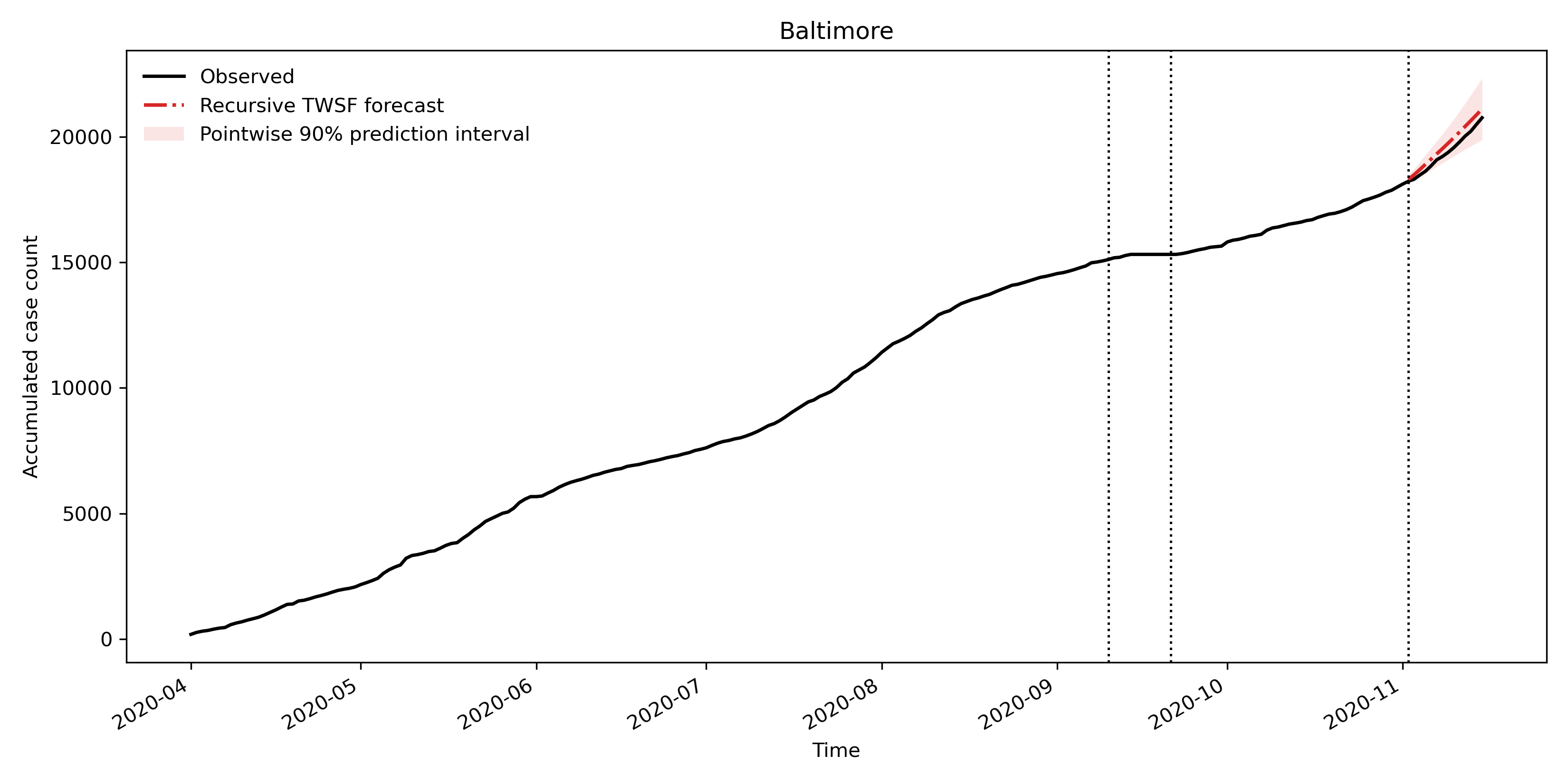}
		\caption{Baltimore: 11-01.} 
	\end{subfigure} 
	\begin{subfigure}[b]{0.32\textwidth}
		\centering 
		\includegraphics[width=\linewidth]{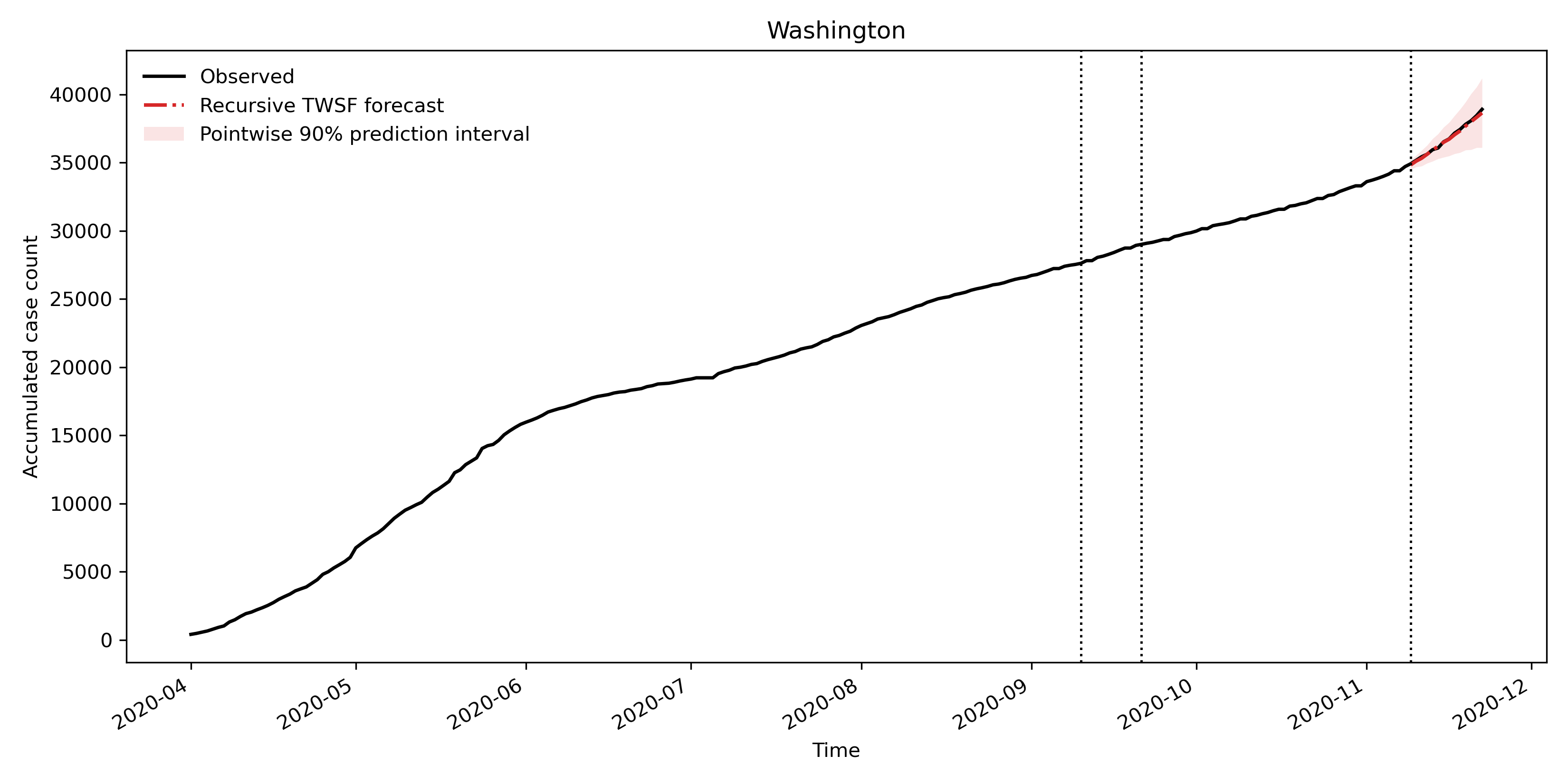}
		\caption{Washington: 11-08.} 
	\end{subfigure} 
	\caption{Validation forecasts for cities that admitted fans later in the season. 
	Black lines show observed cumulative cases, and red dash-dotted lines show \TWSF~forecasts. 
	Shaded regions are pointwise 90\% prediction intervals. 
	The first dotted vertical line marks September 10, the end of the unit-side pre-treatment window; the second marks September 21, the beginning of the treated-donor window used to estimate temporal dynamics; and the third marks the target city’s actual opening date. Forecasts cover the following 14 days.}
	\label{fig:validation.same} 
\end{figure}

\subsubsection{Counterfactual Forecasts} 
We next consider the 11 cities that kept their stadiums closed. 
Each city is assigned a hypothetical opening date equal to its first home game after October 1: October 4: \{Chicago, Detroit, Las Vegas, Los Angeles, San Francisco\}, October 11: \{New York, Seattle\}, October 18: \{Minnesota, New England\}, November 1: \{Buffalo, Green Bay\}. 
For each target city, we estimate the 14-day cumulative trajectory that would have followed the hypothetical opening. 
At lead $m \in [14]$, the estimand is $S_{i \tau} + \theta_{i, \eta^{(m)}}(\tau)$. 
Because this quantity is a conditional mean rather than a future realization, we use the pointwise confidence interval in \eqref{eq:ci.horizon} without the additional future-innovation term. 
The donor pool remains unchanged. 
For each hypothetical opening date, we use the hyperparameters selected for the corresponding validation date in Table~\ref{tab:parameters}.

Figure~\ref{fig:counter.same} reports the results for all cities. 
For most cities, the estimated open-stadium trajectory remains close to the observed closed-stadium trajectory. 
The largest visible differences arise in Seattle, Buffalo, and Green Bay, where observed counts rise above the counterfactual forecasts.
Green Bay is a known problematic case: \cite{nfl_pnas} attribute a large shock in the relevant period to a Wisconsin Department of Health Services reporting error \citep{green_bay}. 
Buffalo is also difficult to interpret because its forecast window overlaps a pronounced Erie County surge and the introduction of New York’s micro-cluster restrictions \citep{buffalo1, buffalo2}. 

Except for New England, the estimated open-stadium trajectories generally lie below the observed closed-stadium trajectories. 
This pattern should not be interpreted as evidence that stadium opening reduced infections. 
A more defensible conclusion is that the estimated open-state mean is not systematically or substantially above the realized closed-state path under the conditions studied. 
This finding is broadly consistent with \cite{nfl_pnas}, but the present design does not identify the mechanism.
The results also vary with the hypothetical opening date because both local epidemic conditions and the available treated-donor history change over time. 
Thus, although the results do not indicate a large systematic increase in case counts, the timing of an opening decision may still have substantive consequences. 

\begin{figure}[t]
	\centering 
	\begin{subfigure}[b]{0.32\textwidth}
		\centering 
		\includegraphics[width=\linewidth]{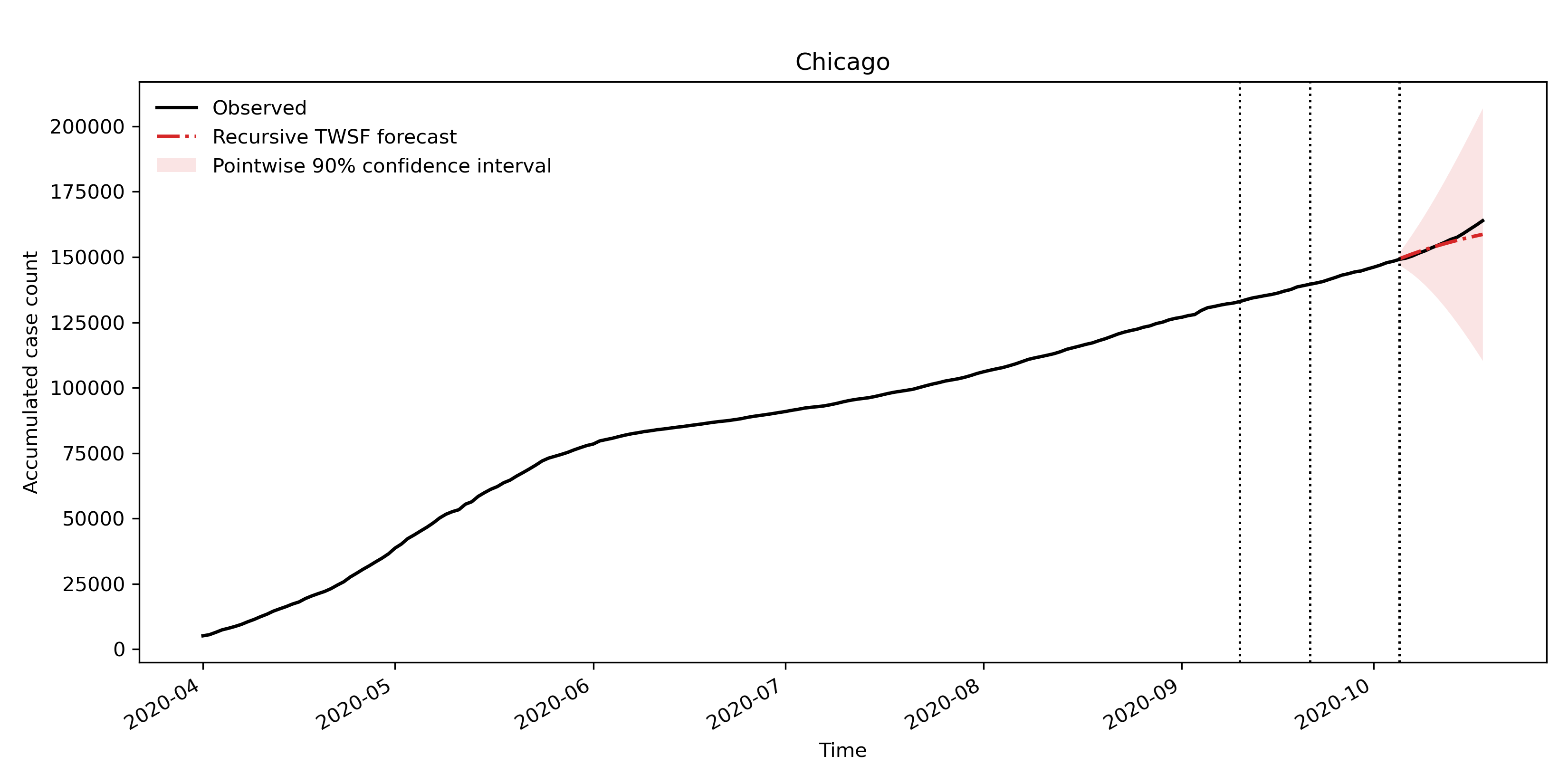}
		\caption{Chicago: 10-04.} 
	\end{subfigure} 
	\begin{subfigure}[b]{0.32\textwidth}
		\centering 
		\includegraphics[width=\linewidth]{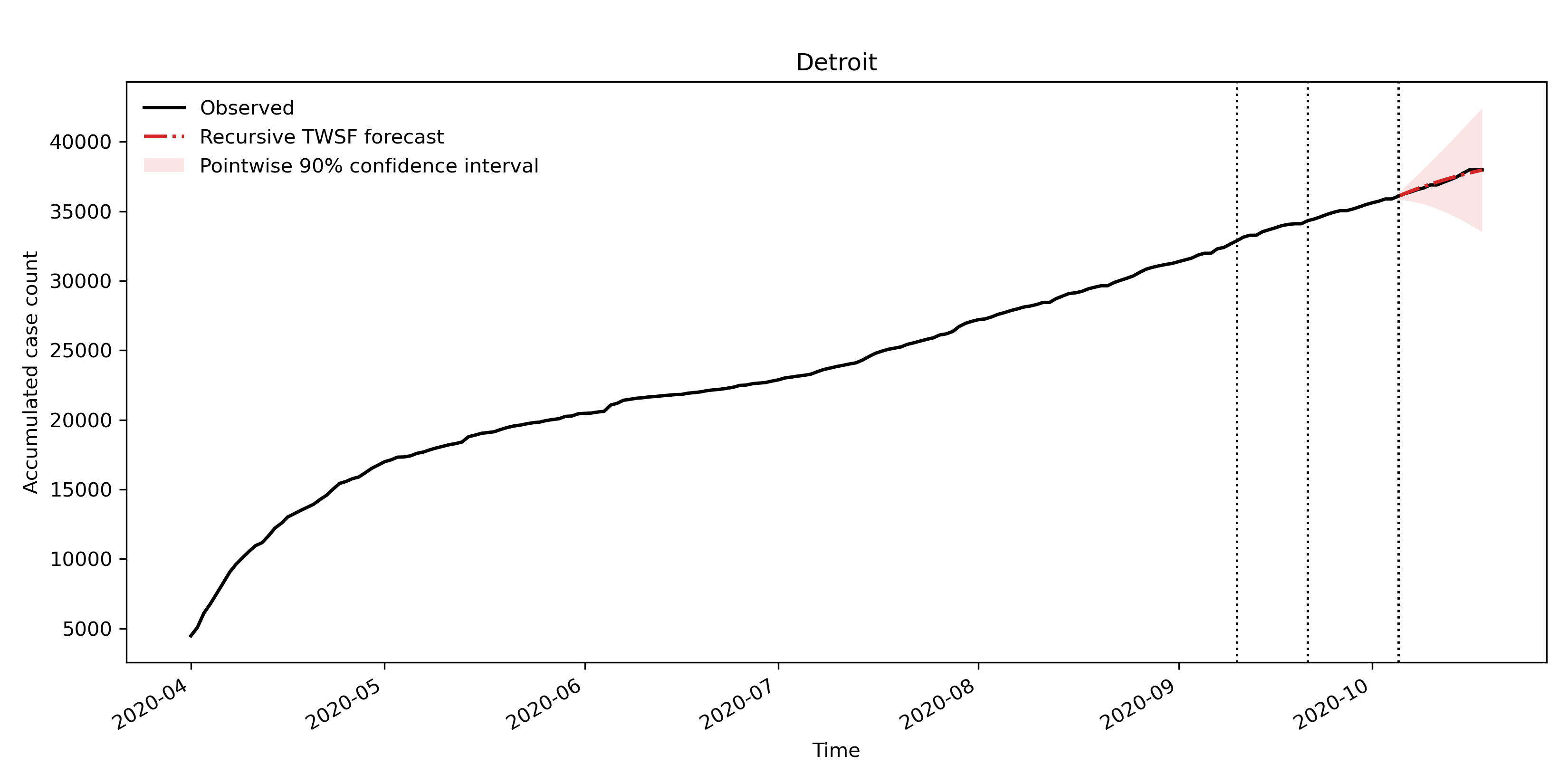}
		\caption{Detroit: 10-04.} 
	\end{subfigure} 
	\begin{subfigure}[b]{0.32\textwidth}
		\centering 
		\includegraphics[width=\linewidth]{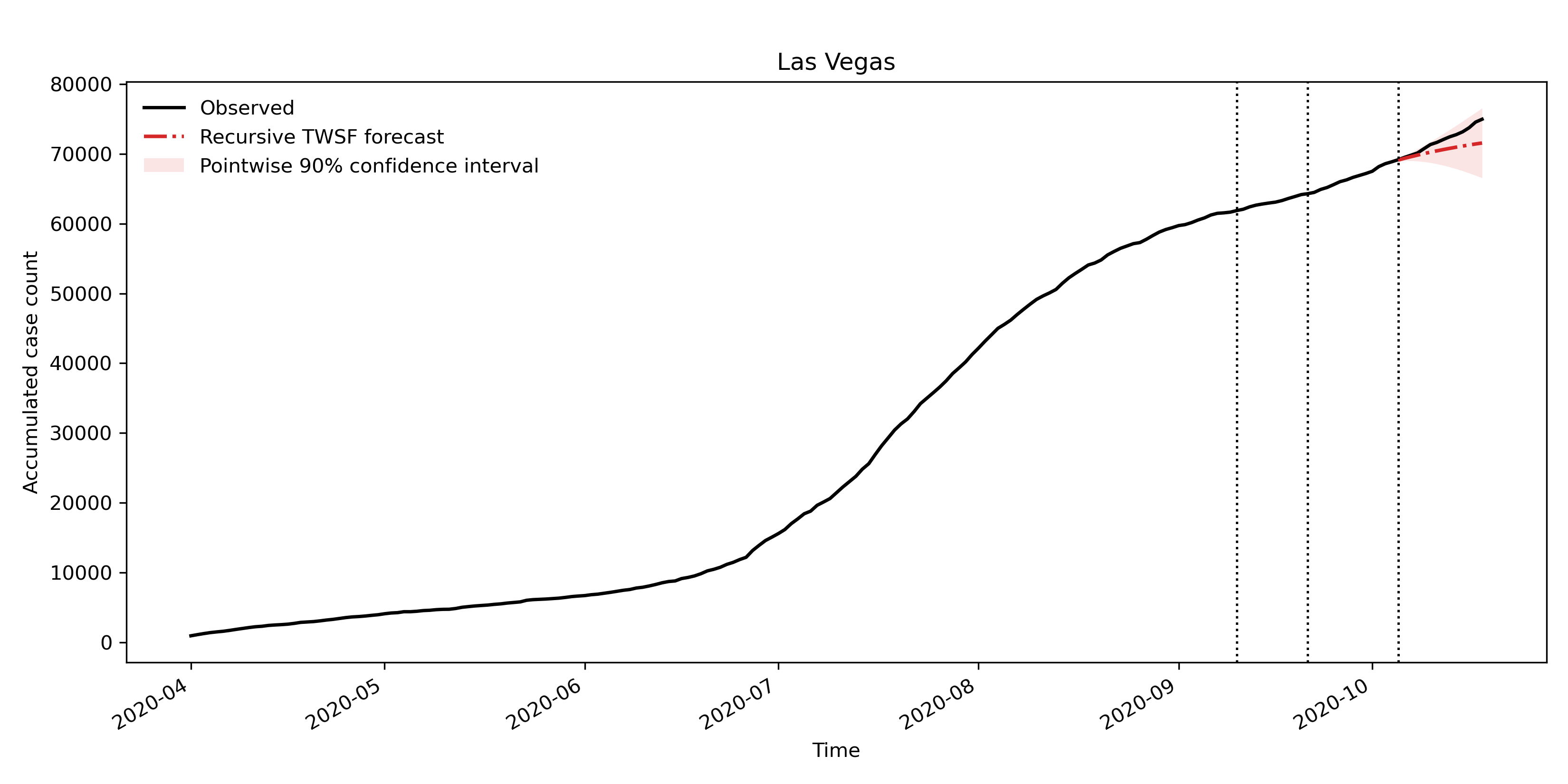}
		\caption{Las Vegas: 10-04.} 
	\end{subfigure} 
	\\
	\begin{subfigure}[b]{0.32\textwidth}
		\centering 
		\includegraphics[width=\linewidth]{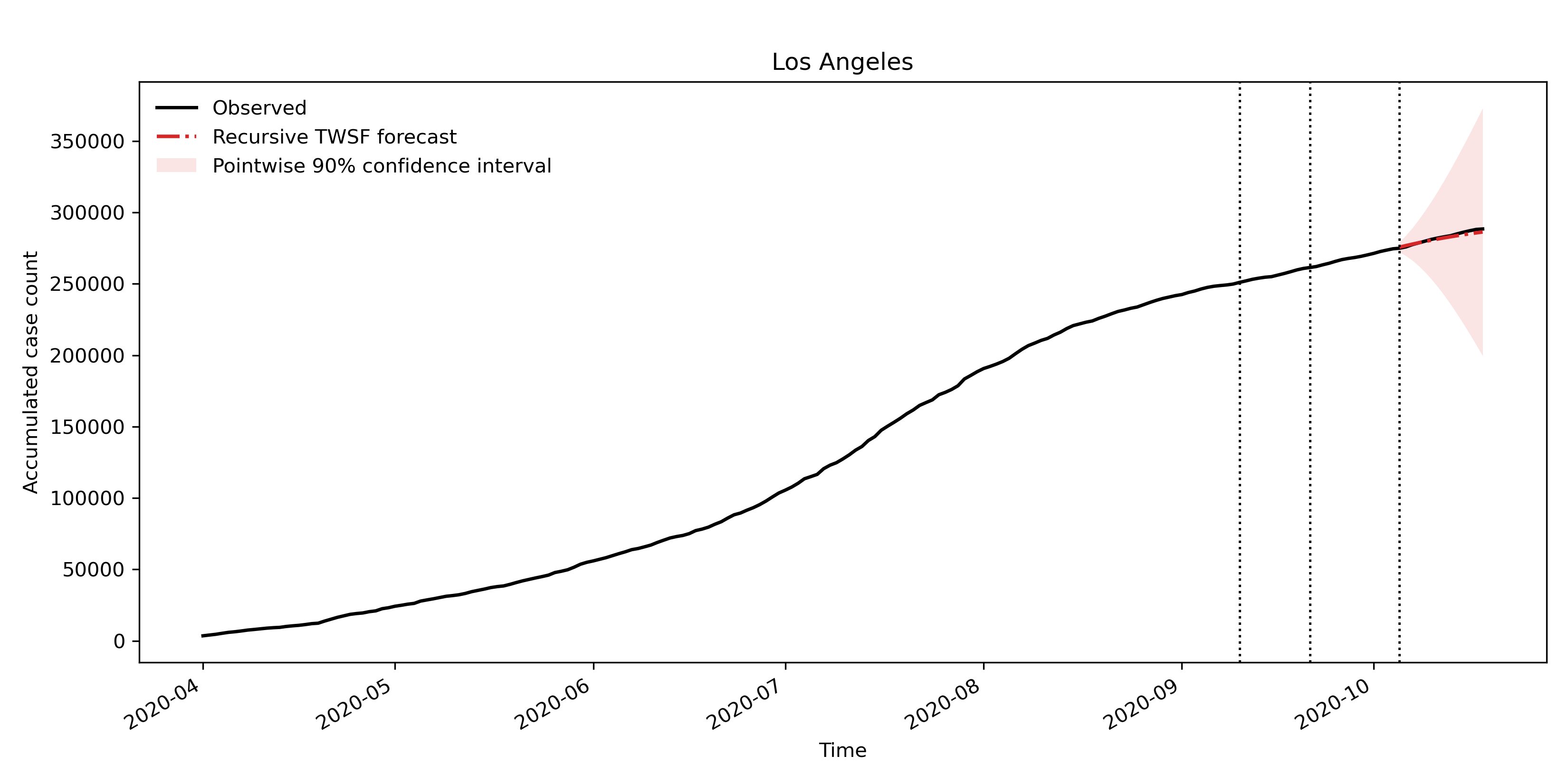}
		\caption{Los Angeles: 10-04.} 
	\end{subfigure} 
	\begin{subfigure}[b]{0.32\textwidth}
		\centering 
		\includegraphics[width=\linewidth]{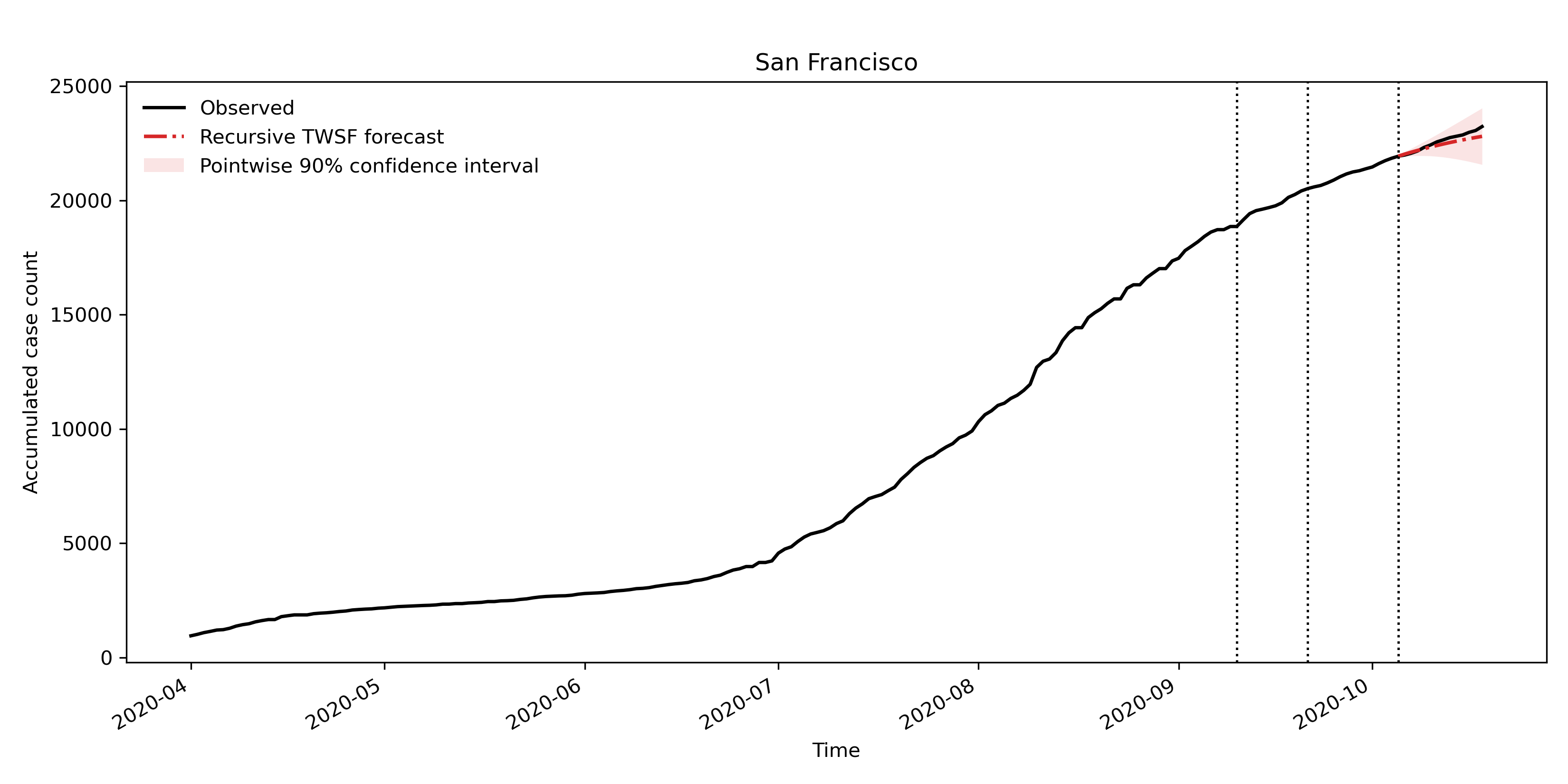}
		\caption{San Francisco: 10-04.} 
	\end{subfigure} 
	\begin{subfigure}[b]{0.32\textwidth}
		\centering 
		\includegraphics[width=\linewidth]{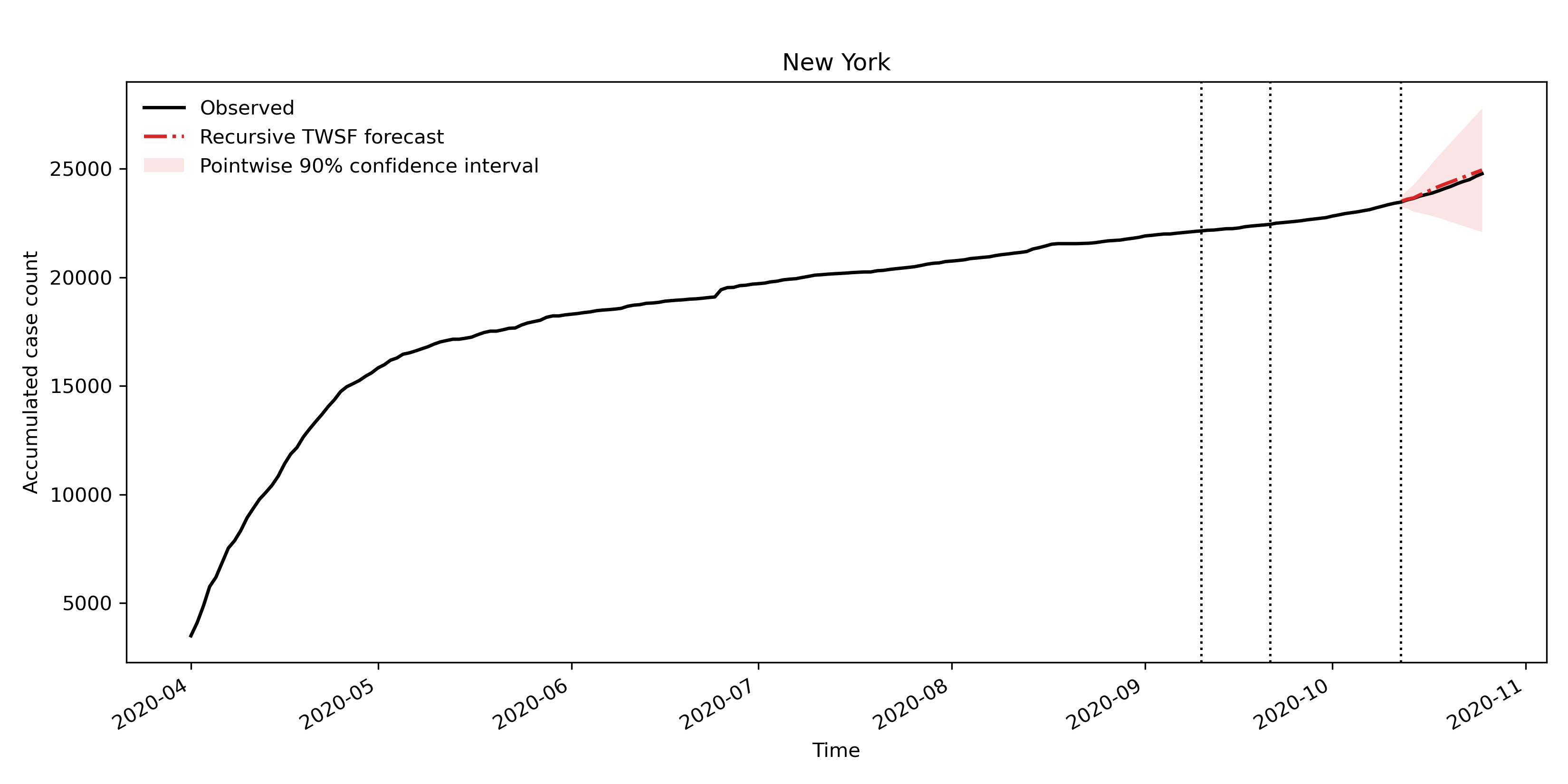}
		\caption{New York: 10-11.} 
	\end{subfigure} 
	\\
	\begin{subfigure}[b]{0.32\textwidth}
		\centering 
		\includegraphics[width=\linewidth]{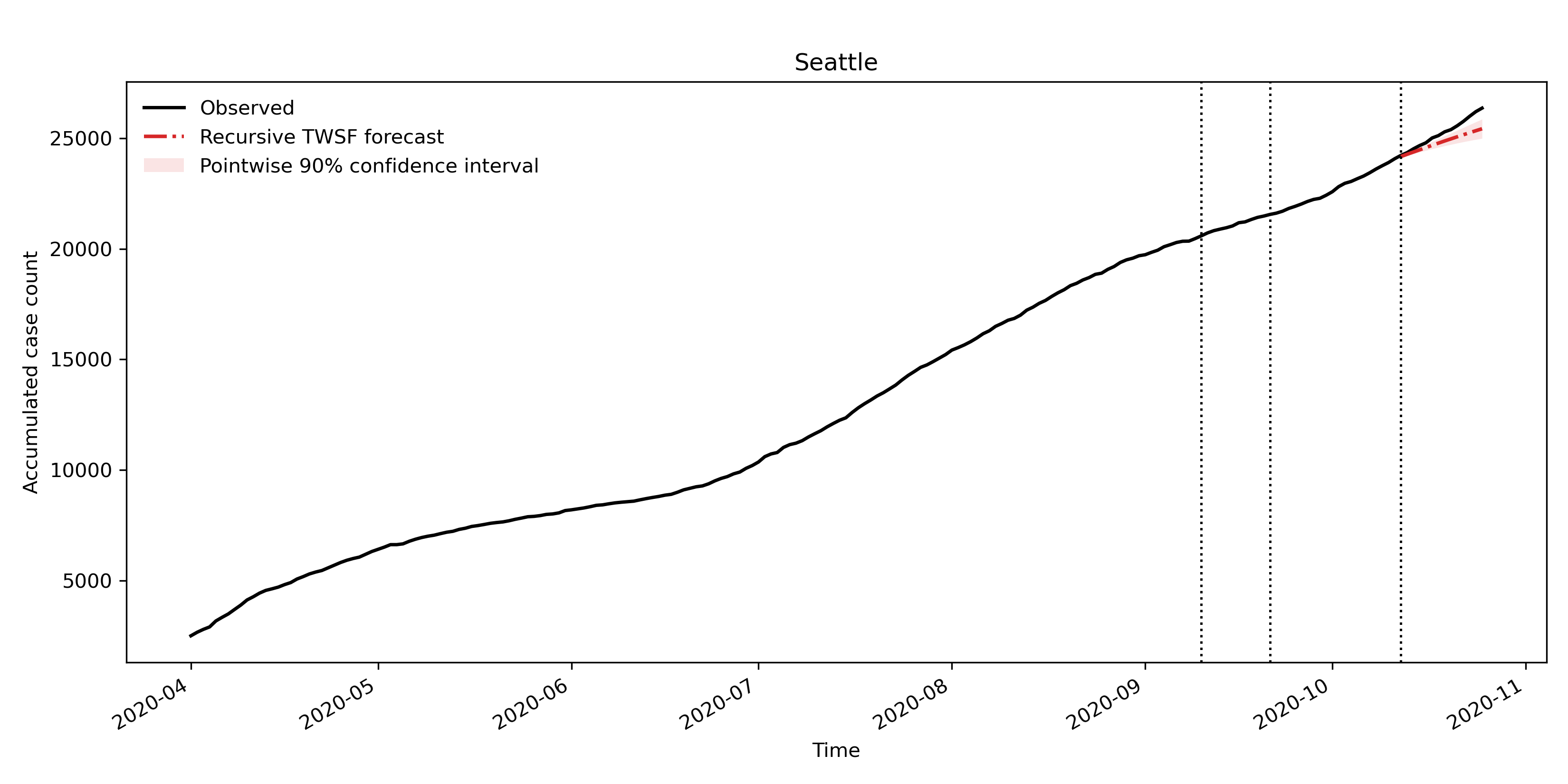}
		\caption{Seattle: 10-11.} 
	\end{subfigure} 
	\begin{subfigure}[b]{0.32\textwidth}
		\centering 
		\includegraphics[width=\linewidth]{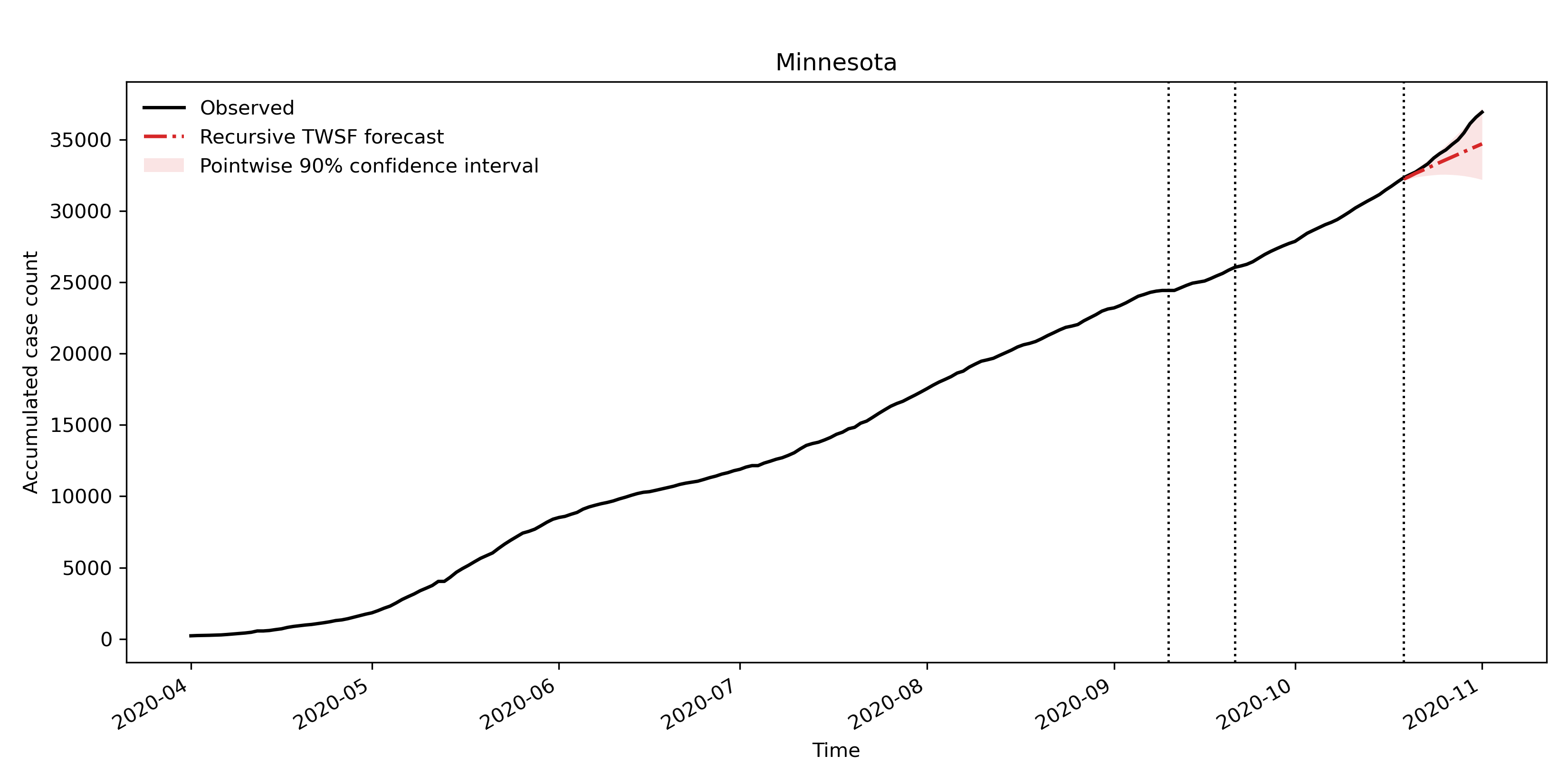}
		\caption{Minnesota: 10-18.} 
	\end{subfigure} 
	\begin{subfigure}[b]{0.32\textwidth}
		\centering 
		\includegraphics[width=\linewidth]{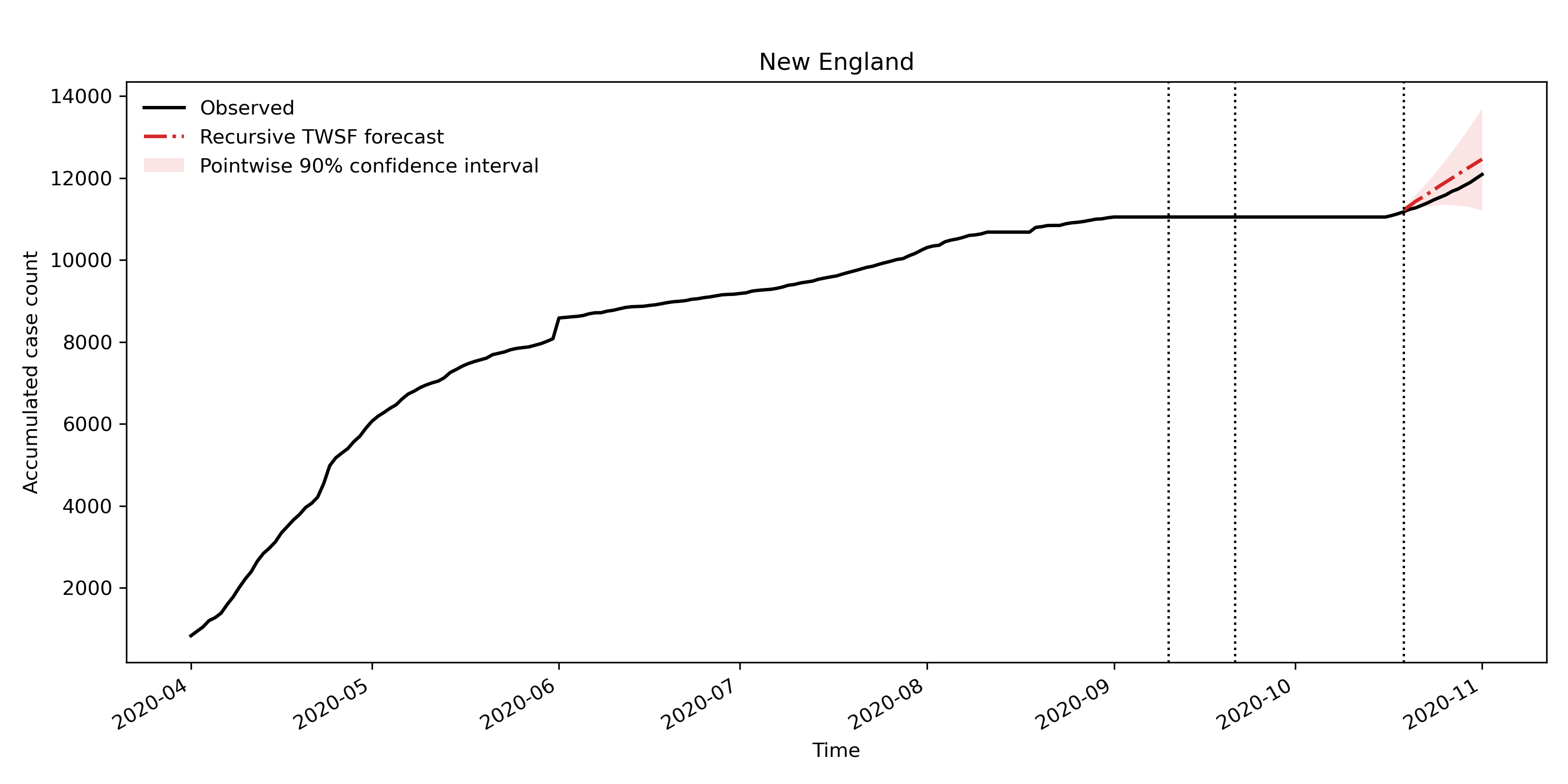}
		\caption{New England: 10-18.} 
	\end{subfigure} 
	\\
	\begin{subfigure}[b]{0.32\textwidth}
		\centering 
		\includegraphics[width=\linewidth]{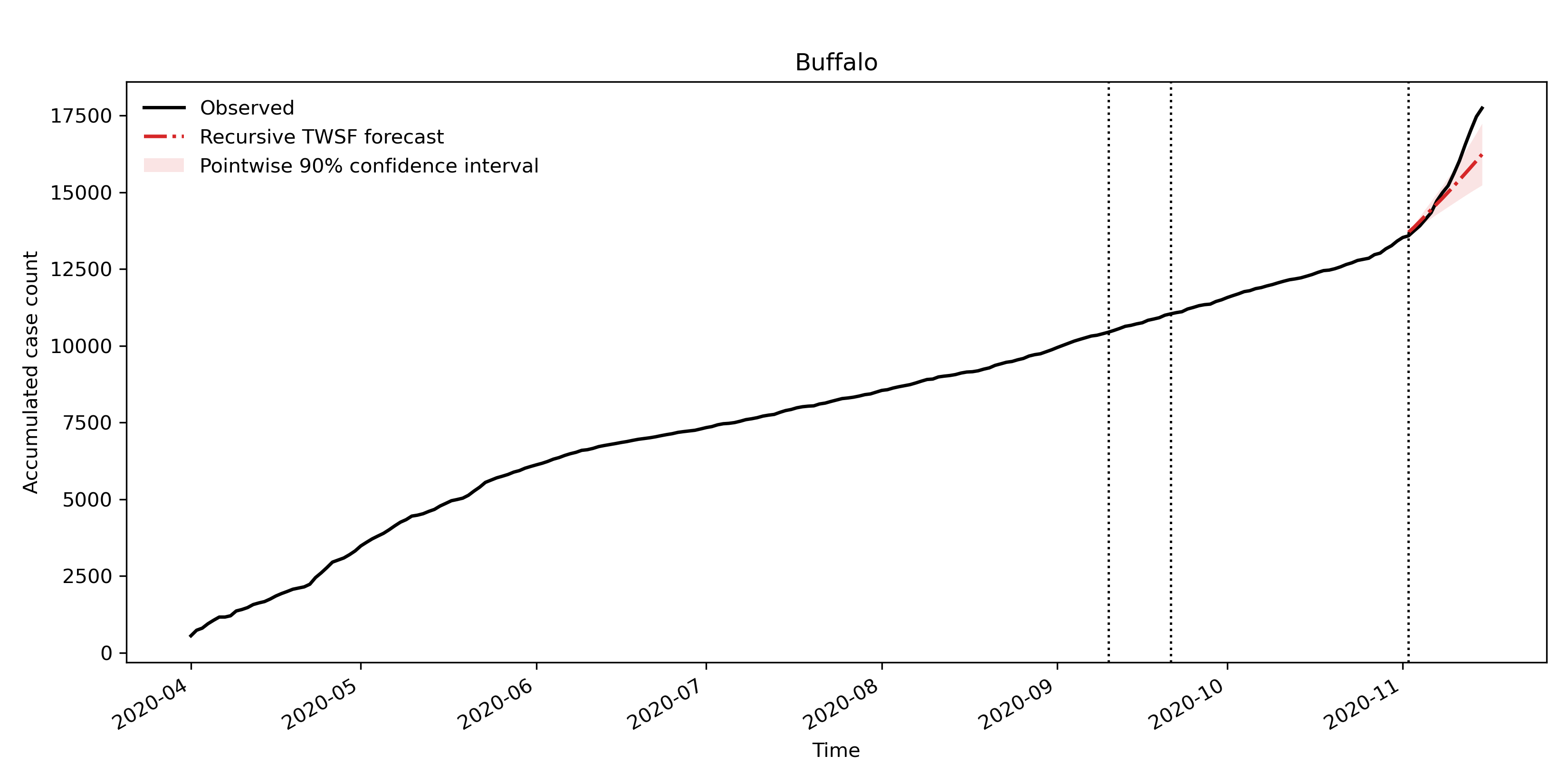}
		\caption{Buffalo: 11-01.} 
	\end{subfigure} 
	\begin{subfigure}[b]{0.32\textwidth}
		\centering 
		\includegraphics[width=\linewidth]{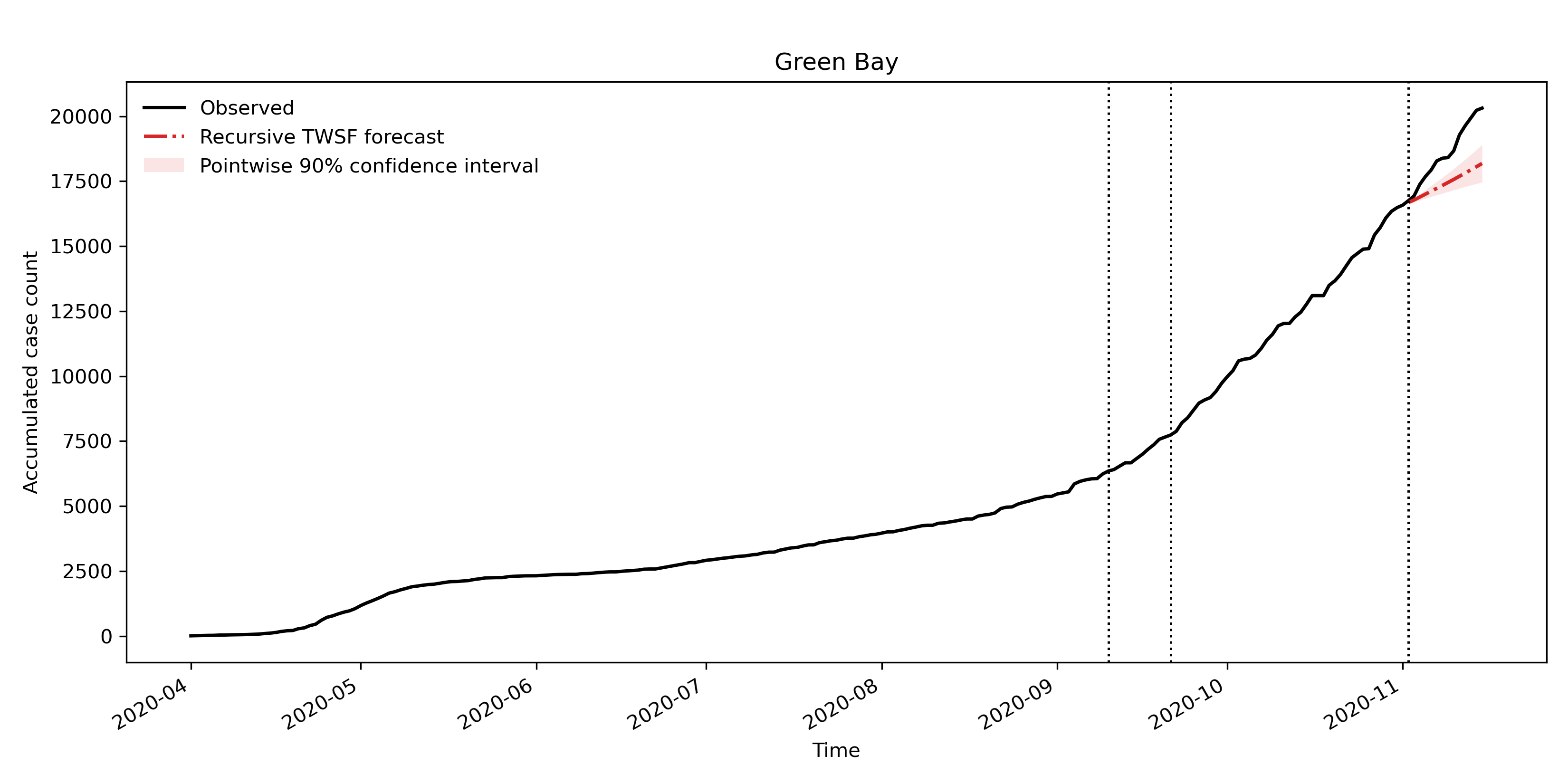}
		\caption{Green Bay: 11-01.} 
	\end{subfigure} 
	\caption{Counterfactual forecasts for cities that kept their stadiums closed, under the hypothetical policy of admitting fans at the first home game after October 1. 
	Black lines show observed cumulative cases, and red dash-dotted lines show \TWSF~forecasts. 
	Shaded regions are pointwise 90\% confidence intervals. 
	The first dotted vertical line marks September 10, the end of the unit-side pre-treatment window; the second marks September 21, the beginning of the treated-donor window used to estimate temporal dynamics; and the third marks the target city’s actual opening date. Forecasts cover the following 14 days.}
	\label{fig:counter.same} 
\end{figure}

\section{Conclusion} \label{sec:conclusion}
This article introduces a framework for causal forecasting in panel data. 
%
The proposed \TWSF~estimator bridges two familiar but distinct ideas: the counterfactual reconstruction logic of \SC~and \SI, and the extrapolative structure of \mssa. 
Under the stated conditions, we establish finite-sample pointwise guarantees and feasible inference for fixed-horizon summaries. 
A sharper \PCR~error bound underlies these results and may also be useful in related low-rank panel methods. 
%
%
%
%
%
Simulations support the theory, while the NFL application illustrates \TWSF's value for prospective policy decisions. 

To the best of our knowledge, this article provides one of the first panel data approaches to prospective causal forecasting for target units that remain untreated throughout the observed sample. \TWSF~should therefore be viewed both as a concrete estimator and a foundation for a broader research agenda. 
%
%
%
Alternative temporal models may yield complementary estimators in settings where \mssa~is less suitable. Other important extensions include allowing network interference and developing theory for growing horizons or simultaneous confidence bands. These advances would strengthen the framework for sequential decision-making and could enable TWSF to serve as a counterfactual forecasting module within reinforcement-learning or dynamic policy-learning pipelines, shifting the emphasis from evaluating proposed interventions to optimizing adaptive policies under uncertainty.

\section*{Acknowledgements}
We are sincerely grateful to Peng Ding for their thoughtful feedback and helpful
suggestions. 

\bibliographystyle{alpha}
\bibliography{bib}

\newpage  
\appendix

\renewcommand{\thepage}{S\arabic{page}}
\renewcommand{\theequation}{S\arabic{equation}}
\renewcommand{\theassumption}{S\arabic{assumption}}
\renewcommand{\thelemma}{S\arabic{lemma}}
\setcounter{equation}{0}
\setcounter{section}{0}
\setcounter{assumption}{0}
\setcounter{lemma}{0}
\setcounter{page}{1}

\makeatletter
\def\@seccntformat#1{\@ifundefined{#1@cntformat}%
   {\csname the#1\endcsname\space}
   {\csname #1@cntformat\endcsname}}
\makeatother
\renewcommand{\thesection}{S\arabic{section}}

\begin{center}
    \LARGE \sc Supplementary material 
\end{center}
\vspace{12pt}

Appendix~\ref{sec:tau} shows how the theoretical results in the main text extend to pointwise prospective treatment effects. 
Appendix~\ref{sec:sims.fine.print} provides additional details on the study in Section~\ref{sec:simulations}.
%

The remaining sections collect the proofs. 
Appendix~\ref{sec:proofs.notation} introduces the notation used throughout the proofs. 
Appendix~\ref{sec:proofs.concentration} collects several standard results from matrix perturbation theory and concentration inequalities that are repeatedly invoked in the arguments. 
Appendix~\ref{sec:proofs.framework} proves the identification results from Section~\ref{sec:framework}, namely Propositions~\ref{prop:beta}--\ref{prop:alpha} and Corollary~\ref{cor:identification}. 
Appendix~\ref{sec:proofs.pcr} develops a general \PCR~error bound, which is then used to prove Theorem~\ref{prop:parameter.recovery}. 
Appendix~\ref{sec:proof.error} proves the high-probability error bound of Theorem~\ref{thm:error}, and Section~\ref{sec:proofs.inference} proves the asymptotic normality result of Theorem~\ref{thm:inference}. 
Appendices~\ref{sec:proof.var.est.2} and \ref{sec:proof.var.est.1} establish the variance estimation results in Propositions~\ref{prop:var.sigma} and \ref{prop:var.asymp}. 
Finally, Appendices~\ref{sec:proof.normality.direct} and \ref{sec:proofs.recursive} prove the fixed-horizon multi-step results for the direct and recursive estimators, corresponding to Theorems~\ref{thm:direct.normality} and \ref{thm:recursive.normality}.

\section{Estimating Pointwise Treatment Effects} \label{sec:tau}

\subsection{Setup and Identification}
The main text forecasts the target unit’s treated potential outcome beyond the observed panel. 
We now combine that forecast with a forecast of the unit’s untreated potential outcome. 
Define $\Fc_\tau \coloneqq \sigma\big(\bu_N, \bv_{T+1}(0), \bv_{T+1}(1) \big)$, and, for $d \in \{0,1\}$, let $\theta_d \coloneqq \Ex[Y_{N, T+1}(d) \mid \Fc_\tau ]$. 
The prospective treatment effect estimand is
\begin{align}
	\tau \coloneqq \Ex\left[Y_{N, T+1}(1) - Y_{N, T+1}(0) \mid \Fc_\tau \right] = \theta_1 - \theta_0. 
\end{align}
Under Assumptions~\ref{assump:lfm}-\ref{assump:mean_ind}, $\theta_d = \langle \bu_N, \bv_{T+1}(d) \rangle$. 
The treated component $\theta_1$ coincides with the estimand $\theta$ defined in \eqref{eq:estimand} from the main text. 
It therefore remains to forecast $\theta_0$. 

For each $a \in [r]$, define the control-state temporal factor $f_a^0(t) \coloneqq V_{ta}(0)$. 
We impose a control-state analogue of Assumption~\ref{assump:hankel}. 

\begin{assumption} \label{assump:hankel.0}
For every $a \in [r]$, positive integers $m, n \in \mathbb{N}$, and shift $s \in \mathbb{Z}$, let $\emph{rank}(\bH(f_a^0; m, n, s)) \le Q_0$.
\end{assumption}

As in Proposition~\ref{prop:alpha}, Assumption~\ref{assump:hankel.0} implies that, for any lag length $L_0 \ge r Q_0$, there exists a coefficients vector $\balpha_0 \in \Rb^{L_0}$, common across units, such that 
\begin{align}
	\Ex\left[Y_{i, t+1}(0) \mid \Ec \right] = \sum_{\ell=1}^{L_0} \alpha_{0\ell} \cdot \Ex\left[Y_{i, t-L_0+\ell}(0) \mid \Ec \right]. 
	\label{eq:tau.1}
\end{align}
Because the target unit remains under control through $T$, its most recent control outcomes are observed.
Thus, \eqref{eq:tau.1} identifies
\begin{align}
	\theta_0 = \sum_{\ell=1}^{L_0} \alpha_{0\ell} \cdot \Ex\left[Y_{N, T-L_0+\ell} \mid \Ec \right]. \label{eq:tau.2}
\end{align}
Unlike the treated forecast, this control-side forecast requires neither cross-unit weights nor an analogue of Assumption~\ref{assump:units}. 
Persistent control units may nevertheless be pooled to estimate the common temporal rule more precisely. 
When the target unit is the only persistent control, it is sufficient to impose finite Hankel rank directly on the target unit’s scalar sequence of control-state conditional means.

\subsection{Prospective Treatment Effect Estimator}
Recall that $\Ic_0$ denotes the set of units observed under control throughout the post-treatment period, and let $N_0 \coloneqq | \Ic_0|$. 
Whereas the main text studies the baseline case of $N_0 = 1$, we now allow any $N_0 \ge 1$. 

%
Fix a lag length $L_0$, and define
$B_0 \coloneqq T_1 / (L_0 + 1) \in \Nb$ and $B_0^\full \coloneqq T / (L_0+1) \in \Nb$. 
For the target unit $N$, form a Page matrix from the first $B_0-1$ complete post-treatment blocks: 
$\big(\bP^0_{N, \post}\big)_{ab} \coloneqq Y_{N, T_0+(b-1)(L_0+1)+a}$ for $a \in [L_0+1]$ and $b \in [B_0-1]$. 
For each non-target persistent control $i \in \Ic_0 \setminus \{N\}$, form a Page matrix from all complete blocks in its full observed history: 
$\big(\bP^0_{i, \full}\big)_{ab} \coloneqq Y_{i, (b-1)(L_0+1)+a}$ for $a \in [L_0+1]$ and $b \in [B_0^{\full}]$. 
Stacking these matrices horizontally gives 
\begin{align}
	\bP^0 \coloneqq
	\left[ \bP^0_{N, \post} ~  \left\{\bP^0_{i, \full}: i \in \Ic_0 \setminus \{N\} \right\} \right]
	\eqqcolon 
	\begin{pmatrix}
		\bZ_{0, \lag}
		\\
		\hdashline
		\bz_{0, \tnext}^\top
	\end{pmatrix}
	\in \Rb^{(L_0+1) \times M_0},
	\label{eq:Zlag.znext.0} 
\end{align}
where $M_0 \coloneqq (B_0-1) + (N_0-1)B_0^{\full}$. 
Here, $\bZ_{0, \lag} \in \Rb^{L_0 \times M_0}$ contains the lag vectors while $\bz_{0, \tnext} \in \Rb^{M_0}$ contains the corresponding one-step-ahead responses. 
The target unit's final $L_0$ observations form the forecast state 
$\bw_0 \coloneqq \big[ Y_{N, T-L_0+1}, \dots, Y_{NT} \big]^\top \in \Rb^{L_0}$. 
The estimator proceeds in three steps. 

\begin{enumerate} [label=(\alph*)]

	\item For any $k_0 \le \min\{L_0, M_0\}$, define $\bZ_{0, \lag}^{(k_0)} \coloneqq \HSVT(\bZ_{0, \lag}, k_0)$.

	\item Estimate the common control temporal rule by 
	$\hbalpha_0 \coloneqq \left(\bZ_{0, \lag}^{(k_0)} \right)^{\top, \dagger} \bz_{0, \tnext}$.

	\item Let $\htheta_1$ be the treated forecast defined in \eqref{eq:htheta}. The untreated forecast is $\htheta_0 \coloneqq \left\langle \hbalpha_0, \bw_0 \right 	\rangle$, and the final prospective treatment effect estimate is $\htau \coloneqq \htheta_1 - \htheta_0$.

\end{enumerate}

%
The target unit’s pre-treatment history is deliberately omitted from the control-side training matrix because those observations already enter the unit-side regression used to estimate $\htheta_1$. 
By contrast, non-target persistent controls do not enter the treated-side \TWSF~fit. 
Their full pre- and post-treatment histories may therefore be pooled when estimating the control-state temporal rule. 
The untreated component consequently uses \TWSF’s time-side \PCR~procedure but does not require its unit-side synthesis step.
When $N_0 = 1$, the temporal rule is estimated solely from the target unit’s earlier post-treatment blocks. 
%


\subsection{Feasible Inference} 
Let $\hVc_1^2$ denote the feasible variance estimator for $\htheta_1$ from Section~\ref{sec:inference}. 
Define the analogous control-side variance estimator for $\htheta_0$ by 
\begin{align}
	\hVc_0^2 \coloneqq \hsigma^2 \left\{ \| \hbalpha_0 \|_2^2 + \| \hbq_0 \|_2^2 \left(1 + \| \hbalpha_0 \|_2^2 \right) \right\},
\end{align}
where $\hbq_0 \coloneqq \big(\bZ_{0, \lag}^{(k_0)} \big)^\dagger \bw_0$.  
Under homoskedasticity, the pooled estimator $\hsigma^2$ in \eqref{eq:var.sigma.pool} may be used for both the treated- and control-side components. 
Alternatively, the control-side residual sum of squares may be incorporated into \eqref{eq:var.sigma.pool}, with each component weighted by its residual degrees of freedom. 
The corresponding control-side residual estimator is
\begin{align}
	\hsigma^2_0 \coloneqq \frac{\| \bz_{0,\tnext} - \bZ_{0,\lag}^\top \hbalpha_0 \|_2^2}{M_0 - k_0}.
\end{align}

Suppose that, in addition to Assumption~\ref{assump:hankel.0}, the control-side analogs of Assumptions~\ref{assump:transport} and \ref{assump:spectra} hold. 
Under the homoskedastic version of Assumption~\ref{assump:subg}, it can be established that 
\begin{align}
	\frac{\htau - \tau}{\left(\hVc_1^2 + \hVc_0^2\right)^{1/2}} \rightsquigarrow \Nc(0,1). \label{eq:tau.inf} 
\end{align}
Consequently, a pointwise $(1-\gamma)100\%$ confidence interval for $\tau$  is 
\begin{align}
	\texttt{CI}_{\tau}(\gamma) \coloneqq 
	\left[ \htau \pm \Phi^{-1}(1-\gamma/2) \cdot \left(\hVc_1^2 + \hVc_0^2 \right)^{1/2} \right]. 
\end{align}
To obtain \eqref{eq:tau.inf}, apply the orthogonalization argument of Theorem~\ref{thm:inference} to the single control-side \PCR~stage. 
Within the control estimator, the Page construction separates the blocks used to estimate the temporal rule from the terminal observations used as the forecast state. 
Across the two components, the control-side fit uses persistent-control observations and the target unit’s post-treatment control outcomes, whereas treated-side \TWSF~uses the treated-donor blocks and the target unit’s pre-treatment history.
Under Assumption~\ref{assump:subg}, the resulting leading errors are conditionally independent, so their variances add.
Consistency of the two plug-in variance estimators, followed by Slutsky’s theorem, yields \eqref{eq:tau.inf}.

If the Page blocks overlap, or if the idiosyncratic shocks are dependent across units or time, the variance must include an estimated covariance term. 
The same construction extends to the fixed-horizon estimands in Section~\ref{sec:forecast.horizon} by applying the direct or recursive temporal procedure separately to the control outcomes.



\section{Simulation Studies: Fine Print} \label{sec:sims.fine.print}

\subsection{Simulation Design} \label{sec:sims.design}

\subsubsection{Setup}  \label{sec:sims.setup} 
Consider horizons $h \in \{1, 5, 10\}$ and index simulations by $n \in \{25, 50, 75, 100, 125, 150\}$. 
For each $n$, set $N_1=T_0=L=n$, and define $T_1(n) \coloneqq  8(n+10)$ with $T(n) \coloneqq T_0+T_1(n)$.  
At horizon $h$, direct \TWSF~uses Page blocks of length $L + h$, giving $B_h = \lfloor T_1(n) / (L+h) \rfloor$ complete blocks. 
Recursive \TWSF~uses blocks of length $L+1$, giving $B_{\mathrm{rec}} = \lfloor T_1(n) / (L+1) \rfloor$. 
Incomplete trailing blocks are discarded. 

To align comparisons across values of $n$, we use a nested design with a common forecast origin. 
Each latent replication is first generated at $n_{\max} = 150$. 
A smaller panel retains the first $n$ donors and a trailing window of length $T(n) + 10$, with its observed portion ending at the same forecast date as the $n_{\max}$ panel. 
Thus, every value of $n$ forecasts the same future dates while $N_1$, $T_0$, $L$, and $T_1$ increase with $n$. 

We generate 100 independent latent factor realizations and, conditional on each one, 10 independent noise realizations. 
This yields $R=100\times 10=1000$ replications per design point. 
Panels of different dimensions share the corresponding entries of the same maximum-dimension latent and noise arrays. 
Monte Carlo standard errors are computed over the 100 latent factor clusters, with the 10 conditional noise draws treated as observations within each cluster.

\subsubsection{Data-Generating Process} \label{sec:sims.dgp}
For each latent replication, generate 150 donor factors $\bu_j \coloneqq [1,\xi_{j2},\xi_{j3},\xi_{j4}]^\top$, where the last three coordinates are independent standard normal draws. 
Each of these coordinates is then standardized across the 150 donors to have empirical mean zero and empirical standard deviation one.
Next, sample a set $\Ac$ of eight donors uniformly from the first 25 donor rows. 
Draw convex weights $(\lambda_j)_{j \in \Ac}$ from a symmetric Dirichlet distribution and set  $\bu_N=\sum_{j \in \Ac} \lambda_j\bu_j$. 
Because every simulated panel contains the first 25 donors, the target factor lies in the donor span for every value of $n$.

Define the maximum time index by $T_\star=T(150)+10$. 
Set $\omega_0=2\pi/(8T_\star)$, $\omega_1=2\pi/12$, $\omega_2=2\pi/37$, and $\omega_3=2\pi/91$. 
For $x \in \Rb$, define the harmonic pair $\psi(x) \coloneqq [\sin(x), \cos(x) ]^\top$. 
The harmonic basis is $\boldsymbol{b}(t) \coloneqq \big[ \psi(\omega_0t - \pi/3), \psi(\omega_1 t), \psi(\omega_2 t), \psi(\omega_3 t) \big]^\top$. 
The control- and treated-state time factors are then $\bv_t(0)\coloneqq\bA_0\boldsymbol{b}(t)$ and $\bv_t(1)\coloneqq\bA_1\boldsymbol{b}(t)$, where $\bA_0,\bA_1\in\Rb^{4\times 8}$ are fixed across replications. 
The lowest-frequency harmonic is absent under control and present under treatment, producing a smooth treatment-specific component. 
The remaining harmonics generate medium- and long-period oscillations under both states. 
Smaller panels inherit the corresponding trailing subwindow of these factors.

After constructing the raw signals at the maximum dimension, we apply a common scale factor so that $\max_{i,t,d}|\langle \bu_i, \bv_t(d) \rangle|\leq 0.8$. 
The same scaling is retained in every smaller panel. 
We then add mutually independent Gaussian noise with mean zero and standard deviation $\sigma=0.10$. 
The noise arrays are generated at $n_{\max}$ and subset in the same way as the signal arrays, so corresponding entries are shared across dimensions within a replication but remain independent across units, times, intervention states, and noise realizations.
The population ranks are fixed throughout the grid: $r_y=4$ and $r_z=8$. 

\begin{figure}[!tp]
	\centering 
	\begin{subfigure}[b]{0.32\textwidth}
		\centering 
		\includegraphics[width=\linewidth]{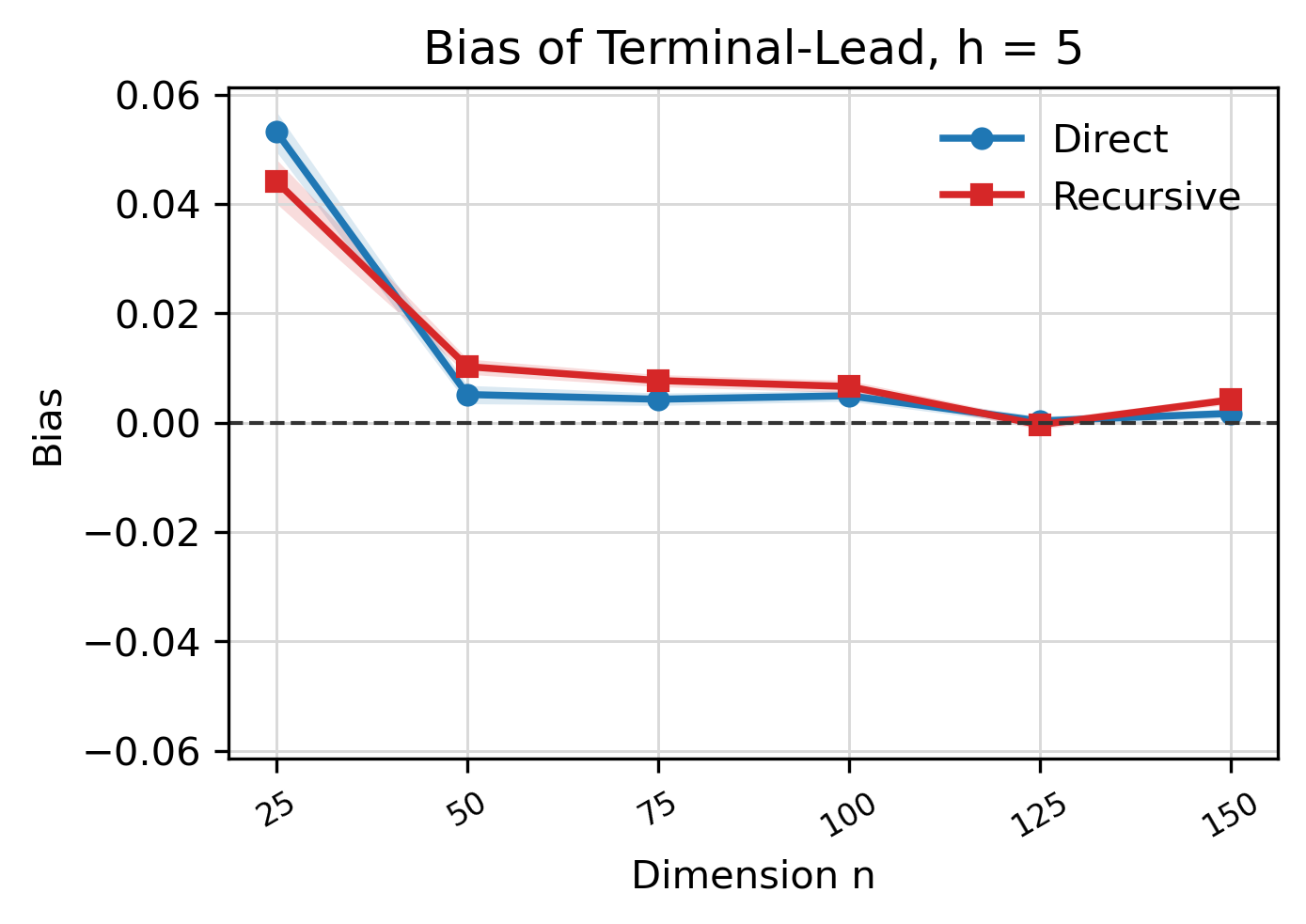}
		\caption{Bias $h=5$} 
		\label{fig:terminal.bias.5}
	\end{subfigure} 
	\begin{subfigure}[b]{0.32\textwidth}
		\centering 
		\includegraphics[width=\linewidth]{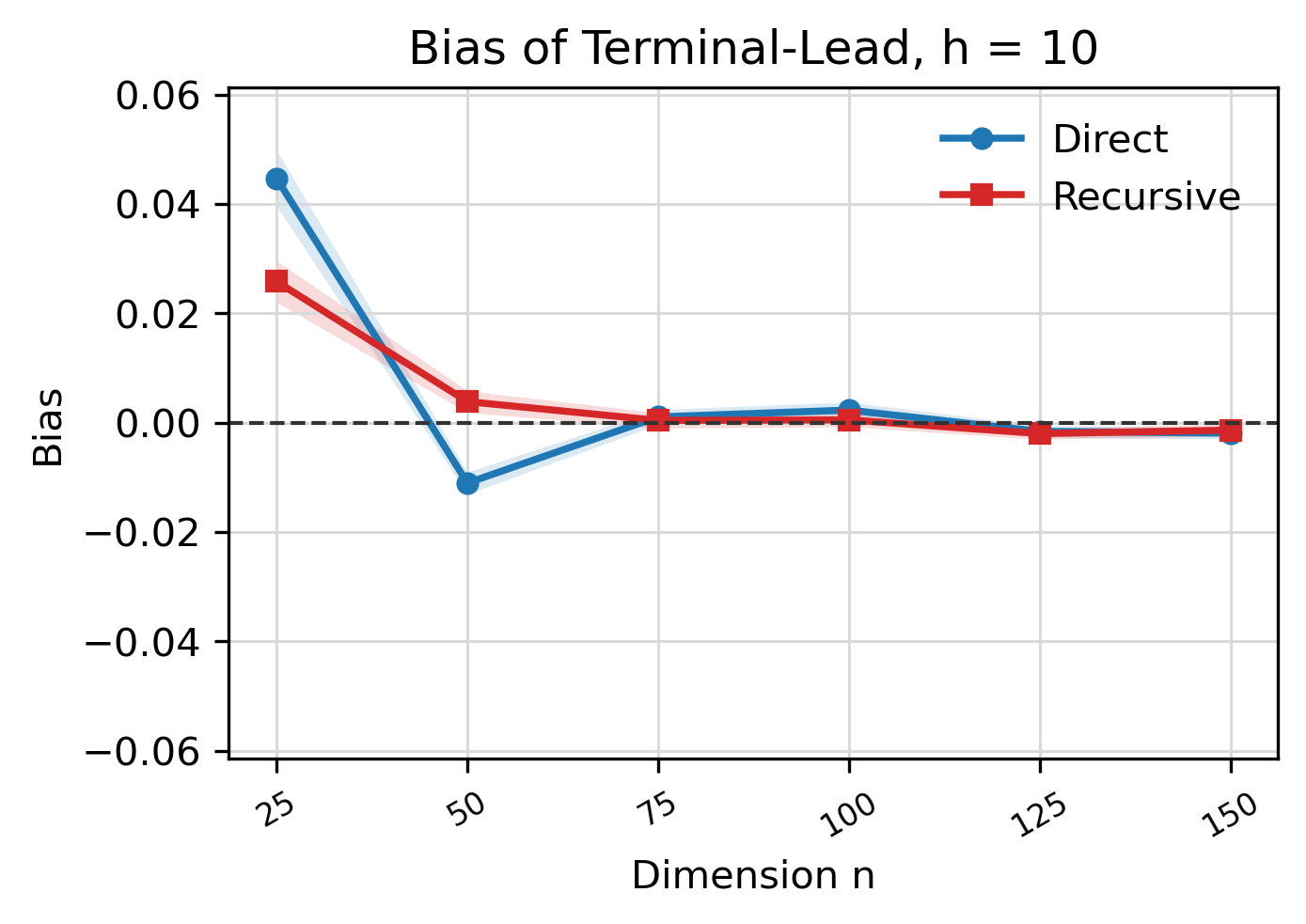}
		\caption{Bias $h=10$} 
		\label{fig:terminal.bias.10}
	\end{subfigure}
	\\
	\begin{subfigure}[b]{0.32\textwidth}
		\centering 
		\includegraphics[width=\linewidth]{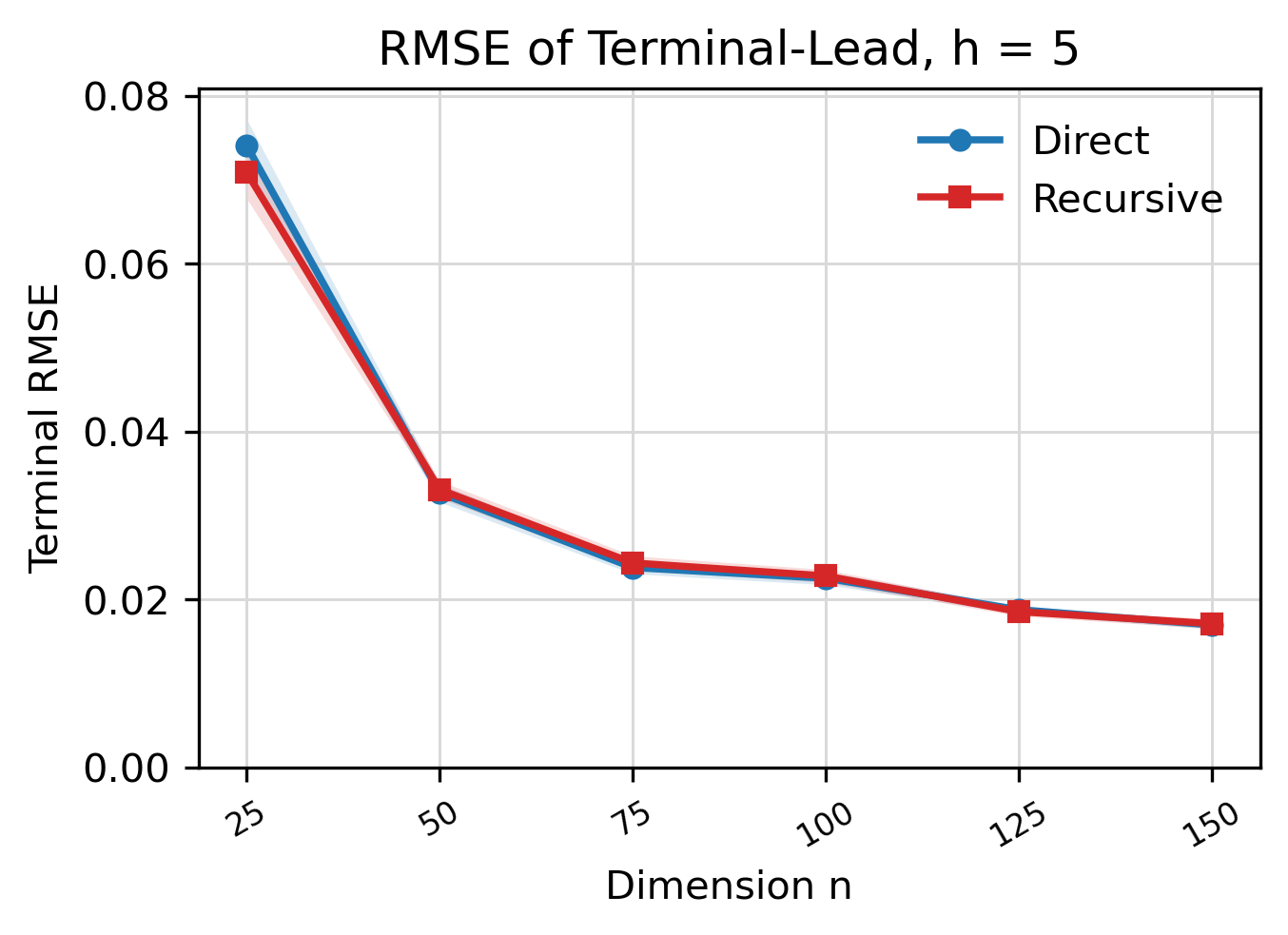}
		\caption{RMSE $h=5$} 
		\label{fig:terminal.rmse.5}
	\end{subfigure} 
	\begin{subfigure}[b]{0.32\textwidth}
		\centering 
		\includegraphics[width=\linewidth]{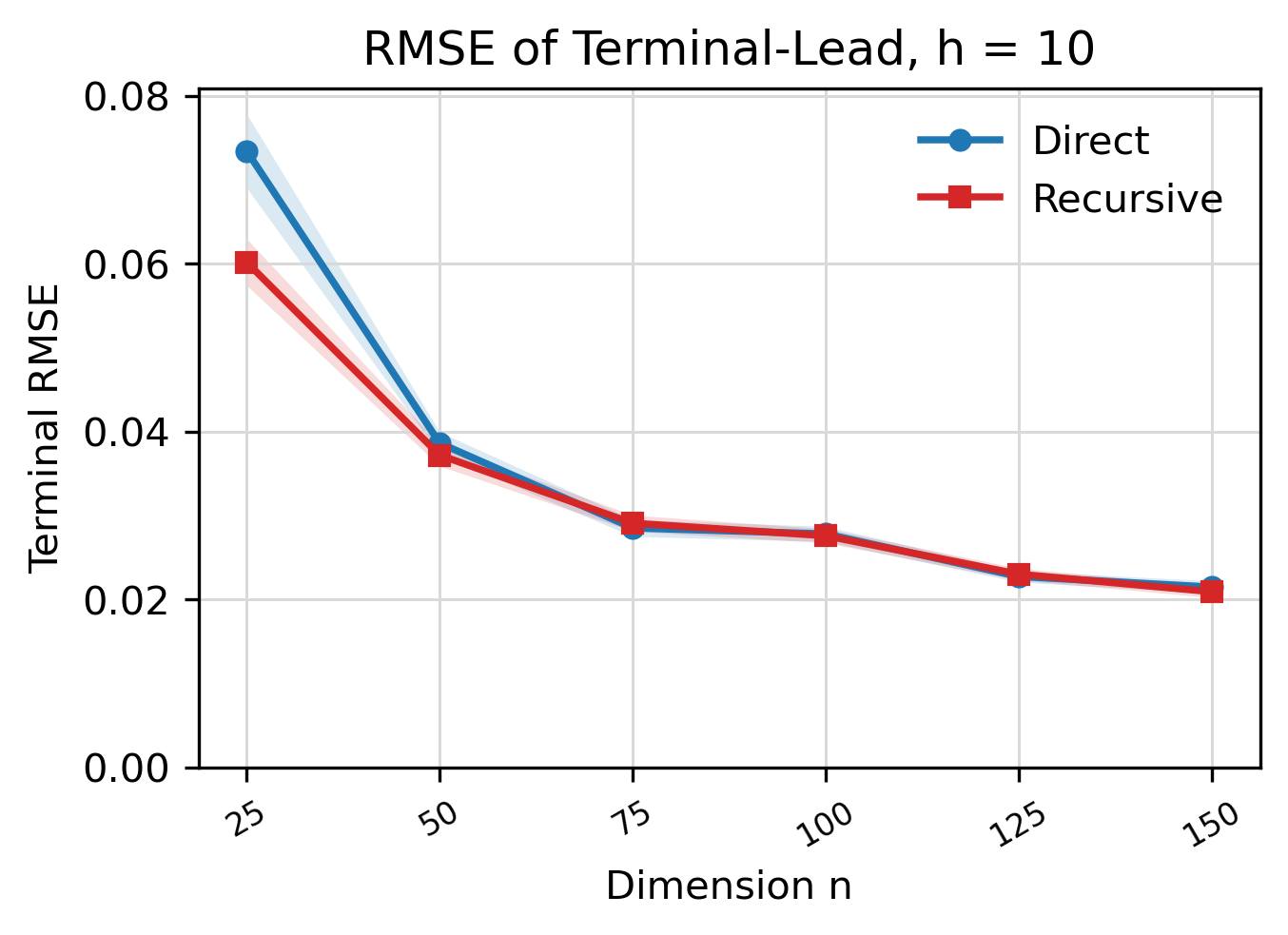}
		\caption{RMSE $h=10$} 
		\label{fig:terminal.rmse.10}
	\end{subfigure} 
	\\
	\begin{subfigure}[b]{0.32\textwidth}
		\centering 
		\includegraphics[width=\linewidth]{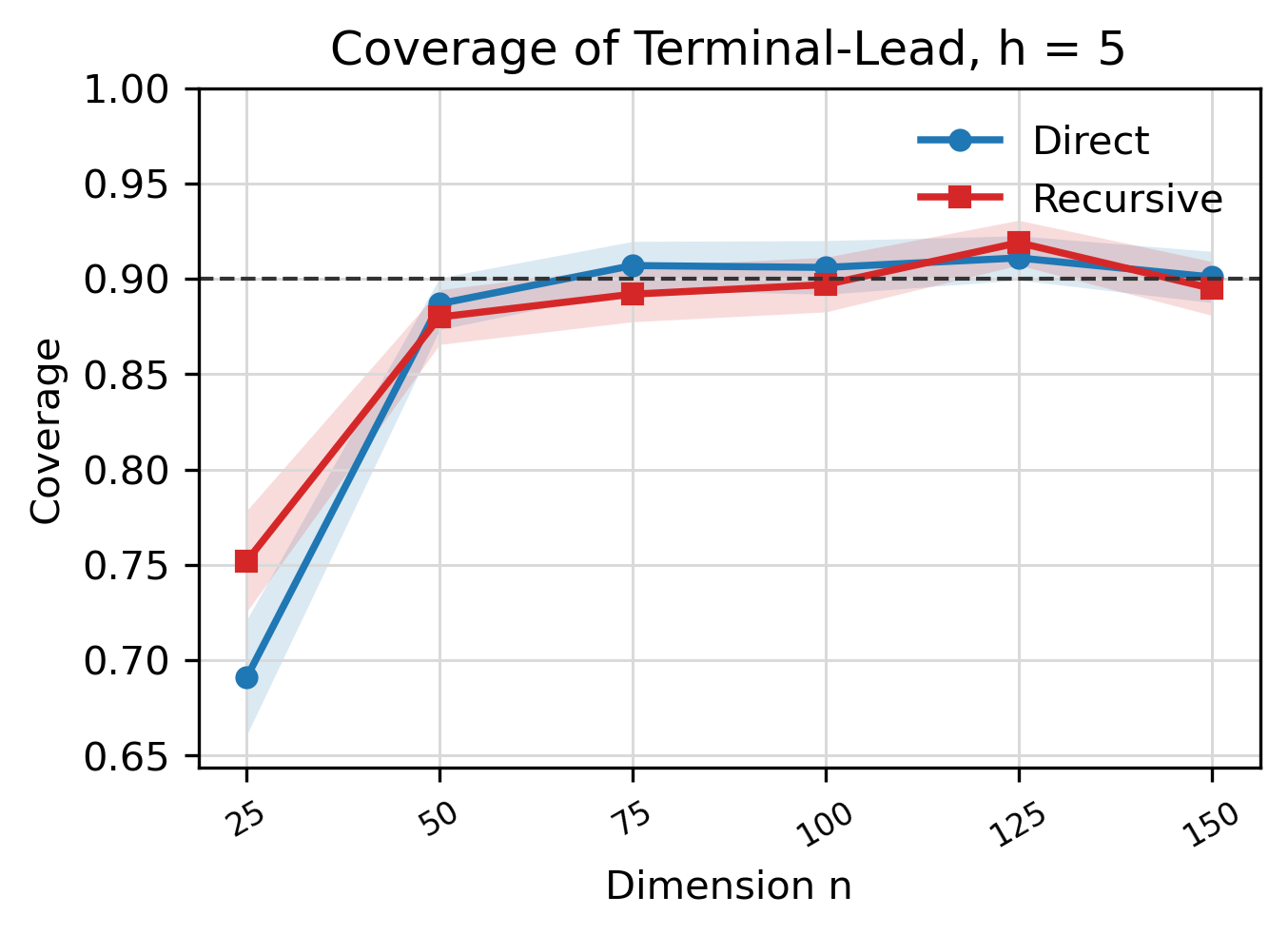}
		\caption{Coverage $h=5$} 
		\label{fig:terminal.cov.5}
	\end{subfigure} 
	\begin{subfigure}[b]{0.32\textwidth}
		\centering 
		\includegraphics[width=\linewidth]{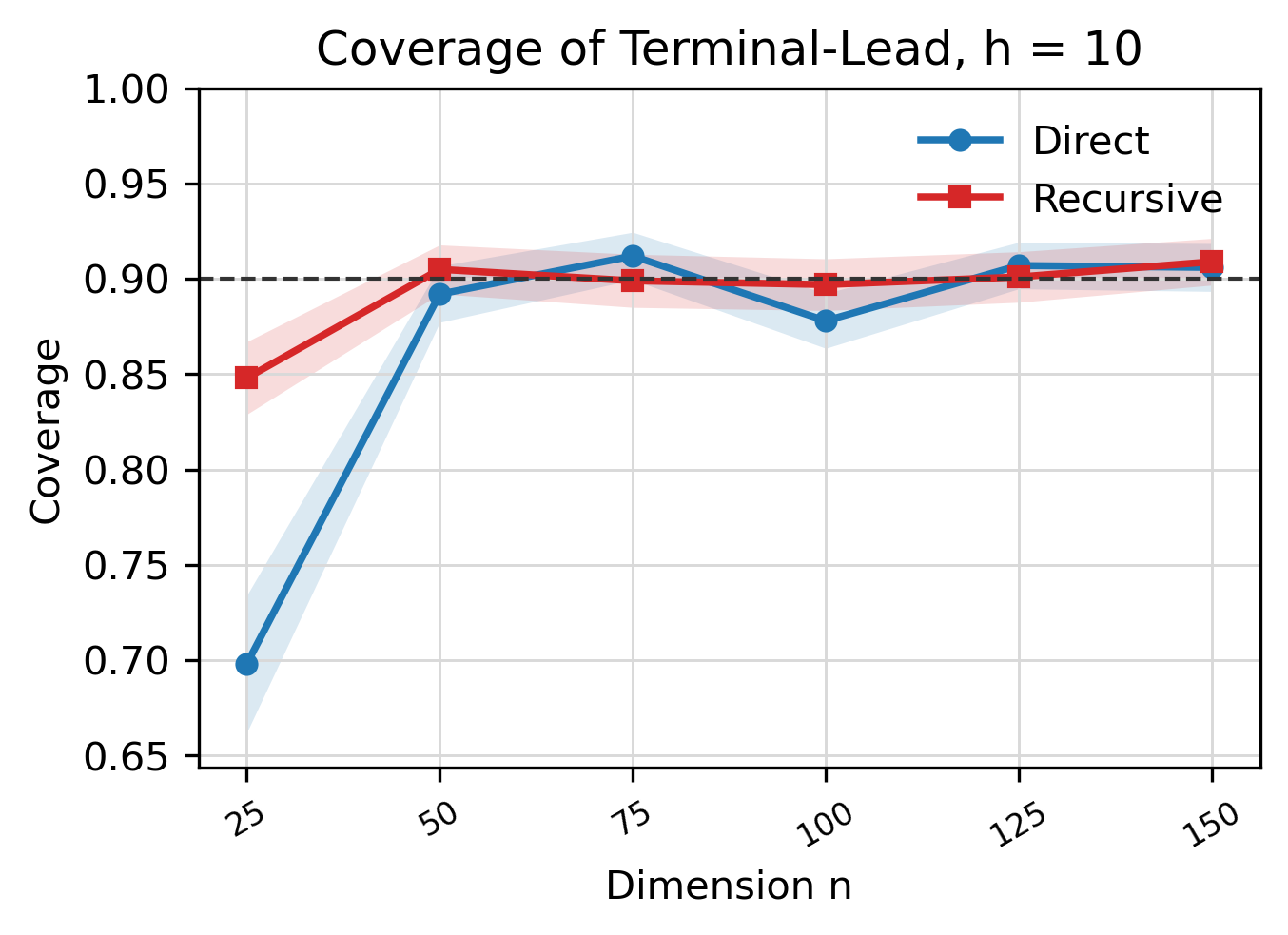}
		\caption{Coverage $h=10$} 
		\label{fig:terminal.cov.10}
	\end{subfigure}
	\\
	\begin{subfigure}[b]{0.32\textwidth}
		\centering 
		\includegraphics[width=\linewidth]{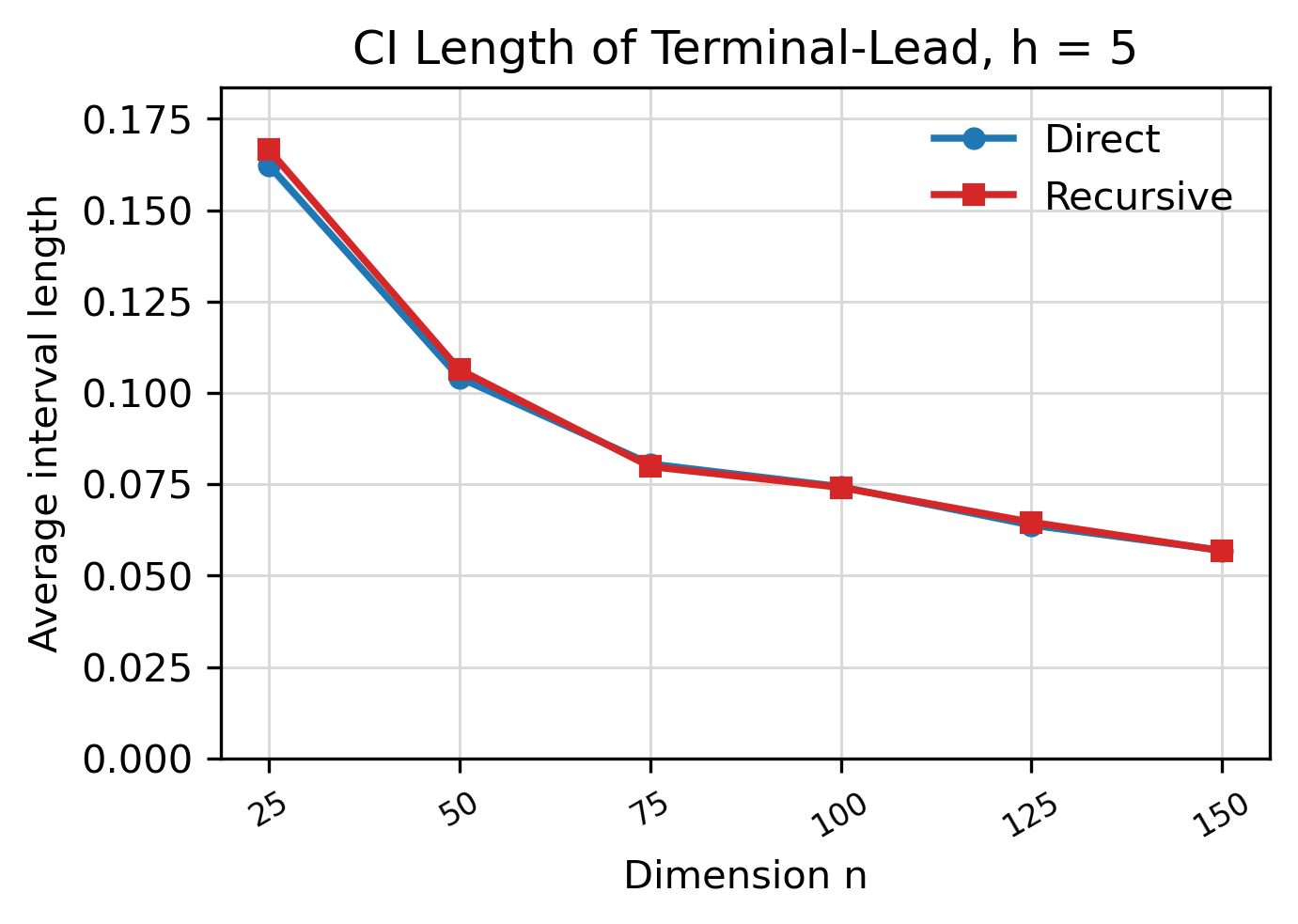}
		\caption{Length $h=5$} 
		\label{fig:terminal.len.5}
	\end{subfigure} 
	\begin{subfigure}[b]{0.32\textwidth}
		\centering 
		\includegraphics[width=\linewidth]{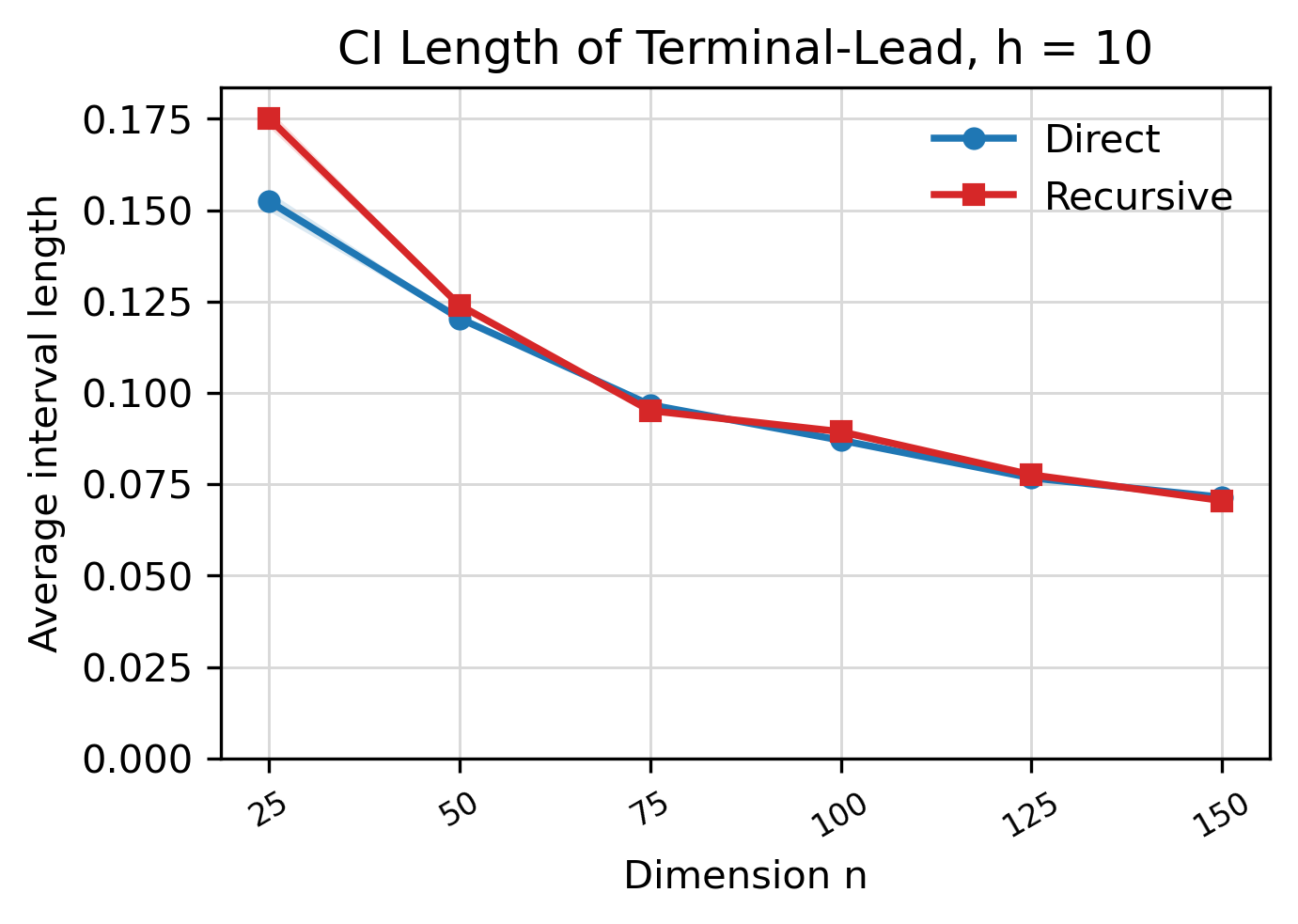}
		\caption{Length $h=10$} 
		\label{fig:terminal.len.10}
	\end{subfigure}
	\caption{Finite-sample performance of direct and recursive \TWSF~for terminal-lead estimands. Panels~\ref{fig:terminal.bias.5}--\ref{fig:terminal.bias.10} report bias, \ref{fig:terminal.rmse.5}-\ref{fig:terminal.rmse.10} report RMSE, \ref{fig:terminal.cov.5}-\ref{fig:terminal.cov.10} report 90\% pointwise coverage, and \ref{fig:terminal.len.5}-\ref{fig:terminal.len.10} report average interval length. Blue and red denote direct and recursive \TWSF; shaded bands are clustered Monte Carlo standard errors. Results are based on 1000 replications.}
	\label{fig:terminal.sims}  
\end{figure}

\subsubsection{Performance Measures} \label{sec:sims.metrics}
For horizon $h$, method $m \in \{\direct, \rec\}$, and replication $s \in [R]$, define the forecast error by $\Delta_{\eta_h}^{(s,m)} \coloneqq \htheta^{(s,m)}_{\eta_h} - \theta_{\eta_h}^{(s)}$. 
We measure centering and estimation accuracy using 
\begin{align}
	\textsf{Bias}_{\eta_h}^{(m)} \coloneqq \frac{1}{R} \sum_{s=1}^R \Delta_{\eta_h}^{(s,m)},
	\quad
	\textsf{RMSE}_{\eta_h}^{(m)} \coloneqq \left\{ \frac{1}{R} \sum_{s=1}^R \left(\Delta_{\eta_h}^{(s,m)}\right)^2 \right\}^{1/2}. 
\end{align}
Let $\hVc_{\eta_h}^{(s,m)}$ denote the direct or recursive plug-in standard error defined in \eqref{eq:direct.var.est} and \eqref{eq:rec.var.est}. 
Using the pointwise 90\% confidence intervals in Section~\ref{sec:h.summary}, define 
\begin{align}
	\textsf{Coverage}_{\eta_h}^{(m)} \coloneqq \frac{1}{R} \sum_{s=1}^R \mathds{1} \left\{ \theta^{(s)}_{\eta_h} \in \texttt{CI}^{(s,m)}_{\theta_{\eta_h}}(0.1)\right\},
	~~
	\textsf{Length}_{\eta_h}^{(m)} \coloneqq \frac{2 \Phi^{-1}(0.95)}{R} \sum_{s=1}^R \hVc_{\eta_h}^{(s,m)}. 
\end{align} 
%


\subsection{Simulation Results} \label{sec:sims.results.terminal}
%
Figure~\ref{fig:terminal.sims} reports the terminal-lead results. 
The qualitative conclusions from the horizon-average analysis carry over to the terminal-lead analysis.

\section{Notation for Proofs} \label{sec:proofs.notation}
%
%
Recall $B \coloneqq T_1 / (L+1) \in \Nb$ and $M \coloneqq (B-1) N_1$. 

\subsection{Observed Blocks} \label{sec:notation.obs}
Denote the observed blocks as 
\begin{align}
	\bY &\coloneqq \bY_{\Ic_1, \pre} \in \Rb^{N_1 \times T_0},
	\quad
	\by \coloneqq \by_{N, \pre} \in \Rb^{T_0}, 
	\\
	\bZ &\coloneqq \bZ_\lag \in \Rb^{L \times M}, 
	\quad
	\bz \coloneqq \bz_{\tnext} \in \Rb^{M}.
\end{align} 
Recall $\bW \coloneqq \big[Y_{j, T-L+ \ell}: j \in \Ic_1, \ell \in [L] \big] \in \Rb^{N_1 \times L}$.

\subsection{Population Blocks} \label{sec:notation.pop}
Define $\bbX \coloneqq \Ex[\bX \mid \Ec]$ for any random object $\bX$, e.g., $\bbY \coloneqq \Ex[\bY \mid \Ec] = \Ex[\bY_{\Ic_1, \pre} \mid \Ec]$ and $\bby \coloneqq \Ex[\by \mid \Ec] = \Ex[\by_{N, \pre} \mid \Ec]$. 
With this notation, let $r_y = \rank(\bbY)$ and $r_z = \rank(\bbZ)$. 
Further, let $\lambda_y$ and $\lambda_z$ denote the $r_y$ and $r_z$-th singular values of $\bbY$ and $\bbZ$. 
%
%

\subsection{Noise Blocks} \label{sec:notation.noise}
Denote the noise blocks as 
\begin{align}
	\bXi_y &\coloneqq \bY - \bbY,
	\quad \bxi_y \coloneqq \by - \bby,
	\\ 
	\bXi_z &\coloneqq \bZ - \bbZ, 
	\quad \bxi_z \coloneqq \bz - \bbz,
	\\
	\bXi_w &\coloneqq \bW - \bbW. 
\end{align}
%
%
Choose $C_\noise > 0$ as an absolute constant, and let 
\begin{align}
	\eta_y &\coloneqq C_\noise \cdot \sigma  \left(\sqrt{N_1} + \sqrt{T_0} + \sqrt{\log(N_1 T_0)} \right),
	\\
	\eta_z &\coloneqq C_\noise \cdot \sigma  \left(\sqrt{L} + \sqrt{M} + \sqrt{\log(LM)} \right),
	\\
	\eta_w &\coloneqq C_\noise \cdot \sigma  \left(\sqrt{N_1} + \sqrt{L} + \sqrt{\log(N_1L)} \right). 
	\label{eq:noise.eta}
\end{align}
Define
\begin{align}
	\Gc_{\noise, y} \coloneqq \left\{ \| \bXi_y \|_{\txtop} \le \eta_y \right\},
	~
	\Gc_{\noise, z} \coloneqq \left\{ \| \bXi_z \|_\txtop \le \eta_z \right\},
	~
	\Gc_{\noise, w} \coloneqq \left\{ \| \bXi_w \|_\txtop \le \eta_w \right\}. \label{eq:event.noise}
\end{align}
Define the joint event as $\Gc_{\noise} \coloneqq \Gc_{\noise, y} \cap \Gc_{\noise, z} \cap \Gc_{\noise, w}$. 

\subsection{\PCR~Estimates} \label{sec:notation.est}
Denote the de-noised population blocks as 
\begin{align}
	\hbY \coloneqq \bY^{(r_y)},
	\quad
	\hbZ \coloneqq \bZ^{(r_z)}. 
\end{align}
Recall from \eqref{eq:pcr.beta.hat} and \eqref{eq:pcr.alpha.hat} the \PCR~parameter estimates as
\begin{align}
	\hbbeta &= \hbY^{\top, \dagger} \by,
	\quad
	\hbalpha = \hbZ^{\top, \dagger} \bz. 
\end{align}
%
%
%
%
%
Define the parameter estimation errors as 
\begin{align}
	\bDelta_{\alpha} &\coloneqq \hbalpha - \balpha^*,
	\quad
	\bDelta_{\beta} \coloneqq \hbbeta - \bbeta^*. 
\end{align}
%
%
Recall from Section~\ref{sec:param.est}, 
\begin{align}
	\Lambda_{\alpha} &\coloneqq \frac{ r_z}{\min\{\sqrt{L}, \sqrt{M}\}} + \frac{ \sqrt{ r_z (1 + \log(LM))}}{\sqrt{M}}, 
	\\
	\Lambda_{\beta} &\coloneqq \frac{ r_y}{\min\{\sqrt{N_1}, \sqrt{T_0} \}} + \frac{ \sqrt{r_y (1 + \log(N_1 T_0))}}{\sqrt{T_0}}.  
\end{align} 
Define the events 
\begin{align}
	\Gc_{\PCR, \alpha} &\coloneqq \left\{ \| \bDelta_{\alpha} \|_2 \le \frac{C_\PCR \cdot \sigma \Lambda_{\alpha}}{\sqrt{L}} \right\},
	\quad
	\Gc_{\PCR, \beta} \coloneqq \left\{ \| \bDelta_{\beta} \|_2 \le \frac{C_\PCR \cdot \sigma \Lambda_{\beta}}{\sqrt{N_1}} \right\},
	\label{eq:event.pcr}
\end{align}
where $C_\PCR > 0$ is an absolute constant.
Define the joint event as $\Gc_\PCR \coloneqq \Gc_{\PCR, \alpha} \cap \Gc_{\PCR, \beta}$. 
%


\section{Useful Known Results} \label{sec:proofs.concentration}
We collect several standard tools that are used repeatedly throughout the proofs. 
%
The first result gives an exact perturbation identity for Moore-Penrose pseudoinverses. 
\begin{lemma} {\cite[Theorem 3.2]{stewart1977perturbation}} \label{lemma:pseudoinverse.stewart}
Let $\bA, \bB \in \Rb^{m \times n}$. Then, we have 
\begin{align}
	\bB^\dagger - \bA^\dagger =
	- \bB^\dagger \bP_B (\bB - \bA) \bPi_A \bA^\dagger
	+ (\bB^\top \bB)^\dagger \bPi_B (\bB - \bA)^\top \bP_A^\perp
	- \bPi_B^\perp (\bB - \bA)^\top \bP_A (\bA \bA^\top)^\dagger. 
\end{align}
Here, $\bP_M \coloneqq \bM \bM^\dagger$ and $\bPi_M \coloneqq \bM^\dagger \bM$ for any matrix $\bM \in \Rb^{m \times n}$. 
\end{lemma}

The next result is known as Weyl’s inequality. It controls how much singular values can move under an operator-norm perturbation. 
\begin{lemma} [Weyl's inequality] \label{lemma:weyl}
Given $\bA, \bB \in \Rb^{m \times n}$, let $\sigma_i$ and $\widehat{\sigma}_i$ be the $i$-th singular values of $\bA$ and $\bB$, respectively, in decreasing order and repeated by multiplicities. 
Then for all $i \le \min\{m, n\}$, we have $\abs{ \sigma_i - \widehat{\sigma}_i} \le \|\bA - \bB \|_\eop.$
\end{lemma} 

The next result is a general Hoeffding inequality for weighted sums of independent sub-Gaussian variables. 
\begin{lemma} {\cite[Theorem 2.6.3]{vershynin2018high}} \label{lemma:hoeffding}
Let $X_1, \cdots, X_N$ be independent, mean-zero, sub-Gaussian random variables, and $\ba = [a_i] \in \Rb^N$. Then, for every $t \ge 0$, we have
\begin{align}
	\Pb \left( \left|\sum_{i=1}^N a_i X_i \right| \ge t \right) \le 2 \cdot \exp(-\frac{ct^2}{K^2 \| \ba \|_2^2}),
\end{align}
where $K = \max_i \| X_i \|_{\psi_2}$. 
\end{lemma}

The following matrix concentration inequality controls the operator norm of a random matrix with independent sub-Gaussian entries.
\begin{lemma} {\cite[Theorem 4.4.5]{vershynin2018high}} \label{lemma:subg_matrix}
Let $\bA = [A_{ij}]$ be an $m \times n$ random matrix where the entries $A_{ij}$ are independent, mean zero, sub-Gaussian variables satisfying $K = \max_{i,j} \| A_{ij} \|_{\psi_2}$. 
Then for any $t > 0$, we have 
\begin{align}
	\| \bA \|_{\eop} \le CK \left(\sqrt{m} + \sqrt{n} + \sqrt{t} \right)
\end{align}
with probability at least $1-2\exp(-t)$ for some absolute constant $C > 0$. 
\end{lemma} 

The final result is Bernstein’s inequality for sums of independent sub-exponential variables.
\begin{lemma} {\cite[Theorem 2.8.1]{vershynin2018high}} \label{lemma:bernstein}
Let $X_1, \cdots, X_N$ be independent, mean zero, sub-exponential random variables. Then, for every $t \ge 0$, we have
\begin{align}
	\Pb \left(\left| \sum_{i=1}^N X_i \right| \ge t \right) \le 2 \exp \left( -c \min \left\{ \frac{t^2}{\sum_{i=1}^N \| X_i \|^2_{\psi_1}}, \frac{t}{\max_i \| X_i \|_{\psi_1}} \right\} \right),
\end{align}
where $c > 0$ is an absolute constant. 
\end{lemma}

\section{Proofs on Identification} \label{sec:proofs.framework}
%

\subsection{Proof of Proposition~\ref{prop:beta}} 

\begin{proof}
Condition on $\Ec$. 
By Assumption~\ref{assump:units}, there exists a weights vector $\bbeta \in \Rb^{N_1}$ such that 
\begin{align}
	\bu_N &= \sum_{j \in \Ic_1} \beta_j \cdot \bu_j. 
	\label{eq:units.temp}
\end{align}
Hence, for all $i \in [N]$, $t \in [T+1]$, and $d \in \{0,1\}$, 
\begin{align}
	\Ex\left[Y_{it}(d) \mid \bu_i, \bv_t(d) \right] &= \Ex\left[ \langle \bu_i, \bv_{t}(d) \rangle + \varepsilon_{it}(d) \mid \bu_i, \bv_{t}(d)\right]  &&\because \text{Assumption~\ref{assump:lfm}}
	\\
	&= \langle \bu_i, \bv_t(d) \rangle \mid \{ \bu_i, \bv_t(d) \} && \because \text{Assumption~\ref{assump:mean_ind}}
	\\
	&= \langle \bu_i, \bv_t(d) \rangle \mid \Ec
	\\
	&= \sum_{j \in \Ic_1} \beta_j \cdot \langle \bu_j, \bv_t(d) \rangle \mid \Ec &&\because \text{\eqref{eq:units.temp}}
	\\
	&= \sum_{j \in \Ic_1} \beta_j \cdot \Ex\left[ \langle \bu_j, \bv_t(d) \rangle + \varepsilon_{jt}(d) \mid \Ec \right] && \because \text{Assumption~\ref{assump:mean_ind}}
	\\
	&= \sum_{j \in \Ic_1} \beta_j \cdot \Ex\left[Y_{jt}(d) \mid \Ec \right]. && \because \text{Assumption~\ref{assump:lfm}}. 
\end{align}
We note that the third equality holds since $\langle \bu_i, \bv_t(d) \rangle$ is deterministic conditional on $\{\bu_i, \bv_t(d)\}$. 
Taking $i = N$, $t = T+1$, and $d = 1$ gives Proposition~\ref{prop:beta}(a).
Moreover, taking $i = N$, $t \in [T_0]$, and $d=0$, and noting $Y_{jt} \coloneqq Y_{jt}(0)$, due to \eqref{eq:sutva}, gives Proposition~\ref{prop:beta}(b). 
%
%
\end{proof} 

%
%

\subsection{Proof of Proposition~\ref{prop:alpha}} \label{sec:proofs.alpha}

\begin{proof}
For each latent component $a \in [r]$, define $f_a(t) \coloneqq V_{ta}(1)$ for $t \in \Zb$. 
Let $S$ denote the shift operator on bi-infinite sequences, i.e., $(Sg)(t) \coloneqq g(t+1)$. 
By the standard finite-rank Hankel recurrence characterization, Assumption~\ref{assump:hankel} implies that, for each $a \in [r]$, there is a monic polynomial $\pi_a$ of degree $q_a \le Q$ such that $\pi_a(S) f_a = 0$. 
Define the common annihilating polynomial 
\begin{align}
	p(z) \coloneqq \prod_{a=1}^r \pi_a(z) = z^q - \sum_{h=1}^q \gamma_h z^{q - h}, 
\end{align}
where $q \coloneqq \sum_{a=1}^r q_a \le rQ$. 
Because the operators $\pi_a(S)$ commute and $\pi_a$ divides $p$, we have $p(S) f_a = 0$ for every $a \in [r]$. 
Evaluating this identity at time $t-q$ gives
\begin{align}
	0 = \left(p(S) f_a \right) (t-q) = f_a(t) - \sum_{h=1}^q \gamma_h f_a(t-h). 
\end{align}
Consequently, for $a \in [r]$ and $t \in \Zb$ 
\begin{align}
	f_a(t) = \sum_{h=1}^q \gamma_h f_a(t-h). 
\end{align}
For every $j \in \Ic_1$, define its treated latent mean sequence by
\begin{align}
	\mu_j(t) \coloneqq \left\langle \bu_j, \bv_t(1) \right \rangle = \sum_{a=1}^r U_{ja} f_a(t). 
\end{align}
Since $\mu_j$ is a linear combination of $f_a$, it satisfies the same common recurrence:
\begin{align}
	\mu_j(t) = \sum_{h=1}^q \gamma_h \mu_j(t-h). 
\end{align}
Because $q \le rQ \le L$, pad these recurrence coefficients to an $L$-dimensional vector $\balpha \in \Rb^L$ by setting $\alpha_{L+1-h} \coloneqq \gamma_h$ for $h \in [q]$, and setting every remaining coordinate to zero. 

For every $j \in \Ic_1$ and $t \in \Zb$, 
\begin{align}
	\sum_{\ell=1}^L \alpha_\ell \mu_j(t-L+\ell) 
	&= \sum_{h=1}^q \gamma_h \mu_j\left(t-L+(L+1-h)\right)
	\\
	&= \sum_{h=1}^q \gamma_h \mu_j(t+1-h)
	\\
	&= \mu_j(t+1). 
\end{align}
Thus, the common $L$-lag relation is 
\begin{align}
	\mu_j(t+1) = \sum_{\ell=1}^L \alpha_\ell \mu_j(t-L+\ell), \quad j \in \Ic_1, t \in \Zb. \label{eq:time.1}
\end{align}
Notably, $\balpha$ is common to all donor units since it is constructed solely from the treated temporal coefficients $f_a$ and not the unit factors. 

Now, condition on $\Ec$. 
As shown in the proof of Proposition~\ref{prop:beta}, Assumptions~\ref{assump:lfm} and \ref{assump:mean_ind} imply that for every $s \in [T+1]$, 
\begin{align}
	\Ex\left[Y_{js}(1) \mid \Ec \right] = \left \langle \bu_j, \bv_s(1) \right \rangle = \mu_j(s). \label{eq:time.2}
\end{align}
Therefore, by \eqref{eq:time.1} and \eqref{eq:time.2}, taking $t = T$ gives 
\begin{align}
	\Ex\left[Y_{j, T+1}(1) \mid \Ec \right] = \sum_{\ell=1}^L \alpha_\ell \Ex\left[Y_{j, T-L+\ell}(1) \mid \Ec \right]. 
\end{align}
Since every lag time $T-L+\ell$ for $\ell \in [L]$ lies in the observed post-treatment period, applying \eqref{eq:sutva} yields $Y_{j, T-L+\ell} \coloneqq Y_{j, T-L+\ell}(1)$, which then proves part (a). 

Finally, let $t \in \{T_0+L+1, \dots, T\}$ and apply \eqref{eq:time.1} and \eqref{eq:time.2} with $t-1$ in place of $t$ to obtain  
\begin{align}
	\Ex\left[Y_{jt}(1) \mid \Ec \right] = \sum_{\ell=1}^L \alpha_\ell \Ex\left[Y_{j, t-L-1+\ell}(1) \mid \Ec \right]. 
\end{align}
Since the response at time $t$ and all of its $L$ regressors lie in the observed post-treatment period, we invoke \eqref{eq:sutva} to conclude part (b). 
This completes the proof. 
\end{proof}

\subsection{Proof of Corollary~\ref{cor:identification}} \label{sec:proof.identification}

\begin{proof}
Condition on $\Ec$. 
Define $\bP_{\alpha} \coloneqq \bbZ^{\top, \dagger}  \bbZ^\top$ and $\bP_{\beta} \coloneqq \bbY^{\top,\dagger}  \bbY^\top$.
Recall from Theorem~\ref{thm:estimation} that 
$$\theta = \left \langle \balpha, \bbW^\top \bbeta \right \rangle.$$ 
By Assumption~\ref{assump:transport}, we have $\col(\bbW) \subseteq \col(\bbY)$. 
Since $\bP_{\beta}$ is the orthogonal projector onto $\col(\bbY)$, it follows that $\bP_{\beta} \bbW = \bbW$. 
Accordingly, we obtain 
\begin{align}
	\theta = \left\langle \balpha, \bbW^\top \bP_{\beta} \bbeta \right \rangle = \left\langle \balpha, \bbW^\top \bbeta^* \right \rangle. 
\end{align}
Following a similar argument, observe that Assumption~\ref{assump:transport} asserts $\row(\bbW) \subseteq \row(\bbZ^\top)$. 
Since $\bP_{\alpha}$ is the orthogonal projector onto $\row(\bbZ^\top)$, we have $\bbW \bP_{\alpha} = \bbW$, which yields 
\begin{align}
	\theta = \left\langle \balpha, \bP_{\alpha} \bbW^\top \bbeta^* \right \rangle= \left \langle \bP_{\alpha} \balpha, \bbW^\top \bbeta^* \right \rangle = \left \langle \balpha^*, \bbW^\top \bbeta^* \right \rangle. 
\end{align}
This completes the proof. 
\end{proof}

\section{\PCR~Parameter Estimation Error Results} \label{sec:proofs.pcr}
To prove Theorem~\ref{prop:parameter.recovery}, we first establish a general \PCR~estimation error bound in Lemma~\ref{lemma:pcr}. 
Appendix~\ref{sec:pcr.setup} introduces the general statistical framework, Appendix~\ref{sec:key.lemmas.pcr} states the auxiliary lemmas, and Appendix~\ref{sec:proof.lemma.pcr} proves Lemma~\ref{lemma:pcr}. 
The proofs of the auxiliary lemmas are collected in Appendix~\ref{sec:pcr.key.lemmas.proofs}. 

\subsection{Statistical Framework} \label{sec:pcr.setup} 
Let $\Fc$ denote a sigma-field. 
Consider the linear model 
\begin{align}
	\by = \bmu + \bzeta \in \Rb^m, \label{eq:linear.model} 
\end{align}
where $\by$ is the observed response vector, $\bmu$ is the unknown signal, and $\bzeta$ the idiosyncratic noise. 
Denote the observed covariates as 
\begin{align}
	\bX = \bA + \bXi \in \Rb^{m \times n}, 
\end{align}
where $\bA$ represents the true covariates and $\bXi$ represents idiosyncratic noise. 
Suppose $\bA$ and $\bmu$ are $\Fc$-measurable, and $\bmu \in \col(\bA)$. 
Define the $n$-dimensional target parameter by 
\begin{align}
	\bgamma^* \coloneqq \bA^\dagger \bmu. \label{eq:gamma.unique}
\end{align}
%

\subsubsection{Assumptions}
Within this framework, we impose the following assumptions, which mirror Assumptions~\ref{assump:bounded}, \ref{assump:subg}, and \ref{assump:spectra} of the main body, respectively.  

\begin{assumption} \label{assump:pcr.bounded}
Conditioned on $\Fc$, suppose the entries of $\bA$ and $\bmu$ are bounded between $[-1,1]$. 
\end{assumption}

\begin{assumption} \label{assump:pcr.noise}
Conditioned on $\Fc$, the entries of $\bXi= [\xi_{ij}]$ and $\bzeta = [\zeta_i]$ are mutually independent, mean-zero, sub-Gaussian variables satisfying $\| \xi_{ij} \|_{\psi_2} \le C_\xi \sigma$ and $\| \zeta_i \|_{\psi_2} \le C_\zeta \sigma$ for some constants $C_\xi, C_\zeta > 0$. 
\end{assumption}

\begin{assumption} \label{assump:pcr.spectra}
Conditioned on $\Fc$, the condition number $\kappa$ of $\bA$ satisfies $\kappa^{-1} \ge c$ and $\| \bA \|_F^2 \ge c' m n$ for some constants $c, c' > 0$. 
\end{assumption} 

\subsubsection{\PCR~Estimator}
To estimate the parameter $\bgamma^*$, the \PCR~estimator proceeds as follows:  
\begin{enumerate} [label=(\alph*)] 
	\item{\em Low-rank approximation.} Define $\bX^{(k)}$ as the rank-$k$ approximation of $\bX$: 
	\begin{align}
		\bX^{(k)} \coloneqq \HSVT(\bX, k). 
	\end{align}
	
	\item{\em Linear representation.} Define the \PCR~estimate as 
	\begin{align}
		\hbgamma \coloneqq \left( \bX^{(k)} \right)^\dagger \by \in \argmin_{\bomega \in \Rb^n} \| \by - \bX^{(k)} \bomega \|_2.  
	\end{align}
\end{enumerate} 

\subsubsection{\PCR~Parameter Estimation Error}
The following lemma is the primary generic result that bounds the \PCR~estimation error in recovering $\bgamma^*$. 

\begin{lemma} \label{lemma:pcr}
Fix $t > 0$. 
Let Assumptions~\ref{assump:pcr.bounded} to \ref{assump:pcr.spectra} hold. 
Suppose $k = r = \rank(\bA)$ and
\begin{align}
	m n \ge c_\gamma \sigma^2 r \cdot \left(\sqrt{m} + \sqrt{n} + \sqrt{t} \right)^2 \label{eq:sep}
\end{align}
for a sufficiently large constant $c_\gamma > 0$.
%
Then, conditional on $\Fc$, with probability at least $1 - \Oc(\exp(-t))$, 
\begin{align}
	\left \| \hbgamma - \bgamma^* \right \|_2 \lesssim \frac{\sigma r \left(\sqrt{m} + \sqrt{n} + \sqrt{t} \right)}{n \sqrt{m}} + \frac{\sigma \sqrt{r} \left( \sqrt{r} + \sqrt{t} \right)}{\sqrt{mn}}.
\end{align}
\end{lemma} 

\subsection{Key Lemmas} \label{sec:key.lemmas.pcr}
We begin by stating key lemmas to prove Lemma~\ref{lemma:pcr}. 
Define $\bhA \coloneqq \bX^{(r)} = \hbU \hbS \hbV^\top$, where $r = \rank(\bA)$. 
Let $\bP_V \coloneqq \bA^\dagger \bA = \bV \bV^\top$ and $\bP_{\hV} \coloneqq \bhA^\dagger \bhA = \hbV \hbV^\top$ denote the orthogonal projectors onto the rowspaces of $\bA$ and $\hbA$. 
Let $\lambda \coloneqq s_r(\bA)$ denote the $r$-th singular value of $\bA$. 
Finally, recall $\bXi \coloneqq \bX - \bA$. 

\begin{lemma} \label{lemma:pcr.spectra}
Let Assumption~\ref{assump:pcr.bounded} hold. 
Then, we have 
\begin{align}
	\| \bgamma^* \|_2 \le \frac{\sqrt{m}}{\lambda}. 
\end{align}
\end{lemma}

\begin{lemma} \label{lemma:pcr.pseudo}
Suppose $k = r = \rank(\bA)$. 
Then, we have 
\begin{align}
	\| \hbA^\dagger \|_{\emph{op}} \le \frac{1}{\lambda - \| \bXi\|_{\emph{op}}},
\end{align}
provided $\lambda > \| \bXi\|_{\emph{op}}$. 
\end{lemma}

\begin{lemma} \label{lemma:pcr.peturb.1}
Suppose $k = r = \rank(\bA)$. 
Then, we have 
\begin{align}
	\| \hbA - \bA \|_{\emph{op}} \le 2 \cdot \| \bXi\|_{\emph{op}}. 
\end{align}
\end{lemma}

\begin{lemma} \label{lemma:pcr.peturb.2}
Suppose $k = r = \rank(\bA)$. 
Then, we have 
\begin{align}
	\| \bP_{\hV}^\perp \cdot \bP_V \|_{\emph{op}} \le \frac{2 \cdot \| \bXi\|_{\emph{op}}}{\lambda}. 
\end{align}
\end{lemma}

\begin{lemma} \label{lemma:pcr.subg.rotate}
Let $\Hc \supseteq \Fc$ be a sigma-field. 
Suppose $\bQ \in\Rb^{m \times r}$ is $\Hc$-measurable and satisfies $\bQ^\top \bQ = \bI$. 
Further suppose that, conditional on $\Hc$, $\bzeta = [\zeta_i] \in \Rb^m$ is a random vector with independent, mean-zero, sub-Gaussian entries satisfying $\| \zeta_i \|_{\psi_2} \le K_\zeta \sigma$ for some constant $K_\zeta > 0$. 
Then there exist constants $C, c > 0$, depending only on $K_\zeta$, such that for every $t > 0$, 
\begin{align}
	\Pb \left( \| \bQ^\top \bzeta \|_2 > C \sigma \left\{\sqrt{r} + \sqrt{t} \right\} \mid \Hc \right) \le 2 \cdot \exp(-ct). 
\end{align}
\end{lemma}

\subsection{Proof of Lemma~\ref{lemma:pcr}} \label{sec:proof.lemma.pcr}

\begin{proof}
By \eqref{eq:gamma.unique}, note that $\bA \bgamma^* = \bA \bA^\dagger \bmu = \bmu$ since $\bmu \in \col(\bA)$. 
Hence, it follows that 
\begin{align}
	\hbgamma - \bgamma^* &= \hbA^\dagger \by - \bgamma^*
	\\
	&= \hbA^\dagger (\bA \bgamma^* + \bzeta ) - \bgamma^* \pm \hbA \bgamma^* && \because \text{\eqref{eq:linear.model}}
	\\
	&= \hbA^\dagger (\bA - \hbA) \bgamma^* + (\bP_{\hV} - \bI_n ) \bgamma^* + \hbA^\dagger \bzeta. && \because \hbA^\dagger \hbA = \bP_{\hV}
\end{align}
Taking norms, we obtain 
\begin{align}
	\| \hbgamma - \bgamma^* \|_2 
	&\le \| \hbA^\dagger (\bA - \hbA) \bgamma^* \|_2 + \| (\bP_{\hV} - \bI_n ) \bgamma^* \|_2 + \| \hbA^\dagger \bzeta \|_2. 
\label{eq:pcr.terms}
\end{align}
We proceed to bound each term in \eqref{eq:pcr.terms} separately. 

\bigskip \noindent {\em Term 1: $\| \hbA^\dagger (\bA - \hbA) \bgamma^* \|_2$.} 
For a fixed $t > 0$, let 
\begin{align}
	\eta_t \coloneqq C_\xi \sigma \left(\sqrt{m} + \sqrt{n} + \sqrt{t} \right), \label{eq:eta.t} 
\end{align}
and define the event $\Gc_\Xi(t) \coloneqq \left\{ \| \bXi\|_{\text{op}} \le \eta_t \right\}$. 
On $\Gc_\Xi(t)$, we have 
\begin{align}
	\| \hbA^\dagger (\bA - \hbA) \bgamma^* \|_2 &\le \| \hbA^\dagger \|_{\text{op}} \cdot \| \hbA - \bA \|_{\text{op}} \cdot \| \bgamma^* \|_2
	\\
	&\le \left( \frac{1}{\lambda - \eta_t} \right) \cdot \| \hbA - \bA \|_{\text{op}} \cdot \| \bgamma^* \|_2 && \because \text{Lemma~\ref{lemma:pcr.pseudo}}
	\\
	& \le \left( \frac{2 \eta_t}{\lambda - \eta_t} \right) \cdot  \| \bgamma^* \|_2 && \because \text{Lemma~\ref{lemma:pcr.peturb.1}}
	\\
	& \le  \frac{2 \eta_t \sqrt{m}}{\lambda (\lambda - \eta_t)}. && \because \text{Lemma~\ref{lemma:pcr.spectra}} \label{eq:pcr.t1}
\end{align}
By Lemma~\ref{lemma:subg_matrix}, we have for some constant $c_\xi > 0$  
\begin{align}
	\Pb(\Gc_\Xi(t) \mid \Fc) \ge 1 - c_\xi \cdot \exp(-t). \label{eq:pcr.prob.1}
\end{align}

\bigskip \noindent {\em Term 2: $\| (\bP_{\hV} - \bI ) \bgamma^*  \|_2$.} 
On $\Gc_\Xi(t)$, we further have 
\begin{align}
	\| (\bP_{\hV} - \bI_n ) \bgamma^* \|_2 &=
	\| (\bP_{\hV} - \bI_n) \cdot \bP_V \gamma^* \|_2 && \because \text{\eqref{eq:gamma.unique}}
	\\
	&= \| - \bP_{\hV}^\perp \cdot \bP_V \cdot \bgamma^* \|_2
	\\
	&\le \| \bP_{\hV}^\perp \cdot \bP_V \|_{\text{op}} \cdot \| \bgamma^* \|_2
	\\
	&\le \left( \frac{2 \eta_t}{\lambda} \right) \cdot \| \bgamma^* \|_2 && \because \text{Lemma~\ref{lemma:pcr.peturb.2}} 
	\\
	&\le \frac{2 \eta_t \sqrt{m}}{\lambda^2}. &&\because \text{Lemma~\ref{lemma:pcr.spectra}}
	\label{eq:pcr.t2}
\end{align}

\noindent {\em Term 3: $\| \hbA^\dagger \bzeta \|_2$.} 
First, note that 
\begin{align}
	\| \hbA^\dagger \bzeta \|_2 &= \| \hbV \hbS^{-1} \hbU^\top \bzeta \|_2
	\\
	&\le \| \hbA^\dagger \|_{\text{op}} \cdot \| \hbU^\top \bzeta \|_2. &&\because \| \hbS^{-1} \|_{\text{op}} = \| \hbA^\dagger \|_{\text{op}} \label{eq:pcr.t3.0}
\end{align} 
To control the inequality above, for a fixed $t > 0$, define the event 
\begin{align}
	\Gc_\zeta(t) \coloneqq \left\{ \| \hbU^\top \bzeta \|_2 \le C_\zeta \sigma \left(\sqrt{r} + \sqrt{t} \right) \right\}. 
\end{align}
Hence, on the joint event $\Gc(t) \coloneqq \Gc_\Xi(t) \cap \Gc_\zeta(t)$, we simplify \eqref{eq:pcr.t3.0} as 
\begin{align}
	\| \hbA^\dagger \bzeta \|_2 
	&\le \| \hbA^\dagger \|_{\text{op}} \cdot \| \hbU^\top \bzeta \|_2
	\\ 
	&\le \frac{\| \hbU^\top \bzeta \|_2}{\lambda - \eta_t} &&\because \text{Lemma~\ref{lemma:pcr.pseudo}}
	\\
	&\le \frac{C_\zeta \sigma \left(\sqrt{r} + \sqrt{t} \right)}{\lambda - \eta_t}. \label{eq:pcr.t3}
\end{align}
Define the sigma-field $\Hc \coloneqq \Fc \vee \sigma(\bXi)$. 
Observe that $\hbU$ is $\Hc$-measurable. 
Further, by Assumption~\ref{assump:pcr.noise}, $\bzeta$ is independent of $\bXi$ conditional on $\Fc$. 
Therefore, conditional on $\Hc$, the coordinates of $\bzeta$ remain independent, mean-zero, and subgaussian with $\psi_2$-norm bounded by $C_\zeta \sigma$. 
This enables us to apply Lemma~\ref{lemma:pcr.subg.rotate} with $\bQ = \hbU$, yielding $\Pb(\Gc_\zeta(t) \mid \Hc ) \ge 1 - c_\zeta \cdot \exp(-t)$
for some constant $c_\zeta > 0$. 
Since $\Gc_{\Xi}(t) \in \Hc$, the tower property then gives 
\begin{align}
	\Pb \left(\Gc_\Xi(t) \cap \Gc_\zeta^c(t) \mid \Fc \right)
	&= \Ex \left[ \mathds{1}\{\Gc_\Xi(t)\} \cdot \Pb \left(\Gc_\zeta^c(t) \mid \Hc \right) \mid \Fc \right]
	\le c_\zeta \cdot \exp(-t). 
	\label{eq:pcr.prob.2}
\end{align}

\noindent {\em Putting everything together.} 
On $\Gc(t)$, combining \eqref{eq:pcr.t1}, \eqref{eq:pcr.t2}, and \eqref{eq:pcr.t3}, we conclude 
\begin{align}
	\| \hbgamma - \bgamma^* \|_2 &\lesssim 
	\frac{\eta_t \sqrt{m}}{\lambda (\lambda - \eta_t)} + \frac{\eta_t \sqrt{m}}{\lambda^2} + \frac{\sigma \left(\sqrt{r} + \sqrt{t}\right)}{\lambda - \eta_t} 
	\\
	&\lesssim \frac{\eta_t \sqrt{m}}{\lambda(\lambda - \eta_t)} + \frac{\sigma \left(\sqrt{r} + \sqrt{t}\right)}{\lambda - \eta_t}.
\end{align}
Observe that $\| \bA \|_F^2 \le r \cdot s^2_1(\bA)$. 
By Assumption~\ref{assump:pcr.spectra}, we further have $\| \bA \|_F^2 \ge c' mn$, and so 
\begin{align}
	s^2_1(\bA) \ge \frac{c' \cdot m n }{r}. 
\end{align}
At the same time, Assumption~\ref{assump:pcr.spectra} implies $\lambda = s_r(\bA) \ge c s_1(\bA)$. 
Combining the above, 
\begin{align}
	\lambda^2 \ge c^2 s^2_1(\bA) \ge \frac{c^2 c' \cdot m n}{r}. \label{eq:rank.sep.1}
\end{align}
Moreover, \eqref{eq:sep} and \eqref{eq:eta.t} suggest
\begin{align}
	\eta_t^2 \le \frac{C^2_\xi \cdot mn}{c_\gamma \cdot r}. \label{eq:rank.sep.2}
\end{align}
Together, \eqref{eq:rank.sep.1} and \eqref{eq:rank.sep.2} yield 
\begin{align}
	\eta_t \le \frac{\lambda}{2},  \label{eq:eta.lambda}
\end{align}
provided $c_\gamma$ is chosen to be large enough. 
In turn, $(\lambda - \eta_t)^{-1} \lesssim \lambda^{-1}$. 
This results in  
\begin{align}
	\| \hbgamma - \bgamma^* \|_2 
	&\lesssim \frac{\eta_t \sqrt{m}}{\lambda^2} + \frac{\sigma \left(\sqrt{r} + \sqrt{t}\right)}{\lambda}.
\end{align}
Plugging in \eqref{eq:eta.t} and \eqref{eq:rank.sep.1}, and simplifying, gives our desired inequality: 
\begin{align}
	\| \hbgamma - \bgamma^* \|_2 
	&\lesssim \frac{\sigma r \left(\sqrt{m} + \sqrt{n} + \sqrt{t} \right)}{n \sqrt{m}} + \frac{\sigma \sqrt{r} \left( \sqrt{r} + \sqrt{t} \right)}{\sqrt{mn}}.
	\label{eq:pcr.joint}
\end{align}
It remains to establish the probability of the event $\Gc(t)$. 
Taking a union bound over \eqref{eq:pcr.prob.1} and \eqref{eq:pcr.prob.2}, we conclude $\Pb \left( \Gc(t) \mid \Fc \right) \ge 1 - \left(c_\xi + c_\zeta \right) \cdot \exp(-t)$. 
This completes the proof. 
\end{proof}

\subsection{Proof of Key Lemmas} \label{sec:pcr.key.lemmas.proofs}
For ease of notation, we suppress the conditioning on $\Fc$ throughout this section.

\subsubsection{Proof of Lemma~\ref{lemma:pcr.spectra}}
\begin{proof}
By Cauchy-Schwarz, 
\begin{align}
	\| \bgamma^* \|_2 \le \| \bA^\dagger \|_{\text{op}} \cdot \| \bmu \|_2
	\le \frac{\sqrt{m}}{\lambda};
\end{align}
for the second inequality, we leverage Assumption~\ref{assump:pcr.bounded}, which implies $\| \bmu \|_2 \le \sqrt{m}$. 
\end{proof}

\subsubsection{Proof of Lemma~\ref{lemma:pcr.pseudo}}
\begin{proof}
By Lemma~\ref{lemma:weyl}, we have
\begin{align}
	| \hs_r(\bX) - \lambda | \le \| \bX - \bA \|_{\text{op}} = \| \bXi\|_{\text{op}}. 
\end{align}
Accordingly, we obtain 
\begin{align}
	\hs_r(\bX) \ge \lambda - \| \bXi\|_{\text{op}}. 
\end{align}
Since $\hs_r(\hbA) = \hs_r(\bX)$, it follows that 
\begin{align}
	\hs_r(\bhA) \ge \lambda - \| \bXi\|_{\text{op}}. 
\end{align}
Therefore, we conclude that
\begin{align}
	\| \bhA^\dagger \|_{\text{op}} = \frac{1}{\hs_r(\bhA)} \le \frac{1}{\lambda - \| \bXi\|_{\text{op}}},
\end{align}
provided $\lambda > \| \bXi\|_{\text{op}}$. 
\end{proof}

\subsubsection{Proof of Lemma~\ref{lemma:pcr.peturb.1}}
\begin{proof}
Since $\rank(\bA) = r$, we have $s_{r+1}(\bA) = 0$. 
Thus, by Lemma~\ref{lemma:weyl}, we have
\begin{align}
	\hs_{r+1}(\bX) - s_{r+1}(\bA) \le \| \bX - \bA \|_{\text{op}} = \| \bXi\|_{\text{op}}. 
\end{align}
As a result, 
\begin{align}
	\| \bhA - \bX \|_{\text{op}} = \hs_{r+1}(\bX) \le \| \bXi\|_{\text{op}}. 
\end{align}
Applying the triangle inequality yields 
\begin{align}
	\| \bhA - \bA \|_{\text{op}} &\le \| \bhA - \bX \|_{\text{op}} + \| \bX - \bA \|_{\text{op}}  \le 2 \cdot \| \bXi\|_{\text{op}}.
\end{align} 
This completes the proof. 
\end{proof}

\subsubsection{Proof of Lemma~\ref{lemma:pcr.peturb.2}}
\begin{proof}
Take any vector $\bu \in \text{range}(\bP_{\hV}^\perp)$. 
Then, $\hbA \cdot \bu = \bzero$ since $\bP_{\hV}$ is the orthogonal row projector onto $\hbA$. 
Therefore, we have 
\begin{align}
	\bA \bu = (\bA - \hbA) \cdot \bu. \label{eq:temp.1}
\end{align}
Further, because $\bA = \bA \bP_V$, we have
\begin{align}
	\bA \bu = \bA \bP_V \bu, \label{eq:temp.2}
\end{align}
yielding 
\begin{align}
	\| \bA  \bP_V \bu \|_2 \ge \lambda \cdot \| \bP_V \bu \|_2. \label{eq:temp.3}
\end{align}
Combining \eqref{eq:temp.1}, \eqref{eq:temp.2}, and \eqref{eq:temp.3}, 
\begin{align}
	\lambda \cdot \| \bP_V \bu \|_2 \le \| (\hbA - \bA ) \cdot \bu \|_2 \le \| \hbA - \bA \|_{\text{op}} \cdot \| \bu \|_2. 
\end{align}
Taking a supremum over all $\bu \in \text{range}(\bP_{\hV}^\perp)$ gives
\begin{align}
	\| \bP_V \cdot \bP_{\hV}^\perp \|_{\text{op}} &\le \frac{ \| \hbA - \bA \|_{\text{op}}} {\lambda} \le \frac{2 \cdot \| \bXi\|_{\text{op}}}{\lambda},
\end{align}
where the final inequality follows from Lemma~\ref{lemma:pcr.peturb.1}. 
Since $\bP_V$ and $\bP_{\hV}^\perp$ are symmetric operators,
\begin{align}
	\|\bP_{\hV}^\perp \cdot  \bP_V \|_{\text{op}} = \| \bP_V \cdot \bP_{\hV}^\perp \|_{\text{op}}. 
\end{align}
This proves the stated bound. 
\end{proof}

\subsubsection{Proof of Lemma~\ref{lemma:pcr.subg.rotate}}

\begin{proof}
Conditional on $\Hc$, $\bQ$ is fixed. 
For any $\bx \in \Sc^{r-1}$, 
\begin{align}
	\bx^\top \bQ^\top \bzeta = \left(\bQ \bx \right)^\top \bzeta. 
\end{align}
Since $\| \bQ \bx \|_2 = \| \bx \|_2 = 1$, this is a weighted sum of independent, mean-zero, sub-Gaussian variables satisfying
\begin{align}
	\left\| \left(\bQ \bx \right)^\top \bzeta \right\|_{\psi_2} 
	& \le A K_\zeta \sigma \cdot \left\| \bQ \bx \right\|_2
	=  A K_\zeta \sigma
\end{align}
for some constant $A > 0$. 
Invoking Lemma~\ref{lemma:hoeffding}, we have for every $u > 0$, 
\begin{align}
	\Pb \left( \left| \bx^\top \bQ^\top \bzeta \right| > u \mid \Hc \right) \le 2 \cdot \exp(-\frac{c u^2}{K^2_\zeta \sigma^2}). 
\end{align}
Let $\Nc$ be a $1/2$-net of $\Sc^{r-1}$ with $|\Nc| \le 5^r$. 
Then, 
\begin{align}
	\left\| \bQ^\top \bzeta \right\|_2 \le 2 \max_{\bx \in \Nc} \left| \bx^\top \bQ^\top \bzeta \right|. 
\end{align}
Therefore, 
\begin{align}
	\Pb \left( \left\| \bQ^\top \bzeta \right\|_2 > 2 u \mid \Hc \right) 
	&\le \Pb\left( \max_{\bx \in \Nc} \left| \bx^\top \bQ^\top \bzeta \right| > u \mid \Hc \right)
	\\
	&\le \sum_{\bx \in \Nc} \Pb\left( \left| \bx^\top \bQ^\top \bzeta \right| > u \mid \Hc \right)
	\\
	&\le 2 \cdot 5^r \cdot \exp(-c\frac{u^2}{K^2_\zeta \sigma^2}). 
\end{align}
Choose $u = A K_\zeta \sigma (\sqrt{r} + \sqrt{t})$ with $A > 0$ large enough. 
Then,
\begin{align}
	\frac{c u^2}{K_\zeta^2 \sigma^2} = c A^2 \left(\sqrt{r} + \sqrt{t} \right)^2 \ge c A^2 (r + t). 
\end{align}
Therefore,
\begin{align}
	2 \cdot 5^r \exp(-c\frac{u^2}{K^2_\zeta \sigma^2}) \le 2 \cdot \exp(r \log(5) - c A^2 r - cA^2 t). 
\end{align}
Choose $A$ such that $c A^2 \ge 2 \log(5)$ and $c A^2 \ge 2$. 
Then, $r \log(5) - cA^2 r \le 0$ and $-cA^2t \le -2t$. 
Thus,
\begin{align}
	2 \cdot 5^r \exp(-c\frac{u^2}{K^2_\zeta \sigma^2}) \le 2 \exp(-2t) \le 2 \exp(-t). 
\end{align}
As a result, we conclude 
\begin{align}
	\Pb \left( \| \bQ^\top \bzeta \|_2 > 2AC K_\zeta \sigma \left\{\sqrt{r} + \sqrt{t} \right\} \mid \Hc \right) \le 2 \cdot \exp(-t). 
\end{align}
The proof is complete. 
\end{proof}

\section{Proof of Theorem~\ref{prop:parameter.recovery}} \label{sec:proof.param.recovery}

\begin{proof}
We prove the unit-side result by applying Lemma~\ref{lemma:pcr} with $\Fc = \Ec$, $\bA = \Ex[\bY_{\Ic_1, \pre}^\top \mid \Ec]$, $\bX = \bY_{\Ic_1, \pre}^\top$, $\by = \by_{N, \pre}$, and $\bgamma^* = \bbeta^*$, which yields $r = r_y$, $m = T_0$, $n = N_1$. 
Choosing $t = C \log(N_1T_0)$ for a large enough $C > 0$ gives our desired result. 

Analogously, we prove the time-side result by applying Lemma~\ref{lemma:pcr} with $\Fc = \Ec$, $\bA = \Ex[\bZ_{\lag}^\top \mid \Ec]$, $\bX = \bZ_{\lag}^\top$, $\by = \bz_{\tnext}$, and $\bgamma^* = \balpha^*$, which yields $r = r_z$, $m = M$, $n = L$. 
Recall that $L \ge r Q$ ensures that $\balpha$ exists. 
Choosing $t = C \log(LM)$ for a large enough $C > 0$ completes the proof. 
\end{proof}

\section{Proof of Theorem~\ref{thm:error}} \label{sec:proof.error}
%

%
We begin by stating a key lemma that aids in our proof of Theorem~\ref{thm:error}. 
We relegate its proof to Appendix~\ref{sec:proof.main.noise}. 

\begin{lemma} \label{lemma:noise}
Let $\Ec$ and $\Fc$ be sigma-fields. Let $\bXi = [\xi_{ij}] \in \Rb^{n \times k}$ be a random matrix such that, conditional on $\Ec$, its entries are independent, mean-zero, subgaussian random variables satisfying $\| \xi_{ij} \|_{\psi_2} \le K_\Xi \sigma$ for $K_\Xi > 0$. Assume that $\bXi \independent \Fc \mid \Ec$. 
Let $\bu \in \Rb^k$ and $\bv \in \Rb^n$ be $\Fc$-measurable random vectors. Let $\Gc \in \Fc$ be an event such that, on $\Gc$, $\| \bu \|_2 \le A$ and $\| \bv \|_2 \le B$, where $A, B \ge 0$ are fixed. Then there is a universal constant $C >0$ depending only on $K_\Xi$, such that for every $t > 0$, 
\begin{align}
	\Pb \left( \Gc \cap \left\{ | \bu^\top \bXi^\top \bv | > C \sigma A B t \right\} ~\Big|~ \Ec \right) \le 2 \exp(-t^2). 
\end{align}
\end{lemma} 

\subsection{Completing Proof of Theorem~\ref{thm:error}} 
\begin{proof}
Condition on $\Ec$. 
Observe that $\htheta = \left \langle \hbalpha, \bW^\top \hbbeta \right \rangle$ and $\theta = \left \langle \balpha^*, \bbW^\top \bbeta^* \right \rangle$ by Corollary~\ref{cor:identification}. 
Therefore, 
\begin{align}
	\htheta - \theta &= 
		\left \langle \bDelta_{\alpha}, \bbW^\top \bDelta_{\beta} \right \rangle + \left \langle \bDelta_{\alpha}, \bbW^\top \bbeta^* \right \rangle + \left \langle \balpha^*, \bbW^\top \bDelta_{\beta} \right \rangle
	\\
	& \quad~ + \left\langle \bDelta_{\alpha}, \bXi_w^\top \bDelta_{\beta} \right \rangle + \left \langle \bDelta_{\alpha}, \bXi_w^\top \bbeta^* \right \rangle + \left \langle \balpha^*, \bXi_w^\top \bDelta_{\beta} \right \rangle + \left \langle \balpha^*, \bXi_w^\top \bbeta^* \right \rangle . 
	\label{eq:general}
\end{align}
We proceed to bound each term in \eqref{eq:general} separately. 
Throughout, set $\rho \coloneqq (N_1T_0)^{-10} + (LM)^{-10}$ and $\varphi_\rho \coloneqq \sqrt{\log(8/\rho)}$. 

\bigskip \noindent {\em Deterministic $\bbW$-terms}: 
Define $\Gc_\PCR$ as in \eqref{eq:event.pcr}. 
Applying the Cauchy-Schwarz inequality, we obtain
\begin{align}
	\left| \left \langle \bDelta_{\alpha}, \bbW^\top \bDelta_{\beta} \right \rangle \right| &\le 
	\| \bDelta_{\alpha} \|_2 \cdot \| \bDelta_{\beta} \|_2 \cdot \| \bbW \|_{\text{op}}. 
\end{align} 
By Assumption~\ref{assump:bounded}, it follows that 
\begin{align}
	\| \bbW \|_{\text{op}} \le \| \bbW \|_F \le  \sqrt{N_1L}. \label{eq:barw.bound.op}
\end{align} 
On the event $\Gc_\PCR$, we have 
\begin{align}
	\left| \left \langle \bDelta_{\alpha}, \bbW^\top \bDelta_{\beta} \right \rangle \right| &\lesssim \sigma^2 \Lambda_{\alpha}  \Lambda_{\beta}. 
	\label{eq:term1} 
\end{align}
Moreover, we have 
\begin{align}
	\left| \left \langle \bDelta_{\alpha}, \bbW^\top \bbeta^* \right \rangle \right| &\le 
	\| \bDelta_{\alpha} \|_2 \cdot \| \bbW^\top \bbeta^* \|_2.
\end{align} 
Following the arguments in the proof of Corollary~\ref{cor:identification}, Assumption~\ref{assump:transport} yields $\bbW^\top \bbeta^* = \bbW^\top \bbeta$. 
Thus, for any $\ell \in [L]$, we write 
\begin{align}
	\left(\bbW^\top \bbeta \right)_\ell &= \sum_{j \in \Ic_1} \beta_j \cdot \Ex[Y_{j, T-L+\ell}(1) \mid \Ec] 
	\\ &= \sum_{j \in \Ic_1} \beta_j \cdot \left \langle \bu_j, \bv_{T - L + \ell}(1) \right \rangle && \because \text{Assumptions~\ref{assump:lfm} and \ref{assump:mean_ind}}
	\\ &= \left \langle \bu_N, \bv_{T - L + 1 + \ell}(1) \right \rangle &&\because \text{Assumption~\ref{assump:units}} 
	\\ &= \Ex[Y_{N, T-L+\ell}(1) \mid \Ec], &&\because \text{Assumptions~\ref{assump:lfm} and \ref{assump:mean_ind}} 
	\label{eq:wbeta.norm}
\end{align}
and so Assumption~\ref{assump:bounded} implies 
\begin{align}
	\| \bbW^\top \bbeta^* \|_2 \le \sqrt{L}. \label{eq:barwbeta.bound}
\end{align} 
Therefore, on $\Gc_\PCR$, 
\begin{align}
	\left| \left \langle \bDelta_{\alpha}, \bbW^\top \bbeta^* \right \rangle \right| &\lesssim \sigma \Lambda_{\alpha}. 
	\label{eq:term2}
\end{align}
Similarly, 
\begin{align}
	\left| \left \langle \balpha^*, \bbW^\top \bDelta_{\beta} \right \rangle \right| &\le 
	\| \bDelta_{\beta} \|_2 \cdot \| \bbW \balpha^* \|_2.
\end{align} 
Following the arguments in the proof of Corollary~\ref{cor:identification}, Assumption~\ref{assump:transport} yields $\bbW \balpha^* = \bbW \balpha$. 
Hence, for each $j \in \Ic_1$, Proposition~\ref{prop:alpha} gives
\begin{align}
	\left(\bbW \balpha \right)_j = \sum_{\ell=1}^L \alpha_\ell \cdot \Ex[Y_{j, T-L+\ell} (1) \mid \Ec] = \Ex[Y_{j, T+1}(1) \mid \Ec], \label{eq:walpha} 
\end{align}
and so Assumption~\ref{assump:bounded} implies 
\begin{align}
	\| \bbW \balpha^* \|_2 \le \sqrt{N_1}. \label{eq:barwalpha.bound}
\end{align} 
Thus, on $\Gc_\PCR$, 
\begin{align}
	\left| \left \langle \balpha^*, \bbW^\top \bDelta_{\beta} \right \rangle \right| &\lesssim \sigma \Lambda_{\beta}.  
	\label{eq:term3}
\end{align}
By Theorem~\ref{prop:parameter.recovery}, note that 
\begin{align}
	\Pb(\Gc_\PCR^c \mid \Ec) \lesssim \rho. \label{eq:prob.pcr}
\end{align} 

\noindent {\em Stochastic $\bXi_w$-terms}: 
Next, we record two useful deterministic bounds. 
By Lemma~\ref{lemma:pcr.spectra}, 
\begin{align}
	\| \balpha^* \|_2 &\le \frac{\sqrt{M}}{\lambda_z},
	\quad
	\| \bbeta^* \|_2 \le \frac{\sqrt{T_0}}{\lambda_y}. \label{eq:alpha.beta.bound}
\end{align} 
Define four $\bXi_w$-events as  
\begin{align}
	\Gc_{\Delta \Delta} &\coloneqq \left\{ \left| \left \langle \bDelta_{\alpha}, \bXi_w^\top \bDelta_{\beta} \right \rangle \right| \le  
	\frac{C_w \cdot \sigma^3 \varphi_\rho \cdot \Lambda_{\alpha} \Lambda_{\beta}}{\sqrt{N_1L}} \right\},
	\\
	\Gc_{\Delta \beta} &\coloneqq \left\{ \left| \left \langle \bDelta_{\alpha}, \bXi_w^\top \bbeta^* \right \rangle \right| \le  
	\frac{C_w \cdot \sigma^2 \varphi_\rho \cdot \sqrt{T_0} \Lambda_{\alpha}}{\sqrt{L} \lambda_y} \right\},
	\\
	\Gc_{\alpha \Delta} &\coloneqq \left\{\left| \left \langle \balpha^*, \bXi_w^\top \bDelta_{\beta} \right \rangle \right| \le  
	\frac{C_w \cdot \sigma^2 \varphi_\rho \cdot \sqrt{M} \Lambda_{\beta}}{\sqrt{N_1} \lambda_z} \right\},
	\\
	\Gc_{\alpha \beta} &\coloneqq \left\{ \left| \left \langle \balpha^*, \bXi_w^\top \bbeta^* \right \rangle \right| \le  
	\frac{C_w \cdot \sigma \varphi_\rho \cdot \sqrt{M T_0}}{\lambda_y \lambda_z} \right\},
	\label{eq:term.w}
\end{align}
where $C_w > 0$ is a sufficiently large constant. 

Define the joint $\bXi_w$-events as $\Gc_{w} \coloneqq \Gc_{\Delta \Delta} \cap \Gc_{\Delta \beta} \cap \Gc_{\alpha \Delta} \cap \Gc_{\balpha \bbeta}$. 
Define the sigma-field $\Fc \coloneqq \Ec \vee \sigma(\by, \bY, \bZ, \bz)$.
Note that $\bDelta_{\alpha}, \bDelta_{\beta}, \balpha^*, \bbeta^*$ are $\Fc$-measurable and that $\bXi_w$ is independent of $\Fc$, conditional on $\Ec$. 
This holds because $\bW$ uses the final forecasting Page lags, whereas $\bZ$ and $\bz$ are built from the earlier training Page columns, and the pre-treatment data is time-disjoint; for a refresher, see Section~\ref{sec:alg.setup}. 
Hence, under Assumption~\ref{assump:subg}, these noise entries are conditionally independent. 
We can accordingly apply Lemma~\ref{lemma:noise} four times to control the failure probabilities for the events within $\Gc_w$. 
In particular, applying Lemma~\ref{lemma:noise} with $\bu = \bDelta_{\alpha}$ and $\bv = \bDelta_{\beta}$, 
\begin{align}
	\Pb\left( \Gc_\PCR \cap \Gc^c_{\Delta \Delta} \mid \Ec \right) \le 2 \cdot \exp(-\varphi_\rho^2). 
\end{align}
Similarly, applying Lemma~\ref{lemma:noise} with $\bu = \bDelta_{\alpha}$ and $\bv = \bbeta^*$, coupled with \eqref{eq:alpha.beta.bound}, 
\begin{align}
	\Pb\left( \Gc_\PCR \cap \Gc^c_{\Delta \beta} \mid \Ec \right) \le 2 \cdot \exp(-\varphi_\rho^2). 
\end{align}
Next, applying Lemma~\ref{lemma:noise} with $\bu = \balpha^*$ and $\bv = \bDelta_{\beta}$, coupled with \eqref{eq:alpha.beta.bound}, 
\begin{align}
	\Pb\left( \Gc_\PCR \cap \Gc^c_{\alpha \Delta} \mid \Ec \right) \le 2 \cdot \exp(-\varphi_\rho^2). 
\end{align}
Applying Lemma~\ref{lemma:noise} with $\bu = \balpha^*$ and $\bv = \bbeta^*$, coupled with \eqref{eq:alpha.beta.bound}, 
\begin{align}
	\Pb\left( \Gc_\PCR \cap \Gc^c_{\alpha \beta} \mid \Ec \right) \le 2 \cdot \exp(-\varphi_\rho^2). 
\end{align}
Observing $\Gc^c_w = \Gc^c_{\Delta \Delta} \cup \Gc^c_{\Delta \beta} \cup \Gc^c_{\alpha \Delta} \cup \Gc^c_{\balpha \bbeta}$ and taking the union bound yields 
\begin{align}
	\Pb \left(\Gc_\PCR \cap \Gc^c_w \mid \Ec \right) \le 8 \cdot \exp(-\varphi_\rho^2) = \rho. \label{eq:rho.temp}
\end{align}
On $\Gc_\PCR \cap \Gc_w$, the $\bXi_w$-block stochastic terms satisfy the bounds in \eqref{eq:term.w}. 

\bigskip \noindent 
{\em Putting everything together.} 
Under Assumption~\ref{assump:spectra}, we follow the arguments that led to \eqref{eq:rank.sep.1} to obtain 
\begin{align}
	&\lambda_y^2 \gtrsim \frac{N_1 T_0}{r_y}, 
	\quad 
	\lambda_z^2 \gtrsim \frac{LM}{r_z}.
\end{align}
Leveraging these inequalities in conjunction with \eqref{eq:term1}, \eqref{eq:term2}, \eqref{eq:term3}, and \eqref{eq:term.w}, we have on $\Gc_\PCR \cap \Gc_w$
\begin{align}
	\left| \htheta - \theta \right| &\lesssim \sigma \Lambda_{\alpha} + \sigma \Lambda_{\beta} + \sigma^2 \Lambda_{\alpha} \Lambda_{\beta} + \frac{\sigma \varphi_\rho \cdot \left(\sigma\Lambda_{\alpha} + \sqrt{r_z} \right) \left(\sigma \Lambda_{\beta} + \sqrt{r_y}\right)}{\sqrt{N_1L}}.
\end{align}
Since $\rho \ge \min\{N_1T_0, LM\}^{-10}$, we have $(1/\rho) \le \min\{N_1T_0, LM\}^{10}$. 
Thus, 
\begin{align}
	\varphi^2_\rho \le \log(8) + 10 \log( \min\{N_1T_0, LM \}) \lesssim 1 + \log(\min\{N_1T_0, LM \}), 
\end{align}
which gives $\varphi_\rho \lesssim \tau$, where $\tau^2 \coloneqq 1 + \log(\min\{N_1 T_0, LM \})$. 
This gives our desired inequality: 
\begin{align}
	\left| \htheta - \theta \right| &\lesssim \sigma\Lambda_{\alpha} + \sigma\Lambda_{\beta} + \sigma^2 \Lambda_{\alpha} \Lambda_{\beta} + \frac{\sigma \sqrt{\tau} \cdot \left(\sigma \Lambda_{\alpha} + \sqrt{r_z} \right) \left(\sigma \Lambda_{\beta} + \sqrt{r_y}\right)}{\sqrt{N_1L}}.
\end{align}
It remains to bound the probability of the joint event $\Gc_\PCR \cap \Gc_w$. 
Taking a union bound gives 
\begin{align}
	\Pb \left( (\Gc_\PCR \cap \Gc_w)^c \mid \Ec \right) 
	&\le \Pb \left( \Gc^c_\PCR \mid \Ec \right) + \Pb \left( \Gc_\PCR \cap \Gc^c_w \mid \Ec \right) 
	\lesssim \rho,	
\end{align}
where the final inequality leverages \eqref{eq:prob.pcr} and \eqref{eq:rho.temp}.
This completes the proof. 
\end{proof}


\subsection{Proof of Lemma~\ref{lemma:noise}} \label{sec:proof.main.noise}

\begin{proof}
Define the sigma-field $\Hc \coloneqq \Ec \vee \Fc$. 
Conditional on $\Hc$, the vectors $\bu$ and $\bv$ are fixed. 
Because $\bXi$ is independent of $\Fc$ conditional on $\Ec$, conditioning further on $\Fc$ does not alter the conditional distribution of $\bXi$. 
Thus, conditional on $\Hc$, the entries $\xi_{ij}$ of $\bXi$ remain independent, mean-zero, sub-Gaussian variables with 
$\psi_2$-norm bounded by $K_\Xi \sigma$. 

Note that the bilinear form can be expanded as 
\begin{align}
	\bu^\top \bXi^\top \bv = \sum_{i=1}^n \sum_{j=1}^k v_i \cdot u_j \cdot \xi_{ij}. 
\end{align}
Conditional on $\Hc$, $\bu^\top \bXi^\top \bv$ is a deterministic weighted sum of independent, mean-zero, sub-Gaussian random variables with weights $a_{ij} \coloneqq v_i u_j$, and hence
\begin{align}
	\sum_{i=1}^n \sum_{j=1}^k a_{ij}^2 = \left( \sum_{i=1}^n v_i^2 \right) \cdot \left(\sum_{j=1}^k u_j^2 \right) = \| \bv \|_2 ^2 \cdot \| \bu \|_2^2. 
\end{align}
This gives 
\begin{align}
	\| \bu^\top \bXi^\top \bv \|_{\psi_2} \le K_\Xi \sigma \cdot \| \bu \|_2 \cdot \| \bv \|_2. 
\end{align}
Hence, Lemma~\ref{lemma:hoeffding} states that for every $s > 0$, 
\begin{align}
	\Pb \left( \left| \bu^\top \bXi^\top \bv \right| > s \mid \Hc \right) \le 2 \cdot \exp(- \frac{c s^2}{K^2_\Xi \sigma^2 \cdot \| \bu \|_2^2 \cdot \| \bv \|_2^2} ). 
\end{align}
On the event $\Gc$, $\| \bu \|_2 \le A$ and $\| \bv \|_2 \le B$. 
Thus, on $\Gc$, taking $s =  C \sigma A B t$, 
\begin{align}
	\frac{c s^2}{K^2_\Xi \sigma^2 \cdot \| \bu \|_2^2 \cdot \| \bv \|_2^2} &= \frac{c C^2 A^2 B^2 t^2}{K^2_\Xi \sigma^2 \cdot \| \bu \|_2^2 \cdot \| \bv \|_2^2}
	\ge \frac{c C^2 t^2}{K^2_\Xi}. 
\end{align}
Choose $C$ large enough so that $cC^2 / K^2_\Xi \ge 1$. 
Then, 
\begin{align}
	\mathds{1}\{\Gc \} \cdot \Pb \left( \left| \bu^\top \bXi^\top \bv \right| > C \sigma A B t ~\Big|~ \Hc \right) \le \mathds{1}\{\Gc \} \cdot 2 \exp(-t^2).
\end{align}
Because $\Gc \in \Hc$, taking conditional expectation given $\Ec$ gives 
\begin{align}
	\Pb \left( \Gc \cap \left\{ \left| \bu^\top \bXi^\top \bv \right| > C \sigma A B t \right\} ~\Big|~ \Ec \right)
	&= \Ex \left[ \mathds{1}\{\Gc\} \cdot \Pb \left( \left| \bu^\top \bXi^\top \bv \right| > C \sigma A B t ~\Big|~ \Hc \right) ~\Big|~ \Ec \right] 
	\\ &\le 2 \exp(-t^2) \cdot \Ex\left[ \mathds{1} \{\Gc \} \mid \Ec \right]
	\\ &\le 2 \exp(-t^2). 
\end{align}
This completes the proof. 
\end{proof}

\section{Proof of Theorem~\ref{thm:inference}} \label{sec:proofs.inference} 
Theorem~\ref{thm:error} establishes that the \TWSF~estimator can achieve vanishing causal forecast error under suitable conditions. 
For inference, however, we must identify more than the magnitude of the error, but rather determine which part of that error survives at first order and whether the remaining terms are negligible on the scale of its standard deviation. 

A direct expansion of the bilinear estimator exposes the difficult. In particular, it contains the terms 
\begin{align}
	\left \langle \hbalpha - \balpha^*, \bbW^\top  \bbeta^* \right \rangle,
	\quad
	\left \langle \balpha^*, \bbW^\top  \left(\hbbeta - \bbeta^*\right) \right\rangle. 
\end{align}
These contributions are linear in the two \PCR~estimation errors and are not automatically of smaller order than the stochastic fluctuation of $\htheta$. Moreover, their laws are cumbersome because both \PCR~estimators depend on noisy design matrices and their estimated singular subspaces. A direct central limit argument would therefore require a delicate joint analysis of the two nonlinear first-stage procedures. 

Motivated by the semiparametric and debiased machine-learning literature on locally robust scores \citep{Newey1994, ChernozhukovEtAl2018DML, ChernozhukovEtAl2022LocallyRobust, ChernozhukovNeweySingh2022Automatic, ChernozhukovNeweySingh2022Riesz, HirshbergWager2021}, the key move is to reveal the \TWSF~estimator through an {\em orthogonal} score. Consider the augmented population functional 
\begin{align}
	\Sc(\ba, \boldsymbol{b}; \bq_{\alpha}, \bq_{\beta}) \coloneqq \left\langle \ba, \bbW^\top \boldsymbol{b} \right\rangle
	+ \left\langle \bq_{\beta}, \bby - \bbY^\top  \boldsymbol{b} \right\rangle
	+ \left\langle \bq_{\alpha}, \bbz - \bbZ^\top  \ba \right\rangle.
	\label{eq:theta.decomp}
\end{align}
At $(\ba, \boldsymbol{b}) = (\balpha^*, \bbeta^*)$, the two residual moments vanish by Corollary~\ref{cor:identification} with Propositions~\ref{prop:beta} and \ref{prop:alpha}. Consequently, $\Sc(\balpha^*, \bbeta^*; \bq_{\alpha}, \bq_{\beta}) = \theta$ for every choice of $(\bq_{\alpha}, \bq_{\beta})$. 
We choose these representers so that the score is first-order insensitive to perturbations in $(\balpha^*, \bbeta^*)$.
Differentiating with respect to the nuisance parameters yield the orthogonality requirements
\begin{align}
	\bbY \bq^*_{\beta} = \bbW \balpha^*,
	\quad
	\bbZ \bq^*_{\alpha} = \bbW^\top \bbeta^*.  
\end{align}
Assumption~\ref{assump:transport} ensures that these equations are feasible. Indeed, $\bbW \balpha^* \in \col(\bbW) \subseteq \col(\bbY_{\Ic_1, \pre})$, and thus, 
\begin{align}
	\bbY \bq^*_{\beta} = \bbY \bbY^\dagger \bbW \balpha^* = \bbW \balpha^*. \label{eq:riesz.exist.b}
\end{align}
Analogously, $\bbW^\top \bbeta^* \in \row(\bbW) \subseteq \row(\bbZ^\top)$, and hence,
\begin{align}
	\bbZ \bq^*_{\alpha} = \bbZ  \bbZ^\dagger \bbW^\top \bbeta^* = \bbW^\top \bbeta^*.  \label{eq:riesz.exist.a}
\end{align}
We take the canonical minimum $\ell_2$-norm solutions,
\begin{align}
	\bq^*_{\alpha} \coloneqq \bbZ^\dagger \bbW^\top \bbeta^*,
	\quad
	\bq^*_{\beta} &\coloneqq \bbY^\dagger \bbW \balpha^*.   \label{eq:riesz.proofs}
\end{align} 
These vectors are the Riesz representers of the two nuisance derivatives. They do not alter the population target. Instead, they express the forecast’s sensitivity to $\bbeta^*$ and $\balpha^*$ in the respective residual geometries, allowing the nuisance-linear terms to be canceled by residual corrections. Their norms also make the relevant geometry visible: by Lemma~\ref{lemma:pcr.spectra}, we recall \eqref{eq:wbeta.norm} and \eqref{eq:walpha} with Assumption~\ref{assump:bounded} to obtain 
\begin{align}
	\| \bq_{\alpha}^* \|_2 &\le \frac{\sqrt{L}}{\lambda_z},
	\quad
	\| \bq_{\beta}^* \|_2 \le \frac{\sqrt{N_1}}{\lambda_y}. \label{eq:riesz.bound}
\end{align}
Thus, the representers remain controlled when the forecast-relevant directions are well supported by the nonzero spectra of the two designs.

To transfer this population orthogonality to the empirical expansion, we introduce the proof-only projected representers 
\begin{align} 
	\tbq_{\alpha} \coloneqq \hbP_z \bq_{\alpha}^*,
	\quad
	\tbq_{\beta} \coloneqq \hbP_y \bq_{\beta}^*, \label{eq:riesz.proj.def}
\end{align}
where $\hbP_z \coloneqq \hbZ^\dagger \hbZ$ and $\hbP_y \coloneqq \hbY^\dagger \hbY$. 
We denote the corresponding estimation errors as 
\begin{align}
	\bDelta_{\tq_{\alpha}} \coloneqq \tbq_{\alpha} - \bq^*_{\alpha},
	\quad
	\bDelta_{\tq_{\beta}} \coloneqq \tbq_{\beta} - \bq^*_{\beta}, \label{eq:riesz.proj.est}
\end{align}
and primitive residual noises from the unit- and time-side population regressions as 
\begin{align}
	\bdelta_{\alpha} \coloneqq \bxi_z - \bXi_z^\top \balpha^*.
	\quad
	\bdelta_{\beta} \coloneqq \bxi_y - \bXi_y^\top \bbeta^*. 
	\label{eq:delta.noise}
\end{align} 
With this context in mind, our strategy is as follows: we will show that the estimation error separates into
\begin{align}
	\htheta - \theta = G + R. 
\end{align}
The leading term is a sum of three linear forms in mutually independent noise blocks; the remainder collects the products of \PCR~errors, representer-projection errors, and design perturbations. The remainder of the proof formalizes this separation. 
Below, we first state the key lemmas that aid in our proof of Theorem~\ref{thm:inference} in Appendix~\ref{sec:proofs.inference.lemmas}; we relegate their proofs to Appendix~\ref{sec:proofs.inference.lemmas.collect}. We then complete the proof of Theorem~\ref{thm:inference} in Appendix~\ref{sec:proofs.inference.final}. 


\subsection{Key Lemmas} \label{sec:proofs.inference.lemmas}

Lemma~\ref{lemma:riesz.0} controls the empirical projection of the Riesz representers. 
\begin{lemma} \label{lemma:riesz.0}
Let the setup of Corollary~\ref{cor:identification} hold. 
Further, let Assumption~\ref{assump:bounded} hold, as well as $k_y = r_y$ and $k_z = r_z$.
Then, conditional on $\Ec$, 
\begin{align}
	\| \bDelta_{\tq_{\alpha}} \|_2 
	\lesssim 
	\frac{\sqrt{L} \cdot \| \bXi_z \|_\eop}{\lambda_z^2},
	\quad 
	\| \bDelta_{\tq_{\beta}} \|_2 
	\lesssim 
	\frac{\sqrt{N_1} \cdot \| \bXi_y \|_\eop}{\lambda_y^2},
\end{align}
where $\bDelta_{\tq_{\alpha}}$ and $\bDelta_{\tq_{\beta}}$ are defined as in \eqref{eq:riesz.proj.est}. 
\end{lemma}

\noindent Next, Lemma~\ref{lemma:decomp} gives the exact orthogonal decomposition of $\htheta - \theta$ that will set the stage for the following analyses. 
\begin{lemma} \label{lemma:decomp}
Let the setup of Lemma~\ref{lemma:riesz.0} hold. 
Define 
\begin{align}
	G &\coloneqq \left \langle \balpha^*, \bXi_w^\top \bbeta^* \right \rangle + \left \langle \bq^*_{\beta}, \bdelta_{\beta} \right\rangle + \left \langle \bq^*_{\alpha}, \bdelta_{\alpha} \right\rangle, \label{eq:lead}
	\\
	R &\coloneqq 
	 \left\langle \bDelta_{\alpha}, \bbW^\top \bDelta_{\beta} \right\rangle
	+\left \langle \bDelta_{\alpha}, \bXi_w^\top \bbeta^* \right\rangle
	+\left \langle \balpha^*, \bXi_w^\top \bDelta_{\beta} \right\rangle
	+ \left\langle \bDelta_{\alpha}, \bXi_w^\top \bDelta_{\beta} \right\rangle 
	\\ &\qquad
	-  \left\langle \bq^*_{\beta}, \bXi_y^\top \bDelta_{\beta} \right\rangle
	- \left\langle \bq^*_{\alpha}, \bXi_z^\top \bDelta_{\alpha} \right\rangle
	\\ &\qquad
	+ \left\langle \bDelta_{\tq_{\beta}}, \bdelta_{\beta} \right\rangle - \left\langle \bDelta_{\tq_{\beta}}, \bbY^\top \bDelta_{\beta} \right\rangle - \left\langle \bDelta_{\tq_{\beta}}, \bXi_y^\top \bDelta_{\beta} \right\rangle
	\\ &\qquad
	+\left \langle \bDelta_{\tq_{\alpha}}, \bdelta_{\alpha} \right\rangle - \left\langle \bDelta_{\tq_{\alpha}}, \bbZ^\top \bDelta_{\alpha} \right\rangle - \left\langle \bDelta_{\tq_{\alpha}}, \bXi_z^\top \bDelta_{\alpha} \right\rangle.  \label{eq:remain}
\end{align}
Then, conditional on $\Ec$, $\htheta - \theta = G + R$. 
\end{lemma}

\noindent Lemma~\ref{lemma:lead} identifies the Gaussian law of the leading term $G$. 
\begin{lemma} \label{lemma:lead}
Let the setup of Theorem~\ref{thm:inference} hold. 
Then, conditional on $\Ec$, $G \sim \mathcal{N}(0, \Vc^2)$.
\end{lemma}

\noindent Finally, we establish a high probability bound on the remainder term $R$. 
\begin{lemma} \label{lemma:remainder}
Let the setup of Theorem~\ref{thm:inference} hold. 
Then, conditional on $\Ec$, with probability at least $1 - \Oc(\rho)$, $| R | \lesssim \Psi$.
\end{lemma}

\subsection{Completing Proof of Theorem~\ref{thm:inference}} \label{sec:proofs.inference.final}

\begin{proof}
Condition on $\Ec$. 
By Lemma~\ref{lemma:decomp}, 
\begin{align}
	\htheta - \theta = G + R. 
\end{align}
Lemma~\ref{lemma:lead} gives $G \sim \mathcal{N}(0, \Vc^2)$, and thus $G / \Vc \sim \mathcal{N}(0,1)$. 
At the same time, Lemma~\ref{lemma:remainder} states that
\begin{align}
	\Pb \left( | R | \le C \Psi \mid \Ec \right) \ge 1 - \Oc(\rho)
\end{align}
for a sufficiently large constant $C >0$. 
Therefore, for any $\epsilon > 0$,
\begin{align}
	\Pb \left( \frac{ | R | }{\Vc} > \epsilon \mid \Ec \right)
	&\le C \rho + \mathds{1}\left\{ \frac{C \Psi}{\Vc} > \epsilon \right\}. 
\end{align}
Because $\rho \rightarrow 0$ as $N_1, T_0, L, M \rightarrow \infty$ and $\Psi / \Vc \rightarrow 0$ by assumption, it follows that $R / \Vc \xrightarrow{p} 0$. 
Therefore, combining the above gives  
\begin{align}
	\frac{\htheta - \theta}{\Vc}
	&= \frac{G}{\Vc} + \frac{R}{\Vc}
	= \frac{G}{\Vc} + o_p(1), 
\end{align}
and conditional Slutsky gives
\begin{align}
	\frac{\htheta - \theta}{\Vc} \rightsquigarrow \mathcal{N}(0,1). 
\end{align}
The proof is complete. 
\end{proof}

\subsection{Proof of Key Lemmas} \label{sec:proofs.inference.lemmas.collect} 

\subsubsection{Proof of Lemma~\ref{lemma:riesz.0}}

\begin{proof}
Condition on $\Ec$. 
Since $\bq^*_{\beta} \in \row(\bbY)$, we obtain 
\begin{align}
	\bDelta_{\tq_{\beta}} &= - \left(\bI_{T_0} - \hbP_y \right) \cdot \bP_y \bq^*_{\beta},
\end{align}
where $\bP_y \coloneqq \bbY^\dagger \bbY$. 
Taking norms, 
\begin{align}
	\| \bDelta_{\tq_{\beta}} \|_2 &\le \| \hbP_y^\perp \bP_y \|_\txtop \cdot \| \bq^*_{\beta} \|_2
	\\ &\le \frac{2  \| \bXi_y \|_\txtop}{\lambda_y} \cdot \| \bq^*_{\beta} \|_2  && \because \text{Lemma~\ref{lemma:pcr.peturb.2}}
	\\ &\le \frac{2 \sqrt{N_1} \| \bXi_y \|_\txtop}{\lambda_y^2}. && \because \text{\eqref{eq:riesz.bound}} 
\end{align}
By a symmetrical argument, we also conclude 
\begin{align}
	\| \bDelta_{\tq_{\alpha}} \|_2 &\le  \frac{2 \sqrt{L} \| \bXi_z \|_\txtop}{\lambda_z^2}. 
\end{align}
This concludes the proof. 
\end{proof} 

\subsubsection{Proof of Lemma~\ref{lemma:decomp}}

\begin{proof}
Condition on $\Ec$. 
Motivated by the population score in \eqref{eq:theta.decomp}, we consider a sample argument that implements the same idea using the proof-only projections of the population representers in \eqref{eq:riesz.proj.def} onto the empirical \PCR~row spaces. This yields, 
\begin{align}
	\htheta + \left \langle \tbq_{\beta}, \by - \bY^\top \hbbeta \right \rangle + \left \langle \tbq_{\alpha}, \bz - \bZ^\top \hbalpha \right \rangle. 
\end{align}
By the construction of \eqref{eq:riesz.proj.def} and the facts that $\bY^\top \hbbeta = \hbY^\top \hbbeta$ and $\bZ^\top \hbalpha = \hbZ^\top \hbalpha$, it follows that 
\begin{align}
	\left \langle \tbq_{\alpha}, \bz - \bZ^\top \hbalpha \right \rangle = 0,
	\quad
	\left \langle \tbq_{\beta}, \by - \bY^\top \hbbeta \right \rangle = 0. 
\end{align}
That is, the residual corrections are exactly zero by the \PCR~normal equations, and hence may be added to $\htheta$ without any alteration. Therefore, 
\begin{align}
	\htheta = \htheta + \left \langle \tbq_{\beta}, \by - \bY^\top \hbbeta \right \rangle + \left \langle \tbq_{\alpha}, \bz - \bZ^\top \hbalpha \right \rangle. 
\end{align}
To enable further progress, note that 
\begin{align}
	\by - \bY^\top \hbbeta &= 
	\left(\bby + \bxi_y \right) - \left(\bbY + \bXi_y\right)^\top \left(\bbeta^* + \bDelta_{\beta}\right)
	\\
	&= \bxi_y - \bXi_y^\top \bbeta^* - \bbY^\top \bDelta_{\beta} - \bXi_y^\top \bDelta_{\beta} && \because \bby = \bbY^\top \bbeta^*
	\\
	&= \bdelta_{\beta} - \bbY^\top \bDelta_{\beta} - \bXi_y^\top \bDelta_{\beta}. 
\end{align}
Accordingly, we have
\begin{align}
	\left\langle \tbq_{\beta}, \by - \bY^\top \hbbeta \right\rangle
	&= \left(\bq^*_{\beta} + \bDelta_{\tq_{\beta}} \right)^\top \left(\bdelta_{\beta} - \bbY^\top \bDelta_{\beta} - \bXi_y^\top \bDelta_{\beta} \right)
	\\
	&= \left \langle \bq^*_{\beta}, \bdelta_{\beta} \right\rangle - \left\langle \bq^*_{\beta}, \bbY^\top \bDelta_{\beta} \right\rangle - \left\langle \bq^*_{\beta}, \bXi_y^\top \bDelta_{\beta} \right\rangle 
	\\ &\qquad + \left\langle \bDelta_{\tq_{\beta}}, \bdelta_{\beta} \right\rangle - \left\langle \bDelta_{\tq_{\beta}}, \bbY^\top \bDelta_{\beta} \right\rangle - \left\langle \bDelta_{\tq_{\beta}}, \bXi_y^\top \bDelta_{\beta} \right\rangle. \label{eq:unit.correct}
\end{align}
Following similar arguments, we arrive at the decomposition
\begin{align}
	\left\langle \tbq_{\alpha}, \bz - \bZ^\top \hbalpha \right \rangle
	&= \left(\bq^*_{\alpha} + \bDelta_{\tq_{\alpha}} \right)^\top \left(\bdelta_{\alpha} - \bbZ^\top \bDelta_{\alpha} - \bXi_z^\top \bDelta_{\alpha} \right)
	\\
	&= \left\langle \bq^*_{\alpha}, \bdelta_{\alpha} \right \rangle - \left\langle \bq^*_{\alpha}, \bbZ^\top \bDelta_{\alpha}  \right\rangle - \left\langle \bq^*_{\alpha}, \bXi_z^\top \bDelta_{\alpha} \right \rangle 
	\\ &\qquad + \left\langle \bDelta_{\tq_{\alpha}}, \bdelta_{\alpha} \right\rangle - \left\langle \bDelta_{\tq_{\alpha}}, \bbZ^\top \bDelta_{\alpha} \right\rangle - \left\langle \bDelta_{\tq_{\alpha}}, \bXi_z^\top \bDelta_{\alpha} \right\rangle. \label{eq:time.correct}
\end{align}
Further, by \eqref{eq:riesz.exist.b}, 
\begin{align}
	\left \langle \bq^*_{\beta}, \bbY^\top \bDelta_{\beta} \right \rangle
	&= \left\langle \bDelta_{\beta}, \bbY \bq^*_{\beta} \right\rangle
	= \left\langle \bDelta_{\beta}, \bbW \balpha^* \right\rangle
	= \left\langle \balpha^*, \bbW^\top \bDelta_{\beta} \right\rangle. 
	\label{eq:cancel.1}
\end{align}
Similarly, \eqref{eq:riesz.exist.a} yields
\begin{align}
	\left\langle \bq^*_{\alpha}, \bbZ^\top \bDelta_{\alpha} \right \rangle
	= \left \langle \bDelta_{\alpha}, \bbW^\top \bbeta^* \right \rangle. 
	\label{eq:cancel.2}
\end{align}
Therefore, adding \eqref{eq:general}, \eqref{eq:unit.correct}, and \eqref{eq:time.correct}, and applying the cancellations \eqref{eq:cancel.1} and \eqref{eq:cancel.2}, gives our desired result, whereby the two troublesome nuisance-linear terms cancel algebraically.
\end{proof}

\subsubsection{Proof of Lemma~\ref{lemma:lead}}

\begin{proof}
Condition on $\Ec$. 
Recall from \eqref{eq:lead} that 
\begin{align}
	G &\coloneqq \left \langle \balpha^*, \bXi_w^\top \bbeta^* \right \rangle + \left \langle \bq^*_{\beta}, \bdelta_{\beta} \right \rangle + \left \langle \bq^*_{\alpha}, \bdelta_{\alpha} \right \rangle.  
\end{align}
We will analyze each term independently. 

\bigskip \noindent {\em Term 1: $\left \langle \balpha^*, \bXi_w^\top \bbeta^* \right \rangle$}. 
Conditional on $\Ec$, Assumption~\ref{assump:subg} with homoskedastic Gaussian noise implies $\left \langle \balpha^*, \bXi_w^\top \bbeta^* \right \rangle$ is a mean-zero normal random variable with variance
\begin{align}
	\Var\left( \left \langle \balpha^*, \bXi_w^\top \bbeta^* \right \rangle \right)
	= \sigma^2 \cdot \| \balpha^* \|_2^2 \cdot \| \bbeta^* \|_2^2. 
	\label{eq:ln.1}
\end{align}

\noindent {\em Term 2: $\left \langle \bq^*_{\beta}, \bdelta_{\beta} \right \rangle$}. 
Inserting \eqref{eq:delta.noise} and expanding terms, we have
\begin{align}
	\left \langle \bq^*_{\beta}, \bdelta_{\beta} \right \rangle &= 
	\left \langle \bq^*_{\beta}, \bxi_y - \bXi_y^\top \bbeta^* \right \rangle.
\end{align}
Assumption~\ref{assump:subg} under Gaussianity yields
\begin{align}
	\bxi_y \sim \mathcal{N}\left(\bzero, \sigma^2 \bI_{T_0} \right),
	\quad
	\bXi_y^\top \bbeta^* \sim \mathcal{N}\left(\bzero, \sigma^2 \cdot \| \bbeta^* \|_2^2 \cdot \bI_{T_0} \right). 
\end{align}
Because $\bxi_y$ and $\bXi_y$ are independent, 
\begin{align}
	\bdelta_{\beta} \sim \mathcal{N}\left( \bzero, \sigma^2\cdot(1 + \| \bbeta^* \|_2^2) \cdot \bI_{T_0} \right),
\end{align}
and hence, $\left \langle \bq^*_{\beta}, \bdelta_{\beta} \right \rangle$ is a mean-zero normal random variable with 
\begin{align}
	\Var\left(\left \langle \bq^*_{\beta}, \bdelta_{\beta} \right \rangle \right) = \sigma^2 \cdot \| \bq^*_{\beta} \|_2^2 \cdot ( 1 + \| \bbeta^* \|_2^2 ). 
	\label{eq:ln.2}
\end{align}

\noindent {\em Term 3: $\left \langle \bq^*_{\alpha}, \bdelta_{\alpha} \right \rangle$}. 
Following the arguments that led to \eqref{eq:ln.2}, we conclude that $\left \langle \bq^*_{\alpha}, \bdelta_{\alpha} \right \rangle$ is a mean-zero normal random variable with 
\begin{align}
	\Var\left(\left \langle \bq^*_{\alpha}, \bdelta_{\alpha} \right \rangle \right) = \sigma^2 \cdot \| \bq^*_{\alpha} \|_2^2 \cdot (1 + \| \balpha^* \|_2^2). 
\end{align}

\noindent {\em Putting everything together.}
By Assumption~\ref{assump:subg} under Gaussianity and the Page construction, $\bXi_w$, $(\bxi_y, \bXi_y)$, and $(\bxi_z, \bXi_z)$ are mutually independent noise blocks, and hence summing their variances gives $\Vc^2$ and our desired result. 
\end{proof}

\subsubsection{Proof of Lemma~\ref{lemma:remainder}}

\begin{proof}
Condition on $\Ec$. 
Throughout, let 
\begin{align}
	\varphi_\rho &\coloneqq \sqrt{C_\varphi \log(C_\varphi / \rho)}. \label{eq:varphi.update}
\end{align}
We state several useful bounds on the terms above, which will help us simplify the terms later on. 
Following the arguments that led to \eqref{eq:eta.lambda}, we leverage \eqref{eq:sep.y} and \eqref{eq:sep.z} to conclude
\begin{align}
	\eta_y \lesssim \lambda_y,
	\quad
	\eta_z \lesssim \lambda_z. \label{eq:eta.lambda.simple}
\end{align}
Moreover, Assumptions~\ref{assump:bounded} and ~\ref{assump:spectra} yield
\begin{align}
	&\lambda_y^2 \asymp \frac{N_1 T_0}{r_y}, 
	\quad 
	\lambda_z^2 \asymp \frac{LM}{r_z}, \label{eq:lmda.lb} 
	\\
	&\| \bbY \|_\txtop  \lesssim \lambda_y, 
	\quad
	\| \bbZ \|_\txtop \lesssim \lambda_z. \label{eq:pop.op}
\end{align}
%
%
We proceed to bound the individual error terms in $R$ based on \eqref{eq:remain}. 

\bigskip \noindent {\em Deterministic $\bbW$-term}: 
Let $\Gc_{\PCR, \alpha}$ and $\Gc_{\PCR, \beta}$ be defined as in \eqref{eq:event.pcr}. 
On $\Gc_\PCR \coloneqq \Gc_{\PCR, \alpha} \cap \Gc_{\PCR, \beta}$, we obtain 
\begin{align}
	\left| \left \langle \bDelta_{\alpha}, \bbW^\top \bDelta_{\beta} \right \rangle \right| 
	&\le \| \bDelta_{\alpha} \|_2 \cdot \| \bbW \|_\txtop \cdot \| \bDelta_{\beta} \|_2
	\lesssim \sigma^2 \  \Lambda_{\alpha} \Lambda_{\beta}. 
	\label{eq:remain.t1}
\end{align} 
Note that the above uses \eqref{eq:barw.bound.op}. 
By Theorem~\ref{prop:parameter.recovery}, note that 
\begin{align}
	\Pb\left(\Gc_\PCR^c \mid \Ec \right) \lesssim (N_1T_0)^{-10} + (LM)^{-10} \le \rho. \label{eq:inf.prob.pcr} 
\end{align}

\noindent {\em Stochastic $\bXi_w$-terms}: 
Define three $\bXi_w$-events as  
\begin{align}
	\Gc_{\Delta \Delta} &\coloneqq \left\{ \left| \left \langle \bDelta_{\alpha}, \bXi_w^\top \bDelta_{\beta} \right \rangle \right| \le  
	\frac{C_w \cdot \sigma^3 \varphi_\rho \cdot \Lambda_{\alpha} \Lambda_{\beta}}{\sqrt{N_1L}} \right\},
	\\
	\Gc_{\Delta \beta} &\coloneqq \left\{ \left| \left \langle \bDelta_{\alpha}, \bXi_w^\top \bbeta^* \right \rangle \right| \le  
	\frac{C_w \cdot \sigma^2 \varphi_\rho \cdot \sqrt{T_0} \Lambda_{\alpha}}{\sqrt{L} \lambda_y} \right\},
	\\
	\Gc_{\alpha \Delta} &\coloneqq \left\{\left| \left \langle \balpha^*, \bXi_w^\top \bDelta_{\beta} \right \rangle \right| \le  
	\frac{C_w \cdot \sigma^2 \varphi_\rho \cdot \sqrt{M} \Lambda_{\beta}}{\sqrt{N_1} \lambda_z} \right\},
	\label{eq:inf.term.w}
\end{align}
where $C_w > 0$ is a sufficiently large constant, and $\Gc_w \coloneqq \Gc_{\Delta \Delta} \cap \Gc_{\Delta \beta} \cap \Gc_{\alpha \Delta}$. 
On $\Gc_\PCR \cap \Gc_w$, the $\bXi_w$-block stochastic terms satisfy the bounds in \eqref{eq:inf.term.w}. 
Moreover, following the arguments that led to \eqref{eq:rho.temp}, we obtain 
\begin{align}
	\Pb\left(\Gc_\PCR \cap \Gc_w^c \mid \Ec \right) \lesssim \rho. \label{eq:inf.prob.w}
\end{align}

\noindent {\em Riesz errors: unit-side}. 
Define $\Gc_\noise$ as in \eqref{eq:event.noise}, and define 
\begin{align}
	\Gc_{\xi \beta} \coloneqq \left\{ | \left \langle \bDelta_{\tq_{\beta}}, \bxi_y \right \rangle | \le C_\xi \cdot \sigma \varphi_\rho \cdot \tQ_{\beta} \right\},
\end{align}
where $C_\xi > 0$ is a sufficiently large constant and 
\begin{align}
	\tQ_{\beta} \coloneqq \frac{\eta_y\sqrt{N_1}}{\lambda_y^2}. \label{eq:tQ.beta} 
\end{align}
On $\Gc_{\xi \beta}$, we have 
\begin{align}
	\left| \left \langle \bDelta_{\tq_{\beta}}, \bdelta_{\beta} \right \rangle \right| 
	&\le \left| \left \langle \bDelta_{\tq_{\beta}}, \bxi_y \right \rangle \right| +\left | \left \langle \bDelta_{\tq_{\beta}}, \bXi_y^\top \bbeta^* \right \rangle \right|
	\lesssim \sigma \varphi_\rho \cdot \tQ_{\beta} + \left| \left \langle \bDelta_{\tq_{\beta}}, \bXi_y^\top \bbeta^* \right \rangle \right|. 
\end{align}
To control the second term, observe that on $\Gc_{\noise, y}$, Lemma~\ref{lemma:riesz.0} gives $\| \bDelta_{\tq_{\beta}} \|_2 \le \tQ_{\beta}$. 
Moreover, on $\Gc_\noise$, \eqref{eq:alpha.beta.bound} states 
\begin{align}
	\| \bXi_y \|_\txtop \cdot \| \bbeta^* \|_2 \le \frac{\eta_y \sqrt{T_0}}{\lambda_y}. 
\end{align}
Combining the above, we conclude that on $\Gc_{\noise, y} \cap \Gc_{\xi \beta}$, 
\begin{align}
	\left| \left \langle \bDelta_{\tq_{\beta}}, \bdelta_{\beta} \right \rangle \right| &\lesssim \tQ_{\beta} \cdot \left\{ \sigma \varphi_\rho + \frac{\eta_y \sqrt{T_0}}{\lambda_y} \right\}. 
	\label{eq:t7.1}
\end{align} 
Define the sigma-field $\Hc_y \coloneqq \Ec \vee \sigma(\bY)$. 
Then, $\bDelta_{\tq_{\beta}}$ is $\Hc_y$-measurable and $\bxi_y$ is independent of $\Hc_y$, conditional on $\Ec$. 
Moreover, conditional on $\Hc_y$, $\left \langle \bDelta_{\tq_{\beta}}, \bxi_y \right \rangle \sim \mathcal{N}(0, \sigma^2 \| \bDelta_{\tq_{\beta}} \|_2^2)$. 
Thus, by Lemma~\ref{lemma:hoeffding} and the tower property, we have 
\begin{align}
	\Pb\left( \Gc_{\noise, y}  \cap \Gc_{\xi \beta}^c \mid \Ec \right) 
	&= \Ex\left[ \mathds{1}\left\{\Gc_{\noise, y} \right\} \cdot \Pb \left( \Gc^c_{\xi \beta} \mid \Hc_y \right) \mid \Ec \right]
	\le 2 \cdot \exp(-\varphi_\rho^2) \lesssim \rho. 
	\label{eq:inf.prob.riesz.unit} 
\end{align} 
Continuing, note that on $\Gc_\noise \cap \Gc_{\PCR}$, \eqref{eq:pop.op} yields  
\begin{align}
	\left| \left \langle \bDelta_{\tq_{\beta}}, \bbY^\top \bDelta_{\beta} \right \rangle \right|
	&\le \| \bDelta_{\tq_{\beta}} \|_2 \cdot \| \bbY \|_\txtop \cdot \| \bDelta_{\beta}\|_2
	\lesssim \tQ_{\beta} \cdot \lambda_y \cdot \frac{\sigma \Lambda_{\beta}}{\sqrt{N_1}}, 
	\label{eq:t7.2}
\end{align}
as well as 
\begin{align}
	\left| \left \langle \bDelta_{\tq_{\beta}}, \bXi^\top_y \bDelta_{\beta} \right \rangle \right|
	&\le \| \bDelta_{\tq_{\beta}} \|_2 \cdot \| \bXi_y \|_\txtop \cdot \| \bDelta_{\beta}\|_2
	\lesssim \tQ_{\beta} \cdot \eta_y \cdot \frac{\sigma \Lambda_{\beta}}{\sqrt{N_1}}. 
	\label{eq:t7.3}
\end{align}
As a result, \eqref{eq:t7.1}, \eqref{eq:t7.2}, and \eqref{eq:t7.3} assert that on $\Gc_\noise \cap \Gc_{\PCR} \cap \Gc_{\xi \beta}$, 
\begin{align}
	\left| \left \langle \bDelta_{\tq_{\beta}}, \bdelta_{\beta} \right \rangle \right| + \left| \left \langle \bDelta_{\tq_{\beta}}, \bbY^\top \bDelta_{\beta} \right \rangle \right| + 
	\left| \left \langle \bDelta_{\tq_{\beta}}, \bXi^\top_y \bDelta_{\beta} \right \rangle \right|
	&\lesssim 
	\tQ_{\beta} \cdot \left\{ \sigma \varphi_\rho + \frac{\eta_y \sqrt{T_0}}{\lambda_y} + \frac{\sigma \lambda_y \Lambda_{\beta}}{\sqrt{N_1}} \right\}. 
	\label{eq:remain.t7}
\end{align}
Further, on $\Gc_\noise \cap \Gc_\PCR$, applying \eqref{eq:riesz.bound} gives 
\begin{align}
	\left| \left \langle \bq^*_{\beta}, \bXi_y^\top \bDelta_{\beta} \right \rangle \right| 
	&\le \| \bq^*_{\beta} \|_2 \cdot \| \bXi_y \|_\txtop \cdot \| \bDelta_{\beta} \|_2
	\lesssim \frac{\sigma \eta_y \Lambda_{\beta}}{\lambda_y}. 
	\label{eq:remain.t5}
\end{align} 
Using Lemma~\ref{lemma:subg_matrix}, \eqref{eq:inf.prob.pcr}, and taking a union bound gives 
\begin{align}
	\Pb \left( (\Gc_\PCR \cap \Gc_\noise)^c \mid \Ec \right) 
	&\le \Pb \left( \Gc^c_\PCR \mid \Ec \right) + \Pb \left( \Gc^c_\noise \mid \Ec \right) 
	\lesssim \rho. 	
	\label{eq:inf.prob.noise} 
\end{align}

\noindent {\em Riesz errors: time-side}. 
Following the arguments above, we first define the event 
\begin{align}
	\Gc_{\xi \alpha} \coloneqq \left\{ | \left \langle \bDelta_{\tq_{\alpha}}, \bxi_z \right \rangle | \le C_\xi \cdot \sigma \varphi_\rho \cdot \tQ_{\alpha} \right\},
\end{align}
where
\begin{align}
	\tQ_{\alpha} \coloneqq \frac{\eta_z \sqrt{L}}{\lambda_z^2}. \label{eq:tQ.alpha} 
\end{align}
On $\Gc_\noise \cap \Gc_{\PCR} \cap \Gc_{\xi \alpha}$, we conclude  
\begin{align}
	\left| \left \langle \bDelta_{\tq_{\alpha}}, \bdelta_{\alpha} \right \rangle \right| + \left| \left \langle \bDelta_{\tq_{\alpha}}, \bbZ^\top \bDelta_{\alpha} \right \rangle \right| + 
	\left| \left \langle \bDelta_{\tq_{\alpha}}, \bXi^\top_z \bDelta_{\alpha} \right \rangle \right|
	&\lesssim 
	\tQ_{\alpha} \cdot \left\{ \sigma \varphi_\rho + \frac{\eta_z \sqrt{M}}{\lambda_z} + \frac{\sigma \lambda_z \Lambda_{\alpha}}{\sqrt{L}} \right\}.
	\label{eq:remain.t8}
\end{align}
%
Define the sigma-field $\Hc_z \coloneqq \Ec \vee \sigma(\bZ)$.
Then, $\bDelta_{\tq_{\alpha}}$ is $\Hc_z$-measurable, while $\bxi_z$ is independent of $\Hc_z$, conditional on $\Ec$. 
Mirroring the logic behind \eqref{eq:inf.prob.riesz.unit}, we assert 
\begin{align}
	\Pb\left( \Gc_{\noise, z} \cap \Gc_{\xi \alpha}^c \mid \Ec \right) \lesssim \rho. 
	\label{eq:inf.prob.riesz.time} 
\end{align}
Finally, on $\Gc_\noise \cap \Gc_\PCR$, 
\begin{align}
	\left| \left \langle \bq^*_{\alpha}, \bXi_z^\top \bDelta_{\alpha} \right \rangle \right| 
	&\lesssim \frac{\sigma \eta_z \Lambda_{\alpha}}{\lambda_z}. 
	\label{eq:remain.t6}
\end{align}  

\noindent {\em Putting everything together.}
On the master event $\Gc_\star \coloneqq \Gc_\noise \cap \Gc_\PCR \cap \Gc_w \cap \Gc_{\xi \beta} \cap \Gc_{\xi \alpha}$, inserting \eqref{eq:remain.t1}, \eqref{eq:inf.term.w}, \eqref{eq:remain.t5}, \eqref{eq:remain.t6}, \eqref{eq:remain.t7}, and \eqref{eq:remain.t8}, into \eqref{eq:remain}, we arrive at the inequality 
\begin{align}
	\left| R \right| 
	&\lesssim 
	\sigma^2  \Lambda_{\alpha} \Lambda_{\beta}
	+ \frac{\sigma^3 \varphi_\rho \Lambda_{\alpha} \Lambda_{\beta}}{\sqrt{N_1L}}
	+ \frac{\sigma^2 \varphi_\rho \sqrt{T_0} \Lambda_{\alpha}}{\sqrt{L} \lambda_y}
	+ \frac{\sigma^2 \varphi_\rho \sqrt{M} \Lambda_{\beta}}{\sqrt{N_1} \lambda_z} 
	+ \frac{\sigma \eta_y \Lambda_{\beta}}{\lambda_y} 
	+ \frac{\sigma \eta_z \Lambda_{\alpha}}{\lambda_z}
	\\
	&\quad 
	+ \tQ_{\beta} \cdot \left\{ \sigma \varphi_\rho + \frac{\eta_y \sqrt{T_0}}{\lambda_y} + \frac{\sigma \lambda_y \Lambda_{\beta}}{\sqrt{N_1}} \right\}
	+ \tQ_{\alpha} \cdot \left\{ \sigma \varphi_\rho + \frac{\eta_z \sqrt{M}}{\lambda_z} + \frac{\sigma \lambda_z \Lambda_{\alpha}}{\sqrt{L}} \right\}. 
\end{align}
Plugging \eqref{eq:lmda.lb} and \eqref{eq:noise.eta} into the above, and simplifying, gives our desired result. 
It remains to bound the probability of the master event $\Gc_\star$. 
Taking a union bound over  \eqref{eq:inf.prob.pcr}, \eqref{eq:inf.prob.w}, \eqref{eq:inf.prob.riesz.unit}, \eqref{eq:inf.prob.noise}, and \eqref{eq:inf.prob.riesz.time} we conclude 
\begin{align}
	\Pb\left(\Gc^c_\star \mid \Ec \right)
	&\le 
	\Pb\left(\Gc_\noise^c \mid \Ec \right)
	+ \Pb\left(\Gc_\PCR^c \mid \Ec \right)
	+ \Pb\left(\Gc_\PCR \cap \Gc_w^c \mid \Ec \right)
	\\ &\quad 
	+ \Pb\left( \Gc_{\noise, y} \cap \Gc_{\xi \beta}^c \mid \Ec \right)
	+ \Pb\left( \Gc_{\noise, z}  \cap \Gc_{\xi \alpha}^c \mid \Ec \right)
	\lesssim \rho. 
\end{align}
The proof is complete. 
\end{proof}

\section{Proof of Proposition~\ref{prop:var.sigma}} \label{sec:proof.var.est.2}
Define two variance estimators based on the unit- and time-sided \PCR~residuals:
\begin{align}
	\hsigma^2_{\beta} &\coloneqq \frac{ \| \by - \bY^\top \hbbeta \|_2^2}{T_0 - k_y}, 
	\quad
	\hsigma^2_{\alpha} \coloneqq \frac{ \| \bz - \bZ^\top \hbalpha \|_2^2}{M - k_z}. \label{eq:var.sigma.ab}
\end{align}
In turn, these yield our pooled estimator defined in \eqref{eq:var.sigma.pool}: 
\begin{align}
	\hsigma^2 &\coloneqq \frac{(T_0 - k_y) \cdot \hsigma^2_{\beta} + (M - k_z) \cdot \hsigma^2_{\alpha}}{(T_0 - k_y) + (M - k_z)}. 
\end{align}
Our strategy will be to control the estimation errors of $\hsigma^2_{\beta}$ and $\hsigma^2_{\alpha}$, and then to combine these results into a statement on $\hsigma^2$. 
Before doing so, we first state a useful lemma; its proof is relegated to Appendix~\ref{sec:proof.subexp.2}. 

\begin{lemma} \label{lemma:subexp.2}
Let $\bxi \sim \mathcal{N}(0, \sigma^2 \bI)$. 
Let $\Hc$ be a sigma-field independent of $\bxi$. 
Let $\bP \in \Rb^{m \times m}$ be an $\Hc$-measurable orthogonal projection matrix. 
Then for any $t > 0$,
\begin{align}
	\Pb \left( \left| \bxi^\top \bP \bxi - \sigma^2 \tr(\bP) \right| > C \sigma^2 \cdot \left\{ \sqrt{\tr(\bP) \cdot t} + t \right\} \mid \Hc \right) \le 2 \exp(-ct),
\end{align}
for constants $C, c > 0$. 
\end{lemma}

\subsection{Completing Proof of Proposition~\ref{prop:var.sigma}}

\begin{proof}
We begin with the proof for $\hsigma^2_{\beta}$ since the proof for $\hsigma^2_{\alpha}$ is analogous. 
Condition on $\Ec$. 
Let $d_y = T_0 - r_y$, 
\begin{align}
	\hbP_y \coloneqq \hbY^\top \hbY^{\top, \dagger},
	\quad
	\hbP_y^\perp \coloneqq \bI_{T_0} - \hbP_y. 
\end{align}
With this notation, we can write
\begin{align}
	\hsigma^2_{\beta} &= \frac{\| \hbP^\perp_y  \by \|_2^2}{d_y}
	= \frac{\| \hbP_y^\perp  \bby \|_2^2}{d_y} + \frac{ \| \hbP_y^\perp  \bxi_y \|_2^2}{d_y} + \frac{2 \bxi_y^\top \hbP^\perp_y  \bby}{d_y}.  
\end{align}
Dividing by $\sigma^2$ and subtracting one gives
\begin{align}
	\frac{\hsigma_{\beta}^2}{\sigma^2} - 1 = \frac{\| \hbP_y^\perp  \bby \|_2^2}{\sigma^2 d_y} + \left( \frac{\| \hbP^\perp_y \bxi_y \|_2^2}{\sigma^2 d_y} - 1 \right)+  \frac{2 \bxi_y^\top \hbP^\perp_y  \bby}{\sigma^2 d_y}.
\end{align} 
We proceed to bound each term separately. 

\bigskip \noindent {\em Term 1}: 
Define the population projection onto $\col(\bbY)$ as $\bP_y \coloneqq \bbY^\top \bbY^{\top, \dagger}$. 
Since $\bby \in \col(\bbY^\top)$, we can write 
\begin{align}
	\| \hbP_y^\perp  \bby \|_2 
	&= \| \hbP_y^\perp \bP_y \bby \|_2
	\\
	&\le \| \hbP_y^\perp \bP_y \|_{\txtop} \cdot \| \bby \|_2
	\\
	&\le \left( \frac{2 \cdot \| \bXi_y \|_\txtop }{\lambda_y} \right) \cdot \| \bby \|_2 && \because \text{Lemma~\ref{lemma:pcr.peturb.2}}
	\\
	&\le \frac{2 \sqrt{T_0} \cdot \| \bXi_y\|_\txtop}{\lambda_y}. && \because \text{Assumption~\ref{assump:bounded}}
\end{align}
Define $\Gc_{\noise, y}$ as in \eqref{eq:event.noise}. 
By Lemma~\ref{lemma:subg_matrix}, 
\begin{align}
	\Pb\left( \Gc_{\noise, y}^c \mid \Ec \right) \lesssim (N_1 T_0)^{-10}. 
\end{align}
On this event $\Gc_{\noise, y}$, 
\begin{align}
	\frac{\| \hbP_y^\perp  \bby \|_2^2}{\sigma^2 d_y} &\le \frac{4 \eta_y^2 T_0}{\sigma^2 \lambda^2_y d_y}. 
	\label{eq:var.sigma.t1}
\end{align} 

\noindent {\em Term 2}:
Letting $\tau_y \coloneqq 1 + \log(N_1 T_0)$, define the event
\begin{align}
	\Gc_{\proj, y} \coloneqq \left\{ \left| \bxi_y^\top \hbP_y^\perp \bxi_y - \sigma^2 \tr(\hbP_y^\perp) \right| 
	\le C_\proj \sigma^2 \left( \sqrt{\tr(\hbP^\perp_y) \tau_y} + \tau_y \right) \right\},
\end{align}
where $C_\proj > 0$ is an absolute constant. 
Conditional on $\Ec$ and $\bY$, the projection $\hbP_y$ is fixed, and $\bxi_y \sim \mathcal{N}(0, \sigma^2 \bI_{T_0})$. 
Applying Lemma~\ref{lemma:subexp.2} with the sigma-field $\Hc \coloneqq \Ec \vee \sigma(\bY)$ gives
\begin{align}
	\Pb \left( \Gc_{\proj, y}^c \mid \Hc \right) \le 2 \exp(-c \tau_y). 
\end{align} 
Taking conditional expectation with respect to $\Ec$, 
\begin{align}
	\Pb\left(\Gc_{\proj, y}^c \mid \Ec \right) = \Ex\left[ \Pb\left(\Gc_{\proj, y}^c \mid \Hc \right) \mid \Ec \right] \le 2 \exp(- c \tau_y) \lesssim (N_1T_0)^{-10}. 
\end{align}
Since $\tr(\hbP^\perp_y) = d_y$, on the event $\Gc_{\proj, y}$, we have
\begin{align}
	\frac{\| \hbP^\perp_y \bxi_y \|_2^2}{\sigma^2 d_y} - 1
	&\lesssim \frac{\sqrt{\tau_y}}{\sqrt{d_y}} + \frac{\tau_y}{d_y}. 
	\label{eq:var.sigma.t2}
\end{align}

\noindent {\em Term 3}: 
Define the event
\begin{align}
	\Gc_{\tcross, y} \coloneqq \left\{ \left| \bxi_y^\top \hbP_y^\perp \bby \right| \le C_\tcross \sigma \| \hbP^\perp_y \bby \|_2 \sqrt{\tau_y} \right\},
\end{align}
where $C_\tcross > 0$ is an absolute constant. 
As aforementioned, conditional on $\Ec$ and $\bY$, the projection $\hbP_y$ is fixed. 
Moreover, by Assumption~\ref{assump:subg} with homoskedastic Gaussian noise,
\begin{align}
	\bxi_y^\top \hbP_y^\perp \bby \sim \mathcal{N}\left(0, \sigma^2 \cdot \| \hbP^\perp_y \bby \|_2^2 \right).
\end{align}
Hence, by Lemma~\ref{lemma:hoeffding}, 
\begin{align}
	\Pb \left( \Gc^c_{\tcross, y} \mid \Hc \right) \le 2\exp(-c C_\tcross^2 \tau_y) \lesssim (N_1 T_0)^{-10}. 
\end{align}
Taking conditional expectation gives 
\begin{align}
	\Pb \left( \Gc^c_{\tcross, y} \mid \Ec \right) \lesssim (N_1 T_0)^{-10}. 
\end{align}
On the event $\Gc_{\tcross, y}$, we have 
\begin{align}
	\frac{2 \left| \bxi_y^\top \hbP^\perp_y  \bby \right|}{\sigma^2 d_y}
	&\le \frac{2 C_\tcross \| \hbP^\perp_y \bby \|_2 \sqrt{\tau_y}}{\sigma d_y}
	\lesssim
	\frac{\| \hbP^\perp_y \bby \|_2^2}{\sigma^2 d_y} + \frac{\tau_y}{d_y},
\end{align}
which follows from the elementary fact $2ab \le a^2 + b^2$ for any $a, b \ge 0$. 
By \eqref{eq:var.sigma.t1}, on $\Gc_{\noise, y} \cap \Gc_{\tcross, y}$, 
\begin{align}
	\frac{2 \left| \bxi_y^\top \hbP^\perp_y  \bby \right|}{\sigma^2 d_y}
	\lesssim
	\frac{\eta_y^2 T_0}{\sigma^2 \lambda^2_y d_y} + \frac{\tau_y}{d_y}.
	\label{eq:var.sigma.t3}
\end{align} 
Notably, 
\begin{align}
	\Pb \left(\Gc_{\noise, y} \cap \Gc_{\tcross, y}^c \mid \Ec \right) 
	&= \Ex\left[ \mathds{1}\{\Gc_{\noise, y}\} \cdot \Pb \left(\Gc_{\tcross, y}^c \mid \Hc \right) \mid \Ec \right]
	\lesssim (N_1 T_0)^{-10}. 
\end{align}

\noindent {\em Putting everything together}: 
Now, define the master event as $\Gc_{\star, y} \coloneqq \Gc_{\noise, y} \cap \Gc_{\proj, y} \cap \Gc_{\tcross, y}$. 
We can bound the probability of $\Gc_{\star, y}$ by taking a union bound  
\begin{align}
	\Pb \left(\Gc_{\star, y}^c \mid \Ec \right)
	&\le \Pb\left(\Gc_{\noise,y}^c \mid \Ec \right) + \Pb \left(\Gc_{\proj, y}^c \mid \Ec \right) + \Pb \left(\Gc_{\noise, y} \cap \Gc_{\tcross, y}^c \mid \Ec \right) \lesssim (N_1 T_0)^{-10}. 
\end{align}
On $\Gc_{\star, y}$, \eqref{eq:var.sigma.t1}, \eqref{eq:var.sigma.t2}, and \eqref{eq:var.sigma.t3} yield 
\begin{align}
	\left| \frac{\hsigma_{\beta}^2}{\sigma^2} - 1 \right|
	\lesssim 
	\frac{ \eta_y^2 T_0}{\sigma^2 \lambda^2_y d_y}
	+ \frac{\sqrt{\tau_y}}{\sqrt{d_y}} + \frac{\tau_y}{d_y} \eqqcolon \epsilon_{\beta}
\end{align}
By analogous arguments, we conclude that with probability at least $1 - \Oc((LM)^{-10})$, 
\begin{align}
	\left| \frac{\hsigma^2_{\alpha}}{\sigma^2} - 1 \right| &\lesssim 
	\frac{\eta_z^2 M}{\sigma^2 \lambda^2_z d_z}
	+ \frac{\sqrt{\tau_z}}{\sqrt{d_z}} + \frac{\tau_z}{d_z} \eqqcolon \epsilon_{\alpha},
\end{align}
where $d_z \coloneqq M - r_z$ and $\tau_z \coloneqq 1 + \log(LM)$. 
Therefore, to analyze the pooled estimator $\hsigma^2$, we first note  
\begin{align}
	\frac{\hsigma^2}{\sigma^2} - 1
	&=
	\frac{d_y}{d_y + d_z} \cdot \left( \frac{\hsigma^2_{\beta}}{\sigma^2} - 1 \right)
	+ \frac{d_z}{d_y + d_z} \cdot \left( \frac{\hsigma^2_{\alpha}}{\sigma^2} - 1 \right).
\end{align}
Taking absolute values, 
\begin{align}
	\left| \frac{\hsigma^2}{\sigma^2} - 1 \right|
	&\le
	\left| \frac{d_y}{d_y + d_z} \cdot \left( \frac{\hsigma^2_{\beta}}{\sigma^2} - 1 \right) \right|
	+ \left| \frac{d_z}{d_y + d_z} \cdot \left( \frac{\hsigma^2_{\alpha}}{\sigma^2} - 1 \right) \right|.
\end{align}
Therefore, with probability at least $1 - \Oc((N_1 T_0)^{-10} + (LM)^{-10})$, 
\begin{align}
	\left| \frac{\hsigma^2}{\sigma^2} - 1 \right| \lesssim \frac{d_y \cdot \epsilon_{\beta} + d_z \cdot \epsilon_{\alpha}}{d_y + d_z}. 
\end{align}
Plugging \eqref{eq:lmda.lb} and \eqref{eq:noise.eta} into the above, and simplifying, completes the proof.  
\end{proof}

\subsection{Proof of Lemma~\ref{lemma:subexp.2}} \label{sec:proof.subexp.2}

\begin{proof}
Conditional on $\Hc$, $\bP$ is fixed. 
Define $d \coloneqq \tr(\bP) = \rank(\bP)$. 
Let $\bQ \in \Rb^{m \times d}$ be a matrix of orthonormal columns spanning $\col(\bP)$ such that $\bP = \bQ \bQ^\top$ and $\bQ^\top \bQ = \bI_d$. 
Hence, we can write 
\begin{align}
	\bxi^\top \bP \bxi = \| \bQ^\top \bxi \|_2^2. 
\end{align}
Since $\bxi \sim \mathcal{N}(0, \sigma^2 \bI_m)$ conditional on $\Hc$, we have that
\begin{align}
	\bg \coloneqq \frac{1}{\sigma} \bQ^\top \bxi \sim \mathcal{N}(0, \bI_d). 
\end{align}
Hence,
\begin{align}
	\bxi^\top \bP \bxi - \sigma^2 d = \sigma^2 \sum_{j=1}^d \left(g_j^2 - 1 \right),
\end{align}
where $g_j \sim \mathcal{N}(0,1)$. 
Thus, $X_j \coloneqq g_j^2 -1$ is a mean-zero, sub-exponential random variable with $\| X_i \|_{\psi_1} \le C$ for some constant $C > 0$. 
By Lemma~\ref{lemma:bernstein}, we have
\begin{align}
	\Pb\left(  \left| \sum_{j=1}^d X_j  \right| > C \sigma^2 \left\{\sqrt{\tr(\bP)t} + t \right\} \mid \Hc \right) \le 2 \exp(-ct). 
\end{align}
This completes the proof. 
\end{proof}

\section{Proof of Proposition~\ref{prop:var.asymp}} \label{sec:proof.var.est.1} 
Under our notation, we write the Riesz representer estimates as 
\begin{align}
	\hbq_{\alpha} \coloneqq \hbZ^\dagger  \bW^\top  \hbbeta, 
	\quad
	\hbq_{\beta} &\coloneqq \hbY^\dagger  \bW  \hbalpha,
\end{align}
and their corresponding estimation errors as
\begin{align} 
	\bDelta_{q_{\alpha}} &\coloneqq \hbq_{\alpha} - \bq^*_{\alpha},
	\quad 
	\bDelta_{q_{\beta}} \coloneqq \hbq_{\beta} - \bq^*_{\beta},
	\label{eq:delta.riesz}
\end{align} 
where $\bq^*_{\alpha}$ and $\bq^*_{\beta}$ are defined as in \eqref{eq:riesz.proofs}. 
The following bounds these recovery errors. 
\begin{lemma} \label{lemma:riesz}
Let the setup of Corollary~\ref{cor:identification} hold. 
Further, let Assumption~\ref{assump:bounded} hold, as well as $k_y = r_y$ and $k_z = r_z$. 
Then, conditional on $\Ec$, 
\begin{align}
	\| \bDelta_{q_{\beta}} \|_2 
	&\lesssim 
	\frac{ \sqrt{N_1} \cdot \| \bXi_y \|_{\emph{op}}}{\left( \lambda_y - \| \bXi_y \|_{\emph{op}} \right)^2}
	+
	\frac{ \sqrt{M} \cdot \| \bXi_w \|_{\emph{op}}}{\lambda_z \left(\lambda_y - \| \bXi_y \|_{\emph{op}}\right)}
	+ 
	\left( \frac{ \| \bbW \|_{\emph{op}} +  \| \bXi_w \|_{\emph{op}}}{\lambda_y - \| \bXi_y \|_{\emph{op}}} \right) \cdot \| \bDelta_{\alpha} \|_2,
	\\
	\| \bDelta_{q_{\alpha}} \|_2 
	&\lesssim 
	\frac{ \sqrt{L} \cdot \| \bXi_z \|_{\emph{op}}}{\left( \lambda_z - \| \bXi_z \|_{\emph{op}} \right)^2}
	+
	\frac{ \sqrt{T_0} \cdot \| \bXi_w \|_{\emph{op}}}{\lambda_y \left(\lambda_z - \| \bXi_z \|_{\emph{op}}\right)}
	+ 
	\left( \frac{ \| \bbW \|_{\emph{op}} +  \| \bXi_w \|_{\emph{op}}}{\lambda_z - \| \bXi_z \|_{\emph{op}}} \right) \cdot \| \bDelta_{\beta} \|_2,
\end{align}
provided $\lambda_y > \| \bXi_y \|_{\emph{op}}$ and $\lambda_z > \| \bXi_z \|_{\emph{op}}$. 
\end{lemma}

\subsection{Completing Proof of Proposition~\ref{prop:var.asymp}} 

\begin{proof}
Condition on $\Ec$. 
Recall $\Upsilon \coloneqq \Vc^2/ \sigma^2$ and $\hUpsilon \coloneqq \hVc^2 / \hsigma^2$. 
Observe that
\begin{align}
	\hUpsilon - \Upsilon 
	&= \left\{ \| \hbalpha \|_2^2  \cdot \| \hbbeta \|_2^2 - \| \balpha^* \|_2^2 \cdot \| \bbeta^* \|_2^2 \right\}
	 + \left\{ \| \hbq_{\beta} \|_2^2 \cdot  \left( 1 + \| \hbbeta \|_2^2 \right)  - \| \bq_{\beta}^* \|_2^2  \cdot\left(1 + \| \bbeta^* \|_2^2 \right) \right\}
	\\ &\quad+ \left\{ \| \hbq_{\alpha} \|_2^2 \cdot \left( 1 + \| \hbalpha \|_2^2 \right)  - \| \bq_{\alpha}^* \|_2^2 \cdot \left(1 + \| \balpha^* \|_2^2 \right) \right\}. 
\end{align}
To analyze the above, we first record a useful elementary inequality: for any vectors $\hbx$ and $\bx$, 
\begin{align}
	\left| \| \hbx \|_2^2 - \| \bx \|_2^2 \right| &\le \| \hbx - \bx \|_2  \cdot \left( 2  \| \bx \| + \| \hbx - \bx \|_2 \right).
	\label{eq:alg.ineq} 
\end{align}
We proceed to bound each term separately. 

\bigskip \noindent {\em Term 1}: 
Define $\Gc_\noise$ and $\Gc_\PCR$ as in \eqref{eq:event.noise} and \eqref{eq:event.pcr}, respectively. 
Leveraging \eqref{eq:alg.ineq}, we obtain 
\begin{align}
	&\left| \| \hbalpha \|_2^2 \cdot \| \hbbeta \|_2^2 - \| \balpha^* \|_2^2 \cdot \| \bbeta^* \|_2^2  \right|
	\le 
	\left| \| \hbalpha \|_2^2 - \| \balpha^* \|_2^2 \right| \cdot \| \hbbeta \|_2^2 + \left| \| \hbbeta \|_2^2 - \| \bbeta^* \|_2^2 \right| \cdot \| \balpha^* \|_2^2
	\\
	&\qquad \le \left| \| \hbalpha \|_2^2 - \| \balpha^* \|_2^2 \right| \cdot \left(\| \bbeta^* \|_2 + \| \bDelta_{\beta} \|_2 \right)^2 + \left| \| \hbbeta \|_2^2 - \| \bbeta^* \|_2^2 \right| \cdot \| \balpha^* \|_2^2
	\\
	&\qquad \le  \| \bDelta_{\alpha} \|_2 \cdot  \left(2  \| \balpha^* \|_2 + \| \bDelta_{\alpha} \|_2 \right) \cdot \left(\| \bbeta^* \|_2 + \| \bDelta_{\beta} \|_2 \right)^2
	+ \| \bDelta_{\beta} \|_2 \cdot  \left( 2  \| \bbeta^* \|_2 + \| \bDelta_{\beta} \|_2 \right) \cdot \| \balpha^* \|_2^2. 
\end{align}
On the event $\Gc_\PCR$, \eqref{eq:alpha.beta.bound} gives 
\begin{align}
	&\left| \| \hbalpha \|_2^2 \cdot  \| \hbbeta \|_2^2 - \| \balpha^* \|_2^2 \cdot \| \bbeta^* \|_2^2  \right|
	\\
	&\qquad\lesssim \left( \frac{\sqrt{T_0}}{\lambda_y} + \frac{\sigma \Lambda_{\beta}}{\sqrt{N_1}} \right)^2 
	\cdot \left(\frac{\sigma \sqrt{M} \Lambda_{\alpha}}{\sqrt{L} \lambda_z} + \frac{\sigma^2 \Lambda_{\alpha}^2}{L} \right)
	+ \frac{\sigma M \Lambda_{\beta}}{\sqrt{N_1} \lambda_z^2} \cdot \left( \frac{\sqrt{T_0}}{\lambda_y} + \frac{\sigma \Lambda_{\beta}}{\sqrt{N_1}} \right). 
	\label{eq:var.est.t1}
\end{align}

\noindent {\em Term 2}: 
By a similar argument, applying \eqref{eq:alg.ineq} yields 
\begin{align}
	&\left| \| \hbq_{\beta} \|_2^2 \cdot \left( 1 + \| \hbbeta \|_2^2 \right)  - \| \bq_{\beta}^* \|_2^2 \cdot \left(1 + \| \bbeta^* \|_2^2 \right) \right|
	\\
	&\qquad \le \| \bDelta_{q_{\beta}} \|_2  \cdot \left( 2  \| \bq^*_{\beta} \|_2 + \| \bDelta_{q_{\beta}} \|_2 \right) \cdot \left\{ 1 + \left( \| \bbeta^* \|_2 + \| \bDelta_{\beta} \|_2 \right)^2 \right \} + \| \bq^*_{\beta} \|^2_2 \cdot  \| \bDelta_{\beta} \|_2 \cdot \left( 2  \| \bbeta^* \|_2 + \| \bDelta_{\beta} \|_2 \right). 
\end{align}
On $\Gc_\noise \cap \Gc_\PCR$, Lemma~\ref{lemma:riesz} gives 
\begin{align}
	\| \bDelta_{q_{\beta}} \|_2 \lesssim 
	\frac{\eta_y \sqrt{N_1}}{\lambda_y^2} + \frac{\eta_w \sqrt{M}}{\lambda_y \lambda_z} + \frac{\sigma \left( \| \bbW \|_\txtop + \eta_w \right) \Lambda_{\alpha}}{\lambda_y \sqrt{L}}
	\eqqcolon Q_{\beta}. \label{eq:delta.b}
\end{align}
%
Thus, 
\begin{align}
	&\left| \| \hbq_{\beta} \|_2^2 \cdot \left( 1 + \| \hbbeta \|_2^2 \right)  - \| \bq_{\beta}^* \|_2^2 \cdot \left(1 + \| \bbeta^* \|_2^2 \right) \right|
	\\
	&\qquad \lesssim \left( \frac{\sqrt{N_1} Q_{\beta}}{\lambda_y} + Q^2_{\beta} \right) \cdot \left\{ 1 + \left( \frac{\sqrt{T_0}}{\lambda_y} + \frac{\sigma \Lambda_{\beta}}{\sqrt{N_1}} \right)^2 \right\}
	+ \frac{\sigma \sqrt{N_1} \Lambda_{\beta}}{\lambda_y^2} \cdot \left( \frac{\sqrt{T_0}}{\lambda_y} + \frac{\sigma \Lambda_{\beta}}{\sqrt{N_1}} \right). 
	\label{eq:var.est.t2}
\end{align}

\noindent {\em Term 3}: 
Analogous to \eqref{eq:var.est.t2}, we have that on $\Gc_\noise \cap \Gc_\PCR$,
\begin{align}
	&\| \hbq_{\alpha} \|_2^2 \cdot \left( 1 + \| \hbalpha \|_2^2 \right)  - \| \bq_{\alpha}^* \|_2^2 \cdot \left(1 + \| \balpha^* \|_2^2 \right)
	\\
	&\qquad \lesssim \left( \frac{\sqrt{L} Q_{\alpha}}{\lambda_z} + Q^2_{\alpha} \right) \cdot \left\{ 1 + \left( \frac{\sqrt{M}}{\lambda_z} + \frac{\sigma\Lambda_{\alpha}}{\sqrt{L}} \right)^2 \right\}
	+ \frac{\sigma \sqrt{L} \Lambda_{\alpha}}{\lambda_z^2} \cdot \left( \frac{\sqrt{M}}{\lambda_z} + \frac{\sigma \Lambda_{\alpha}}{\sqrt{L}} \right),
	\label{eq:var.est.t3}
\end{align}
where
\begin{align}
	Q_{\alpha} \coloneqq \frac{\eta_z \sqrt{L}}{\lambda_z^2} + \frac{\eta_w \sqrt{T_0}}{\lambda_y \lambda_z} + \frac{\sigma \left( \| \bbW \|_\txtop + \eta_w \right) \Lambda_{\beta}}{\lambda_z \sqrt{N_1}}. 
	\label{eq:delta.a}
\end{align}

\noindent {\em Putting everything together}: 
On $\Gc_\noise \cap \Gc_\PCR$, \eqref{eq:var.est.t1}, \eqref{eq:var.est.t2}, and \eqref{eq:var.est.t3} imply 
\begin{align}
	\left| \hUpsilon - \Upsilon \right|
	&\lesssim 
	\left( \frac{\sqrt{T_0}}{\lambda_y} + \frac{\sigma \Lambda_{\beta}}{\sqrt{N_1}} \right)^2 
	\cdot \left(\frac{\sigma \sqrt{M} \Lambda_{\alpha}}{\sqrt{L} \lambda_z} + \frac{\sigma^2 \Lambda_{\alpha}^2}{L} \right)
	+ \frac{\sigma M \Lambda_{\beta}}{\sqrt{N_1} \lambda_z^2} \cdot \left( \frac{\sqrt{T_0}}{\lambda_y} + \frac{\sigma \Lambda_{\beta}}{\sqrt{N_1}} \right)\
	\\
	&\quad + \left( \frac{\sqrt{N_1} Q_{\beta}}{\lambda_y} + Q^2_{\beta} \right) \cdot \left\{ 1 + \left( \frac{\sqrt{T_0}}{\lambda_y} + \frac{\sigma \Lambda_{\beta}}{\sqrt{N_1}} \right)^2 \right\}
	+ \frac{\sigma \sqrt{N_1} \Lambda_{\beta}}{\lambda_y^2} \cdot \left( \frac{\sqrt{T_0}}{\lambda_y} + \frac{\sigma\Lambda_{\beta}}{\sqrt{N_1}} \right)
	\\
	&\quad + \left( \frac{\sqrt{L} Q_{\alpha}}{\lambda_z} + Q^2_{\alpha} \right) \cdot \left\{ 1 + \left( \frac{\sqrt{M}}{\lambda_z} + \frac{\sigma\Lambda_{\alpha}}{\sqrt{L}} \right)^2 \right\}
	+ \frac{\sigma \sqrt{L} \Lambda_{\alpha}}{\lambda_z^2} \cdot \left( \frac{\sqrt{M}}{\lambda_z} + \frac{\sigma \Lambda_{\alpha}}{\sqrt{L}} \right). 
\end{align}
Plugging \eqref{eq:lmda.lb} and \eqref{eq:noise.eta} into the above, and simplifying, gives our desired result in \eqref{eq:asymp.var.ub}. 

It remains to bound the probability of the joint event $\Gc_\noise \cap \Gc_\PCR$. 
%
%
By Theorem~\ref{prop:parameter.recovery}, note that $\Pb(\Gc_\PCR^c \mid \Ec) \lesssim \rho$. 
Moreover, by Lemma~\ref{lemma:subg_matrix} and the union bound, $\Pb(\Gc_\noise^c \mid \Ec) \lesssim \rho$. 
As a result, 
\begin{align}
	\Pb\left((\Gc_\noise \cap \Gc_\PCR)^c \mid \Ec \right) &\le \Pb\left(\Gc_\noise^c \mid \Ec\right) + \Pb\left( \Gc_\PCR^c \mid \Ec \right) \lesssim \rho. 
\end{align}
This establishes \eqref{eq:asymp.var.ub}. 
To complete the proof, note that since, by assumption, $\hsigma^2 /\sigma^2  \xrightarrow{p} 1$ and $\Gamma / \Upsilon = o(1)$, 
\begin{align}
	\frac{\hVc^2}{\Vc^2}
	&= \frac{\hsigma^2 }{\sigma^2} \cdot \frac{\hUpsilon}{\Upsilon} \xrightarrow{p} 1. 
\end{align}
Because the square root map is continuous on $(0, \infty)$, we have $\hVc / \Vc \xrightarrow{p} 1$.
%
%
\end{proof} 


\subsection{Proof of Lemma~\ref{lemma:riesz}}

\begin{proof}
Without loss of generality, we focus on the $\bDelta_{q_{\beta}}$ bound as the $\bDelta_{q_{\alpha}}$ bound is analogous. 
Condition on $\Ec$. 
Observe that
\begin{align}
	\bDelta_{q_{\beta}}
	&= (\hbY^\dagger - \bbY^\dagger) \bbW \balpha^*
	+ \hbY^\dagger (\bW - \bbW ) \balpha^*
	+ \hbY^\dagger \bW \bDelta_{\alpha}. 
\end{align}
Taking norms, we have
\begin{align}
	\| \bDelta_{q_{\beta}} \|_2
	&\le 
	\| (\hbY^\dagger - \bbY^\dagger) \bbW \balpha^* \|_2
	+ \| \hbY^\dagger (\bW - \bbW ) \balpha^* \|_2
	+ \| \hbY^\dagger \bW \bDelta_{\alpha} \|_2. \label{eq:delta.beta.general}
\end{align}
We proceed to bound each term separately. 

\bigskip \noindent
{\em Term 1: $\| (\hbY^\dagger - \bbY^\dagger) \bbW \balpha^* \|_2$.} 
Firstly, we have
\begin{align}
	\| (\hbY^\dagger - \bbY^\dagger) \bbW \balpha^* \|_2
	&\le \| \hbY^\dagger - \bbY^\dagger \|_{\text{op}} \cdot \| \bbW \balpha^* \|_2. \label{eq:qb.term.temp}
\end{align} 
By \eqref{eq:barwalpha.bound}, $\| \bbW \balpha^* \|_2 \le \sqrt{N_1}$. 
To bound $\| \hbY^\dagger - \bbY^\dagger \|_{\text{op}}$, we recall a pseudoinverse perturbation decomposition. 


By Lemma~\ref{lemma:pseudoinverse.stewart}, we can write 
\begin{align}
	\hbY^\dagger - \bbY^\dagger = \hbY^\dagger \left(\hbY - \bbY \right) \bbY^\dagger
	+ \hbY^\dagger \hbY^{\dagger, \top} \left(\hbY - \bbY \right)^\top \left(\bI_{N_1} - \bbY \bbY^\dagger\right)
	+ \left(\bI_{N_1} - \hbY^\dagger \hbY \right) \left(\hbY - \bbY \right)^\top \bbY^{\dagger, \top} \bbY^\dagger. 
\end{align}
Taking norms and noting that projection matrices have operator norms bounded by $1$, 
\begin{align}
	\| \hbY^\dagger - \bbY^\dagger \|_{\text{op}}
	&\le \| \hbY^\dagger \|_{\text{op}} \cdot \| \hbY - \bbY \|_{\text{op}} \cdot \| \bbY^\dagger \|_{\text{op}}
	+ \| \hbY^\dagger \|_{\text{op}}^2 \cdot \| \hbY - \bbY \|_{\text{op}}
	+ \| \bbY^\dagger \|_{\text{op}}^2 \cdot \| \hbY - \bbY \|_{\text{op}}. 
	\label{eq:pseudo.y.0}
\end{align}
Observe that 
\begin{align}
	\| \bbY^\dagger \|_{\text{op}} = \frac{1}{\lambda_y},
	\quad
	\| \hbY^\dagger \|_{\text{op}} \le \frac{2}{\lambda_y - \| \bXi_y \|_{\text{op}}},
	\label{eq:pseudo.y.1}
\end{align}
where the second inequality follows from Lemma~\ref{lemma:pcr.pseudo}. 
Additionally, by Lemma~\ref{lemma:pcr.peturb.1},
\begin{align}
	\| \hbY - \bbY \|_{\text{op}} \le 2 \cdot \| \bXi_y \|_{\text{op}}. 
	\label{eq:pseudo.y.2}
\end{align}
Plugging \eqref{eq:pseudo.y.1} and \eqref{eq:pseudo.y.2} into \eqref{eq:pseudo.y.0} yields 
\begin{align}
	\| \hbY^\dagger - \bbY^\dagger \|_{\text{op}}
	&\le  \frac{4 \cdot \| \bXi_y \|_{\text{op}}}{\lambda_y (\lambda_y - \| \bXi_y \|_{\text{op}})}
	+ \frac{8 \cdot \| \bXi_y \|_{\text{op}}}{( \lambda_y - \| \bXi_y \|_{\text{op}})^2}
	+ \frac{2 \cdot \| \bXi_y \|_{\text{op}}}{\lambda_y^2}. 
	\label{eq:pseudo.y.full}
\end{align}
Applying \eqref{eq:pseudo.y.0} to \eqref{eq:qb.term.temp}, 
\begin{align}
	\| (\hbY^\dagger - \bbY^\dagger) \bbW \balpha^* \|_2
	&\le 
	\sqrt{N_1} \cdot \left\{
	\frac{4 \cdot \| \bXi_y \|_{\text{op}}}{\lambda_y (\lambda_y - \| \bXi_y \|_{\text{op}})}
	+ \frac{8 \cdot \| \bXi_y \|_{\text{op}}}{( \lambda_y - \| \bXi_y \|_{\text{op}})^2}
	+ \frac{2 \cdot \| \bXi_y \|_{\text{op}}}{\lambda_y^2}
	\right\}. 
	\label{eq:qb.term.1}
\end{align}

\noindent
{\em Term 2: $\| \hbY^\dagger (\bW - \bbW ) \balpha^* \|_2$.} 
Next, we note that
\begin{align}
	\| \hbY^\dagger (\bW - \bbW ) \balpha^* \|_2
	&\le \| \hbY^\dagger \|_{\text{op}} \cdot \| \bW - \bbW \|_{\text{op}} \cdot \| \balpha^* \|_2
	\\
	&\le \left( \frac{2}{\lambda_y - \| \bXi_y \|_{\text{op}}} \right) \cdot \| \bW - \bbW \|_{\text{op}} \cdot \| \balpha^* \|_2
	&& \because \text{Lemma~\ref{lemma:pcr.pseudo}}
	\\
	&\le \left( \frac{2 \cdot \| \bXi_w \|_{\text{op}}}{\lambda_y - \| \bXi_y \|_{\text{op}}} \right) \cdot  \| \balpha^* \|_2
	&& 
	\\
	&\le \frac{2 \sqrt{M} \cdot \| \bXi_w \|_{\text{op}}}{\lambda_z \left(\lambda_y - \| \bXi_y \|_{\text{op}}\right)}. &&\because \text{\eqref{eq:alpha.beta.bound}}
	\label{eq:qb.term.2}
\end{align} 

\noindent
{\em Term 3: $\| \hbY^\dagger \bW \bDelta_{\alpha} \|_2$.} 
Taking norms, 
\begin{align}
	\| \bW \|_{\text{op}} \le \| \bbW \|_{\text{op}} + \| \bW - \bbW \|_{\text{op}}
	\le \| \bbW \|_{\text{op}} + \| \bXi_w \|_{\text{op}}. \label{eq:w.bound}
\end{align}
Accordingly, we have 
\begin{align}
	\| \hbY^\dagger \bW \bDelta_{\alpha} \|_2
	&\le \| \hbY^\dagger \|_{\text{op}} \cdot \| \bW \|_{\text{op}} \cdot \| \bDelta_{\alpha} \|_2
	\\
	&\le \left( \frac{2}{\lambda_y - \| \bXi_y \|_{\text{op}}} \right) \cdot \| \bW \|_{\text{op}} \cdot \| \bDelta_{\alpha} \|_2
	&& \because \text{Lemma~\ref{lemma:pcr.pseudo}}
	\\
	&\le \left( \frac{2\cdot \| \bbW \|_{\text{op}} + 2 \cdot \| \bXi_w \|_{\text{op}}}{\lambda_y - \| \bXi_y \|_{\text{op}}} \right) \cdot \| \bDelta_{\alpha} \|_2. 
	\label{eq:qb.term.3}
	&&\because \text{\eqref{eq:w.bound}}
\end{align} 

\bigskip \noindent
{\em Putting everything together.}
Inserting \eqref{eq:qb.term.1}, \eqref{eq:qb.term.2}, and \eqref{eq:qb.term.3} into \eqref{eq:delta.beta.general}, we conclude
\begin{align}
	\| \bDelta_{q_{\beta}} \|_2 
	&\le 
	\sqrt{N_1} \cdot \left\{
	\frac{4 \cdot \| \bXi_y \|_{\text{op}}}{\lambda_y (\lambda_y - \| \bXi_y \|_{\text{op}})}
	+ \frac{8 \cdot \| \bXi_y \|_{\text{op}}}{( \lambda_y - \| \bXi_y \|_{\text{op}})^2}
	+ \frac{2 \cdot \| \bXi_y \|_{\text{op}}}{\lambda_y^2}
	\right\}
	\\
	&\qquad +
	\frac{2 \sqrt{M} \cdot \| \bXi_w \|_{\text{op}}}{\lambda_z \left(\lambda_y - \| \bXi_y \|_{\text{op}}\right)}
	+ 
	\left( \frac{2\cdot \| \bbW \|_{\text{op}} + 2 \cdot \| \bXi_w \|_{\text{op}}}{\lambda_y - \| \bXi_y \|_{\text{op}}} \right) \cdot \| \bDelta_{\alpha} \|_2,
\end{align}
provided $\lambda_y > \| \bXi_y \|_{\text{op}}$. 
The proof for $\bDelta_{\alpha}$ is identical, yielding
\begin{align}
	\| \bDelta_{q_{\alpha}} \|_2 
	&\le 
	\sqrt{L} \cdot \left\{
	\frac{4 \cdot \| \bXi_z \|_{\text{op}}}{\lambda_z (\lambda_z - \| \bXi_z \|_{\text{op}})}
	+ \frac{8 \cdot \| \bXi_z \|_{\text{op}}}{( \lambda_z - \| \bXi_z \|_{\text{op}})^2}
	+ \frac{2 \cdot \| \bXi_z \|_{\text{op}}}{\lambda_z^2}
	\right\}
	\\
	&\qquad +
	\frac{2 \sqrt{T_0} \cdot \| \bXi_w \|_{\text{op}}}{\lambda_y \left(\lambda_z - \| \bXi_z \|_{\text{op}}\right)}
	+ 
	\left( \frac{2\cdot \| \bbW \|_{\text{op}} + 2 \cdot \| \bXi_w \|_{\text{op}}}{\lambda_z - \| \bXi_z \|_{\text{op}}} \right) \cdot \| \bDelta_{\beta} \|_2,
\end{align}
provided $\lambda_z > \| \bXi_z \|_{\text{op}}$.
Simplifying the bounds above completes the proof.
\end{proof}


\section{Proof of Theorem~\ref{thm:direct.normality}} \label{sec:proof.normality.direct}
Throughout this section, all assumptions are understood as conditional assumptions given $\Ec_{\eta}$. 
The proof of Theorem~\ref{thm:direct.normality} reduces the direct multi-step estimator to the one-step analysis. 
Lemma~\ref{lemma:dir.alpha} first shows that, for each $\ell \in \Sc_\eta$, the corresponding training response and target conditional mean admit the same $L$-lag population representation. 
Because all active leads share the same temporal design matrix, their weighted responses and coefficients aggregate into a single regression, to which the arguments of Theorem~\ref{thm:inference} and Proposition~\ref{prop:var.asymp} apply after normalization. 
The proof of Lemma~\ref{lemma:dir.alpha} is given in Appendix~\ref{sec:proof.dir.identification}.

\begin{lemma} \label{lemma:dir.alpha}
Let Assumptions~\ref{assump:lfm}, \ref{assump:mean_ind}, and \ref{assump:hankel} hold after conditioning on $\Ec_{\eta}$. Fix a lag length $L \in \Zb_+$ satisfying $L \ge rQ$. 
Then, for each $\ell \in \Sc_\eta$, there exists a coefficients vector $\balpha^\edirect_\ell \in \Rb^{L}$ such that 
\begin{enumerate} [label=(\alph*)]
	\item $\Ex[Y_{j, T+\ell}(1) \mid \Ec_{\eta}] = \sum_{a =1}^{L} \alpha^\edirect_{\ell, a} \cdot \Ex[Y_{j, T - L + a} \mid \Ec_{\eta}]$ for all $j \in \Ic_1$, 
	
	\item $\Ex[Y_{j, T_0 + (b-1)(L+h_\eta) + L + \ell} \mid \Ec_{\eta}] = \sum_{a = 1}^{L} \alpha^\edirect_{\ell, a} \cdot \Ex[ Y_{j, T_0 + (b-1)(L+h_\eta) + a} \mid \Ec_{\eta}]$ for all $j \in \Ic_1$ and $b \in [B_{\eta}-1]$.  
\end{enumerate}
\end{lemma}
%

\subsection{Completing the Proof of Theorem~\ref{thm:direct.normality}} 
\begin{proof}
Condition on $\Ec_{\eta}$, and write $\bbX = \Ex[\bX \mid \Ec_{\eta}]$ for any random object $\bX$. 
Further suppress the superscript ``\texttt{dir}.'' 
Set $\bZ \coloneqq \bZ_{\direct\text{-}\lag}$, $\bz_{\ell} \coloneqq \bz_{\direct\text{-}\tnext, \ell}$ for each $\ell \in \Sc_\eta$, and let $\hbZ$ be the rank-$r^{\direct}_z$ truncation of $\bZ$.  

For each $\ell \in \Sc_\eta$, Lemma~\ref{lemma:dir.alpha} supplies a population coefficient $\balpha_\ell \in \Rb^L$ satisfying $\bbz_\ell = \bbZ^\top \balpha_\ell$ and $\bbW \balpha_\ell$. 
Define the aggregate response and population coefficient by 
\begin{align}
	\bz_{\eta} \coloneqq \sum_{\ell \in \Sc_\eta} \eta_\ell \bz_\ell,
	\quad
	\balpha_{\eta} \coloneqq \sum_{\ell \in \Sc_\eta} \eta_\ell \balpha_\ell. 
\end{align}
Linearity gives $\bbz_{\eta} = \bbZ^\top \balpha_{\eta}$. 
Because all active leads use the same truncated design matrix, the \PCR~estimator is also linear in the response: 
\begin{align}
	\hbalpha_{\eta} \coloneqq \sum_{\ell \in \Sc_\eta} \eta_\ell \hbalpha_\ell 
	= \hbZ^{\top, \dagger} \bz_{\eta}. 
\end{align}
Consequently, we have
\begin{align}
	\htheta_{\eta} = \left\langle \hbalpha_{\eta}, \bW^\top \hbbeta \right \rangle,
\end{align}
so the direct estimator is itself a single bilinear \TWSF~estimator based on the aggregate response $\bz_{\eta}$. 

Let $\balpha^*_{\eta} \coloneqq \bbZ^{\top, \dagger} \bbZ^\top \balpha_{\eta}$, and let $\bbeta^*$ be the recoverable coefficients from Corollary~\ref{cor:identification}. 
The direct version of Assumption~\ref{assump:transport} yields
\begin{align}
	\bbW \balpha^*_{\eta} = \bbW \balpha_{\eta},
	\quad
	\bbW^\top \bbeta^* = \bbW^\top \bbeta. 
\end{align}
Lemma~\ref{lemma:dir.alpha} and Proposition~\ref{prop:beta} therefore imply
\begin{align}
	\theta_{\eta} 
	= \left\langle \balpha_{\eta}, \bbW^\top \bbeta \right \rangle
	= \left\langle \balpha^*_{\eta}, \bbW^\top \bbeta^* \right \rangle. 
\end{align}
Define the corresponding minimum-norm Riesz representers by
\begin{align}
	\bq^*_{\beta, \eta} \coloneqq \bbY^\dagger \bbW \balpha^*_{\eta},
	\quad
	\bq^*_{\alpha} \coloneqq \bbZ^\dagger \bbW^\top \bbeta^*. 
\end{align}
By the direct version of Assumption~\ref{assump:transport}, they satisfy
\begin{align}
	\bbY \bq^*_{\beta, \eta} = \bbW \balpha^*_{\eta},
	\quad
	\bbZ \bq^*_{\alpha} = \bbW^\top \bbeta^*,
\end{align}

We now apply the proof of Theorem~\ref{thm:inference} to the direct regression with response $\bz_{\eta} / \| \boldeta \|_1$ and population coefficient $\balpha_{\eta} / \| \boldeta \|_1$. 
The corresponding recoverable coefficient and unit-side representer are $\balpha^*_{\eta} / \| \boldeta \|_1$ and $\bq^*_{\beta, \eta}/\| \boldeta \|_1$, whereas $\bq^*_{\alpha}$ is unchanged. 
Assumption~\ref{assump:bounded} implies, coordinatewise, $\bz_{\eta} / \| \boldeta \|_1$ and $\bbW \balpha^*_{\eta} / \| \boldeta \|_1$ are bounded in absolute value by one. 
Thus, the bounded-signal conditions required by the one-step analysis continue to hold.

Let $\bxi_\ell \coloneqq \bz_\ell - \bbz_\ell$. 
The direct Page construction places the lag rows and the response rows in disjoint unit-time cells. 
Hence, Assumption~\ref{assump:subg} implies that the vectors $\{\bxi_\ell: \ell \in \Sc_\eta \}$ are mutually independent and independent of the lag-design noise. 
Under the Gaussian specialization,
\begin{align}
	\frac{1}{\| \boldeta \|_1} \sum_{\ell \in \Sc_\eta} \eta_\ell \bxi_\ell 
	\sim 
	\Nc\left(0, \sigma^2 \frac{\| \boldeta \|_2^2}{\| \boldeta \|_1^2} \bI_{M_{\eta}} \right). 
\end{align}
Its coordinates also have subgaussian norm bounded by a constant multiple of $\sigma$, since $\| \boldeta \|_2 / \| \boldeta \|_1 \le 1$. 
Therefore, every concentration and remainder estimate in the proof of Theorem~\ref{thm:inference} remains unchanged. 
The only modification to the exact Gaussian variance calculation is that the factor $1$ contributed by the time-side response innovation is replaced by $\| \boldeta \|_2^2 / \| \boldeta \|_1^2$. 
The same replacement in the proof of Proposition~\ref{prop:var.asymp} does not enlarge its bounds because this factor is at most one. 

Multiplying the resulting normalized expansion by $\| \boldeta \|_1$ gives $\htheta_{\eta} - \theta_{\eta} = G_{\eta} + R_{\eta}$, where $G_{\eta} \sim \Nc(0, \Vc^2_{\eta})$ with 
\begin{align}
	\Vc^2_{\eta} \coloneqq \sigma^2 \cdot \left\{
		\| \balpha^*_{\eta} \|_2^2 \cdot \| \bbeta^* \|_2^2
		+ \| \bq^*_{\beta, \eta} \|_2^2 \cdot \left(1 + \| \bbeta^* \|_2^2 \right)
		+ \| \bq^*_{\alpha} \|_2^2 \cdot \left( \| \boldeta \|_2^2 + \| \balpha^*_{\eta} \|_2^2 \right)
	\right\}. 
\end{align}
Moreover, $| R_{\eta} | \lesssim \| \boldeta \|_1 \Psi_{\eta}$ with probability at least $1 - \Oc(\rho_{\eta})$. 

It remains to verify the plug-in variance bound. 
For the normalized regression, the empirical temporal coefficient and unit-side Riesz representer are $\hbalpha_{\eta} / \| \boldeta \|_1$ and $\hbq_{\beta, \eta} / \| \boldeta \|_1$, while $\hbq_{\alpha}$ is unchanged. 
Its time-side response-variance factor is $\| \boldeta \|_2^2 / \| \boldeta \|_1^2$. 
Consequently, multiplying the normalized scale-free plug-invariance by $\| \boldeta \|_1^2$ gives exactly the expression in \eqref{eq:direct.var.est}. 
Repeating the proof of Proposition~\ref{prop:var.asymp} with the response-variance factor above therefore gives 
\begin{align}
	\left| \frac{\hVc^2_{\eta}}{\hsigma^2} - \frac{\Vc^2_{\eta}}{\sigma^2} \right| \lesssim \| \boldeta \|_1^2 \Gamma_{\eta} 
\end{align}
with probability at least $1 - \Oc(\rho_{\eta})$. 

Finally, if $\| \boldeta \|_1 \Psi_{\eta} / \Vc_{\eta} = o(1)$, then $R_{\eta} / \Vc_{\eta} \xrightarrow{p} 0$. 
Since $G_{\eta} / \Vc_{\eta}$ is conditionally standard normal,
\begin{align}
	\frac{\htheta_{\eta} - \theta_{\eta}}{\Vc_{\eta}} = \frac{G_{\eta}}{\Vc_{\eta}} + \frac{R_{\eta}}{\Vc_{\eta}} \rightsquigarrow \Nc(0,1). 
\end{align}
Moreover, the condition $\sigma^2 \| \boldeta \|_1^2 \Gamma_{\eta} / \Vc^2_{\eta} = o(1)$, implies $(\hVc_{\eta} / \hsigma^2) / (\Vc_{\eta} / \sigma^2) \xrightarrow{p} 1$. 
Together with $\hsigma^2 / \sigma^2 \xrightarrow{p} 1$,
\begin{align}
	\frac{\hVc^2_{\eta}}{\Vc^2_{\eta}} = \frac{\hsigma^2}{\sigma^2} \cdot \frac{\hVc_{\eta} / \hsigma^2}{\Vc_{\eta} / \sigma^2} \xrightarrow{p} 1,
\end{align}
and hence, $\hVc_{\eta} / \Vc_{\eta} \xrightarrow{p} 1$. 
Conditional Slutsky then yields 
\begin{align}
	\frac{\htheta_{\eta} - \theta_{\eta}}{\hVc_{\eta}} \rightsquigarrow \Nc(0,1). 
\end{align}
Restoring the superscript ``\direct'' gives all the conclusions of Theorem~\ref{thm:direct.normality}. 

\end{proof}

\subsection{Proof of Lemma~\ref{lemma:dir.alpha}} \label{sec:proof.dir.identification} 

\begin{proof}
Condition on $\Ec_{\eta}$. 
Let $\mu_j(t) \coloneqq \langle \bu_j, \bv_t(1) \rangle$. 
Use the common $L$-lag recurrence furnished by the proof of Proposition~\ref{prop:alpha}. 
Thus, for a conditionally fixed vector $\bphi  \coloneqq [\phi_1, \dots, \phi_L]^\top \in \Rb^L$, for every $j \in \Ic_1$ and $t \in \Zb$ 
\begin{align}
	\mu_j(t+1) = \sum_{a=1}^L \phi_a \mu_j(t-L+a).
\end{align}
Define the lag state $\bx_j(t) \coloneqq \big[ \mu_j(t-L+1), \dots, \mu_j(t) \big]^\top$. 
By the definition of the companion map in \eqref{eq:pi.matrix}, $\bx_j(t+1) = \bPi(\bphi) \bx_j(t)$. 
Hence, for every $\ell \ge 1$, 
\begin{align}
	\mu_j(t+\ell) = \be_L^\top \bPi(\bphi)^\ell \bx_j(t). 
\end{align}
Set $\balpha^\direct_\ell \coloneqq \{\bPi(\bphi)^\ell\}^\top \be_L$. 
Taking $t= T$ gives
\begin{align}
	\mu_j(T+\ell) = \sum_{a=1}^L \alpha^{\direct}_{\ell, a} \cdot \mu_j(T-L+a),
\end{align}	
which gives part (a).
For a training block $b$, take $t = T_0 + (b-1)(L+h_\eta) + L$. 
Then,
\begin{align}
	\mu_j\left(T_0 + (b-1)(L+h_\eta) + L + \ell \right) = \sum_{a=1}^L \alpha^{\direct}_{\ell, a} \cdot \mu_j \left(T_0+(b-1)(L+h_\eta) + a \right),
\end{align}
which gives part (b). 
Assumptions~\ref{assump:lfm} and \ref{assump:mean_ind} convert these latent means into the corresponding conditional expectations; in the right-hand-side of part (a) and entirety of part (b), the donor outcomes are observed treated outcomes. 

\end{proof}

\section{Proof of Theorem~\ref{thm:recursive.normality}} \label{sec:proofs.recursive} 
Throughout this section, all assumptions are understood as conditional assumptions given $\Ec_{\eta}$. 
Appendix~\ref{sec:proofs.recursive.notation} introduces the notation for the recursive analysis, and Appendix~\ref{sec:recursive.lemmas} states the key intermediate results, culminating in Lemmas~\ref{lemma:rec.normality} and \ref{lemma:rec.var.est}; these are the recursive counterparts of Theorem~\ref{thm:inference} and Proposition~\ref{prop:var.asymp}. 
Appendix~\ref{sec:proofs.rec.complete} then completes the proof of Theorem~\ref{thm:recursive.normality}, while Appendix~\ref{sec:recursive.lemmas.proofs} contains the proofs of the supporting lemmas. 

\subsection{Recursive Notation} \label{sec:proofs.recursive.notation}
We carry over the notation from Appendix~\ref{sec:proofs.notation} and introduce bespoke notation for the recursive framework. 
Notably, we now define $\bbX \coloneqq \Ex[\bX \mid \Ec_\eta]$ for any random $\bX$.  
Recall $\Sc_\eta \coloneqq \{\ell \in [h]: \eta_\ell \neq 0\}$ and $h_{\eta} \coloneqq \max \Sc_\eta$. 
For any vector $\bx \in \Rb^L$ and $\ell \ge 1$, further recall $\bg_\ell(\bx) \coloneqq \{ \bPi(\bx)^\ell \}^\top \be_L$ and $\bJ_\ell(\bx) = \nabla \bg_\ell(\bx)$. 
Define the aggregate recursive map and its Jacobian by
\begin{align}
	\bg_{\eta}(\bx) \coloneqq \sum_{\ell \in \Sc_\eta} \eta_\ell \bg_\ell (\bx),
	\quad
	\bJ_{\eta}(\bx) \coloneqq \sum_{\ell \in \Sc_\eta} \eta_\ell \bJ_\ell(\bx). 
\end{align}
For all $\bu, \bx \in \Rb^L$, define the aggregate Taylor remainder by 
\begin{align}
	\bvarrho_{\eta}(\bx, \bu) \coloneqq \bg_{\eta}(\bx + \bu) - \bg_{\eta}(\bx) - \bJ_{\eta}(\bx)  \bu. 
\end{align}
For compactness, we denote 
\begin{align}
	\balpha^*_{\eta} \coloneqq \bg_{\eta}(\balpha^*),
	\quad
	\bJ^*_{\eta} \coloneqq \bJ_{\eta}(\balpha^*),
	\quad
	\bvarrho_{\eta} \coloneqq \bvarrho_{\eta}(\balpha^*, \bDelta_{\alpha}). 
\end{align}
For ease of notation, we drop all superscripts ``$\rec$.'' 
With this convention, we denote the estimates as 
\begin{align}
	\hbalpha_{\eta} &\coloneqq \hbalpha^\rec_{\eta} = \bg_{\eta}(\hbalpha), 
	\quad
	\hbJ_{\eta} \coloneqq \bJ_{\eta}(\hbalpha). 
\end{align}
We define the population aggregate minimum $\ell_2$-norm Riesz representers by 
\begin{align}
	\bq^*_{\beta \eta} \coloneqq  \bbY^\dagger \bbW \balpha^*_{\eta}, 
	\quad
	\bq^*_{\alpha \eta} \coloneqq \bbZ^\dagger \left(\bJ^*_{\eta} \right)^\top \bbW^\top \bbeta^*. \label{eq:riesz.rep.rec}
\end{align} 
Following the proof of Theorem~\ref{thm:inference}, we introduce proof-only projection representers
\begin{align}
	\tbq_{\beta \eta} \coloneqq \hbP_y \bq^*_{\beta \eta},
	\quad
	\tbq_{\alpha \eta} \coloneqq \hbP_z \bq^*_{\alpha \eta}, \label{eq:rec.riesz.proj.def}
\end{align}
where $\hbP_y$ and $\hbP_z$ are defined as in \eqref{eq:riesz.proj.def}. 
We write the empirical estimates as 
\begin{align}
	\hbq_{\beta \eta} &\coloneqq \hbq^\rec_{\beta \eta} = \hbY^\dagger \bW \hbalpha_{\eta},
	\quad
	\hbq_{\alpha \eta} \coloneqq \hbq^\rec_{\alpha \eta} = \hbZ^\dagger \hbJ_{\eta}^\top \bW^\top \hbbeta. 
\end{align}
Finally, denote the estimation errors as  
\begin{align}
	\bDelta_{\tq_{\beta \eta}} &\coloneqq \tbq_{\beta \eta} - \bq^*_{\beta \eta},
	\quad
	\bDelta_{\tq_{\alpha \eta}} \coloneqq \tbq_{\alpha \eta} - \bq^*_{\alpha \eta},
	\\
	\bDelta_{q_{\beta \eta}} &\coloneqq \hbq_{\beta \eta} - \bq^*_{\beta \eta},
	\quad
	\bDelta_{q_{\alpha \eta}} \coloneqq \hbq_{\alpha \eta} - \bq^*_{\alpha \eta}. 
\end{align}

\subsection{Key Lemmas} \label{sec:recursive.lemmas}

\noindent Lemma~\ref{lemma:jacobian} gives deterministic calculus bounds for the recursive companion-map coefficient $\bg_{\eta}$: its size, Jacobian, Jacobian stability, and second-order Taylor remainder. These bounds quantify how the first-stage estimation error in $\hbalpha$ propagates through the recursive forecast coefficients. 
\begin{lemma} \label{lemma:jacobian}
For any $\bx \in \Rb^L$, 
\begin{align}
	\| \bg_{\eta} (\bx) \|_2 &\le \| \boldeta \|_1  \left(1 + \| \bx \|_2 \right)^{h_{\eta}-1}  \| \bx \|_2, \label{eq:jacobian.0}
	\\
	\| \bJ_{\eta}(\bx) \|_{\emph{op}} &\le h_{\eta}  \| \boldeta \|_1  \left( 1 + \| \bx \|_2 \right)^{h_{\eta}-1}. \label{eq:jacobian.1}
\end{align}
If $h_{\eta} \ge 2$, then for every $\bx, \bu \in \Rb^L$, 
\begin{align}
	\| \bJ_{\eta}(\bx+\bu) - \bJ_{\eta}(\bx) \|_{\emph{op}} &\le h_{\eta}(h_{\eta}-1)  \| \boldeta \|_1  \left( 1 + \| \bx \|_2 + \| \bu \|_2 \right)^{h_{\eta}-2}  \| \bu \|_2, \label{eq:jacobian.2}
	\\
	\| \bvarrho_{\eta}(\bx, \bu) \|_2 &\le \frac{h_{\eta}(h_{\eta}-1)}{2}  \| \boldeta \|_1  \left(1 + \| \bx \|_2 + \| \bu \|_2 \right)^{h_{\eta} - 2}  \| \bu \|_2^2. \label{eq:jacobian.3} 
\end{align}
If $h_{\eta} =1$, then $\bg_{\eta}(\bx)= \eta_1 \bx$, $\bJ_{\eta} (\bx)= \eta_1 \bI$, and $\bvarrho_{\eta} (\bx, \bu)= \bzero$. 
\end{lemma} 

\noindent Lemma~\ref{lemma:rec.identification} is an analog of Corollary~\ref{cor:identification} as it establishes identification of the aggregate estimand $\theta_{\eta}$ within the recursive framework. 
\begin{lemma} \label{lemma:rec.identification}
%
%
Let the setup of Corollary~\ref{cor:identification} hold after conditioning on $\Ec_{\eta}$, with Assumption~\ref{assump:recursive} replacing Assumption~\ref{assump:transport}. 
Then, conditional on $\Ec_{\eta}$, 
\begin{align}
	\theta_{\eta} &= \left\langle \balpha^*_{\eta}, \bbW^\top \bbeta^* \right \rangle. \label{eq:identification.rec}
\end{align}
Moreover, the Riesz equations below are feasible
\begin{align}
	\bbY \bq^{*}_{\beta \eta} = \bbW \balpha^*_{\eta},
	\quad
	\bbZ \bq^{*}_{\alpha \eta} = \left(\bJ^*_{\eta} \right)^\top \bbW^\top \bbeta^*, \label{eq:riesz.rec}
\end{align}
where the minimum $\ell_2$-norm Riesz representers are defined as in \eqref{eq:riesz.rep.rec}. 
\end{lemma} 

\noindent Lemma~\ref{lemma:rec.decomp} is an analog of Lemma~\ref{lemma:decomp} in that it decomposes the pointwise error in terms of two distinct components. 

\begin{lemma} \label{lemma:rec.decomp}
Let the setup of Lemma~\ref{lemma:rec.identification} hold. 
Further, let Assumption~\ref{assump:bounded} hold, as well as $k_y = r_y$ and $k_z = r_z$. 
Define
\begin{align}
	G_{\eta} &\coloneqq \left\langle \balpha^*_{\eta}, \bXi_w^\top \bbeta^* \right\rangle + \left\langle \bq^*_{\beta \eta}, \bdelta_{\beta} \right\rangle + \left\langle \bq^*_{\alpha \eta}, \bdelta_{\alpha} \right \rangle, \label{eq:lead.rec}
	\\
	R_{\eta} &\coloneqq
	\left\langle \bJ_{\eta}^* \bDelta_{\alpha}, \bbW^\top \bDelta_{\beta} \right\rangle
	+ \left\langle \bJ_{\eta}^* \bDelta_{\alpha}, \bXi_w^\top \bbeta^* \right\rangle
	+ \left\langle \balpha^*_{\eta}, \bXi_w^\top \bDelta_{\beta} \right\rangle
	+ \left\langle \bJ_{\eta}^* \bDelta_{\alpha}, \bXi_w^\top \bDelta_{\beta} \right\rangle
	\\ &\qquad
		+ \left\langle \bvarrho_{\eta}, \bbW^\top \bbeta^* \right \rangle
		+ \left\langle \bvarrho_{\eta}, \bbW^\top \bDelta_{\beta} \right \rangle
		+ \left\langle \bvarrho_{\eta}, \bXi_w^\top \bbeta^* \right \rangle
		+ \left\langle \bvarrho_{\eta}, \bXi_w^\top \bDelta_{\beta} \right \rangle
	\\ &\qquad
	-  \left\langle \bq^{*}_{\beta \eta}, \bXi_y^\top \bDelta_{\beta} \right\rangle
	- \left\langle \bq^{*}_{\alpha \eta}, \bXi_z^\top \bDelta_{\alpha} \right\rangle
	\\ &\qquad
	+ \left\langle \bDelta_{\tq_{\beta \eta}}, \bdelta_{\beta} \right\rangle - \left\langle \bDelta_{\tq_{\beta \eta}}, \bbY^\top \bDelta_{\beta} \right\rangle - \left\langle \bDelta_{\tq_{\beta \eta}}, \bXi_y^\top \bDelta_{\beta} \right\rangle
	\\ &\qquad
	+  \left\langle \bDelta_{\tq_{\alpha \eta}}, \bdelta_{\alpha} \right\rangle - \left\langle \bDelta_{\tq_{\alpha \eta}}, \bbZ^\top \bDelta_{\alpha} \right\rangle - \left\langle \bDelta_{\tq_{\alpha \eta}}, \bXi_z^\top \bDelta_{\alpha} \right\rangle.  \label{eq:remain.rec}
\end{align}
Then, conditional on $\Ec_{\eta}$, we have $\htheta_{\eta} - \theta_{\eta} = G_{\eta} + R_{\eta}$. 
\end{lemma} 

\noindent Lemma~\ref{lemma:rec.normality} is an analog of Theorem~\ref{thm:inference} that establishes the Gaussian law of the lead term and offers a high-probability bound on the remainder term. 
\begin{lemma} \label{lemma:rec.normality}
Let the setup of Theorem~\ref{thm:inference} hold after conditioning on $\Ec_{\eta}$, with Assumption~\ref{assump:recursive} replacing Assumption~\ref{assump:transport}. 
Then, conditional on $\Ec_{\eta}$, $G_{\eta} \sim \mathcal{N}(0, \Vc^2_{\eta})$, where
\begin{align}
	\Vc_{\eta}^2 &\coloneqq \sigma^2  
	\left\{ 
	\| \balpha_{\eta}^* \|_2^2  \cdot \| \bbeta^* \|_2^2 
	+ \| \bq_{\beta \eta}^* \|_2^2  \left(1 + \| \bbeta^* \|_2^2 \right)
 	+ \| \bq_{\alpha \eta}^* \|_2^2  \left(1 + \| \balpha^* \|_2^2 \right)
	\right\};
	\label{eq:rec.asymp.var} 
\end{align}
moreover, with probability at least $1 - \Oc(\rho)$, $|R_{\eta}| \lesssim \mathfrak{C} \| \boldeta \|_1 \Psi$.
%
\end{lemma}

\noindent Finally, Lemma~\ref{lemma:rec.var.est} is an analog of Proposition~\ref{prop:var.asymp} bounding the variance estimation error. 
\begin{lemma} \label{lemma:rec.var.est}
Let the setup and conditions of Lemma~\ref{lemma:rec.normality} hold. 
Then, conditional on $\Ec_{\eta}$, with probability at least $1 - \Oc(\rho)$, 
\begin{align}
	\left| \frac{\hVc^2_{\eta}}{\hsigma^2} - \frac{\Vc^2_{\eta}}{\sigma^2} \right| &\lesssim \mathfrak{C}^2 \| \boldeta \|_1^2 \Gamma.
\end{align}
\end{lemma}

\subsection{Completing Proof of Theorem~\ref{thm:recursive.normality}} \label{sec:proofs.rec.complete}

\begin{proof}
The result is immediate from following the proof of Theorem~\ref{thm:inference} using Lemmas~\ref{lemma:rec.normality} and \ref{lemma:rec.var.est}, and  
restoring the superscript ``\rec.'' 
\end{proof}

\subsection{Proofs of Key Lemmas} \label{sec:recursive.lemmas.proofs}

\subsubsection{Proof of Lemma~\ref{lemma:jacobian}}

\begin{proof}
We prove each result separately. 

\bigskip \noindent \underline{\em Derivation of \eqref{eq:jacobian.0}}: 
Observe that only the last row of $\bPi(\bx)$ depends on $\bx^\top$. 
As a result, $\bPi(\bx)^\top \be_L = \bx$. 
Therefore, for every $\ell \ge 1$,
\begin{align}
	\bg_\ell(\bx) = \left\{\bPi(\bx)^\ell\right\}^\top \be_L = \left\{\bPi(\bx)^\top \right\}^{\ell-1} \bx. \label{eq:rec.temp.0}
\end{align}
Let $\bS \in \{0,1\}^{L \times L}$ denote a shift matrix with ones on the superdiagonal. 
Then, 
\begin{align}
	\bPi(\bx) = \bS + \be_L  \bx^\top. 
\end{align}
Because $\| \bS \|_\txtop \le 1$, it follows that
\begin{align}
	\| \bPi(\bx) \|_\txtop \le \| \bS \|_\txtop + \| \be_L \bx^\top \|_\txtop \le 1 + \| \bx \|_2. \label{eq:jac.1}
\end{align}
Combining \eqref{eq:rec.temp.0} with \eqref{eq:jac.1} yields
\begin{align}
	\|\bg_\ell(\bx)\|_2
	\le 
	\|\bPi(\bx) \|_\txtop^{\ell-1}  \cdot \| \bx \|_2
	\le \left(1 + \| \bx \|_2 \right)^{\ell-1}  \| \bx \|_2. 
\end{align}
Summing with weights $\eta_\ell$ gives 
\begin{align}
	\| \bg_{\eta}(\bx) \|_2 \le \sum_{\ell \in \Sc_\eta} | \eta_\ell | \cdot  \| \bg_\ell(\bx) \|_2
	\le \| \boldeta \|_1  \left(1 + \| \bx \|_2 \right)^{h_{\eta}-1}  \| \bx \|_2
\end{align}

\noindent \underline{\em Derivation of \eqref{eq:jacobian.1}}:
For any $\bu \in \Rb^L$, observe that 
\begin{align}
	\bPi(\bx + \bu) - \bPi(\bx) = \be_L  \bu^\top. \label{eq:jac.diff}
\end{align}
For any $t \in [0,1]$, the product rule for matrix powers gives
\begin{align}
	\frac{d}{dt} \bPi(\bx + t \bu)^\ell = \sum_{a=0}^{\ell-1} \bPi(\bx + t \bu)^a  \left\{ \bPi(\bx + \bu) - \bPi(\bx) \right\}  \bPi(\bx + t\bu)^{\ell-1-a}.
\end{align}
Accordingly, we obtain
\begin{align}
	\frac{d}{dt} \bg_\ell(\bx + t \bu) = \left( \sum_{a=0}^{\ell-1} \bPi(\bx + t \bu)^a  \left\{ \bPi(\bx + \bu) - \bPi(\bx) \right\}  \bPi(\bx + t\bu)^{\ell-1-a} \right)^\top \be_L. 
\end{align}
Evaluating at $t=0$ and applying \eqref{eq:jac.diff} yields 
\begin{align}
	\bJ_{\ell}(\bx)  \bu = \sum_{a=0}^{\ell-1} \left\{ \be_L^\top  \bPi(\bx)^a  \be_L \right\} \left\{ \bPi(\bx)^{\ell-1-a} \right\}^\top \bu.
	\label{eq:jacobian}
\end{align}
Taking norms of \eqref{eq:jacobian}, we obtain
\begin{align}
	\| \bJ_{\ell}(\bx) \|_\txtop \le \sum_{a=0}^{\ell-1} \left| \be_L^\top  \bPi(\bx)^a  \be_L \right| \cdot \| \bPi(\bx)^{\ell-1-a} \|_\txtop. 
\end{align}
Since $| \be_L^\top  \bPi(\bx)^a  \be_L| \le \| \bPi(\bx)^a \|_\txtop \le \| \bPi(\bx) \|_\txtop^a$, we further have
\begin{align}
	\| \bJ_{\ell}(\bx) \|_\txtop \le \sum_{a=0}^{\ell-1} \| \bPi(\bx) \|_\txtop^{\ell-1}
	\le \ell  \left(1 + \| \bx \|_2 \right)^{\ell-1}, 
\end{align}
where the second inequality uses \eqref{eq:jac.1}. Thus,
\begin{align}
	\| \bJ_{\eta}(\bx) \|_\txtop &\le \sum_{\ell \in \Sc_\eta} | \eta_\ell | \cdot \| \bJ_\ell(\bx)\|_\txtop
	\le h_{\eta}  \| \boldeta \|_1  \left( 1 + \| \bx \|_2 \right)^{h_{\eta}-1}. 
\end{align}

\noindent \underline{\em Derivation of \eqref{eq:jacobian.2}}:
Suppose $h_{\eta} \ge 2$. 
Let $\bA \coloneqq \bPi(\bx)$ and $\bB \coloneqq \bPi(\bx + \bu)$. 
For any integer $m \ge 1$, the telescoping identity gives 
\begin{align}
	\bB^m - \bA^m = \sum_{a=0}^{m-1} \bB^a  \left(\bB - \bA\right)  \bA^{m-1-a}. \label{eq:telescope}
\end{align}
Let $D_x(\bu) \coloneqq 1 + \| \bx \|_2 + \| \bu \|_2$. 
Then,
\begin{align}
	\| \bA \|_\txtop \le D_x(\bu),
	\quad
	\| \bB \|_\txtop \le D_x(\bu),
	\quad
	\| \bB - \bA \|_\txtop = \| \bu \|_2,
\end{align}
where the third inequality uses \eqref{eq:jac.diff}. Hence, 
\begin{align}
	\| \bB^m - \bA^m \|_\txtop \le m  D_x^{m-1}(\bu)  \| \bu \|_2. \label{eq:power}
\end{align}
Now, fix a direction $\bv \in \Rb^L$. 
By \eqref{eq:jacobian}, 
\begin{align}
	\left\{ \bJ_{\ell}(\bx+\bu) - \bJ_{\ell}(\bx) \right\}  \bv 
	&= \left[ \sum_{a=0}^{\ell-1} \left( \bB^a  \be_L  \bv^\top  \bB^{\ell-1-a} - \bA^a  \be_L  \bv^\top  \bA^{\ell-1-a} \right) \right]^\top \be_L.
\end{align}
For a fixed $a$, we add and subtract $\bA^a \be_L \bv^\top \bB^{\ell-1-a}$ to rewrite each summand as 
\begin{align}
	&\bB^a  \be_L  \bv^\top  \bB^{\ell-1-a} - \bA^a  \be_L  \bv^\top  \bA^{\ell-1-a}
	\\&\qquad = \left(\bB^a - \bA^a\right)  \be_L \bv^\top  \bB^{\ell-1-a} + \bA^a  \be_L \bv^\top  \left(\bB^{\ell-1-a} - \bA^{\ell-1-a} \right). 
\end{align}
From \eqref{eq:power}, each summand has norm at most $ (\ell-1)  D_x^{\ell-2}(\bu) \cdot \| \bu \|_2 \cdot \| \bv \|_2$. 
Therefore, summing over all $a$ and taking the supremum over $\| \bv \|_2 = 1$ yields
\begin{align}
	\| \bJ_{\ell}(\bx + \bu) - \bJ_{\ell}(\bx) \|_\txtop \le \ell \left(\ell-1\right)  D_x^{\ell-2} \cdot \| \bu \|_2. 
\end{align}
Consequently,
\begin{align}
	\| \bJ_{\eta}(\bx + \bu) - \bJ_{\eta}(\bx) \|_\txtop &\le \sum_{\ell \in \Sc_\eta} |\eta_\ell| \cdot \ell (\ell-1)  D_x^{\ell-2} \cdot \| \bu \|_2
	\\ &\le h_{\eta} (h_{\eta}-1)  D_x^{h_{\eta}-2} \cdot \| \boldeta \|_1 \cdot \| \bu \|_2. 
\end{align}

\noindent \underline{\em Derivation of \eqref{eq:jacobian.3}}:
To complete the proof, note that $\bg_{\eta}$ is continuously differentiable and thus, the fundamental theorem of calculus states 
\begin{align}
	\bg_{\eta}(\bx + \bu) - \bg_{\eta}(\bx) = \int_{0}^1 \bJ_{\eta}(\bx + t \bu)   \bu  dt. 
\end{align}
Therefore, 
\begin{align}
	\bvarrho_{\eta}(\bx, \bu) = \int_0^1 \left\{\bJ_{\eta}(\bx + t\bu) - \bJ_{\eta}(\bx) \right\}  \bu  dt. 
\end{align}
Invoking \eqref{eq:jacobian.2} with $t\bu$ in place of $\bu$, we conclude
\begin{align}
	\| \bvarrho_{\eta}(\bx, \bu) \|_2 &\le \int_0^1 \| \bJ_{\eta}(\bx + t\bu) - \bJ_{\eta}(\bx) \|_\txtop  \cdot \| \bu \|_2  t dt
	\\
	&\le \frac{h_{\eta} (h_{\eta}-1)}{2}  \| \boldeta \|_1  \left( 1 + \| \bx \|_2 + \| \bu \|_2 \right)^{h_{\eta}-2}  \| \bu \|_2^2. 
\end{align}
This completes the proof. 
\end{proof}

\subsubsection{Proof of Lemma~\ref{lemma:rec.identification}}

\begin{proof}
Condition on $\Ec_{\eta}$. 
First, because $\Fc_{\eta} \subseteq \Ec_{\eta}$, the tower property gives for each $\ell \in \Sc_\eta$, 
\begin{align}
	\Ex\left[ \varepsilon_{N, T+\ell}(1) \mid \Fc_{\eta} \right]
	= \Ex\left[ \Ex\left[ \varepsilon_{N, T+\ell}(1) \mid \Ec_{\eta} \right] \mid \Fc_{\eta} \right]
	= 0. 
\end{align}
Consequently, it follows that 
\begin{align}
	\theta_{\eta} = \sum_{\ell \in \Sc_\eta} \eta_\ell  \left\langle \bu_N, \bv_{T+\ell}(1) \right \rangle.
\end{align}
By Proposition~\ref{prop:alpha} and its proof, Assumption~\ref{assump:hankel} implies that there exists a fixed $\balpha \in \Rb^L$ satisfying the common recurrence: for all $j \in \Ic_1$ and $t \in \Zb$,
\begin{align}
	\mu_j(t+1) = \sum_{a=1}^L \alpha_a \mu_j(t-L+a),
\end{align}
where we recall $\mu_j(t) \coloneqq \langle \bu_j, \bv_t(1) \rangle$. 
Let us define the state vector $\bx_j(t) \coloneqq \left[\mu_j(t-L+1), \dots, \mu_j(t) \right]^\top \in \Rb^L$. 
Then the recurrence is equivalent to $\bx_j(t+1) = \bPi(\balpha)  \bx_j(t)$;
iterating $\ell$ times gives
\begin{align}
	\bx_j(T+\ell) = \bPi(\balpha)^\ell  \bx_j(T). 
\end{align}
The final coordinate can then be expressed as 
\begin{align}
	\mu_j(T+\ell) = \be_L^\top  \bPi(\balpha)^\ell  \bx_j(T) = \left \langle \bg_\ell(\balpha), \bx_j(T) \right \rangle. 
\end{align}
Note that the $j$th row of $\bbW$ is precisely $\bx_j(T)^\top$, and thus for every $\ell \in \Sc_\eta$, $\bbW \bg_\ell(\balpha) = [\mu_j(T+\ell): j \in \Ic_1]$. 
%
Therefore, following the proofs of Proposition~\ref{prop:beta} and Corollary~\ref{cor:identification}, we conclude
\begin{align}
	\theta_{\eta} = \left\langle \bbeta, \bbW \bg_{\eta}(\balpha) \right \rangle = \left\langle \bbeta^*, \bbW \bg_{\eta}(\balpha) \right \rangle. \label{eq:rec.ident.0}
\end{align}
It remains to replace $\alpha$ with $\alpha^*$ in \eqref{eq:rec.ident.0}. 
In this pursuit, note that $(\balpha^* - \balpha)$ is orthogonal to $\row(\bbZ^\top)$. 
Since Assumption~\ref{assump:recursive} states that for all $m \in \{0, \dots, h_{\eta}-1\}$, $\row(\bbW \{\bPi(\balpha)^m \}^\top ) \subseteq \row(\bbZ^\top)$, we have 
\begin{align}
	\bbW  \left\{ \bPi(\balpha)^m \right\}^\top  \left(\balpha^* - \balpha \right) = \bzero. \label{eq:rec.identity.0}
\end{align}
Now, by the telescoping identity \eqref{eq:telescope}, we have that for all $\ell \in [h_{\eta}]$
\begin{align}
	\bPi(\balpha^*)^\ell - \bPi(\balpha)^\ell &= \sum_{a=0}^{\ell-1} \bPi(\balpha^*)^a  \left\{\bPi(\balpha^*) - \bPi(\balpha) \right\}  \bPi(\balpha)^{\ell-1-a}
	\\
	&= \sum_{a=0}^{\ell-1} \bPi(\balpha^*)^a  \be_L \left(\balpha^* - \balpha \right)^\top  \bPi(\balpha)^{\ell-1-a}. &&\because \text{\eqref{eq:jac.diff}}
\end{align}
Taking transposes and multiplying by $\bbW$ further gives
\begin{align}
	\bbW \left\{ \bPi(\balpha^*)^\ell - \bPi(\balpha)^\ell \right\}^\top &= 
	\sum_{a=0}^{\ell-1} \bbW  \left\{ \bPi(\balpha)^{\ell-1-a} \right\}^\top  \left(\balpha^* - \balpha \right) \be_L^\top  \left\{ \bPi(\balpha^*)^a \right\}^\top. 
\end{align}
For each summand, $m = \ell - 1 - a \in \{0, \dots, \ell-1\} \subseteq \{0, \dots, h_{\eta}-1\}$. 
Hence, by \eqref{eq:rec.identity.0}, every summand vanishes so that we have for all $\ell \in \{0, \dots, h_{\eta}\}$
\begin{align}
	\bbW  \left\{ \bPi(\balpha^*)^\ell \right\}^\top = \bbW  \left\{ \bPi(\balpha)^\ell \right\}^\top. \label{eq:rec.useful}
\end{align}
Multiplying \eqref{eq:rec.useful} by $\be_L$, we then obtain $\bbW \bg_\ell(\balpha^*) = \bbW \bg_\ell(\balpha)$ for each $\ell \in \Sc_\eta$, and hence, $\bbW \bg_{\eta}(\balpha^*) = \bbW \bg_{\eta}(\balpha)$. Therefore,
\begin{align}
	\theta_{\eta} = \left \langle \bbeta^*, \bbW \bg_{\eta}(\balpha) \right \rangle = \left \langle \bbeta^*, \bbW \bg_{\eta}(\balpha^*) \right \rangle. 
\end{align}
This completes the derivation for \eqref{eq:identification.rec}. 

We now prove the feasibility of the two recursive Riesz equations. 
Define the recursive score as 
\begin{align}
	S_{\eta}(\ba, \boldsymbol{b}) \coloneqq \left\langle \bg_{\eta}(\ba), \bbW^\top \boldsymbol{b} \right \rangle + \left\langle \bq_{\beta \eta}, \bby - \bbY^\top \boldsymbol{b} \right\rangle + \left\langle \bq_{\alpha \eta}, \bbz - \bbZ^\top \ba \right \rangle. 
\end{align}
At $(\ba, \boldsymbol{b}) = (\balpha^*, \bbeta^*)$, the residuals vanish. 
Beginning with the unit-side, differentiating $S_{\eta}$ with respect to $\boldsymbol{b}$ and evaluating at $(\balpha^*, \bbeta^*)$ gives the orthogonality requirement
\begin{align}
	\bbY \bq_{\beta \eta} = \bbW \bg_{\eta}(\balpha^*). 
\end{align}
This equation is feasible because Assumption~\ref{assump:recursive} suggests $\bbW \bg_{\eta}(\balpha^*) \in \col(\bbW) \subseteq \col(\bbY)$, and thus the minimum-norm Riesz representer is
\begin{align}
	\bq^{*}_{\beta \eta} = \bbY^\dagger \bbW \bg_{\eta}(\balpha^*). 
\end{align}
Moving on the time-side, we fix a direction $\bu \in \Rb^L$. 
By Lemma~\ref{lemma:jacobian}, for every sufficiently small $t$ 
\begin{align}
	\bg_{\eta}(\balpha^* + t \bu) = \bg_{\eta}(\balpha^*) + t \bJ_{\eta}(\balpha^*)  \bu + \bvarrho_{\eta}(\balpha^*, t\bu),
\end{align}
where $\| \bvarrho_\eta(\balpha^*, t\bu) \|_2 / |t| \rightarrow 0$. Hence, 
%
%
differentiating $S_{\eta}$ with respect to $\ba$ and evaluating at $(\balpha^*, \bbeta^*)$ gives the orthogonality requirement
\begin{align}
	\bbZ \bq_{\alpha \eta} = \left\{ \bJ_{\eta}(\balpha^*) \right\}^\top \bbW^\top \bbeta^*. \label{eq:riesz.time.1}
\end{align}
Applying \eqref{eq:jacobian}, we obtain 
\begin{align}
	\bbW \bJ_{\eta}(\balpha^*) = \sum_{\ell \in \Sc_\eta} \eta_\ell \sum_{a=0}^{\ell-1} \left\{ \be_L^\top  \bPi(\balpha^*)^a  \be_L \right\}  \bbW  \left\{ \bPi(\balpha^*)^{\ell-1-a} \right\}^\top,
\end{align}
For each $a$, take $m \coloneqq \ell - 1 -a$. Then, $0 \le m \le \ell-1 \le h_{\eta}-1$. By \eqref{eq:rec.useful} and Assumption~\ref{assump:recursive}, $\row(\bbW \{\bPi(\balpha^*)^m\}^\top) \subseteq \row(\bbZ^\top)$. 
Each summand in $\bbW \bJ_{\eta}(\balpha^*)$ consequently has its row space contained in $\row(\bbZ^\top)$, and hence, $\row(\bbW \bJ_{\eta}(\balpha^*)) \subseteq \row(\bbZ^\top)$. Equivalently, $\{\bJ_{\eta}(\balpha^*) \}^\top \bbW^\top \bbeta^* \in \col(\bbZ)$. 
Thus, \eqref{eq:riesz.time.1} is feasible and the minimum-norm solution is 
\begin{align}
	\bq^{*}_{\alpha \eta} = \bbZ^\dagger  \left\{ \bJ_{\eta}(\balpha^*)\right\}^\top  \bbW^\top \bbeta^*. 
\end{align}
The proof is complete. 
\end{proof}

\subsubsection{Proof of Lemma~\ref{lemma:rec.decomp}} 

\begin{proof}
Condition on $\Ec_{\eta}$. 
Following the proof of Lemma~\ref{lemma:decomp}, note that by construction of \eqref{eq:rec.riesz.proj.def}, we have  
\begin{align}
	\left \langle \tbq_{\alpha \eta}, \bz - \bZ^\top \hbalpha \right \rangle = 0,
	\quad
	\left \langle \tbq_{\beta \eta}, \by - \bY^\top \hbbeta \right \rangle = 0. 
\end{align}
Thus, we can write 
\begin{align}
	\htheta_{\eta} = \htheta_{\eta} + \left \langle \tbq_{\beta \eta}, \by - \bY^\top \hbbeta \right \rangle + \left \langle \tbq_{\alpha \eta}, \bz - \bZ^\top \hbalpha \right \rangle. 
\end{align}
Next, we apply Lemma~\ref{lemma:jacobian} to obtain 
\begin{align}
	\hbalpha_{\eta} = \balpha^*_{\eta} + \bJ^*_{\eta}  \bDelta_{\alpha} + \bvarrho_{\eta}. 
\end{align}
Expand the unit-side correction as 
\begin{align}
	\left\langle \tbq_{\beta \eta}, \by - \bY^\top \hbbeta \right \rangle
	&=  \left \langle \bq^*_{\beta \eta}, \bdelta_{\beta} \right \rangle 
	- \left\langle \bq^*_{\beta \eta}, \bbY^\top \bDelta_{\beta} \right \rangle - \left \langle \bq^*_{\beta \eta}, \bXi_y^\top \bDelta_{\beta} \right \rangle 
	\\
	&\quad +  \left \langle \bDelta_{\tq_{\beta \eta}}, \bdelta_{\beta} \right \rangle - \left \langle \bDelta_{\tq_{\beta \eta}}, \bbY^\top \bDelta_{\beta} \right \rangle - \left \langle \bDelta_{\tq_{\beta \eta}}, \bXi_y^\top \bDelta_{\beta} \right \rangle .  \label{eq:decomp.rec.1}
\end{align}
At the same time, the time-side correction expands as 
\begin{align}
	\left\langle \tbq_{\alpha \eta}, \bz - \bZ^\top \hbalpha \right \rangle
	&= \left \langle \bq^*_{\alpha \eta}, \bdelta_{\alpha} \right \rangle 
	- \left\langle \bq^*_{\alpha \eta}, \bbZ^\top \bDelta_{\alpha} \right \rangle - \left \langle \bq^*_{\alpha \eta}, \bXi_z^\top \bDelta_{\alpha} \right \rangle 
	\\
	&\quad +  \left \langle \bDelta_{\tq_{\alpha \eta}}, \bdelta_{\alpha} \right \rangle - \left \langle \bDelta_{\tq_{\alpha \eta}}, \bbZ^\top \bDelta_{\alpha} \right \rangle - \left \langle \bDelta_{\tq_{\alpha \eta}}, \bXi_z^\top \bDelta_{\alpha} \right \rangle .   \label{eq:decomp.rec.2}
\end{align}
Moving on to the targets, we follow the arguments that led to \eqref{eq:general} to obtain 
\begin{align}
	&\left\langle \hbalpha_{\eta}, \bW^\top \hbbeta \right \rangle - \left\langle \balpha^*_{\eta}, \bbW^\top \bbeta^* \right \rangle
	\\
	&\qquad = \left\langle \balpha^*_{\eta}, \bbW^\top \bDelta_{\beta} \right \rangle + \left\langle \bJ^*_{\eta} \bDelta_{\alpha}, \bbW^\top \bbeta^* \right \rangle + \left\langle \balpha^*_{\eta}, \bXi_w^\top \bbeta^* \right\rangle 
	\\
	&\qquad \quad + \left\langle \balpha^*_{\eta}, \bXi_w^\top \bDelta_{\beta} \right\rangle + \left\langle \bJ^*_{\eta} \bDelta_{\alpha}, \bbW^\top \bDelta_{\beta} \right \rangle + \left\langle \bJ^*_{\eta} \bDelta_{\alpha}, \bXi_w^\top \bbeta^* \right \rangle + \left\langle \bJ^*_{\eta} \bDelta_{\alpha}, \bXi_w^\top \bDelta_{\beta} \right \rangle
	\\
	&\qquad \quad + \left\langle \bvarrho_{\eta}, \bbW^\top \bbeta^* \right \rangle + \left\langle \bvarrho_{\eta}, \bbW^\top \bDelta_{\beta} \right \rangle + \left\langle \bvarrho_{\eta}, \bXi_w^\top \bbeta^* \right\rangle + \left\langle \bvarrho_{\eta}, \bXi_w^\top \bDelta_{\beta} \right\rangle. 
	\label{eq:decomp.rec.3}
\end{align}
Moreover, by \eqref{eq:riesz.rec}, 
\begin{align}
	\left\langle \bq^*_{\beta \eta}, \bbY^\top \bDelta_{\beta} \right\rangle &= \left\langle \balpha^*_{\eta}, \bbW^\top \bDelta_{\beta} \right \rangle,
	\\
	\left\langle \bq^*_{\alpha \eta}, \bbZ^\top \bDelta_{\alpha} \right\rangle &= \left\langle \bJ^*_{\eta} \bDelta_{\alpha}, \bbW^\top \bbeta^* \right \rangle. 
	\label{eq:decomp.rec.4}
\end{align}
Merging \eqref{eq:decomp.rec.1}, \eqref{eq:decomp.rec.2}, and \eqref{eq:decomp.rec.3}, and applying the cancellations in \eqref{eq:decomp.rec.4} gives the desired result. 
\end{proof}

\subsubsection{Proof of Lemma~\ref{lemma:rec.normality}} 

\begin{proof}
Condition on $\Ec_{\eta}$. 

\bigskip \noindent \underline{\em Normality of $G_{\eta}$}: 
Observe that
\begin{align}
	\Var\left(\left\langle \balpha^*_{\eta}, \bXi_w^\top \bbeta^* \right\rangle\right) 
	&= \sigma^2  \| \balpha^*_{\eta} \|_2^2  \cdot \| \bbeta^* \|_2^2,
	\\
	\Var\left(\left\langle \bq^*_{\beta \eta}, \bdelta_{\beta} \right\rangle \right)
	&= \sigma^2  \| \bq^*_{\beta \eta} \|_2^2  \left(1 + \| \bbeta^* \|_2^2 \right),
	\\
	\Var\left(\left\langle \bq^*_{\alpha \eta}, \bdelta_{\alpha} \right \rangle\right)
	&= \sigma^2  \| \bq^*_{\alpha \eta} \|_2^2  \left(1 + \| \balpha^* \|_2^2 \right). 
\end{align}
With this, we obtain our desired result from following the proof of Lemma~\ref{lemma:lead} verbatim to conclude \eqref{eq:rec.asymp.var}. 

\bigskip \noindent \underline{\em High probability bound on $R_{\eta}$}: 
We follow the proof of Lemma~\ref{lemma:remainder}. To that end, let 
\begin{align}
	\varphi_\rho \coloneqq \sqrt{C_\varphi \log(C_\varphi / \rho)}. 
	\label{eq:rec.varphi}
\end{align}
We begin by recording some useful algebraic facts. 
Recall \eqref{eq:eta.lambda.simple}, \eqref{eq:lmda.lb}, and \eqref{eq:pop.op}.  
Define the envelope
\begin{align}
	A_{\eta} \coloneqq C_\eta  \| \boldeta \|_1  \left(1 + \frac{\sqrt{M}}{\lambda_z} + \frac{\sigma \Lambda_{\alpha}}{\sqrt{L}} \right)^{h_{\eta}-1}, \label{eq:ch}
\end{align}
where $C_\eta > 0$ depends only on $h_\eta$ and the fixed constants in Assumption~\ref{assump:spectra}. 
Define the events $\Gc_{\PCR, \alpha}$ and $\Gc_{\PCR, \beta}$ as in \eqref{eq:event.pcr}, and the joint event $\Gc_\PCR \coloneqq \Gc_{\PCR, \alpha} \cap \Gc_{\PCR, \beta}$.  
A straightforward adaptation of the proof Theorem~\ref{prop:parameter.recovery} with the conditioning on $\Ec_{\eta}$ yields 
\begin{align}
	\Pb(\Gc_\PCR^c \mid \Ec_{\eta}) \lesssim \rho. \label{eq:rec.event.pcr} 
\end{align} 
Notably, on $\Gc_{\PCR, \alpha}$, Lemma~\ref{lemma:jacobian} with Assumption~\ref{assump:spectra} states
\begin{align}
	\| \balpha^*_{\eta} \|_2 &\le A_{\eta}  \frac{\sqrt{M}}{\lambda_z},     \label{eq:rec.alpha.bound} 
	\\
	\| \bJ^*_{\eta} \|_\txtop &\le A_{\eta}, \label{eq:rec.J.bound}
	\\
	\| \hbJ_{\eta} - \bJ^*_{\eta} \|_\txtop &\le A_{\eta}  \frac{\sigma \Lambda_{\alpha}}{\sqrt{L}}, \label{eq:rec.dJ.bound}
	\\
	\| \bvarrho_{\eta} \|_2 &\le A_{\eta}  \frac{\sigma^2 \Lambda_{\alpha}^2}{K}. \label{eq:rec.R.bound}
\end{align}
As a direct consequence, we have 
\begin{align}
	\| \bq^*_{\beta \eta} \|_2 \le \| \bbY^\dagger \|_\txtop  \cdot \| \bbW \balpha^*_{\eta} \|_2
	\le \frac{\| \bbW \balpha^*_{\eta} \|_2}{\lambda_y}. 
\end{align}
By the proof of Lemma~\ref{lemma:rec.identification}, the $j$th coordinate of $(\bbW \balpha^*_{\eta})_j = \sum_{\ell \in \Sc_\eta} \eta_\ell \Ex[Y_{j, T+\ell}(1) \mid \Ec_{\eta}]$. Assumption~\ref{assump:bounded}, therefore, gives $\| \bbW \balpha^*_{\eta} \|_2 \le \| \boldeta \|_1 \sqrt{N_1}$ and hence,
\begin{align}
	\| \bq^*_{\beta \eta} \|_2 \lesssim \| \boldeta \|_1 \frac{\sqrt{N_1}}{\lambda_y}. 	\label{eq:rec.qb.bound}
\end{align}
At the same time, by \eqref{eq:barwbeta.bound} and \eqref{eq:rec.J.bound} 
\begin{align}
	\| \bq^*_{\alpha \eta} \|_2 
	\le \| \bbZ^\dagger \|_\txtop  \cdot \| \bJ^*_{\eta} \|_\txtop  \cdot \| \bbW^\top \bbeta^* \|_2
	\le 
	A_{\eta}  \frac{\sqrt{L}}{\lambda_z}. 
	\label{eq:rec.qa.bound} 
\end{align}
Armed with these results, we proceed to bound each term in $R_{\eta}$ based on \eqref{eq:remain.rec}. 

\bigskip \noindent {\em Stochastic $\bXi_w$-terms.} 
Define the events
\begin{align}
	\Gc_{\alpha \Delta} &\coloneqq \left\{ \left| \left\langle \balpha^*_{\eta}, \bXi_w^\top \bDelta_{\beta} \right \rangle \right| \le C_w A_{\eta}  \frac{\sigma^2 \varphi_\rho \sqrt{M} \Lambda_{\beta}}{\lambda_z \sqrt{N_1}} \right\}, 
	\\
	\Gc_{\Delta \beta} &\coloneqq \left\{ \left| \left\langle \bJ^*_{\eta} \bDelta_{\alpha}, \bXi_w^\top \bbeta^* \right \rangle \right| \le C_w A_{\eta}  \frac{\sigma^2 \varphi_\rho \sqrt{T_0} \Lambda_{\alpha}}{\lambda_y \sqrt{L}} \right\}, 
	\\
	\Gc_{\Delta \Delta} &\coloneqq \left\{ \left| \left\langle \bJ^*_{\eta} \bDelta_{\alpha}, \bXi_w^\top \bDelta_{\beta} \right \rangle \right| \le C_w A_{\eta}  \frac{\sigma^3 \varphi_\rho \Lambda_{\alpha} \Lambda_{\beta}}{\sqrt{N_1L}} \right\}, 
	\\
	\Gc_{\varrho \beta} &\coloneqq \left\{ \left| \left\langle \bvarrho_{\eta}, \bXi_w^\top \bbeta^* \right \rangle \right| \le C_w A_{\eta}  \frac{\sigma^3 \varphi_\rho \Lambda^2_{\alpha} \sqrt{T_0}}{\lambda_y L} \right\}, 
	\\
	\Gc_{\varrho \Delta} &\coloneqq \left\{ \left| \left\langle \bvarrho_{\eta}, \bXi_w^\top \bDelta_{\beta} \right \rangle \right| \le C_w A_{\eta}  \frac{\sigma^4 \varphi_\rho \Lambda^2_{\alpha} \Lambda_{\beta}}{L \sqrt{N_1}} \right\},
	\label{eq:rec.w.bounds} 
\end{align}
where $C_w > 0$ is a sufficiently large constant. 
Define $\Gc_w \coloneqq \Gc_{\alpha \Delta} \cap \Gc_{\Delta \beta} \cap \Gc_{\Delta \Delta} \cap \Gc_{\varrho \beta} \cap \Gc_{\varrho \Delta}$. 
Following the arguments that led to \eqref{eq:rho.temp},
\begin{align}
	\Pb\left(\Gc_\PCR \cap \Gc_w^c \mid \Ec_{\eta} \right) \lesssim \rho. \label{eq:rec.event.w} 
\end{align}
On $\Gc_\PCR \cap \Gc_w$, the $\bXi_w$-block stochastic terms satisfy exactly the bounds in \eqref{eq:rec.w.bounds}.

\bigskip \noindent {\em Deterministic $\bbW$-terms.} 
On the event $\Gc_\PCR$, note that \eqref{eq:rec.J.bound} gives 
\begin{align}
	\left| \left\langle \bJ_{\eta}^* \bDelta_{\alpha}, \bbW^\top \bDelta_{\beta} \right \rangle \right| 
	&\le \| \bJ^*_{\eta} \|_\txtop \cdot \| \bDelta_{\alpha} \|_2  \cdot \| \bbW \|_\txtop \cdot \| \bDelta_{\beta} \|_2
	\lesssim A_{\eta}  \sigma^2 \Lambda_{\alpha} \Lambda_{\beta}.
	\label{eq:rec.det.w.1}
\end{align}
Note that the above uses \eqref{eq:barw.bound.op} to simplify $\| \bbW \|_\txtop \le \sqrt{N_1L}$. 
Similarly, by \eqref{eq:rec.R.bound} and \eqref{eq:barwbeta.bound}, 
\begin{align}
	\left| \left\langle \bvarrho_{\eta}, \bbW^\top \bbeta^* \right \rangle \right| &\le \| \bvarrho_{\eta} \|_2  \cdot \| \bbW^\top \bbeta^* \|_2
	\lesssim A_{\eta}  \frac{\sigma^2 \Lambda^2_{\alpha}}{\sqrt{L}},
	\label{eq:rec.det.w.2} 
	\\
	\left| \left\langle \bvarrho_{\eta}, \bbW^\top \bDelta_{\beta} \right \rangle \right| &\le \| \bvarrho_{\eta} \|_2  \cdot \| \bbW \|_\txtop \cdot \| \bDelta_{\beta} \|_2
	\lesssim A_{\eta}  \frac{\sigma^3 \Lambda^2_{\alpha} \Lambda_{\beta}}{\sqrt{L}}.
	\label{eq:rec.det.w.3} 
\end{align}

\noindent {\em Riesz error terms: unit-side.} 
Define the event $\Gc_{\noise}$ as in \eqref{eq:event.noise}. 
By Lemma~\ref{lemma:subg_matrix} and the union bound,  
\begin{align}
	\Pb\left(\Gc_{\noise, y}^c \mid \Ec_{\eta} \right) \lesssim \rho. \label{eq:rec.event.noise} 
\end{align} 
Following the proof of Lemma~\ref{lemma:riesz.0}, on $\Gc_{\noise, y}$, we use \eqref{eq:rec.qb.bound} to obtain
\begin{align}
	\| \bDelta_{\tq_{\beta \eta}} \|_2 \lesssim \| \boldeta \|_1 \tQ_{\beta},
\end{align}
where $\tQ_{\beta}$ is defined as in \eqref{eq:tQ.beta}. 
%
%
Next, define the event
\begin{align}
	\Gc_{\xi \beta} \coloneqq \left\{ \left| \left\langle \bDelta_{\tq_{\beta \eta}}, \bxi_y \right\rangle \right| \le C_\xi  \sigma \varphi_\rho  \| \boldeta \|_1 \tQ_{\beta} \right\}
\end{align}
where $C_\xi > 0$ is a sufficiently large constant. 
Following the arguments that led to \eqref{eq:remain.t7}, we have on $\Gc_{\noise, y} \cap \Gc_{\PCR} \cap \Gc_{\xi \beta}$,
\begin{align}
	\left| \left\langle \bDelta_{\tq_{\beta \eta}}, \bdelta_{\beta} \right \rangle \right|
	+ \left| \left\langle \bDelta_{\tq_{\beta \eta}}, \bbY^\top \bDelta_{\beta} \right \rangle \right|
	+ \left| \left\langle \bDelta_{\tq_{\beta \eta}}, \bXi_y^\top \bDelta_{\beta} \right \rangle \right|
	&\lesssim 
	 \| \boldeta \|_1 \tQ_{\beta}  \left( \sigma \varphi_\rho + \frac{\eta_y \sqrt{T_0}}{\lambda_y} + \frac{\sigma \lambda_y \Lambda_{\beta}}{\sqrt{N_1}} \right).
	\label{eq:rec.unit.1} 
\end{align} 
Define the sigma-field $\Hc_y \coloneqq \Ec_{\eta} \vee \sigma(\bY)$. 
Critically, $\bDelta_{\tq_{\beta \eta}}$ is $\Hc_y$-measurable and $\bxi_y$ is conditionally independent of $\Hc_y$ given $\Ec_{\eta}$. 
Hence, Lemma~\ref{lemma:hoeffding} and the tower property yield 
\begin{align}
	\Pb\left(\Gc_{\noise, y}  \cap \Gc_{\xi \beta}^c \mid \Ec_{\eta} \right) 
	&= \Ex\left[ \mathds{1}\left\{\Gc_{\noise, y}  \right\}  \Pb \left( \Gc^c_{\xi \beta} \mid \Hc_y \right) \mid \Ec_{\eta} \right]
	 \lesssim \rho. 
	\label{eq:rec.event.unit} 
\end{align} 
Moreover, on $\Gc_{\noise,y} \cap \Gc_{\PCR, \beta}$ and applying \eqref{eq:rec.qb.bound} 
\begin{align}
	\left| \left \langle \bq^*_{\beta \eta}, \bXi_y^\top \bDelta_{\beta} \right \rangle \right|
	\le \| \bq^*_{\beta \eta} \|_2 \cdot \| \bXi_y \|_\txtop \cdot  \| \bDelta_{\beta} \|_2
	\lesssim \| \boldeta \|_1 \frac{\sigma \eta_y \Lambda_{\beta}}{\lambda_y}. 
	\label{eq:rec.unit.2}
\end{align}

\noindent {\em Riesz error terms: time-side.} 
Following the proof of Lemma~\ref{lemma:riesz.0}, on $\Gc_{\noise, z}$, we use \eqref{eq:rec.qa.bound} to obtain 
\begin{align}
	\| \bDelta_{\tq_{\alpha \eta}} \|_2 \lesssim A_{\eta}  \tQ_{\alpha},
\end{align}
where $\tQ_{\alpha}$ is defined as in \eqref{eq:tQ.alpha}. 
Next, define the event
\begin{align}
	\Gc_{\xi \alpha} \coloneqq \left\{ \left| \left\langle \bDelta_{\tq_{\alpha \eta}}, \bxi_z \right\rangle \right| \le C_\xi A_{\eta}  \sigma \varphi_\rho  \tQ_{\alpha} \right\}.
\end{align} 
Following the arguments that led to \eqref{eq:remain.t8}, we have on $\Gc_{\noise, z} \cap \Gc_{\PCR} \cap \Gc_{\xi \alpha}$,
\begin{align}
	\left| \left \langle \bDelta_{\tq_{\alpha \eta}}, \bdelta_{\alpha} \right \rangle \right|
	+ \left| \left \langle \bDelta_{\tq_{\alpha \eta}}, \bbZ^\top \bDelta_{\alpha} \right \rangle \right|
	+ \left| \left \langle \bDelta_{\tq_{\alpha \eta}}, \bXi_z^\top \bDelta_{\alpha} \right \rangle \right|
	&\lesssim 
	 A_{\eta} \tQ_{\alpha}  \left( \sigma \varphi_\rho + \frac{\eta_z \sqrt{M}}{\lambda_z} + \frac{\sigma \lambda_z \Lambda_{\alpha}}{\sqrt{L}} \right).
	\label{eq:rec.time.1} 
\end{align} 
Define the sigma-field $\Hc_z \coloneqq \Ec_{\eta} \vee \sigma(\bZ)$. 
Note $\bDelta_{\tq_{\alpha \eta}}$ is $\Hc_z$-measurable and $\bxi_z$ is conditionally independent of $\Hc_z$ given $\Ec_{\eta}$. 
Hence, by Lemma~\ref{lemma:hoeffding} and the tower property, 
\begin{align}
	\Pb\left(\Gc_{\noise, z}  \cap \Gc_{\xi \alpha}^c \mid \Ec_{\eta} \right) 
	&= \Ex\left[ \mathds{1}\left\{\Gc_{\noise, z}  \right\}  \Pb \left( \Gc^c_{\xi \alpha} \mid \Hc_z \right) \mid \Ec_{\eta} \right]
	 \lesssim \rho. 
	\label{eq:rec.event.time} 
\end{align} 
Next, on $\Gc_{\noise, z} \cap \Gc_{\PCR, \alpha}$ and applying \eqref{eq:rec.qa.bound} 
\begin{align}
	\left| \left \langle \bq^*_{\alpha \eta}, \bXi_z^\top \bDelta_{\alpha} \right \rangle \right|
	\le \| \bq^*_{\alpha \eta} \|_2 \cdot \| \bXi_z \|_\txtop \cdot \| \bDelta_{\alpha} \|_2
	\lesssim A_{\eta}  \frac{\sigma \eta_z \Lambda_{\alpha}}{\lambda_z}. 
	\label{eq:rec.time.2}
\end{align}

\noindent {\em Putting everything together.} 
On the master event $\Gc_\star \coloneqq \Gc_{\noise, y} \cap \Gc_{\noise, z} \cap \Gc_\PCR \cap \Gc_w \cap \Gc_{\xi \beta} \cap \Gc_{\xi \alpha} $, we combine \eqref{eq:rec.w.bounds}, \eqref{eq:rec.det.w.1}, \eqref{eq:rec.det.w.2}, \eqref{eq:rec.det.w.3}, \eqref{eq:rec.unit.1}, \eqref{eq:rec.unit.2}, \eqref{eq:rec.time.1}, and \eqref{eq:rec.time.2} to arrive at the inequality
\begin{align}
	|R_{\eta}|
	& \lesssim 
	A_{\eta}  \sigma^2 \varphi_\rho 
	 \left\{ \frac{\sqrt{M} \Lambda_{\beta}}{\lambda_z \sqrt{N_1}}
	+  \frac{\sqrt{T_0} \Lambda_{\alpha}}{\lambda_y \sqrt{L}}
	+  \frac{\sigma \Lambda_{\alpha} \Lambda_{\beta}}{\sqrt{N_1L}}
	+  \frac{\sigma \Lambda^2_{\alpha} \sqrt{T_0}}{\lambda_y L}
	+  \frac{\sigma^2 \Lambda^2_{\alpha} \Lambda_{\beta}}{L \sqrt{N_1}} \right\}
	\\ 
	&\quad
	+ A_{\eta}  \sigma^2 \Lambda_{\alpha} \left\{
		 \Lambda_{\beta} + \frac{\Lambda_{\alpha}}{\sqrt{L}}
		 + \sigma \frac{\Lambda_{\alpha} \Lambda_{\beta}}{\sqrt{L}}
	\right\}
	\\
	&\quad
	+ \| \boldeta \|_1\tQ_{\beta}  \left( \sigma \varphi_\rho + \frac{\eta_y \sqrt{T_0}}{\lambda_y} + \frac{\sigma \lambda_y \Lambda_{\beta}}{\sqrt{N_1}} \right)
	+ \| \boldeta \|_1 \frac{\sigma \eta_y \Lambda_{\beta}}{\lambda_y}
	\\
	&\quad
	+ A_{\eta}  \tQ_{\alpha}  \left( \sigma \varphi_\rho + \frac{\eta_z \sqrt{M}}{\lambda_z} + \frac{\sigma \lambda_z \Lambda_{\alpha}}{\sqrt{L}} \right) 
	+ A_{\eta}  \frac{\sigma \eta_z \Lambda_{\alpha}}{\lambda_z}.
	\label{eq:rec.remain.bound}
\end{align}
Using \eqref{eq:noise.eta}, \eqref{eq:eta.lambda.simple}, \eqref{eq:lmda.lb}, \eqref{eq:pop.op}, and \eqref{eq:rec.varphi}, every term in \eqref{eq:rec.remain.bound} is bounded above by $\mathfrak{C}_{\eta} \| \boldeta \|_1 \Psi$. 
Note that this simplification uses \eqref{eq:lmda.lb}. 

It remains to bound the probability of $\Gc_\star$. 
In this pursuit, we leverage \eqref{eq:rec.event.pcr}, \eqref{eq:rec.event.w}, \eqref{eq:rec.event.noise}, \eqref{eq:rec.event.unit}, and \eqref{eq:rec.event.time}, and take a union bound to conclude $\Pb(\Gc_\star^c \mid \Ec_{\eta}) \lesssim \rho$. 

\bigskip \noindent \underline{\em Completing the proof}: 
The desired result is obtained by following the arguments in the proof of Theorem~\ref{thm:inference} verbatim with $(G_{\eta}, R_{\eta}, \mathfrak{C}_{\eta} \| \boldeta \|_1 \Psi, \Vc_{\eta})$ in place of $(G, R, \Psi, \Vc)$. 
\end{proof}

\subsubsection{Proof of Lemma~\ref{lemma:rec.var.est}} 

\begin{proof}
Condition on $\Ec_{\eta}$. 
Define $\Upsilon_{\eta} \coloneqq \Vc^2_{\eta}/ \sigma^2$ and $\hUpsilon_{\eta} \coloneqq \hVc_{\eta}^2 / \hsigma^2$. 
Observe that
\begin{align}
	\hUpsilon_{\eta} - \Upsilon_{\eta}
	&= \left\{ \| \hbalpha_{\eta} \|_2^2  \cdot \| \hbbeta \|_2^2 - \| \balpha_{\eta}^* \|_2^2 \cdot \| \bbeta^* \|_2^2 \right\}
	 + \left\{ \| \hbq_{\beta \eta} \|_2^2   \left( 1 + \| \hbbeta \|_2^2 \right)  - \| \bq_{\beta \eta}^* \|_2^2  \left(1 + \| \bbeta^* \|_2^2 \right) \right\}
	\\ &\quad+ \left\{ \| \hbq_{\alpha \eta} \|_2^2  \left( 1 + \| \hbalpha \|_2^2 \right)  - \| \bq_{\alpha \eta}^* \|_2^2  \left(1 + \| \balpha^* \|_2^2 \right) \right\}. 
\end{align}
Following the proof of Proposition~\ref{prop:var.asymp}, we control each individual term. 

\bigskip \noindent {\em Term 1}: 
Leveraging \eqref{eq:alg.ineq}, we obtain 
\begin{align}
	&\left| \| \hbalpha_{\eta} \|_2^2 \cdot  \| \hbbeta \|_2^2 - \| \balpha_{\eta}^* \|_2^2 \cdot \| \bbeta^* \|_2^2   \right|
	\\ &\quad \lesssim  \| \hbalpha_{\eta} - \balpha^*_{\eta} \|_2   \left( \| \balpha_{\eta}^* \|_2 + \| \hbalpha_{\eta} - \balpha^*_{\eta} \|_2 \right)  \left(\| \bbeta^* \|_2 + \| \bDelta_{\beta} \|_2 \right)^2
 + \| \bDelta_{\beta} \|_2   \left(  \| \bbeta^* \|_2 + \| \bDelta_{\beta} \|_2 \right)  \| \balpha_{\eta}^* \|_2^2. 
\end{align}
To enable progress on the above bound, observe that on $\Gc_{\PCR, \alpha}$, Lemma~\ref{lemma:jacobian} gives
\begin{align}
	\| \hbalpha_{\eta} - \balpha^*_{\eta} \|_2 
	&\le \| \bJ_{\eta}^* \bDelta_{\alpha} \|_2 + \| \bvarrho_{\eta} \|_2
	\\
	&\le \| \bJ_{\eta}^* \|_\txtop \cdot \| \bDelta_{\alpha} \|_2 + \| \bvarrho_{\eta} \|_2
	\\
	&\le A_{\eta}  \frac{\sigma \Lambda_{\alpha}}{\sqrt{L}} \left(1 + \frac{\sigma \Lambda_{\alpha}}{\sqrt{L}} \right). &&\because \text{\eqref{eq:rec.J.bound} and \eqref{eq:rec.R.bound}}
	\label{eq:rec.delta.alpha.bound}
\end{align}
Hence, by \eqref{eq:alpha.beta.bound},  \eqref{eq:rec.alpha.bound}, and \eqref{eq:rec.delta.alpha.bound}, we have on $\Gc_{\PCR}$
\begin{align}
	&\left| \| \hbalpha_{\eta} \|_2^2   \cdot\| \hbbeta \|_2^2 - \| \balpha_{\eta}^* \|_2^2 \cdot \| \bbeta^* \|_2^2   \right|
	\\&\quad \lesssim 
	A^2_{\eta}  \frac{\sigma \Lambda_{\alpha}}{\sqrt{L}} \left(1 + \frac{\sigma \Lambda_{\alpha}}{\sqrt{L}} \right)
	 \left\{ \frac{\sqrt{M}}{\lambda_z} + \frac{\sigma \Lambda_{\alpha}}{\sqrt{L}} \left(1 + \frac{\sigma \Lambda_{\alpha}}{\sqrt{L}} \right) \right\}
	 \left( \frac{\sqrt{T_0}}{\lambda_y} + \frac{\sigma \Lambda_{\beta}}{\sqrt{N_1}} \right)^2
	\\
	&\quad 
	+ A^2_{\eta}  \frac{\sigma M \Lambda_{\beta}}{\lambda^2_z \sqrt{N_1}} 
	 \left( \frac{\sqrt{T_0}}{\lambda_y} + \frac{\sigma \Lambda_{\beta}}{\sqrt{N_1}} \right).
	\label{eq:rec.var.est.1} 
\end{align}

\noindent {\em Term 2}: 
By a similar argument, applying \eqref{eq:alg.ineq} yields 
\begin{align}
	&\left| \| \hbq_{\beta \eta} \|_2^2  \left( 1 + \| \hbbeta \|_2^2 \right)  - \| \bq_{\beta \eta}^* \|_2^2  \left(1 + \| \bbeta^* \|_2^2 \right) \right|
	\\
	&\quad \lesssim \| \bDelta_{q_{\beta \eta}} \|_2   \left(   \| \bq^*_{\beta \eta} \|_2 + \| \bDelta_{q_{\beta \eta}} \|_2 \right)  \left\{ 1 + \left( \| \bbeta^* \|_2 + \| \bDelta_{\beta} \|_2 \right)^2 \right \} + \| \bq^*_{\beta \eta} \|^2_2 \cdot  \| \bDelta_{\beta} \|_2  \left(   \| \bbeta^* \|_2 + \| \bDelta_{\beta} \|_2 \right). 
\end{align}
Following the proof of Lemma~\ref{lemma:riesz}, we have on $\Gc_\noise \cap \Gc_{\PCR, \alpha}$ 
\begin{align}
	\| \bDelta_{q_{\beta \eta}} \|_2 \lesssim A_{\eta} Q_{\beta},   \label{eq:rec.delta.bl}
\end{align}
where $Q_{\beta}$ is defined as in \eqref{eq:delta.b}.
Invoking \eqref{eq:alpha.beta.bound}, \eqref{eq:rec.qb.bound}, and \eqref{eq:rec.delta.bl}, we have on $\Gc_\noise \cap \Gc_\PCR$, 
\begin{align}
	&\left| \| \hbq_{\beta \eta} \|_2^2  \left( 1 + \| \hbbeta \|_2^2 \right)  - \| \bq_{\beta \eta}^* \|_2^2  \left(1 + \| \bbeta^* \|_2^2 \right) \right|
	\\
	&~~ \lesssim 
	A_{\eta}  Q_{\beta}  \left( \| \boldeta \|_1 \frac{\sqrt{N_1}}{\lambda_y}
	+ A_{\eta}  Q_{\beta} \right)
	\left\{1 + \left( \frac{\sqrt{T_0}}{\lambda_y} + \frac{\sigma \Lambda_{\beta}}{\sqrt{N_1}} \right)^2 \right\}
	 + \| \boldeta \|_1^2 \frac{\sigma \sqrt{N_1} \Lambda_{\beta}}{\lambda_y^2} \left( \frac{\sqrt{T_0}}{\lambda_y} + \frac{\sigma \Lambda_{\beta}}{\sqrt{N_1}} \right). 
	\label{eq:rec.var.est.2} 
\end{align}

\noindent {\em Term 3}: 
Again, applying \eqref{eq:alg.ineq} yields 
\begin{align}
	&\left| \| \hbq_{\alpha \eta} \|_2^2  \left( 1 + \| \hbalpha \|_2^2 \right)  - \| \bq_{\alpha \eta}^* \|_2^2  \left(1 + \| \balpha^* \|_2^2 \right) \right|
	\\
	&\quad \lesssim \| \bDelta_{q_{\alpha \eta}} \|_2   \left( \| \bq^*_{\alpha \eta} \|_2 + \| \bDelta_{q_{\alpha \eta}} \|_2 \right)  \left\{ 1 + \left( \| \balpha^* \|_2 + \| \bDelta_{\alpha} \|_2 \right)^2 \right \} + \| \bq^*_{\alpha \eta} \|^2_2 \cdot  \| \bDelta_{\alpha} \|_2  \left(  \| \balpha^* \|_2 + \| \bDelta_{\alpha} \|_2 \right). 
\end{align}
In order to bound the time-side Riesz error terms, we first introduce 
\begin{align}
	\bq^\sharp_{\alpha \eta} \coloneqq \hbZ^\dagger  \left(\bJ^*_{\eta} \right)^\top  \bW^\top \hbbeta. 
\end{align}
This enables us to rewrite
\begin{align}
	\bDelta_{q_{\alpha \eta}} = \underbrace{\left(\hbq_{\alpha \eta} - \bq^\sharp_{\alpha \eta} \right)}_{\eqqcolon \bDelta^J_{q_{\alpha \eta}}} + 	\underbrace{\left(\bq^\sharp_{\alpha \eta} - \bq^*_{\alpha \eta}\right)}_{\eqqcolon \bDelta^\sharp_{q_{\alpha \eta}}}. 
	\label{eq:rec.dalpha.decomp}
\end{align}
We will control each individual term. 
Adapting the proof of Lemma~\ref{lemma:riesz} with $\bJ^*_{\eta}$ treated as a fixed matrix, we have on $\Gc_\noise \cap \Gc_{\PCR, \beta}$,
\begin{align}
	\| \bDelta^\sharp_{q_{\alpha \eta}} \|_2 \lesssim A_{\eta}   Q_{\alpha}, 
	\label{eq:rec.dalpha.sharp} 
\end{align}
where $Q_{\alpha}$ is defined as in \eqref{eq:delta.a}. 
We now turn to $\bDelta^J_{q_{\alpha \eta}}$. 
In this pursuit, observe 
\begin{align}
	\bDelta^J_{q_{\alpha \eta}} = \hbZ^\dagger \left( \hbJ_{\eta} - \bJ_{\eta}^* \right)^\top \bW^\top \hbbeta. 
\end{align}
On $\Gc_{\noise, z} \cap \Gc_{\PCR, \alpha}$, Lemma~\ref{lemma:pcr.pseudo} and \eqref{eq:rec.dJ.bound} yield
\begin{align}
	\| \bDelta^J_{q_{\alpha \eta}} \|_2 \le \| \hbZ^\dagger \|_\txtop  \cdot \| \hbJ_{\eta} - \bJ^*_{\eta} \|_\txtop \cdot \| \bW^\top \hbbeta \|_2
	\le A_{\eta}  \frac{\sigma \Lambda_{\alpha}}{(\lambda_z - \eta_z)  \sqrt{L}}  \| \bW^\top \hbbeta \|_2. 
\end{align} 
Next, we write 
\begin{align}
	\bW^\top \hbbeta &= \bbW^\top \bbeta^* + \bbW^\top \bDelta_{\beta} + \left(\bW - \bbW\right)^\top \bbeta^* + \left(\bW - \bbW\right)^\top \bDelta_{\beta}. 
\end{align}
On $\Gc_{\noise, w} \cap \Gc_{\PCR, \beta}$, we apply \eqref{eq:barw.bound.op}, \eqref{eq:barwbeta.bound}, and \eqref{eq:alpha.beta.bound} to obtain
\begin{align}
	\| \bW^\top \hbbeta \|_2 \lesssim \sqrt{L} + \sigma \sqrt{L} \Lambda_{\beta} + \eta_w \left( \frac{\sqrt{T_0}}{\lambda_y} + \frac{\sigma \Lambda_{\beta}}{\sqrt{N_1}} \right). 
\end{align}
Combining the above, we have on $\Gc_\noise \cap \Gc_\PCR$, 
\begin{align}
	\| \bDelta^J_{q_{\alpha \eta}} \|_2 \lesssim
	A_{\eta}  \frac{\sigma \Lambda_{\alpha}}{(\lambda_z - \eta_z) \sqrt{L}}  \left\{ \sqrt{L} + \sigma \sqrt{L} \Lambda_{\beta} + \eta_w \left( \frac{\sqrt{T_0}}{\lambda_y} + \frac{\sigma \Lambda_{\beta}}{\sqrt{N_1}} \right) \right\}. 
	\label{eq:rec.dalpha.J} 
\end{align}
Applying \eqref{eq:alpha.beta.bound}, \eqref{eq:rec.qa.bound}, \eqref{eq:rec.dalpha.decomp}, \eqref{eq:rec.dalpha.sharp}, and \eqref{eq:rec.dalpha.J}, we have on $\Gc_\noise \cap \Gc_\PCR$, 
\begin{align}
	&\left| \| \hbq_{\alpha \eta} \|_2^2  \left( 1 + \| \hbalpha \|_2^2 \right)  - \| \bq_{\alpha \eta}^* \|_2^2  \left(1 + \| \balpha^* \|_2^2 \right) \right|
	\\
	&\quad \lesssim 
	\tQ_{\alpha \eta}  \left( A_{\eta}  \frac{\sqrt{L}}{\lambda_z} + \tQ_{\alpha \eta} \right)
	 \left\{ 1 + \left(\frac{\sqrt{M}}{\lambda_z} + \frac{\sigma \Lambda_{\alpha}}{\sqrt{L}} \right)^2 \right\}
	+ A^2_{\eta}  \frac{\sigma \sqrt{L} \Lambda_{\alpha}}{\lambda_z^2}  \left( \frac{\sqrt{M}}{\lambda_z} + \frac{\sigma \Lambda_{\alpha}}{\sqrt{L}} \right),
	\label{eq:rec.var.est.3} 
\end{align}
where
\begin{align}
	\tQ_{\alpha \eta}
	\coloneqq
	A_{\eta}   \left\{  Q_{\alpha}
	+ 
	\frac{\sigma \Lambda_{\alpha}}{(\lambda_z - \eta_z) \sqrt{L}}  \left[ \sqrt{L} + \sigma \sqrt{L} \Lambda_{\beta}  + \eta_w \left( \frac{\sqrt{T_0}}{\lambda_y} + \frac{\sigma \Lambda_{\beta}}{\sqrt{N_1}} \right) \right] \right\}. 
\end{align}

\noindent {\em Putting everything together}: 
On $\Gc_\noise \cap \Gc_\PCR$, \eqref{eq:rec.var.est.1}, \eqref{eq:rec.var.est.2}, and \eqref{eq:rec.var.est.3} imply
\begin{align}
	\left| \hUpsilon_{\eta} - \Upsilon_{\eta} \right|
	&\lesssim 
	A^2_{\eta}  \frac{\sigma \Lambda_{\alpha}}{\sqrt{L}} \left(1 + \frac{\sigma \Lambda_{\alpha}}{\sqrt{L}} \right)
	 \left\{ \frac{\sqrt{M}}{\lambda_z} + \frac{\sigma \Lambda_{\alpha}}{\sqrt{L}} \left(1 + \frac{\sigma \Lambda_{\alpha}}{\sqrt{L}} \right) \right\}
	 \left( \frac{\sqrt{T_0}}{\lambda_y} + \frac{\sigma \Lambda_{\beta}}{\sqrt{N_1}} \right)^2
	\\
	&\quad 
	+ A^2_{\eta}  \frac{\sigma M \Lambda_{\beta}}{\lambda^2_z \sqrt{N_1}} 
	 \left( \frac{\sqrt{T_0}}{\lambda_y} + \frac{\sigma \Lambda_{\beta}}{\sqrt{N_1}} \right)
	+ A^2_{\eta}  \frac{\sigma \sqrt{L} \Lambda_{\alpha}}{\lambda_z^2} \left( \frac{\sqrt{M}}{\lambda_z} + \frac{\sigma \Lambda_{\alpha}}{\sqrt{L}} \right)
	\\
	&\quad
	+ A_{\eta} Q_{\beta} \left( \| \boldeta \|_1 \frac{\sqrt{N_1}}{\lambda_y} + A_{\eta} Q_{\beta} \right)
	\left\{1 + \left( \frac{\sqrt{T_0}}{\lambda_y} + \frac{\sigma \Lambda_{\beta}}{\sqrt{N_1}} \right)^2 \right\}
	+ \| \boldeta \|_1^2 \frac{\sigma \sqrt{N_1} \Lambda_{\beta}}{\lambda_y^2} \left( \frac{\sqrt{T_0}}{\lambda_y} + \frac{\sigma \Lambda_{\beta}}{\sqrt{N_1}} \right)
	\\
	&\quad 
	+ \tQ_{\alpha \eta} \left( A_{\eta} \frac{\sqrt{L}}{\lambda_z} + \tQ_{\alpha \eta} \right)
	 \left\{ 1 + \left(\frac{\sqrt{M}}{\lambda_z} + \frac{\sigma \Lambda_{\alpha}}{\sqrt{L}} \right)^2 \right\}. 
	\label{eq:rec.var.est.4}
\end{align}
Simplifying \eqref{eq:rec.var.est.4} via \eqref{eq:noise.eta} and \eqref{eq:lmda.lb} gives the desired inequality.  
Moreover, by \eqref{eq:rec.event.pcr} and \eqref{eq:rec.event.noise}, we take several union bounds to conclude $\Pb(\Gc_\noise^c \cup \Gc_\PCR^c \mid \Ec_{\eta}) \lesssim \rho$.
The desired result is obtained by following the arguments in the proof of Proposition~\ref{prop:var.asymp} verbatim with $(\hUpsilon_{\eta}, \Upsilon_{\eta}, \hVc^2_{\eta}, \Vc^2_{\eta})$ in place of $(\hUpsilon, \Upsilon, \hVc^2, \Vc^2)$ and leveraging Lemma~\ref{lemma:rec.normality}. 
\end{proof}

\end{document}